%% file: main.tex
\documentclass[aps,pra,twocolumn,reprint,groupedaddress,longbibliography,balancelastpage]{revtex4-2}
\usepackage{graphicx}
\usepackage{amsmath}
\usepackage{amssymb}
\usepackage{amsthm}
\usepackage{booktabs}
\usepackage{totcount}
\usepackage{xcolor}
\usepackage{tikz}
\usepackage{quantikz}
\definecolor{blueviolet}{rgb}{0.2, 0.2, 0.6}
\definecolor{webgreen}{rgb}{0,.5,0}
\definecolor{webbrown}{rgb}{.6,0,0}
\usepackage[pdftex,
	bookmarks=false,
	colorlinks=true,
	urlcolor=webbrown,
	linkcolor=blueviolet,
	citecolor=webgreen,
	pdfstartpage=1,
	pdfstartview={FitH},
	bookmarksopen=false
	]{hyperref}

\usepackage{eczoo}

\theoremstyle{plain}
\newtheorem{theorem}{Theorem}[section]
\newtheorem{lemma}[theorem]{Lemma}
\newtheorem{corollary}[theorem]{Corollary}
\newtheorem{proposition}[theorem]{Proposition}
\theoremstyle{remark}

\input{macros}

\makeatletter
\renewcommand{\paragraph}{%
  \@startsection{paragraph}{4}{\z@}%
    {3.25ex \@plus 1ex \@minus .2ex}%
    {-1em}%
    {\normalfont\normalsize\bfseries}%
}
\makeatother

\begin{document}

\title{Beyond transversality: structure of Clifford circuits for CSS codes}

\author{Victor V. Albert}
\affiliation{Joint Center for Quantum Information and Computer Science, NIST/University of Maryland, College Park, Maryland 20742, USA}

\date{\today}

\begin{abstract}
We characterize four groups of Clifford circuits for Calderbank--Shor--Steane (CSS) codes that are relevant to fault-tolerant logical operations.
First, we show that every code-preserving Clifford circuit is a product of Z-diagonal circuits, composed of S and CZ gates, and their X-basis analogues.
Second, we define the two-fold transversal group, generated by depth-one two-local code-preserving circuits, and show that each of its elements can be expressed as a product of layers consisting of either Z-diagonal, X-diagonal, or CNOT gates.
As a corollary, every transversal gate is a product of three transversal diagonal circuits; for connected non-self-dual codes, two such circuits suffice.

We further show that every code-preserving automorphism circuit, consisting of single-qubit Clifford gates and permutations, has a normal form comprising a Hadamard layer, a permutation, and two diagonal circuits.
We also define a two-fold automorphism group, in which a depth-one two-local circuit may be code-preserving up to a permutation, and show that its logical image can be larger than that of the two-fold transversal group.

For \Nallcodes\ CSS codes, plus 45 larger ones in machine-readable records, we provide explicit generators and determine the logical image of the two-fold transversal group.
We tabulate \Nfullcodes\ codes with at most 120 qubits whose full logical Clifford group is generated by two-fold-transversal circuits, including codes of distances \(3\), \(4\), \(5\), \(6\), \(8\), \(10\), \(12\), and \(13\), with respective best rates \(2/5\), \(3/4\), \(34/77\), \(17/28\), \(14/25\), \(1/27\), \(3/56\), and \(1/33\).
We construct three families of CSS codes from bipartite grids, cut-complements, and quadrics, and we census the self-dual codes invariant under rank-3 permutation groups.
Many of these codes realize the full logical Clifford group in this way.
More generally, the induced logical group can be large even when it is not full logical Clifford group: it has order at least 460\,800 for the gross code and roughly \(10^{26}\) for a clustered-cyclic code.

\end{abstract}

\maketitle

\input{fig_fourclasses}

\input{toc}

\input{sec_intro}

\input{sec_classes}

\input{sec_general}

\input{sec_transversal}

\input{sec_depthone}

\input{sec_depth1twolocal}

\input{sec_ldpcdepthone}

\input{sec_permutations}

\input{sec_extension}

\input{sec_conclusion}

\begin{acknowledgments}
	The author acknowledges inspiring and useful discussions with
	K.\@ Agarwal, A.\@ Barg, N.\@ Berthusen, D.\@ Elimelech, M.J.\@ Gullans, D.\@ Hangleiter, A.\@ Holmes, J.M.\@ Koh, J.C.\@ Magdalena de la Fuente, S.\@ Majidy, M.\@ McEwen, and K.\@ Sahay.
	This research was supported in part by grant NSF PHY-2309135 to the Kavli Institute for Theoretical Physics (KITP), where part of this work was completed during the program \textit{AI for Quantum Matter} (\url{https://www.kitp.ucsb.edu/activities/aiqmatter26}).
	This manuscript was prepared with the assistance of ChatGPT, developed by OpenAI, and Claude, developed by Anthropic. These tools were used to calculate code properties, help with proofs, improve clarity, and generate figures, in accordance with the author's instructions. All content, claims, and conclusions have been reviewed and verified by the author to ensure accuracy and originality.
	Certain products, commercial and otherwise, are mentioned in this publication. These mentions are for informational purposes only, and do not imply recommendation or endorsement by NIST.
\end{acknowledgments}

\appendix
\renewcommand{\thesubsection}{\thesection.\arabic{subsection}}
\makeatletter
\renewcommand{\p@subsection}{}
\makeatother

\input{app_codes}
\input{app_setup}
\input{app_general}
\input{app_depthone}
\input{app_transversal}
\input{app_permutations}
\input{app_twolocalaut}

\bibliography{refs}

\end{document}

%% file: macros.tex
\newcommand{\Ftwo}{\mathbb{F}_{2}}                
\newcommand{\tp}{{\mathsf{T}}}                    
\newcommand{\itp}{{-\mathsf{T}}}                  
\newcommand{\Id}{I}                               
\newcommand{\Jm}{J}                               
\newcommand{\ones}{\boldsymbol{1}}                    
\newcommand{\ind}[1]{\boldsymbol{1}_{#1}}             
\newcommand{\pw}{\odot}                           
\newcommand{\suppo}{\operatorname{supp}}          
\newcommand{\diago}{\operatorname{diag}}          
\newcommand{\spano}{\operatorname{span}}
\newcommand{\Homo}{\operatorname{Hom}}
\newcommand{\ranko}{\operatorname{rank}}

\newcommand{\Vsp}{\Ftwo^{2n}}                     
\newcommand{\Xh}{\mathsf{X}}                      
\newcommand{\Zh}{\mathsf{Z}}                      
\newcommand{\sfm}[2]{\langle #1,#2\rangle}        
\newcommand{\Sp}[1]{\mathrm{Sp}(#1)}              
\newcommand{\GLg}[1]{\mathrm{GL}(#1)}             
\newcommand{\Symm}[1]{\mathrm{Sym}_{#1}}           
\newcommand{\Perm}[1]{S_{#1}}                     
\newcommand{\Stb}{\mathrm{Stab}}                  
\newcommand{\Endl}{\mathrm{End}_{\lab}}           

\newcommand{\CX}{C_{X}}                           
\newcommand{\CZ}{C_{Z}}                           
\newcommand{\qcode}[1]{[\![#1]\!]}          
\newcommand{\qcodeb}[1]{\bigl[\!\bigl[#1\bigr]\!\bigr]}  
\newcommand{\lab}{\mathcal{C}}                    
\newcommand{\dst}{\mathcal{D}}                     
\newcommand{\logsp}{\mathcal{Q}}                  
\newcommand{\stabgp}{\mathcal{S}}                 
\newcommand{\Pau}[2]{P(#1|#2)}                    
\newcommand{\IX}{\mathcal{I}_{X}}                 
\newcommand{\IZ}{\mathcal{I}_{Z}}                 
\newcommand{\IL}{\mathcal{I}_{L}}                 

\newcommand{\Ush}[1]{U_{Z}(#1)}                    
\let\Lsh\relax                                    
\newcommand{\Lsh}[1]{U_{X}(#1)}                    
\newcommand{\Uupg}{U^{Z}}                         
\newcommand{\Udng}{U^{X}}                         
\newcommand{\Hd}[1]{H(#1)}                        
\newcommand{\sw}{\sigma}                          
\newcommand{\PX}{\Pi}                             
\newcommand{\pperm}[1]{P_{#1}}                    

\newcommand{\Nfull}{\boldsymbol{N}}                  
\newcommand{\Eup}{\boldsymbol{S}^{Z}}                
\newcommand{\Edn}{\boldsymbol{S}^{X}}                
\newcommand{\Ezx}{\boldsymbol{S}^{Z,X}}              
\newcommand{\Lev}{\boldsymbol{L}}                    
\newcommand{\Mt}{M}                               
\newcommand{\DM}{\Sp{2n}_{\Mt}}                   
\newcommand{\NM}{\boldsymbol{N}_{\Mt}}               
\newcommand{\EupM}{\boldsymbol{S}^{Z}_{\Mt}}           
\newcommand{\EdnM}{\boldsymbol{S}^{X}_{\Mt}}           
\newcommand{\LevM}{\Lev_{\Mt}}                    
\newcommand{\Ntr}{\boldsymbol{N}_{\emptyset}}        
\newcommand{\Ndep}{\boldsymbol{N}_{2\textnormal{fold}}}       
\newcommand{\Eupo}{\boldsymbol{S}^{Z}_{\emptyset}}     
\newcommand{\Edno}{\boldsymbol{S}^{X}_{\emptyset}}     
\newcommand{\Pup}{\mathcal{A}^{Z}}                
\newcommand{\Pdn}{\mathcal{A}^{X}}                
\newcommand{\Gaut}{\mathrm{Aut}(\lab)}                  
\newcommand{\Gaux}{\mathrm{Aut}_{0}(\lab)}              
\newcommand{\PDgp}{\boldsymbol{W}}                   
\newcommand{\rhom}{\rho}                          

\newcommand{\Urad}{\boldsymbol{U}}                   
\newcommand{\gau}{\Gamma}                         
\newcommand{\lgc}{\Lambda}                        
\newcommand{\radF}{\Phi}                          
\newcommand{\radT}{\Xi}                           
\newcommand{\lact}{\lambda}                       

\newcommand{\shp}[1]{\Theta^{Z}(#1)}              
\newcommand{\shm}[1]{\Theta^{X}(#1)}              
\newcommand{\Spz}{S^{Z}}                          
\newcommand{\Spx}{S^{X}}                          
\newcommand{\Spzx}{S^{Z,X}}                       
\newcommand{\Gterm}{G^{\star}}                    
\newcommand{\Lterm}{L^{\star}}                    
\newcommand{\Word}{\Omega}                        
\newcommand{\cell}{b}                             
\newcommand{\eps}{\varepsilon}                    

\newcounter{fullcode}
\regtotcounter{fullcode}
\newcounter{allcodes}
\regtotcounter{allcodes}
\newcommand{\zoolink}[2]{%
  \href{https://errorcorrectionzoo.org/c/#1}%
       {\textcolor{black}{\texttt{#2}}\,{\EczooIconSetScale{1.15}\eczooUseIcon}}}
\newcommand{\fcode}[3][]{\stepcounter{fullcode}\stepcounter{allcodes}%
  $#2$\,\if\relax\detokenize{#1}\relax\texttt{#3}\else\zoolink{#1}{#3}\fi}
\newcommand{\Nfullcodes}{\total{fullcode}}
\newcounter{ldpccode}
\regtotcounter{ldpccode}
\newcommand{\lcode}[3][]{\stepcounter{ldpccode}\stepcounter{allcodes}%
  $#2$\,\if\relax\detokenize{#1}\relax\texttt{#3}\else\zoolink{#1}{#3}\fi}
\newcommand{\Nldpccodes}{\total{ldpccode}}
\newcommand{\Nallcodes}{\total{allcodes}}

%% file: fig_fourclasses.tex
\definecolor{clSp}{HTML}{0072B2}   
\definecolor{clSm}{HTML}{D55E00}   
\definecolor{clLv}{HTML}{009E73}   
\definecolor{clWy}{HTML}{7B52C7}   

\newcommand{\gtband}[5]{%
  \noindent
  \begin{minipage}[t]{0.052\textwidth}\raggedright\scriptsize\textsc{#1}\end{minipage}%
  \hfill\begin{minipage}[t]{0.220\textwidth}\centering #2\end{minipage}%
  \hfill\begin{minipage}[t]{0.220\textwidth}\centering #3\end{minipage}%
  \hfill\begin{minipage}[t]{0.220\textwidth}\centering #4\end{minipage}%
  \hfill\begin{minipage}[t]{0.220\textwidth}\centering #5\end{minipage}%
  \par}
\newcommand{\gtbandc}[5]{%
  \noindent
  \begin{minipage}[c]{0.052\textwidth}\raggedright\scriptsize\textsc{#1}\end{minipage}%
  \hfill\begin{minipage}[c]{0.220\textwidth}\centering #2\end{minipage}%
  \hfill\begin{minipage}[c]{0.220\textwidth}\centering #3\end{minipage}%
  \hfill\begin{minipage}[c]{0.220\textwidth}\centering #4\end{minipage}%
  \hfill\begin{minipage}[c]{0.220\textwidth}\centering #5\end{minipage}%
  \par}
\newcommand{\gtrule}{\vspace{1pt}{\color{black!55}\rule{\linewidth}{0.4pt}}\vspace{1pt}\par}

\begin{figure*}[!t]
\centering

\setlength{\fboxsep}{2.5pt}%

\vspace{5pt}

\gtband{}{%
  \color{clSp}\textbf{(a)}\ $Z$-diagonal ($\Eup$)}{%
  \color{clSm}\textbf{(b)}\ $X$-diagonal ($\Edn$)}{%
  \color{clLv}\textbf{(c)}\ CNOT ($\Lev$)}{%
  \color{clWy}\textbf{(d)}\ partial duality ($\PDgp$)}
\gtrule
\gtband{$\Sp{2n}$}{%
  $\left(\begin{smallmatrix}\Id&{ S}\\[1.5pt]0&\Id\end{smallmatrix}\right)$,\ $S\in\Symm{n}$\\[1pt]{\scriptsize\color{clSp}block $B$}}{%
  $\left(\begin{smallmatrix}\Id&0\\[1.5pt]{ T}&\Id\end{smallmatrix}\right)$,\ $T\in\Symm{n}$\\[1pt]{\scriptsize\color{clSm}block $C$}}{%
  $\left(\begin{smallmatrix}{K}&0\\[1.5pt]0&{K^{\itp}}\end{smallmatrix}\right)$,~{\scriptsize$K{=}\displaystyle\prod_{m}(\Id{+}E_{c_{m}t_{m}})$}\\[1pt]{\scriptsize\color{clLv}blocks $A,D$}}{%
  $\Hd{a}\,\pperm{\pi}$,\ $a\in\Ftwo^{n}$, $\pi\in\Perm{n}$\\[1pt]{\scriptsize\color{clWy}$X\!\leftrightarrow\!Z$ duality}}
\gtrule
\gtband{condition}{%
  {$\CX S\subseteq\CZ$}}{%
  {$\CZ T\subseteq\CX$}}{%
  {$\CX K=\CX$, $\CZ K^{\itp}=\CZ$}}{%
  {$\CX \overset{\pi}{\longleftrightarrow} \CZ$ ($a=\ones$ case)}}
\gtrule
\gtband{gates}{%
  \colorbox{clSp!12}{\makebox[0.86\linewidth]{$\sqrt{Z}$\ and\ $\mathrm{CZ}$}}}{%
  \colorbox{clSm!12}{\makebox[0.86\linewidth]{$\sqrt{X}$,\ $\mathrm{CZ}^{H}$ ($XX$-type)}}}{%
  \colorbox{clLv!12}{\makebox[0.86\linewidth]{\phantom{$C^{H}$}$\mathrm{CNOT}$\phantom{$C^{H}$}}}}{%
  \colorbox{clWy!12}{\makebox[0.86\linewidth]{\phantom{$C^{H}$}Hadamard\ and\ swap\phantom{$C^{H}$}}}}
\gtrule
\gtband{circuit}{%
  {\small$\displaystyle\prod_{i<j\,:\,S_{ij}=1}\mathrm{CZ}_{ij}\;\prod_{i\,:\,S_{ii}=1}\sqrt{Z}_{i}$}}{%
  {\small$\displaystyle\prod_{i<j\,:\,T_{ij}=1}\mathrm{CZ}^{H}_{ij}\;\prod_{i\,:\,T_{ii}=1}\sqrt{X}_{i}$}}{%
  {\small$\displaystyle\prod_{m}\mathrm{CNOT}_{c_{m}\to t_{m}}$}}{%
  {\small$\displaystyle\pperm{\pi}\prod_{i\,:\,a_{i}=1}H_{i}$}}
\vspace{1pt}{\color{black!75}\rule{\linewidth}{0.9pt}}\vspace{2pt}\par

\gtbandc{example}{%
  \begin{quantikz}[ampersand replacement=\&, row sep={15pt,between origins}, column sep=7pt]
    \lstick{$0$}\&\qw\&\qw\&\qw\&\qw\\
    \lstick{$1$}\&\qw\&\qw\&\qw\&\qw\\
    \lstick{$2$}\&\gate{\sqrt{Z}}\&\qw\&\qw\&\qw\\
    \lstick{$3$}\&\gate{\sqrt{Z}}\&\ctrl{1}\&\ctrl{2}\&\qw\\
    \lstick{$4$}\&\qw\&\control{}\&\qw\&\qw\\
    \lstick{$5$}\&\qw\&\qw\&\control{}\&\qw
  \end{quantikz}}{%
  \begin{quantikz}[ampersand replacement=\&, row sep={15pt,between origins}, column sep=4pt]
    \lstick{$0$}\&\gate{H}\&\qw\&\qw\&\qw\&\gate{H}\&\qw\\
    \lstick{$1$}\&\gate{H}\&\qw\&\qw\&\qw\&\gate{H}\&\qw\\
    \lstick{$2$}\&\gate{H}\&\gate{\sqrt{Z}}\&\qw\&\qw\&\gate{H}\&\qw\\
    \lstick{$3$}\&\gate{H}\&\gate{\sqrt{Z}}\&\ctrl{1}\&\ctrl{2}\&\gate{H}\&\qw\\
    \lstick{$4$}\&\gate{H}\&\qw\&\control{}\&\qw\&\gate{H}\&\qw\\
    \lstick{$5$}\&\gate{H}\&\qw\&\qw\&\control{}\&\gate{H}\&\qw
  \end{quantikz}}{%
  \begin{quantikz}[ampersand replacement=\&, row sep={15pt,between origins}, column sep=7pt]
    \lstick{$0$}\&\ctrl{1}\&\targ{}\&\qw\\
    \lstick{$2$}\&\targ{}\&\ctrl{-1}\&\qw\\
    \lstick{$1$}\&\targ{}\&\ctrl{1}\&\qw\\
    \lstick{$4$}\&\ctrl{-1}\&\targ{}\&\qw\\
    \lstick{$3$}\&\ctrl{1}\&\targ{}\&\qw\\
    \lstick{$5$}\&\targ{}\&\ctrl{-1}\&\qw
  \end{quantikz}}{%
  \begin{quantikz}[ampersand replacement=\&, row sep={15pt,between origins}, column sep=8pt]
    \lstick{$0$}\&\gate{H}\&\swap{1}\&\qw\\
    \lstick{$1$}\&\gate{H}\&\targX{}\&\qw\\
    \lstick{$2$}\&\gate{H}\&\swap{1}\&\qw\\
    \lstick{$3$}\&\gate{H}\&\targX{}\&\qw\\
    \lstick{$4$}\&\gate{H}\&\swap{1}\&\qw\\
    \lstick{$5$}\&\gate{H}\&\targX{}\&\qw
  \end{quantikz}}
\gtrule

\caption{\label{fig:four-classes}\textbf{Four Clifford gate families of a CSS
code.} 
The first five rows list the name, symplectic embedding, code-preserving
condition, physical gates, and circuit realization
[the right-hand sides of Eqs.~\eqref{eq:classes-circ-sp},
\eqref{eq:classes-circ-sm}, \eqref{eq:classes-circ-levi},
and~\eqref{eq:classes-circ-weyl}] of each family.
The control--target sequence $c_{m}\to t_{m}$ of a CNOT network is read off a
Gaussian elimination of $K$, one CNOT per row addition, giving the
factorization of $K$ shown in the embedding row.
The bottom row depicts circuits that realize the families on the doubly-even self-dual $\qcode{6,2,2}$ code
$\CX{=}\CZ{=}\spano\{111100,110011\}$, with symplectic logical
basis $L_X{=}L_Z{=}\{100101,010101\}$.
\textbf{(a)} $\sqrt Z_2\,\sqrt Z_3\,\mathrm{CZ}_{34}\,\mathrm{CZ}_{35}$ acts as the
$Z$-diagonal circuit $\bar S_1\bar S_2\,\overline{\mathrm{CZ}}_{12}$;
it is depth two since two $\mathrm{CZ}$'s meet at qubit~$3$.
\textbf{(b)} is the $H^{\otimes6}$-conjugate of (a) and acts as an $X$-diagonal circuit.
The circuit inside the Hadamard layers is the same as that of (a) because the example is self-dual, 
but can be different in the general case.
\textbf{(c)} is a $\mathrm{CNOT}$ network on the matching
$\{\{0,2\},\{1,4\},\{3,5\}\}$.
\textbf{(d)} is $H^{\otimes6}$ followed by the swap $\pperm{\tau}$,
$\tau{=}(0\,1)(2\,3)(4\,5)$, acting as a fold-transversal $H_\tau$.
The condition listed for this family in row three is only for the special case of
such gates, which contain a Hadamard on all qubits.
}
\end{figure*}

%% file: toc.tex
\onecolumngrid
\begingroup
\setlength{\parindent}{0pt}

\newlength{\tocnumwd}\settowidth{\tocnumwd}{\textbf{VIII.}}
\newlength{\tocindent}\setlength{\tocindent}{\dimexpr\tocnumwd+0.6em\relax}

\newcommand{\tocsec}[2]{%
  \makebox[\tocnumwd][r]{\textbf{\ref{#1}.}}\hspace{0.6em}%
  \textbf{\hyperref[#1]{#2}}\dotfill\textbf{\pageref{#1}}\\[0.7ex]}
\newcommand{\tocsub}[2]{\hyperref[#1]{#2}}
\newcommand{\tocsubs}[1]{%
  \hspace*{\tocindent}%
  \begin{minipage}{\dimexpr\linewidth-\tocindent\relax}\small #1\end{minipage}\\[1.5ex]}
\newcommand{\tocdot}{\,\,$\bullet$\,\,}

\rule{\linewidth}{0.8pt}\\[0.4ex]
{\large\textbf{Contents}}\\[1.0ex]

\tocsec{sec:intro}{Introduction \& summary}
\tocsubs{\tocsub{sec:intro-summary}{Summary of results}}
\tocsec{sec:classes}{Four families of code-preserving circuits}
\tocsubs{\tocsub{sec:classes-sp}{$Z$-diagonal circuits}\tocdot
         \tocsub{sec:classes-sm}{$X$-diagonal circuits}\tocdot
         \tocsub{sec:classes-levi}{CNOT circuits}\tocdot
         \tocsub{sec:classes-weyl}{Partial dualities}}
\tocsec{sec:general}{Code-preserving circuit group}
\tocsubs{\tocsub{sec:general-structure}{Structure of $\Nfull$}\tocdot
         \tocsub{sec:general-shear}{Generation by diagonal circuits}}
\tocsec{sec:transversal}{Transversal gate group}
\tocsubs{\tocsub{sec:transversal-nf}{Three-layer decomposition}\tocdot
         \tocsub{sec:transversal-algo}{The size of the transversal group}}
\tocsec{sec:depthone}{Two-fold transversal group}
\tocsubs{\tocsub{sec:depthone-fixed}{Depth-one circuits at a fixed qubit matching}\tocdot
         \tocsub{sec:depthone-union}{Definition and structure of the two-fold transversal group}}
\tocsec{sec:depth1twolocal}{Codes with full-Clifford two-fold transversal gates}
\tocsubs{\tocsub{sec:depth1twolocal-census}{Select full codes}\tocdot
         \tocsub{sec:depth1twolocal-families}{Families of full codes}}
\tocsec{sec:ldpcdepthone}{Other codes with two-fold transversal gates}
\tocsubs{\tocsub{sec:ldpcdepthone-gross}{The gross code}\tocdot
         \tocsub{sec:ldpcdepthone-addressable}{Codes with addressable gates}}
\tocsec{sec:permutations}{Including permutations}
\tocsubs{\tocsub{sec:permutations-nf}{Normal form and how to calculate it}}
\tocsec{sec:extensions}{Possible extensions}
\tocsubs{\tocsub{sec:ext-depth}{Optimizing depth}\tocdot
         \tocsub{sec:ext-homology}{Homology}\tocdot
         \tocsub{sec:ext-noncss}{Extension to non-CSS codes and other alphabets}\tocdot
         \tocsub{sec:ext-perm}{Two-fold automorphism group}}
\tocsec{sec:conclusion}{Conclusion}
\rule{\linewidth}{0.8pt}
\endgroup
\newpage
\twocolumngrid

%% file: sec_intro.tex
\section{Introduction \& summary}
\label{sec:intro}

Large-scale industrial efforts have begun to scale up quantum computers to hundreds of physical qubits and are developing full-stack software for fault-tolerant computation.
Calderbank--Shor--Steane (CSS)~\cite{calderbank1996good,steane1996error} codes are the first choice for many such efforts (Refs.~\cite{bluvstein2024logical,paetznick2024demonstration,bravyi2024high,low2026denser,cain2026shor}, to name a few) because they typically come with simple and naturally fault-tolerant syndrome extraction gadgets and logical gates.

Typical fault-tolerant logical operations for CSS codes can be implemented using circuits whose gate ingredients come from the Clifford group, and much work has been dedicated to finding such circuits for CSS codes of interest.
These circuits can be classified by their depth (number of layers) and the locality of their ingredients (one-local ``transversal'' gates, two-local entangling gates, etc.).

The ``most'' fault-tolerant is the transversal gate set, whose depth-one one-local elements are implemented by applying single-qubit gates to each physical qubit~\cite{gottesman1997stabilizer,zeng2011transversality,anderson2016classification,dasu2025classification,tansuwannont2025clifford,jain2025transversal,barg2026geometric,campsmoreno2026transversal,haruna2026homological}.
This set is also, unfortunately, the most restrictive in terms of the number of logical operations it can realize~\cite{eastin2009restrictions,zeng2011transversality}.

Increasing locality yields depth-one two-local (a.k.a.~``two-fold transversal'') circuits consisting of two-qubit entangling gates defined on a qubit matching~\cite{moussa2016transversal,quintavalle2023partitioning,breuckmann2024fold,eberhardt2024logical,li2025poincare,berthusen2025concatenated,cao2026lego,zheng2026canonical,hong2026ldpc}.
While more powerful than transversal gates, two-fold transversal circuits on a fixed matching remain limited in their set of possible logical operations~\cite{chakraborty2026nogo}.

Increasing depth yields two-local circuits whose layers act on \textit{different} qubit matchings (read: partitions into blocks of size 1 or 2), thereby spreading entanglement and allowing for richer sets of gates that can realize the full logical Clifford group~\cite{tiew2025dehn,zhu2025nonclifford,breuckmann2026cups,gulshen2025symmetric,holmes2026quantum,benhemou2026automated}.
Qubit permutation circuits, whose initial locality is determined by the qubits each permutation moves but which can always be compiled to depth-two SWAP networks~\cite{alon1994routing}, can be added to the mix for architectures where they are cheap to implement~\cite{grassl2013leveraging,sayginel2024autqec,sayginel2025faulttolerant,berthusen2025automorphism,koh2026phantom,chen2026transversal}.

Despite all this progress on specific circuits for specific code families, studies on the structure of the \textit{full} group of code-preserving Clifford circuits, and of how the specific circuits \textit{relate} to this group, have been relatively sparse \cite{dehaene2003clifford,rengaswamy2018synthesis,rengaswamy2020logical,liu2026systematic}.

\subsection{Summary of results}
\label{sec:intro-summary}

We determine a set of minimal ingredients, or \textit{generators}, sufficient to construct code-preserving Clifford circuits with various notions of locality.
This allows us to perform numerical searches that identify logical gates implementable using such circuits for various codes.

\paragraph*{Generators of circuit groups.}
We study four groups of circuits: the group of all code-preserving circuits, the transversal group, the depth-one two-local group on a fixed matching, and a new group we call the \textit{two-fold transversal group}, defined by stacking one or more depth-one two-local layers over arbitrary matchings [Figure~\ref{fig:layers}(c)].
This last group is relevant to fault-tolerant computation because it can spread correlations and compile arbitrary logical Clifford operations~\cite{holmes2026quantum,berthusen2025concatenated}, yet its depth-one two-local circuits can be interleaved with error correction so that errors are contained.

Our main technical result is a presentation of each group in terms of a fixed set of generator circuits.
Each generator circuit comes from one of four families defined by their ingredients: \(Z\)-basis diagonal gates ($\sqrt{Z}$ and CZ), \(X\)-basis diagonal gates, CNOT/swap networks, or ``partial dualities'' consisting of Hadamard blocks and permutations (see Fig.~\ref{fig:four-classes}).
Code-preserving circuits in the two diagonal families are determined by linear conditions~\cite{webster2023transversal,campsmoreno2026transversal,koh2026phantom,bauer2026finding}, forming a vector space that can be found by Gaussian elimination.

Using established results in representation theory and geometry of binary symplectic space~\cite{artin1957geometric,carter1972simple,taylor1992geometry,grove2002classical,malle2011linear}, we show that the two diagonal families suffice to generate the entire group of code-preserving Clifford gates.
In other words, every code-preserving Clifford circuit can be decomposed into a product of $Z$-diagonal and $X$-diagonal circuits.
And since linear conditions entirely determine both diagonal families, the entire group is thus efficiently determinable.

The same result holds for the transversal gate subgroup [Figure~\ref{fig:layers}(a)], which consequently is generated by $Z$-diagonal and $X$-diagonal circuits that are transversal (i.e., consisting of phase gates and dual-phase gates, respectively).
We provide a normal form for transversal gates in terms of such diagonal circuits and a second unique normal form that determines the group's size.
The general efficiency condition reduces to finding two binary linear codes corresponding to the two transversal diagonal sets.

Extending to two-local circuits, we show that \textit{every} code-preserving depth-one two-local layer on a fixed matching of the physical qubits [Figure~\ref{fig:layers}(b)] is generated by diagonal circuits and CNOT layers supported on that matching. The underlying identity in the binary symplectic group (Theorem~\ref{thm:depthone}) has now been independently formalized in Lean~\cite{axiomatic2026cliffordcss}.
The two-fold transversal group [Figure~\ref{fig:layers}(c)] is then generated by the union of these generator families over all matchings.
This allows us to obtain lower bounds on the number of distinct two-fold transversal logical gates by sampling from the spaces of depth-one two-local diagonal circuits and CNOT matchings.

\input{fig_layers}

\paragraph*{Numerical experiments.}
We perform numerical experiments on multiple code families and obtain \Nfullcodes\ codes with at most 120 qubits whose \textit{full} logical Clifford group can be realized by two-fold transversal circuits, plus 45 further full codes with up to 276 qubits in our machine-readable records. Many of these ``full'' codes admit low-weight and low-degree generators (see Table~\ref{tab:full-depth1}).
Notable entries include the $\qcode{18,4,4}$ color code on a torus~\cite{bombin2006topological,bombin2007exact}, the $\qcode{16,6,4}$ tesseract code~\cite{tansuwannont2026full}, the $\qcode{14,3,3}$ phantom/constant-excitation code~\cite{koh2026phantom,lai2025faulttolerant}, the $\qcode{12,2,4}$ carbon code~\cite{paetznick2024demonstration}, and the $\qcode{10,2,3}$ rotated toric code~\cite{albert2026handbook}.
The list also includes two two-block group-algebra codes~\cite{lin2023twoblock}, one directional toric code~\cite{gu2026nearest}, and one copy-cup code~\cite{tiew2026copycup}.

Some full codes also \emph{address} their logical gates.
A single depth-one two-local layer acts as one canonical logical generator, and as the identity on every other logical qubit.
No compilation of several layers is necessary.
The smallest complete example is the $\qcode{10,2,3}$ rotated toric code \texttt{eczoo:dd1f5a7f}.
One layer realizes each of its six canonical generators: both $\overline{\sqrt{Z}}_{i}$, both $\overline{\sqrt{X}}_{i}$, the $\overline{\mathrm{CZ}}$, and its $X$-basis mirror.
The $\qcode{15,4,3}$ code \texttt{cons:blk45k12} addresses every $X$-diagonal gate.
One layer realizes each of its ten $X$-diagonal generators: all four $\overline{\sqrt{X}}_{i}$ and all six $\overline{\mathrm{CZ}^{H}}_{ij}$ (Section~\ref{sec:depth1twolocal-census}).

We cover seven infinite code families that contain many full-Clifford codes, four of which were introduced previously: the ``middle'' Reed--Muller~\cite{tansuwannont2026full}, quantum logic~\cite{holmes2026quantum}, subsystem hypergraph product simplex~\cite{malcolm2025computing}, and concatenated symplectic double~\cite{berthusen2025concatenated} codes.
The three other families are constructed from bipartite grids, cut-complements, and quadrics.
We also census the self-dual codes that a rank-3 permutation group leaves invariant.
This class contains the middle Reed--Muller, cut-complement, and quadric families, and it yields \(69\) further full codes.

We determine a set of two-fold transversal generators for all full-Clifford codes and another set of \Nldpccodes\ codes, mostly with low-density parity-check properties (see Table~\ref{tab:ldpc-depth1}).
We lower-bound the logical image of each such code.
The bounds span $41$ decimal orders of magnitude, from $48$ to about $6\times10^{42}$ gates.
The gross code~\cite{bravyi2024high} reaches at least $460\,800$ logical gates.
A clustered-cyclic $\qcode{24,8,3}$ code~\cite{gu2026qgpu} reaches at least $1.5\times10^{26}$.

\paragraph*{Adding permutations.}
The fourth, duality-based, gate family [Fig.~\ref{fig:four-classes}(d)] becomes relevant when we add permutations to our gate groups.
We show that the code's automorphism group, formed by combinations of single-qubit Clifford gates and qubit permutations, has a normal form consisting of one partial duality and two transversal \(Z,X\)-diagonal circuits.
Once a permutation is fixed, its transversal dressing is, once again, determined by a linear condition.

Permutations can also be adjoined to depth-one two-local layers, yielding what we call the \emph{two-fold automorphism group}.
Its generators need not preserve the code individually: a depth-one two-local circuit that is not code-preserving on its own can become so once a qubit permutation compensates it.
We show that this group is strictly larger than the group generated by the automorphism and two-fold transversal groups put together, yielding more logical gates for architectures where permutations are readily available (see Sec.~\ref{sec:ext-perm}).
We leave further study of this group to future work.

\paragraph*{Outline.}
Section~\ref{sec:classes} introduces the four gate families and their
code-preserving conditions. Section~\ref{sec:general} describes the full group
and proves its generation by diagonal circuits. Sections~\ref{sec:transversal}
and~\ref{sec:depthone} treat transversal and depth-one two-local circuits,
respectively. Section~\ref{sec:depth1twolocal} identifies ``full'' codes, whose
full logical Clifford group is realized by depth-one two-local circuits, and
Section~\ref{sec:ldpcdepthone} lower-bounds the logical image of the depth-one
group for instances of various quantum low-density parity-check (QLDPC) families. Section~\ref{sec:permutations}
incorporates qubit permutations and derives the automorphism normal form.
Section~\ref{sec:extensions} lists possible extensions, and
Section~\ref{sec:conclusion} concludes.

%% file: fig_layers.tex
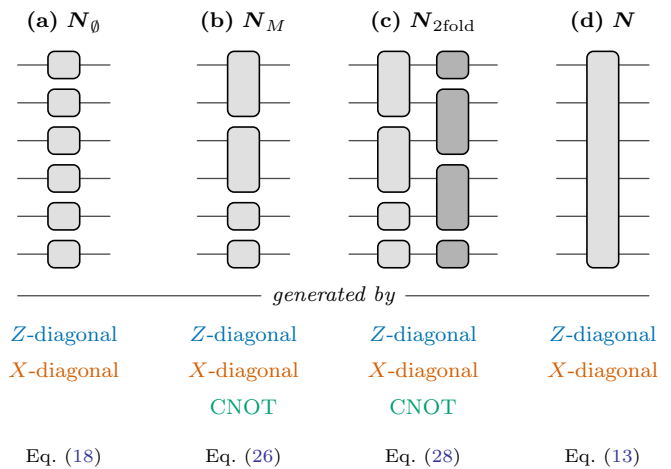
\begin{figure}[t]
\centering
\providecolor{clSp}{HTML}{0072B2}   
\providecolor{clSm}{HTML}{D55E00}   
\providecolor{clLv}{HTML}{009E73}   
\begin{tikzpicture}[
    line cap=round,
    wire/.style={line width=0.5pt, black!70},
    brick/.style={draw=black, line width=0.6pt, rounded corners=2.5pt,
                  fill=black!12},
    brickalt/.style={draw=black, line width=0.6pt, rounded corners=2.5pt,
                     fill=black!30},
  ]
  \def\dy{0.50}      
  \def\gw{0.42}      
  \def\gh{0.36}      
  \def\Wtot{8.35}    
  \def\ybar{-3.05}   
  \def\yZrow{-3.65}  
  \def\yXrow{-4.13}  
  \def\yCrow{-4.61}  
  \def\yref{-5.25}   
  \def\bone#1#2#3{\draw[#1] ({#2-\gw/2},{-#3*\dy-\gh/2})
                  rectangle ({#2+\gw/2},{-#3*\dy+\gh/2});}
  \def\btwo#1#2#3{\draw[#1] ({#2-\gw/2},{-(#3+1)*\dy-\gh/2})
                  rectangle ({#2+\gw/2},{-#3*\dy+\gh/2});}
  \pgfmathsetmacro\cA{0.61}
  \pgfmathsetmacro\cD{\Wtot-0.61}
  \pgfmathsetmacro\cB{\cA+(\cD-\cA)/3}
  \pgfmathsetmacro\cC{\cA+2*(\cD-\cA)/3}
  \foreach \Xc/\Wh in {\cA/0.61, \cB/0.61, \cC/0.98, \cD/0.61}{
    \foreach \i in {0,...,5}{
      \draw[wire] ({\Xc-\Wh},{-\i*\dy}) -- ({\Xc+\Wh},{-\i*\dy});
    }
  }
  \node[font=\footnotesize, anchor=base] at (\cA,0.50) {\textbf{(a)} $\Ntr$};
  \node[font=\footnotesize, anchor=base] at (\cB,0.50) {\textbf{(b)} $\NM$};
  \node[font=\footnotesize, anchor=base] at (\cC,0.50) {\textbf{(c)} $\Ndep$};
  \node[font=\footnotesize, anchor=base] at (\cD,0.50) {\textbf{(d)} $\Nfull$};
  \foreach \i in {0,...,5}{\bone{brick}{\cA}{\i}}
  \btwo{brick}{\cB}{0}
  \btwo{brick}{\cB}{2}
  \bone{brick}{\cB}{4}
  \bone{brick}{\cB}{5}
  \btwo{brick}{\cC-0.39}{0}
  \btwo{brick}{\cC-0.39}{2}
  \bone{brick}{\cC-0.39}{4}
  \bone{brick}{\cC-0.39}{5}
  \bone{brickalt}{\cC+0.39}{0}
  \btwo{brickalt}{\cC+0.39}{1}
  \btwo{brickalt}{\cC+0.39}{3}
  \bone{brickalt}{\cC+0.39}{5}
  \draw[brick] ({\cD-\gw/2},{-5*\dy-\gh/2}) rectangle ({\cD+\gw/2},{\gh/2});
  \node[font=\footnotesize\itshape, inner sep=4pt] (genby)
    at ({0.5*\Wtot},\ybar) {generated by};
  \draw[wire] (0,\ybar) -- (genby.west);
  \draw[wire] (genby.east) -- (\Wtot,\ybar);
  \foreach \Xc in {\cA,\cB,\cC,\cD}{
    \node[font=\footnotesize, text=clSp, anchor=base] at (\Xc,\yZrow)
      {$Z$-diagonal};
    \node[font=\footnotesize, text=clSm, anchor=base] at (\Xc,\yXrow)
      {$X$-diagonal};
  }
  \foreach \Xc in {\cB,\cC}{
    \node[font=\footnotesize, text=clLv, anchor=base] at (\Xc,\yCrow) {CNOT};
  }
  \node[font=\scriptsize, anchor=base] at (\cA,\yref) {Eq.~\eqref{eq:T-gen}};
  \node[font=\scriptsize, anchor=base] at (\cB,\yref) {Eq.~\eqref{eq:dep-gen}};
  \node[font=\scriptsize, anchor=base] at (\cC,\yref) {Eq.~\eqref{eq:dep-genunion}};
  \node[font=\scriptsize, anchor=base] at (\cD,\yref) {Eq.~\eqref{eq:N-shear}};
\end{tikzpicture}
\caption{\label{fig:layers}\textbf{Gate groups and their generating sets.}
\textbf{(a)} A transversal circuit, consisting of a single-qubit gate on each of the six qubits in this sketch. Code-preserving circuits of this shape form the transversal
group $\Ntr$.
\textbf{(b)} A depth-one two-local layer on a fixed matching $\Mt$ of
the qubits; each brick is a generic Clifford gate on a pair of $\Mt$ or on an
unmatched qubit.
The code-preserving layers on
$\Mt$ form the group $\NM$.
\textbf{(c)} A stack of two depth-one layers on different matchings. Stacks of code-preserving layers over
arbitrary matchings generate the \textit{two-fold transversal group} $\Ndep$.
\textbf{(d)} A depth-one six-local layer. Every circuit can be compressed into this shape, and the code-preserving
ones form the group $\Nfull$ of all code-preserving circuits.
Below the divider, each panel lists the gate families that generate its group
--- $Z$-diagonal and $X$-diagonal circuits of the panel's shape, joined in the
two-local cases (b) and (c) by the CNOT circuits --- with the equation stating
each generation result.}
\end{figure}

%% file: sec_classes.tex
\section{Four families of code-preserving circuits}
\label{sec:classes}

We find it convenient to organize circuits into four gate families, which we introduce in this section along with basic data about CSS codes.
The families are distinguished by their \emph{ingredients}, meaning the elementary gates a circuit is built from.
In the main text, we refer to these families by the circuits that they contain --- $Z,X$-diagonal circuits, CNOT networks, and Hadamard-and-permutation circuits (partial dualities) --- but use letters for their corresponding groups in a way that pays tribute to early work on symplectic-group theory by $\boldsymbol{S}$iegel, $\Lev$evi, and $\PDgp$eyl. The appendices use those group-theoretic names throughout.

Two classical binary codes $\CX,\CZ\subseteq\Ftwo^{n}$ that are orthogonal to one another, $\CX\perp\CZ$, define a CSS code on $n$ qubits. 
$Z$-type Pauli operators drawn from $\CZ$ and $X$-type operators drawn from $\CX$ generate its stabilizer group. The two codes are orthogonal, which implies that all stabilizers commute.

In the symplectic formalism (i.e., modulo signs), a
Pauli operator is represented by a symplectic ``label''. The label is a pair $(x|z)$ of binary strings that record where the
operator acts by $X$ and by $Z$.
The stabilizers then form the ``label space'' of the code,
\begin{equation}
  \lab=(\CX|0)\oplus(0|\CZ),\qquad\text{(label space)}
  \label{eq:classes-lab}
\end{equation}
a subspace of the $2n$-dimensional phase space of the $n$ qubits. 
Let $r_X$ and $r_Z$ be the dimensions of the two classical codes, and let $r=r_X+r_Z$.
The code encodes $k=n-r$ logical qubits.

We study Clifford gates that preserve such a code, modulo Pauli operators and global phases. Such a gate reduces to an element of the symplectic group
$\Sp{2n}$~\cite{gottesman1997stabilizer}. That element acts on a label $(x|z)$ by a $2n\times 2n$ binary matrix $g$ that preserves the symplectic inner product.
We thus represent a circuit by the map it induces on symplectic labels and study maps that preserve \(\lab\).

First, we show that every \(\lab\)-preserving symplectic map can be realized as a code-preserving, or logical, physical Clifford circuit via a suitable Pauli correction (Theorem~\ref{thm:dressing}).
We leave the compensating Pauli implicit and work strictly in the symplectic picture from now on.
Note that this approach collapses seemingly different gates to the same class.
For example, \(\sqrt{Z}\otimes\sqrt{Z}^{\dagger}\) is equivalent to \(\sqrt{Z}\otimes\sqrt{Z}\),
up to the Pauli \(I\otimes Z\).

Because a symplectic label map respects the division of phase space into $X$- and
$Z$-parts, we write it in a block form
$g=\bigl(\begin{smallmatrix}A&B\\ C&D\end{smallmatrix}\bigr)$, acting on a label by
\begin{equation}
  (x|z)\,\left(\begin{matrix}A&B\\ C&D\end{matrix}\right)=(xA+zC\,\vert\,xB+zD).
  \label{eq:classes-block}
\end{equation}
A CSS code respects the same $X$/$Z$ split. This split yields simple code-preservation conditions on the four blocks,
\begin{equation}
  \begin{gathered}
  \CX A\subseteq\CX,\qquad \CX B\subseteq\CZ,\\
  \CZ C\subseteq\CX,\qquad \CZ D\subseteq\CZ ,
  \end{gathered}
  \label{eq:classes-gen}
\end{equation}
each enforcing that a block carry the relevant classical code into its proper target.
The block $B$ that turns $X$-labels into $Z$-labels, for instance, must send $\CX$
into $\CZ$.
The same four conditions for projective preservation of the CSS label
space, written in a transposed convention, are used blockwise in
\cite[proof of Prop.~3, Supplementary Material]{holmes2026quantum}.

Three of the four circuit families below are obtained by restricting some of the blocks to be zero.
Each inherits an even simpler special case of Eq.~\eqref{eq:classes-gen}.
The fourth family, the Hadamard-and-permutation gates,
generally stands apart, but in the special case where the Hadamard piece acts on
all qubits, the permutation must map each classical code into the other. 
Figure~\ref{fig:four-classes} summarizes all four families and depicts example
circuits of each one for the \(\qcode{6,2,2}\) doubly-even self-dual code.

\subsection{$Z$-diagonal circuits}
\label{sec:classes-sp}

The first and most popular family consists of the circuits
\begin{equation}
  \Ush{S}=\prod_{i<j\,:\,S_{ij}=1}\mathrm{CZ}_{ij}
          \;\prod_{i\,:\,S_{ii}=1}\sqrt{Z}_{i},
  \label{eq:classes-circ-sp}
\end{equation}
one controlled-$Z$ gate for each off-diagonal entry of a symmetric matrix $S$ and one
single-qubit phase gate $\sqrt{Z}$ for each diagonal entry. $\Ush{S}$ is a valid
symplectic map exactly when $S$ is
symmetric~\cite{dehaene2003clifford,rengaswamy2019unifying}.

The gates in such a circuit commute, so one can apply them in any order. The
circuit forms a single two-local layer, however, if and only if no two CZ gates act on the same physical qubit.

In the symplectic block picture, the parameter $S$ occupies the upper-right block,
\begin{equation}
  \Ush{S}=\begin{pmatrix}\Id&S\\ 0&\Id\end{pmatrix},
  \qquad (x|z)\mapsto(x\,\vert\,z+xS),
  \label{eq:classes-block-sp}
\end{equation}
so that the circuit adds $xS$ to the $Z$-labels and leaves the $X$-labels untouched.
Every gate in Eq.~\eqref{eq:classes-circ-sp} is diagonal in the computational basis,
so we call these \emph{$Z$-diagonal circuits}. Specializing the general
condition~\eqref{eq:classes-gen} leaves the single requirement $\CX S\subseteq\CZ$.
Applied to any $X$-stabilizer, the circuit must produce a $Z$-string already lying in
$\CZ$. The $Z$-diagonal circuits meeting this condition form the family $\Eup$
of the code.

Parameters add, $\Ush{S}\Ush{S'}=\Ush{S+S'}$. The family is
therefore not just an Abelian group, but also a binary vector space. 
The space of admissible patterns $S$ can thus be found by Gaussian elimination. 
This is the Clifford case of an algorithm
that finds the code-preserving diagonal gates at every level of the Clifford
hierarchy~\cite{webster2023transversal,campsmoreno2026transversal,koh2026phantom,bauer2026finding}.
In particular, every diagonal circuit is its own inverse.

A number of gates studied under separate names are individual members of $\Eup$:
\begin{itemize}
\item Circuits with a diagonal binary matrix $S$ are single layers of phase gates. Examples are
  the transversal phase gate that acts as a logical $\sqrt{Z}$ on a doubly-even
  code~\cite{zeng2011transversality,anderson2016classification,dasu2025classification}
  (the identity pattern), phase gates on a subset of qubits, the mixed
  $\sqrt{Z}/\sqrt{Z}^{\dagger}$ patterns that serve self-dual and weakly doubly-even
  codes~\cite{tansuwannont2025clifford,jain2025transversal}, the subcube phase gates
  of quantum Reed--Muller codes~\cite{barg2026geometric}, and the Clifford level of
  the diagonal transversal gates classified for monomial codes and by the homology
  of the CSS chain complex~\cite{campsmoreno2026transversal,haruna2026homological}.
\item A single layer applies a logical $\overline{\mathrm{CZ}}$ to the
  $L\times L$ toric code for every $L\geq3$~\cite{holmes2026quantum}. That layer
  is a diagonal circuit with $L^{2}+2L$ gates.
\item The fold-transversal $\sqrt{Z}$-gate built from a self-inverse symmetry of the
  code~\cite{moussa2016transversal,breuckmann2024fold} is a diagonal circuit.
  Its pattern places a $\mathrm{CZ}$ across each pair of qubits that the symmetry
  swaps and a $\sqrt{Z}$ on each qubit it fixes. Such gates correspond to those $S$ that are also permutation matrices.
\item The same pattern, read off other symmetries, gives the phase-type gates of
  bivariate-bicycle, hypergraph product, and lifted product
  codes~\cite{eberhardt2024logical,quintavalle2023partitioning,zheng2026canonical,hong2026ldpc}. When the symmetry is
  fixed-point free, the diagonal vanishes and no $\sqrt{Z}$ appears at
  all~\cite{berthusen2025concatenated}. A free symmetry of any order contributes one
  $\mathrm{CZ}$ along each edge of its orbit cycles~\cite{gulshen2025symmetric}.
  In the SHYPS codes, the patterns read off automorphism-composed
  dualities yield every $Z$-diagonal logical
  gate~\cite{malcolm2025computing} (Sec.~\ref{sec:depth1twolocal-families}).
\item The cup-product circuits acting between two code
  blocks are \(Z\)-diagonal~\cite{breuckmann2026cups,li2025poincare}.
\item The $\mathrm{CZ}$ between two blocks of a
  code~\cite{dasu2025classification} is a diagonal circuit on the doubled code. Its pattern is
  the identity placed in the off-diagonal block, and restricting the gate to a subset
  of qubit pairs~\cite{cao2026lego} thins that block to a matching. 
\item The $\mathrm{CZ}$ circuits carrying a one-form symmetry of a
  homological code~\cite{zhu2025nonclifford} are also diagonal circuits with structured patterns.
\end{itemize}

\subsection{$X$-diagonal circuits}
\label{sec:classes-sm}

The lower-left symplectic block is a mirror of the upper-right block, generating a second independent family in a similar way.
Its circuits are the \(X\)-basis versions of the $Z$-diagonal ones, and every
gate in them is diagonal in the conjugate basis,
\begin{equation}
  \begin{aligned}
    \Lsh{T}&=\prod_{i<j\,:\,T_{ij}=1}\mathrm{CZ}^{H}_{ij}
             \;\prod_{i\,:\,T_{ii}=1}\sqrt{X}_{i}\\
           &=H^{\otimes n}\,\Ush{T}\,H^{\otimes n},
  \end{aligned}
  \label{eq:classes-circ-sm}
\end{equation}
built from $\sqrt{X}$ gates and the Hadamard-conjugated controlled-$Z$ gate
$\mathrm{CZ}^{H}=H^{\otimes2}\mathrm{CZ}\,H^{\otimes2}$.
The parameter $T$ now occupies the lower-left block,
\begin{equation}
  \Lsh{T}=\begin{pmatrix}\Id&0\\ T&\Id\end{pmatrix},
  \qquad (x|z)\mapsto(x+zT\,\vert\,z),
  \label{eq:classes-block-sm}
\end{equation}
so that the circuit adds $zT$ to the $X$-labels and leaves the $Z$-labels untouched.
Condition~\eqref{eq:classes-gen}
specializes to the mirror requirement $\CZ T\subseteq\CX$. The valid \emph{$X$-diagonal
circuits} form the family $\Edn$.

Everything about this family follows from the $Z$-diagonal one by the formal
$X\leftrightarrow Z$ swap that interchanges the two classical codes.
We use that duality throughout to avoid proving each statement twice. 
The family is, however, independent of the $Z$-diagonal one because it is obtained from a different condition.
The Hadamards
in Eq.~\eqref{eq:classes-circ-sm} are only a bookkeeping change of basis since a general CSS code need not have a 
code-preserving transversal Hadamard. 

On the other hand, when the two classical codes coincide, the
two diagonal families are genuine Hadamard-conjugates of one another. This is why they
appear identical in Figure~\ref{fig:four-classes}. For a code with $\CX\neq\CZ$, they
are different families.

The fold-transversal $\sqrt{Z}$-gate of the previous subsection also has an
$X$-diagonal mirror on the same pairing. For an involutive $XZ$
duality~$\tau$~\cite{breuckmann2024fold}, the pattern placing a
$\mathrm{CZ}^{H}$ across each swapped pair and a $\sqrt{X}$ on each fixed
qubit is a valid $X$-diagonal circuit because the fold
condition $\CX\pperm{\tau}\subseteq\CZ$ coincides with its mirror
$\CZ\pperm{\tau}\subseteq\CX$ for that case.
The mirrors of the SHYPS gates of the previous subsection likewise yield every
$X$-diagonal logical gate~\cite{malcolm2025computing}.

\subsection{CNOT circuits}
\label{sec:classes-levi}

CNOT networks correspond to the two diagonal symplectic blocks [see Figure~\ref{fig:four-classes}(c)], mixing $X$-labels among themselves and $Z$-labels among themselves.
The word ``diagonal'' refers here to the position of the blocks inside the
symplectic matrix, not to the basis in which the gates act.
A CNOT network is
not a diagonal circuit, and the $Z$- and $X$-diagonal circuits of
Sections~\ref{sec:classes-sp} and~\ref{sec:classes-sm} occupy the two
\emph{off}-diagonal symplectic blocks.

A single controlled-NOT with control $c$ and target $t$ is represented by the elementary invertible
map $K=\Id+E_{ct}$, with $E_{ct}$ the matrix unit. 
Any invertible map $K$ factors into such elementary matrices: each row addition performed while Gaussian-eliminating $K$ is one factor $\Id+E_{ct}$, so the elimination itself writes $K$ as a CNOT circuit,
\begin{equation}
  U_{K}=\prod_{m}\mathrm{CNOT}_{c_{m}\to t_{m}},
  \qquad K=\prod_{m}\bigl(\Id+E_{c_{m}t_{m}}\bigr),
  \label{eq:classes-circ-levi}
\end{equation}
at a cost of order $n^{2}/ \log n$ gates~\cite{patel2008efficient}.
The network acts by
$(x|z)\mapsto(xK\,\vert\,zK^{\itp})$, that is, in blocks
$\bigl(\begin{smallmatrix}K&0\\ 0&K^{\itp}\end{smallmatrix}\bigr)$.
The two diagonal blocks must be the correlated pair $K$ and $K^{\itp}$ for the operation to be symplectic.

When the CNOT gates act on \emph{disjoint} pairs, they form an oriented
matching, $N=\sum_{k}E_{c_{k}t_{k}}$ with $N^{2}=0$. In this case the whole network
is a single depth-one layer, and its matrix $K=\Id+N$ squares to one.

A CNOT network preserves the code when its map $K$ preserves each classical code
separately. In this case, condition~\eqref{eq:classes-gen} reads $\CX K=\CX$ and
$\CZ K^{\itp}=\CZ$. These are equalities, not inclusions, because $K$ is invertible.
Although the first is a linear condition on $K$, just like the conditions on $S$ and $T$ for diagonal circuits, the fact that we are restricting to \textit{invertible} $K$ makes the set of valid CNOT networks a nonlinear space.
These gates form a non-Abelian group $\Lev$, a copy of the invertible binary matrices $\GLg{n}$
restricted to only those that preserve the code. 

Qubit permutations form a subgroup of the CNOT gate group since they are assembled from SWAP gates, with each SWAP gate in turn consisting of three CNOT gates.
The permutation symmetries of the code therefore sit inside this family.
These symmetries are the permutations that preserve both classical codes,
$\mathrm{Aut}(\CX)\cap\mathrm{Aut}(\CZ)$~\cite{grassl2013leveraging}.

CSS codes have long been known to admit CNOT gates between blocks~\cite{gottesman1997stabilizer}. More recent members of $\Lev$ include the following:
\begin{itemize}
\item The homomorphic CNOT gadgets that couple two distinct CSS codes through a chain map between their complexes~\cite{benhemou2026automated} are CNOT elements of the direct sum, with block-unitriangular $K$.
\item In a homological product, the permutation symmetries of the input codes lift to permutation symmetries of the product code~\cite{berthusen2025automorphism}.
\item Generalized Dehn twists on hypergraph and balanced products of cyclic codes are CNOT elements, realized by CNOT circuits of depth $O(d)$~\cite{tiew2025dehn}.
\item Phantom codes~\cite{koh2026phantom} realize a logical $\overline{\mathrm{CNOT}}$ between every ordered pair of logical qubits by a physical qubit permutation. For a phantom code, the permutation subgroup of physical CNOT gates realizes the entire logical Levi subgroup $\GLg{k}\subset \Sp{2k}$.
\end{itemize}

\subsection{Partial dualities}
\label{sec:classes-weyl}

The fourth family is generated by a Hadamard gate on an arbitrary subset
of the qubits, followed by a permutation,
\begin{equation}
  (\Hd{a};\pi)=\pperm{\pi}\prod_{i\,:\,a_{i}=1}H_{i},\qquad\text{(partial duality)}
  \label{eq:classes-circ-weyl}
\end{equation}
where the bit string $a\in\Ftwo^{n}$ marks which qubits are Hadamard-transformed, and where
$\pi$ permutes the qubits. A subset Hadamard exchanges $X$ and $Z$ on the chosen
qubits, corresponding to a ``partial'' duality, and the permutation then relabels the qubits. 

This family is not attached to any block of the decomposition, because the other
three families already use all four blocks. It therefore behaves differently from
them. 
For example, the families overlap. A depth-one network of disjoint SWAP gates, where \(a=0\) and \(\pi\) is an involution, is also a CNOT element.

The permutations these gates supply can compensate an otherwise code-breaking transversal gate such that the compensated product preserves the code. This is the mechanism behind
the fold-transversal gates and the code's automorphism group. 

A pure subset Hadamard (without a compensating permutation) can preserve a code only in specific cases. The
all-qubit Hadamard preserves the code precisely when the code is self-dual,
$\CX=\CZ$~\cite{dasu2025classification}. A Hadamard on a proper, nonempty subset can preserve the code only if the
code splits into two independent blocks, that is, only if the code is decomposable
(Proposition~\ref{prop:trans-split}; cf.~\cite{guyot2026addressability}).
For a connected code, then, the only useful
pure dualities are the trivial one and the full one. Every intermediate choice has to be
compensated by a permutation~(Corollary~\ref{cor:trans-connected}).

A CSS code is
\emph{$em$-symmetric} when its two underlying classical codes can be permuted into each other using the same permutation~\cite{breuckmann2024fold,su2024tapestry,english2026duality}. For those codes, there exists a permutation \(\pi\) such that
$\CX=\CZ\pperm{\pi}$ and $\CX\pperm{\pi}=\CZ$, corresponding to a full duality $(\Hd{\ones};\pi)$.
The two conditions agree when $\pi$ is an involution.
The permutations realizing em-dualities, together with the permutation
symmetries $\mathrm{Aut}(\CX)\cap\mathrm{Aut}(\CZ)$, form the group of
permutations preserving the unordered pair $\{\CX,\CZ\}$.
The permutation symmetries are its
sector-preserving subgroup, of index at most two, and the em-dualities, when
present, form a coset of this group.
Lemma~1 of Ref.~\cite{breuckmann2024fold} gives the analogous coset statement for
ZX-dualities. There, the ambient group is the automorphism group of the
\emph{underlying classical code}, obtained by pooling the checks of the two
codes and forgetting their types. That group can be larger than the group
preserving $\{\CX,\CZ\}$.

Partial duality gates appear in the literature in several forms:
\begin{itemize}
\item as gates on the folded surface code~\cite{moussa2016transversal};
\item as fold-transversal Hadamards, arising from the same em-duality that yields fold-transversal $\sqrt{Z}$- and $\sqrt{X}$-gates~\cite{breuckmann2024fold,chen2026transversal};
\item as the Hadamard-SWAP gate of a hypergraph product
  code~\cite{quintavalle2023partitioning}, or of its subsystem
  version~\cite{malcolm2025computing}, a full duality compensated by the twin
  involution;
\item as the $XZ$-duality used to address the second logical block of a
  bivariate-bicycle code~\cite{bravyi2024high};
\item as the generalized $XZ$-dualities catalogued by automorphism-based
  tools~\cite{sayginel2025faulttolerant,sayginel2024autqec}.
\end{itemize}

%% file: sec_general.tex
\section{Code-preserving circuit group}
\label{sec:general}

The four families of Section~\ref{sec:classes} are all subgroups of the
group of all Clifford circuits that preserve a given CSS code. 
We will show that only two families, the $Z$-diagonal and the $X$-diagonal circuits, generate this group.

As before, we work
modulo Pauli operators and global phases. The
code-preserving circuits therefore correspond exactly to the symplectic maps that
fix $\lab$ setwise,
\begin{align}
    \Nfull&=\Stb_{\Sp{2n}}(\lab) \label{eq:N-def}\\
          &=\{\,g\in\Sp{2n}:\lab g=\lab\,\} .\quad\text{(code-preserving group)} \nonumber
\end{align}
A physical circuit realizes every such map, and a compensating Pauli comes
automatically with it (Theorem~\ref{thm:dressing}). 
There are no restrictions on the depth or the locality of the circuit [Figure~\ref{fig:layers}(d)].

Each of the four families of Section~\ref{sec:classes} is
the intersection of $\Nfull$ with a code-independent group of circuits. The valid
$Z$-diagonal and $X$-diagonal circuits $\Eup$ and $\Edn$ come out of all $Z$- and all $X$-diagonal
circuits. The family $\Lev$ comes out of all CNOT networks. The code-preserving
partial dualities come out of the full group $\PDgp$ of Hadamard-and-permutation
gates.

A binary matrix lies in $\Nfull$ if it satisfies two properties. First, its label map
must be symplectic, which makes the circuit a Clifford operation. 
Second, this map
must preserve the code. The two conditions yield a rigid block
form that is true for general stabilizer codes.

\subsection{Structure of $\Nfull$}
\label{sec:general-structure}

The block form occurs in a basis that is matched to the code rather than to the symplectic
$X$/$Z$ split of Section~\ref{sec:classes}. 
Any stabilizer code sorts the Pauli
operators into a product of three kinds. The stabilizers are the operators that fix every
codeword, and a set of $r=n-k$ of them generates the rest. Each stabilizer
generator has a paired destabilizer, which anticommutes with that one
generator and commutes with the others. The $2k$ logical operators of the $k$
encoded qubits carry the remaining freedom.

For example, consider the trivial code that stores logical information in its first $k$ qubits. Its logicals are the $X$ and $Z$ operators of those
qubits. Its stabilizers are the $Z$ operators of the remaining $r$ qubits, and its
destabilizers are the $X$ operators of those same $r$ qubits.
A Clifford change of basis carries any other stabilizer code to this picture. The
change of basis relabels the operators, but it leaves their commutation
relations intact.

Code preservation constrains how a Clifford circuit may mix these three kinds of operators among themselves. The logical action is unconstrained: any Clifford on the logical
operators is allowed. The circuit may also relabel the stabilizer generators among
themselves in any invertible way, and the destabilizers then adjust to maintain the
commutation relations.
But there are other mixing operations the circuit can do.

A circuit in $\Nfull$ fixes the code, so it must send every
stabilizer to a stabilizer. 
In other words, a stabilizer cannot acquire a logical or a
destabilizer tail because it has to continue to commute with all stabilizers. 
However, the circuit
may pad a logical with a stabilizer, but never with a destabilizer,
because the result would anticommute with some stabilizer. 
No such constraint
applies to the destabilizers, and the circuit may pad them with logicals and/or stabilizers. 
These restrictions
yield patterns of zeros in the block form of the circuit, written in the stabilizer,
logical, and destabilizer basis (Lemma~\ref{lem:general-blockform}).

When we assemble the allowed moves, they exhibit $\Nfull$ as a semidirect product,
\begin{equation}
  \Nfull=\underbrace{\Urad}_{\text{padding}}\rtimes\Bigl(
    \underbrace{\GLg{r}}_{\substack{\text{stabilizer}\\\text{relabeling}}}
    \times
    \underbrace{\Sp{2k}}_{\substack{\text{logical}\\\text{action}}}\Bigr) .
  \label{eq:N-structure}
\end{equation}
The two direct factors are the intra-set actions: the invertible (general-linear-map) relabelings of the
$r$ stabilizer generators, and the logical Clifford group $\Sp{2k}$. The normal
factor $\Urad$ collects the padding operations: it adds stabilizers to logicals, and
it adds stabilizers and logicals to destabilizers. This factor acts trivially on
both the generators and the encoded qubits. The three factors are independent, so
the order of $\Nfull$ is the product of their orders. 
At $k=0$, Eq.~\eqref{eq:N-structure}
reduces to the symmetry group of the all-zeros state~\cite{bravyi2021hadamardfree}.

Since we mostly care about nontrivial logical gates, we will often restrict a code-preserving symplectic map to its action on the logical operators alone. This is done through the homomorphism
\begin{equation}
  \lact:\Nfull\twoheadrightarrow\Sp{2k}\qquad\text{(logical action)}
  \label{eq:N-lact}
\end{equation}
onto the Clifford group of the encoded qubits.
This map forgets the actions on the other factors, showing that many physical circuits realize the same logical gate.

\subsection{Generation by diagonal circuits}
\label{sec:general-shear}

A key structural result of the paper is that, for a
nontrivial CSS code, the two diagonal families alone build the entire group of code-preserving Clifford operations. 
If
the code has a nonempty $X$-part and a nonempty $Z$-part, every code-preserving
circuit is a product of $Z$-diagonal and $X$-diagonal circuits,
\begin{equation}
  \Nfull=\langle\,\Eup,\Edn\,\rangle\qquad\text{(generation)}
  \label{eq:N-shear}
\end{equation} 
(Theorem~\ref{thm:sheargen}). 
Another way to say this is that the two diagonal families \textit{generate} the full group.
The only condition required is that the CSS code be nontrivial, that is, \(\CX,\CZ \neq 0\).

A $Z$-diagonal circuit is only a layer of controlled-$Z$ and phase gates, and an
$X$-diagonal circuit is its $X$-basis mirror.
But products of the two families reach well beyond circuits diagonal in either
basis.
A $Z$-diagonal and an $X$-diagonal circuit compose into a Hadamard-type map. Longer products, that
is, group words, express the permutations of the partial dualities and CNOT
networks.

Applying the logical map $\lact$, which forgets about how the map acts on the stabilizers and destabilizers, we obtain the simple corollary
\begin{equation}
  \lact\bigl(\langle\Eup,\Edn\rangle\bigr)=\Sp{2k} ,\qquad\text{(logical completeness)}
  \label{eq:N-logcomplete}
\end{equation}
so the two diagonal families realize \emph{every} logical Clifford operation on the
encoded qubits.

The reverse correspondence is far from one-to-one. Very many physical circuits
realize a single logical
Clifford. Two circuits may agree on the
$\Sp{2k}$ factor and yet differ in the stabilizer relabeling or in the padding. The
number of circuits
that share a given logical action is the same for every logical Clifford. This
number is
the size of one fiber of $\lact$, and it equals the combined order of the other two
factors,
\begin{equation}
  |\Urad|\,|\GLg{r}|=2^{\,r(r+2k)}\prod_{i=1}^{r}\bigl(2^{i}-1\bigr),
  \label{eq:N-fibre}
\end{equation}
in the $r=n-k$ stabilizer generators and $k$ logical qubits (Appendix~\ref{app:general}).

This freedom is a resource: a prescribed logical action leaves the relabeling and padding of Eq.~\eqref{eq:N-fibre} entirely free, and one may optimize over that freedom without altering the logical operation~\cite{rengaswamy2018synthesis,popov2026optimized}. Section~\ref{sec:ext-depth} proposes how to exploit this and mentions related approaches.

The remainder of the paper accordingly restricts attention to circuits that run in a single layer. 
We characterize the groups generated by the two most useful families of such
gates. The transversal gates carry one single-qubit gate per qubit. The
depth-one two-local gates form a single layer that may also pair qubits.

%% file: sec_transversal.tex
\section{Transversal gate group}
\label{sec:transversal}

The first of the depth-one families is the well-known transversal gate group. We define a transversal circuit to consist of one (possibly different)
single-qubit Clifford on each qubit
[Figure~\ref{fig:layers}(a)]. 
Its label map is block
diagonal, with one two-by-two $\Sp{2}$ block per qubit. 

The code-preserving transversal
circuits form
\begin{equation}
  \Ntr=\Sp{2}^{\,n}\cap\Nfull ,\qquad\text{(transversal group)}
  \label{eq:T-def}
\end{equation}
the tuples of single-qubit Cliffords that preserve $\lab$. 
The four families of Sec.~\ref{sec:classes} reduce to three in this case: a transversal $\Lev$ gate acts on each
qubit by an invertible one-by-one matrix over $\Ftwo$, so the transversal $\Lev$ gates
are trivial.

Transversal \(Z,X\)-diagonal circuits are labeled by diagonal matrices $S$ and $T$.
Their binary diagonals determine whether the corresponding qubit carries a gate or not. A $Z$-diagonal circuit is a
$\sqrt{Z}$ on a chosen subset of qubits, and an $X$-diagonal circuit is a $\sqrt{X}$ on such a subset. 

The conditions on \(S\) and \(T\) in terms of the classical codes reduce to the following.
Code-preserving transversal \(Z,X\)-diagonal circuits can be read off from two \emph{parameter codes},
\begin{align}
    \Pup&=\{\,a\in\Ftwo^{n}:a\pw\CX\subseteq\CZ\,\}, \label{eq:T-params}\\
    \Pdn&=\{\,b\in\Ftwo^{n}:b\pw\CZ\subseteq\CX\,\},\qquad\text{(parameter codes)} \nonumber
\end{align}
where $\pw$ is the pointwise product of bit strings. The code $\Pup$ ($\Pdn$) lists the qubit sets
that may carry a $\sqrt{Z}$ ($\sqrt{X}$). Each
defining condition asks that a linear image of the parameter land in a subspace. Both
parameter codes are therefore ordinary binary linear codes.
We define $\Eupo$ and $\Edno$ to be the resulting groups
of valid transversal $Z$-diagonal and $X$-diagonal circuits.

Since the $\Lev$ gates are trivial,
our general depth-one theorem 
(Theorem~\ref{thm:depthone}) reduces to
\begin{equation}
  \Ntr=\langle\Eupo,\Edno\rangle,\qquad\text{(transversal generation)}
  \label{eq:T-gen}
\end{equation}
meaning that the entire transversal gate group is generated only by the above diagonal circuits (Corollary~\ref{cor:transgen}).
In other words, one need only calculate the above two parameter codes to determine the group.

\subsection{Three-layer decomposition}
\label{sec:transversal-nf}

We have so far established that diagonal circuits are sufficient to generate the transversal gate group, but we do not yet know \textit{how many} diagonal circuits are required to express a general group word.
A key result of this section is that three layers of diagonal circuits suffice to realize every transversal gate (Theorem~\ref{thm:transnf}).
In other words, every transversal gate $g$ factors (non-uniquely) as a $Z$-diagonal
circuit, an $X$-diagonal circuit, and a $Z$-diagonal circuit,
\begin{equation}
  g=\Ush{a}\,\Lsh{b}\,\Ush{a'},\qquad a,a'\in\Pup,\; b\in\Pdn.
  \label{eq:T-nf}
\end{equation}

We also provide a way to calculate the above decomposition for a general transversal gate.
Writing the four diagonal
blocks of a general gate $g$ as bit strings $A,B,C,D$, we have
\begin{equation}
  b=C,\qquad a=A\pw B+B\pw C,\qquad a'=B\pw C+B\pw C\pw D.
  \label{eq:T-closed}
\end{equation}
Moreover, two layers suffice exactly when the parameter codes intersect trivially,
\begin{equation}
  \Ntr=\Eupo\cdot\Edno
  \quad\Longleftrightarrow\quad
  \Pup\cap\Pdn=\{0\}
  \label{eq:T-twoblock}
\end{equation}
(Proposition~\ref{prop:twoblock}). 
In that case, every valid
$Z$-diagonal circuit and every valid $X$-diagonal circuit have disjoint supports. They therefore
commute, and
each side of a product collapses to a single layer. 

On the other hand, a nontrivial
$t\in\Pup\cap\Pdn$ yields a subset
Hadamard $\Hd{t}=\Ush{t}\Lsh{t}\Ush{t}$, a partial exchange of $X$ and $Z$ on the
support of $t$. Its shortest diagonal-circuit form is three layers and not two.

A common parameter $t\in\Pup\cap\Pdn$ with
$0\neq t\neq\ones$ implies that the code consists of multiple disconnected code blocks.
In that case, both classical codes split along the partition of the qubits
into
the support of $t$ and its complement, and the stabilizer code is the tensor product
of its restrictions to the two parts (Proposition~\ref{prop:trans-split}). 
A code
that admits no such splitting for any nontrivial partition is \emph{connected}, or
indecomposable. On a
connected code, the only possible common parameters are therefore the empty string and
the all-ones string. The all-ones string implies that the code is self-dual, i.e., $\CX=\CZ$ (Corollary~\ref{cor:trans-connected}).

To summarize, the parameter codes of a connected non-self-dual CSS code intersect trivially, so every valid $Z$-diagonal circuit commutes with every valid $X$-diagonal circuit. By Eq.~\eqref{eq:T-twoblock}, two layers of gates then already suffice to express a general transversal gate, $\Ush{a}\Lsh{b}$.

\subsection{The size of the transversal group}
\label{sec:transversal-algo}

We can determine the \emph{size} of the transversal gate group using a second normal form in terms of diagonal circuits and Hadamard gates $\Hd{a}$ for qubit subsets $a$
(Proposition~\ref{prop:permkernel}). Note that the subset Hadamard is itself a
diagonal-circuit product $\Hd{a}=\Ush{a}\Lsh{a}\Ush{a}$, and hence not a new generator, but it is convenient for our purposes. 
Unlike the
decomposition above, this form is \emph{unique}:
\begin{equation}
\begin{gathered}
  g=\Hd{a}\,\Ush{q}\,\Lsh{p},\\
  a=B\pw C,\quad q=B\pw D,\quad p=A\pw C.
\end{gathered}
  \label{eq:T-uniquenf}
\end{equation}
Here the two diagonal circuits have disjoint support ($q\pw p=0$). We again read the parameters
straight off the diagonal blocks $A,B,C,D$ of $g$.

As a product of sets,
$\Ntr=\{\Hd{a}:a\in\Pup\cap\Pdn\}\cdot\Eupo\cdot\Edno$. Uniqueness of the above normal form thus yields the group's order,
\begin{equation}
  |\Ntr|=\#\{(a,q,p)\in(\Pup\cap\Pdn)\times\Pup\times\Pdn:q\pw p=0\}.
  \label{eq:T-count}
\end{equation}

%% file: sec_depthone.tex
\section{Two-fold transversal group}
\label{sec:depthone}

We extend the transversal gate group and define a new group that allows for two-qubit entangling
Clifford gates. 
First we fix which qubits are
paired and characterize the code-preserving depth-one two-local layers compatible with that choice.
Then we let the pairing vary and characterize the group generated by all such layers
together.
We call this the \emph{two-fold transversal group}, following~\cite{chakraborty2026nogo}: a transversal circuit is one-fold transversal, and a depth-one two-local layer is two-fold transversal with respect to its matching. 
The fold-transversal gates of Section~\ref{sec:classes} are the special case in which the pairing comes from an $XZ$ duality.

\subsection{Depth-one circuits at a fixed qubit matching}
\label{sec:depthone-fixed}

A choice of which qubits are paired in a depth-one two-local circuit is a \emph{matching} $\Mt$ of the $n$
qubits. 
In each such matching, some qubits are paired up, and the rest are left as ``singletons''.  
The ``cells'' of a matching partition the qubits
into blocks of size one or two, say $m_{1}$ singletons and $m_{2}$ pairs with
$n=m_{1}+2m_{2}$. A depth-one two-local layer on $\Mt$ applies an arbitrary
two-qubit Clifford to each pair and an arbitrary single-qubit Clifford to each
singleton. 

We clarify that \emph{two-local} means at most two-local, so a single layer may
mix one-qubit gates and two-qubit gates [Figure~\ref{fig:layers}(b)]. In the label
picture these are the symplectic maps
that are block diagonal for the cells of $\Mt$, with one $\Sp{4}$ block per pair
and one $\Sp{2}$ block per singleton.
There are six possible singleton maps, and 720 two-qubit maps, corresponding to the orders of $\Sp{2}$ and $\Sp{4}$, respectively.
We write $\DM$ for the group formed by all possible combinations of such maps over all cells of a fixed matching $\Mt$.

The code-preserving layers on $\Mt$ are those that in addition fix the label
space [Figure~\ref{fig:layers}(b)],
\begin{equation}
  \NM=\DM\cap\Nfull ,
  \label{eq:dep-NM}
\end{equation}
where \(\Nfull\) is the group of all code-preserving Clifford gates.
The case of empty $\Mt$, where every
cell is a singleton and no qubits are paired, reduces this group to the
transversal group $\Ntr$ of Section~\ref{sec:transversal}.
The group $\NM$ of every other matching contains this group as a subgroup.

We may also restrict each family of Section~\ref{sec:classes}
to its $\Mt$-compatible part. The valid ones form
\begin{equation}
  \Ezx_{\Mt}=\Ezx\cap\DM ,
  \qquad
  \LevM=\Lev\cap\DM .
  \label{eq:dep-families}
\end{equation}
Valid diagonal circuits are those whose matrices $S$ or $T$ are block diagonal in the decomposition induced by the matching $\Mt$. That
is, every controlled-$Z$ they carry must sit inside a pair of qubits in $\Mt$.
CNOT networks include any CNOT or swap circuits compatible with the matching.
The partial dualities $\PDgp\cap\DM$ on the matching consist of any Hadamards
and compatible permutations.
Any permutations in the CNOT or partial-duality families necessarily have to be involutions (read: square to one) since only disjoint SWAP gates are allowed.

The key structural result of this work is that the $Z$-diagonal, $X$-diagonal, and CNOT gate families generate the entire group,
\begin{equation}
  \NM=\bigl\langle\,\EupM,\ \EdnM,\ \LevM\,\bigr\rangle ,
  \label{eq:dep-gen}
\end{equation}
for every CSS code and every matching (Theorem~\ref{thm:depthone}). 
In words, every
code-preserving depth-one two-local gate on a matching \(\Mt\) can be expressed as a product of depth-one diagonal circuits and
CNOT circuits on \(\Mt\).

The proof is an exhaustive reduction by multiplication. Start from an arbitrary layer
$G\in\NM$. One multiplies it on the left and the right by valid diagonal circuits (i.e., ``moves'') on the same
matching, and this simplifies it cell by cell until it reaches a fixed
\emph{terminal}
form. The moves are not fixed in advance, and explicit formulas build
each move from
the blocks of the current symplectic matrix. That the moves reach a terminal
form is the one finite point of the proof. Two-locality caps each cell at two
qubits, so a cell ranges over the $720$ elements of $\Sp{4}$, or over the six
elements of
$\Sp{2}$. We exhaust these finitely many possibilities, and they show that every
block
becomes terminal in at most three moves. Terminality is absorbing, so we clear the
cells one at a time. Inspection shows that each terminal block is either a CNOT
network or a short diagonal-circuit product $\Ush{B}\Lsh{C}\Ush{B}$. We then undo the moves,
and this writes $G$ as a word in the three families. Appendix~\ref{app:depthone}
demonstrates this algorithm on 
the $\qcode{6,2,2}$ code.

The generators in Eq.~\eqref{eq:dep-gen} lie inside $\NM$ for each $\Mt$. This
control at every fixed matching then yields a similar result for the two-fold transversal group.

\subsection{Definition and structure of the two-fold transversal group}
\label{sec:depthone-union}

A transversal gate group can yield the full logical Clifford group only when the code has \(k=1\) logical qubits, and depth-one two-local layers of any single fixed
matching can only realize such a group at \(k\leq 2\)~\cite{chakraborty2026nogo}.
A change of the matching $\Mt$
between layers lifts these restrictions, so we define the group generated by such gates.

The two-fold transversal group is the subgroup of the code-preserving Clifford group
$\Nfull$ generated by all depth-one two-local layers at once,
\begin{equation}
  \Ndep=\Bigl\langle\,\bigcup_{\Mt}\NM\,\Bigr\rangle .\qquad\text{(two-fold transversal group)}
  \label{eq:dep-cup}
\end{equation}
Equivalently, this group contains the code-preserving action of every finite-depth circuit whose layers
are each depth-one, code-preserving, and two-local on some matching. Its elements are group
words, i.e., products of depth-one layers that live on various matchings
[Figure~\ref{fig:layers}(c)]. 
Such a
word is generally not itself a single depth-one two-local layer since its factors pair the qubits in incompatible ways.
 
Our key result from the previous subsection shows that, for a fixed matching \(\Mt\), the corresponding group $\NM$ is generated by $Z$-diagonal, $X$-diagonal, and CNOT depth-one two-local gates.
Since the two-fold transversal group is generated by the union of such groups over all matchings, we obtain our main result (Theorem~\ref{thm:depthone} and Corollary~\ref{cor:depthone}),
\begin{equation}
  \Ndep=\Bigl\langle\,\bigcup_{\Mt}\,\EupM,\ \EdnM,\ \LevM\,\Bigr\rangle .
  \label{eq:dep-genunion}
\end{equation}
In words, this result states that every product of depth-one two-local code-preserving gates on \textit{arbitrary} matchings can be expressed as a product of depth-one two-local $Z$-diagonal circuits, $X$-diagonal circuits, and CNOT circuits.
Conversely, if we can enumerate these three families of generating circuits, then we can generate all possible two-fold transversal circuits.

Since this group includes two-local gates on arbitrary matchings, there is no restriction on implementing arbitrary logical Clifford gates.
In other words, its logical image
$\lact(\Ndep)$ [Eq.~\eqref{eq:N-lact}] can be the full logical Clifford group of the encoded qubits.
This holds, for example, for a subfamily of
quantum Reed--Muller codes~\cite{tansuwannont2026full}, and we describe many more ``full'' codes of this type in the next section. 
By the above result, we merely need to determine the
logical images of the three depth-one families, taken over all matchings, to obtain the logical image.

To reach the full \emph{logical} group is not the same as to reach the full group of code-preserving circuits.
Indeed, we show that $\Ndep$ is in general a proper
subgroup of $\Nfull$.
In other words, some code-preserving gates cannot be decomposed as products of depth-one
two-local layers (see Figure~\ref{fig:hasse}).
If instead $\Ndep$ were equal to $\Nfull$, then every code would be ``full''.

Motivating and complementing this work, Holmes studies the single-layer logical image (a.k.a. ``logical reach'') of
code-preserving $W$-local layers on indecomposable CSS codes
\cite[Def.~1 and Props.~4--5]{holmes2026quantum}, under an addressability
condition and, for the main reach containments, a confinement hypothesis
\cite[Lems.~4--5, Supplementary Material]{holmes2026quantum}.
Their related \emph{ladder of logical reach}
\cite[Thm.~6, Supplementary Material]{holmes2026quantum} orders addressable
single-layer logical-image sets and their automorphism-augmented variants.
The physical counterparts to these logical gates are single addressable layers drawn from the groups $\NM$ defined here.
At the abstract logical level, their Generation Lemma
\cite[Lem.~1]{holmes2026quantum} also shows that the full logical Clifford group can be reached in principle, 
while physical attainment remains
conditional on a code realizing a complete set of generators
\cite[Cor.~3]{holmes2026quantum}.

\input{fig_hasse}

%% file: fig_hasse.tex
\begin{figure}[t]
\centering
\begin{tikzpicture}[
    every node/.style={font=\footnotesize},
    grp/.style={draw, rounded corners, fill=black!4, inner sep=3.5pt, align=center},
    inc/.style={thick, black!70},
  ]
  \node[grp] (N)   at (0,3.3)    {$\Nfull=\Stb_{\Sp{2n}}(\lab)$\\[1pt]
                                  {\scriptsize all code-preserving gates}};
  \node[grp] (cup) at (-2.05,2.0){$\Ndep$\\[1pt]
                                  {\scriptsize$\langle\,\NM:\text{all }\Mt\,\rangle$}};
  \node[grp] (aut) at (2.15,1.45){$\Gaut$\\[1pt]
                                  {\scriptsize$(\Sp{2}\wr\Perm{n})\cap\Nfull$}};
  \node[grp] (M)   at (-2.05,0.75){$\NM$\\[1pt]
                                  {\scriptsize fixed matching $\Mt$}};
  \node[grp] (tr)  at (0,-0.35)  {$\Ntr$\\[1pt]
                                  {\scriptsize transversal group}};
  \draw[inc] (tr)  -- (M);
  \draw[inc] (M)   -- (cup);
  \draw[inc] (cup) -- (N);
  \draw[inc] (tr)  -- (aut);
  \draw[inc] (aut) -- (N);
\end{tikzpicture}

\caption{\label{fig:hasse}\textbf{Gate-group inclusion.} All groups are taken modulo Pauli operators and lie inside
$\Nfull$, the group of code-preserving Cliffords. Lines are
subgroup inclusions, read upward; the two chains meet at the transversal group
$\Ntr$ of Section~\ref{sec:transversal}. Every inclusion drawn is strict in
general, and the groups $\Ndep$ and $\Gaut$ are incomparable. 
Allowing for qubit permutations to compensate non-code-preserving transversal gates gives the automorphism group $\Gaut$.
The transversal and automorphism groups carry no entangling two-local gates, while the group of all matchings includes gates
for all possible matchings $\Mt$.
Some permutations \textit{cannot} be expressed as layers
of code-preserving two-local circuits ($\Gaut\not\subseteq\Ndep$):
the $\qcode{9,3,2:186339}$ code~\cite{cross2025small,crossqiskitqec2025} carries an
order-four permutation $\pperm{\pi=(0\,5\,8\,2)(1\,6)(3\,4)}\in\Gaut$ that is not a group word in depth-one
two-local layers. 
To show this, we exhaust all $9\cdot7!!=945$ maximum matchings to show that
every valid depth-one layer fixes the \(X\)-type stabilizer \(100110011\) that $\pperm{\pi}$
moves.
Since the permutation is outside of the two-fold transversal group, this also implies that it \textit{need not} generate 
all code-preserving Cliffords (\emph{$\Ndep\subsetneq\Nfull$}).
}
\end{figure}

%% file: sec_depth1twolocal.tex
\section{Codes with full-Clifford two-fold transversal gates}
\label{sec:depth1twolocal}

We established in the previous section that a code can, in principle, realize any logical Clifford operation as a product of depth-one two-local code-preserving circuits.
Here, we report results of numerical experiments that identify such codes.

Recalling that \(\Ndep\) is the group generated by all depth-one two-local circuits, we identify codes whose logical image is
\begin{equation}
  \lact(\Ndep)=\Sp{2k} ,\qquad\text{(full code)}
  \label{eq:full-def}
\end{equation}
where $\lact$ sends a code-preserving gate to its action on the $k$ logical
qubits, and $\Sp{2k}$ is the logical Clifford group modulo Pauli operators and
phases. We call a code with this property \emph{full}.

The benefit of full codes is that all their logical Clifford operations are realizable by depth-one two-local circuits interspersed with error-correcting rounds.
However, two-local gates do spread errors, unlike transversal ones. One layer sends an
error on $e$ qubits to an error on at most $2e$ qubits. After such a layer a
code of distance $d$ therefore corrects an error on at most $\lfloor
d/4\rfloor$ qubits, and it detects an error on at most $\lfloor d/2\rfloor$
qubits.
The smallest sensible distance to consider for full codes is thus $d=3$, at which errors can at least be post-selection fault tolerant in the sense of Ref.~\cite{gottesman2016small}.

The above count is the worst case, and a particular matching often does better. 
Some errors heavier than the distance may stay
correctable, and a
finer analysis must decide whether a given partition really yields a fatal error. 
For example, a CNOT layer between blocks of Steane codes doubles
the weight of a single-qubit error from one to two, but the two halves of the
error sit in different blocks, so each block corrects its own half.
For a more non-trivial example, the $\qcode{5,1,2}$ rotated surface code admits a fault-tolerant two-local gate despite being only distance-two~\cite{vasmer2022morphing}.

We study only codes with $k>1$. Full codes at $k=1$ are easy to find since they need not require two-local gates. A self-dual doubly-even code with
one logical qubit, for instance, has a transversal $\sqrt{Z}$ and a transversal
Hadamard, and those two gates already generate $\Sp{2}$~\cite{jain2025transversal}.

\subsection{Select full codes}
\label{sec:depth1twolocal-census}

\begin{table*}[!t]
\caption{\label{tab:full-depth1}%
A list of \Nfullcodes\ indecomposable CSS codes with \(n\le120\) whose full logical Clifford
group can be realized by stacking depth-one two-local code-preserving
circuits.  For each we give the check weight and qubit degree of a
generating set of the stabilizer group.
Codes are
ordered by distance, then by $k$, then by $n$, then by weight and
degree.  Thick lines separate distance groups and thin lines separate
$k$ groups within a distance.  Each code is named by its
parameters followed by its source database and the leading octet of its
uuid or tag; Appendix~\ref{app:codes} expands the database names and discusses the algorithms used.
A further 45 full codes with \(n>120\) appear in the machine-readable records
of Ref.~\cite{albert2026certificate}: 35 \qcode{256,16} codes together with
\qcode{130,52,8}, \qcode{144,100,4}, \qcode{156,24}, \qcode{162,120,6}, \qcode{176,132,6}, \qcode{190,152,4}, \qcode{196,144,4}, \qcode{256,196,4}, \qcode{276,120}, \qcode{276,230,4}.}
\begin{ruledtabular}
\scriptsize
\setlength{\tabcolsep}{1.0pt}
\setlength{\aboverulesep}{0pt}\setlength{\belowrulesep}{0pt}
\renewcommand{\arraystretch}{1.10}
\begin{tabular}{lrr@{\hspace{0.42em}}lrr@{\hspace{0.42em}}lrr@{\hspace{0.42em}}lrr}
Code & Wt & Deg & Code & Wt & Deg & Code & Wt & Deg & Code & Wt & Deg \\
\hline
\fcode[xzzx_10_2_3]{\qcode{10,2,3}}{eczoo:dd1f5a7f} & 4 & 4 & \fcode{\qcode{22,2,4}}{se:d1138262} & 8 & 8 & \fcode{\qcode{40,16,4}}{dd422:a132b525} & 20 & 4 & \fcode{\qcode{32,4,6}}{ml:2fdc0964} & 12 & 6 \\
\cmidrule[0.4pt]{7-9}
\fcode[stab_10_2_3]{\qcode{10,2,3}}{eczoo:1cdbcad5} & 5 & 6 & \fcode{\qcode{22,2,4}}{se:f370ebef} & 8 & 8 & \fcode{\qcode{48,24,4}}{cons:grid6x8} & 12 & 8 & \fcode{\qcode{32,4,6}}{ml:e0decdbe} & 12 & 6 \\
\cmidrule[0.4pt]{10-12}
\fcode{\qcode{14,2,3}}{se:d369fb5b} & 6 & 5 & \fcode{\qcode{22,2,4}}{se:4688b36a} & 10 & 5 & \fcode{\qcode{56,24,4}}{cons:grid4x14} & 16 & 6 & \fcode{\qcode{36,6,6}}{dd422:5c721162} & 12 & 6 \\
\cmidrule[0.4pt]{4-6}
\cmidrule[0.4pt]{7-9}
\cmidrule[0.4pt]{10-12}
\fcode{\qcode{18,2,3}}{se:3f9b05dc} & 6 & 5 & \fcode[biplane_16_4_4]{\qcode{16,4,4}}{se:99fa739a} & 6 & 6 & \fcode{\qcode{72,32,4}}{cons:grid4x18} & 20 & 6 & \fcode{\qcode{40,8,6}}{dd422:03ce3e07} & 12 & 6 \\
\cmidrule[0.4pt]{7-9}
\fcode{\qcode{18,2,3}}{se:463ea74b} & 6 & 5 & \fcode[stab_16_4_4]{\qcode{16,4,4}}{2bga:24768ca1} & 8 & 6 & \fcode{\qcode{64,36,4}}{ar:12eaffb5} & 16 & 7 & \fcode{\qcode{40,8,6}}{ar:ed017efc} & 12 & 8 \\
\cmidrule[0.4pt]{7-9}
\cmidrule[0.4pt]{10-12}
\fcode{\qcode{18,2,3}}{se:4be5b357} & 6 & 5 & \fcode[stab_18_4_4]{\qcode{18,4,4}}{se:8b56945b} & 6 & 8 & \fcode{\qcode{72,40,4}}{cons:grid6x12} & 16 & 7 & \fcode{\qcode{56,16,6}}{ar:22d42765} & 10 & 9 \\
\cmidrule[0.4pt]{10-12}
\fcode{\qcode{18,2,3}}{se:71eaf638} & 6 & 5 & \fcode{\qcode{18,4,4}}{se:339dd3da} & 8 & 8 & \fcode{\qcode{88,40,4}}{cons:grid4x22} & 24 & 7 & \fcode{\qcode{100,56,6}}{ar:f5f0e1a7} & 22 & 12 \\
\cmidrule[0.4pt]{7-9}
\cmidrule[0.4pt]{10-12}
\fcode{\qcode{22,2,3}}{se:745d9b2a} & 6 & 5 & \fcode{\qcode{20,4,4}}{se:68e0b96f} & 6 & 8 & \fcode{\qcode{66,44,4}}{ar:df0dbd7d} & 30 & 8 & \fcode{\qcode{112,68,6}}{ar:781fc730} & 30 & 15 \\
\cmidrule[0.4pt]{7-9}
\cmidrule[1.1pt]{10-12}
\fcode{\qcode{22,2,3}}{se:4166b3e4} & 6 & 6 & \fcode{\qcode{20,4,4}}{se:b9a66cba} & 6 & 8 & \fcode{\qcode{80,48,4}}{cons:grid8x10} & 16 & 8 & \fcode{\qcode{64,2,8}}{ar:202789e7} & 14 & 15 \\
\fcode{\qcode{22,2,3}}{se:76034969} & 6 & 6 & \fcode{\qcode{20,4,4}}{se:1f0d7209} & 8 & 5 & \fcode{\qcode{104,48,4}}{cons:grid4x26} & 28 & 7 & \fcode{\qcode{64,2,8}}{ar:ef873233} & 14 & 15 \\
\cmidrule[0.4pt]{7-9}
\cmidrule[0.4pt]{10-12}
\fcode{\qcode{22,2,3}}{se:f281ebfb} & 6 & 6 & \fcode{\qcode{20,4,4}}{se:0b101385} & 8 & 6 & \fcode{\qcode{96,56,4}}{cons:grid6x16} & 20 & 7 & \fcode{\qcode{64,4,8}}{dd422:24bc40c0} & 12 & 8 \\
\fcode{\qcode{22,2,3}}{se:01b7685f} & 8 & 5 & \fcode{\qcode{20,4,4}}{se:8606ff37} & 8 & 6 & \fcode{\qcode{120,56,4}}{cons:grid4x30} & 32 & 62 & \fcode{\qcode{72,4,8}}{dd422:cebe3fac} & 12 & 8 \\
\cmidrule[0.4pt]{7-9}
\fcode{\qcode{22,2,3}}{se:8fd7fff8} & 8 & 5 & \fcode{\qcode{20,4,4}}{ml:f593c885} & 8 & 7 & \fcode{\qcode{100,64,4}}{ar:ba83a073} & 18 & 8 & \fcode{\qcode{80,4,8}}{dd422:7b0834e2} & 12 & 8 \\
\cmidrule[0.4pt]{7-9}
\cmidrule[0.4pt]{10-12}
\fcode{\qcode{22,2,3}}{se:bc6c761f} & 8 & 5 & \fcode{\qcode{20,4,4}}{ml:2586f00e} & 8 & 8 & \fcode{\qcode{112,72,4}}{cons:grid8x14} & 20 & 7 & \fcode{\qcode{64,6,8}}{ar:4f396e64} & 14 & 13 \\
\cmidrule[0.4pt]{1-3}
\cmidrule[0.4pt]{10-12}
\fcode[phantom_14_3_3]{\qcode{14,3,3}}{eczoo:880eb7d7} & 8 & 6 & \fcode{\qcode{22,4,4}}{se:791cbb1a} & 8 & 6 & \fcode{\qcode{120,72,4}}{cons:grid6x20} & 24 & 7 & \fcode{\qcode{64,8,8}}{ar:097e4c94} & 8 & 8 \\
\cmidrule[0.4pt]{1-3}
\cmidrule[0.4pt]{7-9}
\fcode[stab_15_4_3]{\qcode{15,4,3}}{cons:blk45k12} & 8 & 8 & \fcode{\qcode{22,4,4}}{se:06fce95a} & 8 & 8 & \fcode{\qcode{120,80,4}}{cons:grid10x12} & 20 & 7 & \fcode{\qcode{64,8,8}}{dd422:b81453b9} & 12 & 8 \\
\cmidrule[0.4pt]{7-9}
\fcode{\qcode{16,4,3}}{se:85abc767} & 6 & 8 & \fcode{\qcode{22,4,4}}{se:0ae369f5} & 8 & 8 & \fcode{\qcode{120,90,4}}{cons:johnson16} & 56 & 15 & \fcode{\qcode{64,8,8}}{ar:2aa0f87a} & 12 & 12 \\
\cmidrule[1.1pt]{7-9}
\fcode{\qcode{18,4,3}}{toricdir:c3785f6e} & 7 & 5 & \fcode{\qcode{22,4,4}}{se:0e9b93fb} & 8 & 8 & \fcode[stab_18_2_5]{\qcode{18,2,5}}{ml:7274f0a8} & 8 & 8 & \fcode{\qcode{64,8,8}}{ar:0c41911e} & 14 & 13 \\
\cmidrule[0.4pt]{7-9}
\fcode{\qcode{20,4,3}}{se:2b3dfc02} & 6 & 5 & \fcode{\qcode{22,4,4}}{se:71bb36b3} & 8 & 8 & \fcode{\qcode{77,34,5}}{cons:m22gauged} & 22 & 11 & \fcode{\qcode{64,8,8}}{dd422:0bcf147c} & 16 & 7 \\
\cmidrule[1.1pt]{7-9}
\fcode{\qcode{20,4,3}}{se:a0a0172d} & 6 & 6 & \fcode{\qcode{24,4,4}}{qecdb:4192d9c0} & 6 & 8 & \fcode[stab_20_2_6]{\qcode{20,2,6}}{ml:7020a6c3} & 8 & 8 & \fcode{\qcode{64,8,8}}{ar:2c1b59ac} & 16 & 16 \\
\fcode{\qcode{20,4,3}}{se:e7e4d788} & 6 & 6 & \fcode{\qcode{24,4,4}}{ml:7b49af97} & 8 & 7 & \fcode{\qcode{21,2,6}}{qecdb:52f802ba} & 8 & 8 & \fcode{\qcode{64,8,8}}{ar:0e9e0689} & 16 & 16 \\
\fcode{\qcode{20,4,3}}{se:6ea574c6} & 8 & 5 & \fcode{\qcode{24,4,4}}{ml:04df81f0} & 8 & 8 & \fcode{\qcode{22,2,6}}{se:a00c8506} & 8 & 8 & \fcode{\qcode{72,8,8}}{dd422:f4e61519} & 12 & 8 \\
\fcode{\qcode{20,4,3}}{se:880dc7ea} & 8 & 5 & \fcode{\qcode{24,4,4}}{ml:4083158e} & 8 & 8 & \fcode[stab_22_2_6]{\qcode{22,2,6}}{ml:f660ad96} & 8 & 8 & \fcode{\qcode{72,8,8}}{dd422:f9925f97} & 16 & 8 \\
\fcode{\qcode{20,4,3}}{se:ab724d7d} & 8 & 6 & \fcode{\qcode{24,4,4}}{ml:6d309597} & 8 & 8 & \fcode{\qcode{28,2,6}}{dd422:a44384e0} & 8 & 6 & \fcode{\qcode{80,8,8}}{dd422:5c0525b8} & 12 & 9 \\
\cmidrule[0.4pt]{1-3}
\fcode[stab_15_6_3]{\qcode{15,6,3}}{cons:blk45k18} & 8 & 7 & \fcode{\qcode{24,4,4}}{ml:7b38cf82} & 8 & 8 & \fcode{\qcode{28,2,6}}{ml:13bbbb81} & 8 & 6 & \fcode{\qcode{80,8,8}}{dd422:e4335f61} & 16 & 9 \\
\fcode{\qcode{18,6,3}}{se:61280883} & 8 & 5 & \fcode{\qcode{24,4,4}}{ml:b5b4e98e} & 8 & 8 & \fcode{\qcode{28,2,6}}{ml:1fb1ffb0} & 8 & 6 & \fcode{\qcode{80,8,8}}{dd422:9c1f0fe8} & 16 & 10 \\
\cmidrule[0.4pt]{4-6}
\fcode{\qcode{18,6,3}}{se:9f839398} & 8 & 5 & \fcode[copycup_16_6_4]{\qcode{16,6,4}}{copycup:66fe2995} & 8 & 8 & \fcode{\qcode{28,2,6}}{ml:4e04ab2c} & 8 & 6 & \fcode{\qcode{120,8,8}}{ar:02b8c4d8} & 12 & 20 \\
\fcode{\qcode{18,6,3}}{ml:cdfc0fae} & 8 & 8 & \fcode[stab_16_6_4]{\qcode{16,6,4}}{ml:e1c9cb94} & 8 & 8 & \fcode{\qcode{28,2,6}}{ml:563311cc} & 8 & 6 & \fcode{\qcode{120,8,8}}{ar:e2c52937} & 22 & 33 \\
\fcode{\qcode{22,6,3}}{se:723e2e5d} & 8 & 6 & \fcode{\qcode{22,6,4}}{se:18172690} & 8 & 6 & \fcode{\qcode{28,2,6}}{ml:5d2b3783} & 8 & 6 & \fcode{\qcode{120,8,8}}{ar:66dc0fca} & 24 & 33 \\
\cmidrule[0.4pt]{10-12}
\fcode{\qcode{22,6,3}}{se:e304e3f6} & 8 & 6 & \fcode{\qcode{22,6,4}}{se:dbd8e9e5} & 10 & 6 & \fcode{\qcode{28,2,6}}{ml:7332ade6} & 8 & 6 & \fcode{\qcode{64,12,8}}{dd422:2dd3762b} & 16 & 8 \\
\cmidrule[0.4pt]{1-3}
\cmidrule[0.4pt]{10-12}
\fcode{\qcode{20,8,3}}{se:f61498ec} & 8 & 6 & \fcode{\qcode{24,6,4}}{ml:bbd153de} & 8 & 8 & \fcode{\qcode{28,2,6}}{ml:b47914e7} & 8 & 6 & \fcode{\qcode{56,16,8}}{ar:a49b9550} & 16 & 13 \\
\cmidrule[0.4pt]{4-6}
\cmidrule[0.4pt]{10-12}
\fcode{\qcode{20,8,3}}{se:3a6c1a95} & 10 & 6 & \fcode[stab_20_8_4]{\qcode{20,8,4}}{se:55a6a221} & 10 & 5 & \fcode{\qcode{28,2,6}}{ml:f2f905d2} & 8 & 6 & \fcode{\qcode{64,18,8}}{ar:68c72155} & 16 & 13 \\
\fcode{\qcode{20,8,3}}{se:fa9a8d51} & 10 & 6 & \fcode[stab_22_8_4]{\qcode{22,8,4}}{se:a5e71e83} & 10 & 6 & \fcode{\qcode{28,2,6}}{ml:20b1eb62} & 8 & 7 & \fcode{\qcode{64,18,8}}{ar:b866d0a9} & 18 & 13 \\
\cmidrule[1.1pt]{1-3}
\cmidrule[0.4pt]{10-12}
\fcode[carbon]{\qcode{12,2,4}}{eczoo:9a9707b2} & 6 & 6 & \fcode{\qcode{24,8,4}}{ml:6457a562} & 8 & 8 & \fcode{\qcode{28,2,6}}{ml:4063d985} & 8 & 7 & \fcode{\qcode{64,20,8}}{cons:qrm6} & 16 & 12 \\
\cmidrule[0.4pt]{10-12}
\fcode{\qcode{16,2,4}}{2bga:e4db8568} & 4 & 4 & \fcode{\qcode{24,8,4}}{ml:e35b5235} & 12 & 4 & \fcode{\qcode{28,2,6}}{ml:485ca12d} & 8 & 7 & \fcode{\qcode{120,32,8}}{ar:3ca3253d} & 20 & 21 \\
\cmidrule[0.4pt]{10-12}
\fcode{\qcode{16,2,4}}{se:94f20715} & 6 & 6 & \fcode{\qcode{24,8,4}}{ml:88683fa5} & 12 & 5 & \fcode{\qcode{28,2,6}}{ml:d3123eb0} & 8 & 7 & \fcode[quadric_tower]{\qcode{120,48,8}}{cons:quadric4} & 44 & 46 \\
\cmidrule[0.4pt]{10-12}
\fcode{\qcode{18,2,4}}{se:3a3c1311} & 6 & 8 & \fcode{\qcode{24,8,4}}{ml:78f97bd9} & 12 & 6 & \fcode{\qcode{28,2,6}}{ml:eab6f609} & 8 & 7 & \fcode{\qcode{100,56,8}}{ar:81c359eb} & 34 & 17 \\
\cmidrule[1.1pt]{10-12}
\fcode{\qcode{20,2,4}}{se:7642d51a} & 6 & 8 & \fcode{\qcode{28,8,4}}{ml:9e50d5c5} & 12 & 5 & \fcode{\qcode{28,2,6}}{ml:5e6b63d8} & 8 & 8 & \fcode{\qcode{54,2,10}}{cons:mcr54} & 16 & 7 \\
\cmidrule[1.1pt]{10-12}
\fcode{\qcode{22,2,4}}{se:105def7f} & 6 & 6 & \fcode{\qcode{28,8,4}}{ml:55b90f76} & 12 & 6 & \fcode{\qcode{28,2,6}}{ml:d922fdad} & 8 & 8 & \fcode{\qcode{80,4,12}}{dd422:c9a71361} & 16 & 9 \\
\cmidrule[0.4pt]{10-12}
\fcode{\qcode{22,2,4}}{se:17f1e257} & 6 & 7 & \fcode{\qcode{32,8,4}}{qecdb:2827eae6} & 8 & 8 & \fcode{\qcode{28,2,6}}{ml:ce1ef685} & 12 & 5 & \fcode{\qcode{112,6,12}}{cons:concat112} & 24 & 11 \\
\cmidrule[0.4pt]{4-6}
\cmidrule[1.1pt]{10-12}
\fcode{\qcode{22,2,4}}{se:34fb3f41} & 6 & 8 & \fcode{\qcode{32,12,4}}{dd422:bc8a14c4} & 16 & 4 & \fcode{\qcode{28,2,6}}{ml:60ea3b05} & 12 & 6 & \fcode{\qcode{66,2,13}}{cons:mcr66} & 12 & 17 \\
\cmidrule[0.4pt]{4-6}
\cmidrule[0.4pt]{7-9}
\fcode{\qcode{22,2,4}}{se:fcb10eac} & 6 & 8 & \fcode[quadric_tower]{\qcode{28,14,4}}{cons:quadric3} & 12 & 6 & \fcode{\qcode{32,4,6}}{ml:b248c07a} & 8 & 7 &  & &  \\
\cmidrule[0.4pt]{4-6}
\fcode{\qcode{22,2,4}}{se:d8b634ac} & 8 & 6 & \fcode{\qcode{36,16,4}}{ar:49f03842} & 10 & 7 & \fcode{\qcode{32,4,6}}{ml:237fee9e} & 8 & 8 &  & &  \\
\fcode{\qcode{22,2,4}}{se:2d2eeecc} & 8 & 8 & \fcode{\qcode{40,16,4}}{cons:grid4x10} & 12 & 6 & \fcode{\qcode{32,4,6}}{ml:7bc97ccf} & 8 & 8 &  & &  \\
\end{tabular}
\end{ruledtabular}
\end{table*}

Table~\ref{tab:full-depth1} collects the \Nfullcodes\ full codes with at most 120 qubits that a sweep of several code databases and of the rank-3 permutation-group census returned. Appendix~\ref{app:codes} expands the shorthand
that names the database of each entry. Every entry is indecomposable, so no entry is a
direct sum of smaller codes. The table prints the codes with $n\le120$.
Lengths there run from $n=10$ to $n=120$, and logical counts from $k=2$ to
$k=90$. Distance $4$ dominates the table with $72$ entries, distances $3$,
$6$, and $8$ follow with $33$, $33$, and $29$, and the remainder are spread
over distances $5$, $10$, $12$, and $13$.
The best distance in the table is $13$, at the $\qcode{66,2,13}$ entry, and
the best rate is $3/4$, at $\qcode{120,90,4}$.  The machine-readable records
extend both.  They hold thirty-five $\qcode{256,16}$ codes of
distance up to $16$ and the $\qcode{276,230,4}$ code at rate $5/6$. Most entries are self-dual, $\CX=\CZ$. Many also lie
in the low-density regime of Section~\ref{sec:ldpcdepthone}, with check weight
at most $8$ and qubit degree at most $8$. Within that regime the
best distance and the largest logical count both drop to $8$. A high
distance and a high rate are thus both compatible with fullness, but not yet at
low check weight.

Full codes can produce every logical Clifford gate via several depth-one code-preserving physical gates, but it is additionally interesting to see whether they can compile the canonical ones --- logical $\overline{\sqrt{Z}}$, $\overline{\sqrt{X}}$, $\overline{\mathrm{CZ}}$, and $\overline{\mathrm{CNOT}}$ gates --- with only \textit{one} depth-one physical gate.
For a given basis of the logical Paulis, we say the code
\emph{addresses} a logical gate at depth one when a single layer acts as exactly that
gate, and as the identity on every other logical qubit. 
We discuss four entries of Table~\ref{tab:full-depth1} in this context.
See Ref.~\cite[Def.~5]{holmes2026quantum} for a more focused study of such gate sets, where they are called ``instruction set
architectures''.

The $\qcode{10,2,3}$ rotated toric code \texttt{eczoo:dd1f5a7f} is the smallest entry with a
\emph{complete} addressable set. Both $\overline{\sqrt{Z}}_{i}$, both
$\overline{\sqrt{X}}_{i}$, the single $\overline{\mathrm{CZ}}$, and its $X$-basis mirror are each realized in one depth-one physical circuit. 
Four of the six circuits use three $\mathrm{CZ}$ physical gates and
five single-qubit gates; the other two use five $\mathrm{CZ}$ gates.

The $\qcode{14,3,3}$ \texttt{eczoo:880eb7d7}  phantom/constant-excitation code admits a basis for which all three of its single-qubit logical $\overline{\sqrt{X}}$ gates, all three $\overline{\mathrm{CZ}}_{ij}$ gates, and all three of their $X$-basis mirrors are all addressable.

The $\qcode{16,6,4}$ \texttt{ml:e1c9cb94} ``tesseract'' quantum Reed--Muller code
$\mathrm{QRM}(4)$ addresses pairs but not single qubits. 
Twelve of its
$15$ logical pairs carry $\overline{\mathrm{CZ}}$ \emph{and} its dual $X$-basis gate.
Every one of the $24$ gates is realized by a circuit of exactly four
$\mathrm{CZ}$ physical gates and eight single-qubit gates.
The three logical-qubit pairs that carry neither gate are $(0,3)$, $(1,4)$ and $(2,5)$, and the remaining pairs can be associated with the edges of an octahedron. 
Reference~\cite{tansuwannont2026full}
obtains the same pattern of addressable pairs for quantum Reed--Muller codes
from transversal and fold-transversal circuits.

The $\qcode{15,4,3}$ code \texttt{cons:blk45k12} addresses every $X$-basis ($X$-diagonal) gate. All four $\overline{\sqrt{X}}_{i}$ and all six
$\overline{\mathrm{CZ}^{H}}_{ij}$ $X$-diagonal circuits are realized via depth-one physical circuits, forming the largest
complete addressable set in the table.

\subsection{Families of full codes}
\label{sec:depth1twolocal-families}

We recap existing families containing full codes and conjecture several new ones.
We describe each one so that a reader can rebuild it and state how far fullness is established.

\emph{Middle Reed--Muller codes.} 
Define the physical qubits to be in correspondence with the $2^{m}$ points of $\Ftwo^{m}$ for even $m\ge2$. 
A binary word is then the table of values of a Boolean
function. 
Let both classical codes be the Reed--Muller code of order
$m/2-1$, the evaluations of all polynomials of degree at most $m/2-1$. This gives
\begin{equation}
  \mathrm{QRM}(m)=\qcodeb{\,2^{m},\ \tbinom{m}{m/2},\ 2^{m/2}\,},
  \label{eq:fam-qrm}
\end{equation}
a self-dual doubly-even code with the affine group $\mathrm{AGL}(m,2)$ as a
symmetry. Every member is proven full in Ref.~\cite{tansuwannont2026full}. Table~\ref{tab:full-depth1} mentions
$\mathrm{QRM}(4)=\qcode{16,6,4}$ and $\mathrm{QRM}(6)=\qcode{64,20,8}$, while $\mathrm{QRM}(2)=\qcode{4,2,2}$
falls below our distance cut.
The $\qcode{112,6,12}$ code in the table is the concatenation lift: it places
$\mathrm{QRM}(4)$ inside the even $[7,3,4]$ subcode of the Hamming code, which
triples the distance and carries the whole outer gate group along.
The inner code here is the Steane $\qcode{7,1,3}$ code, whose $X$- and
$Z$-stabilizer spaces are the even $[7,3,4]$ Hamming subcode; compare the
layer-by-layer concatenation theorem of
Holmes~\cite[Thm.~3]{holmes2026quantum}.

\emph{Quantum logic codes.}
Holmes gives 
infinite families of codes that carry the full logical Clifford group with a
complete set of addressable gates \cite[Def.~5 and Prop.~8]{holmes2026quantum}.
Their codes grow from a
``core'' $\qcode{n_{0},2,d_{0}}$ code by two commuting operations: 
\textit{tiling} repeats the core
$r$ times, and \textit{concatenation} puts the code inside a self-dual doubly-even
$\qcode{n_{i},1,d_{i}}$ code. 
The result is the full family
\begin{equation}
  \qcodeb{\,n_{0}\,r\,n_{i}^{\ell},\;2r,\;d_{0}\,d_{i}^{\ell}\,}.
  \label{eq:qlc}
\end{equation}
The three cores used are $\qcode{4,2,2}$, $\qcode{18,2,5}$, and $\qcode{20,2,6}$.
The latter two cores are in Table~\ref{tab:full-depth1}, while the $\qcode{4,2,2}$ core falls below our distance cut.
Many of the codes in our table can be used as new cores for the quantum logic construction.
The tiled quantum logic codes are decomposable, so they are omitted from the table.

\emph{Subsystem hypergraph product simplex codes.}
The SHYPS family~\cite{malcolm2025computing} is the subsystem hypergraph
product~\cite{li2020numerical} of two copies of the $[2^{r}-1,r,2^{r-1}]$
simplex code.
The physical qubits form a $(2^{r}-1)\times(2^{r}-1)$ array, and the code
admits gauge generators of weight three.
The parameters for $r\ge3$ are
\begin{equation}
  \qcodeb{\,(2^{r}-1)^{2},\ r^{2},\ 2^{r-1}\,},
  \qquad
  g_{r}=(2^{r}-1-r)^{2},
  \label{eq:fam-shyps}
\end{equation}
where $g_{r}$ counts the gauge qubits.
A depth-one layer preserves a subsystem code when it preserves the gauge group, with its logical action read off modulo the gauge qubits.
Every member of the family is full.
The simplex code has automorphism group $\GLg{r}$.
Composing each automorphism with the array transpose gives an $XZ$-duality
(Sec.~\ref{sec:classes-weyl}).
Reference~\cite{malcolm2025computing} proves that these dualities induce a
depth-one generating set of all logical diagonal gates, in both bases.
It also proves that such a set yields the entire logical action
$\Sp{2r^{2}}$ (cf.\ Theorem~\ref{thm:sheargen}).

\emph{Symplectic double codes.}
The symplectic double mapping~\cite{bravyi2010majorana,kovalev2013hyperbicycle,liu2024subsystem,burton2024genons}
sends a stabilizer code $C$ to a CSS code
$\mathfrak{D}(C)$ on twice as many qubits. 
This is a genuinely new code when $C$ is not CSS, while for a CSS ``seed'' it decomposes to $C\oplus C$. 
The codes tabulated here are obtained from a
second step~\cite{bravyi2010majorana,berthusen2025concatenated} that removes that decomposability: each pair of qubits related by the
$XZ$-duality $\tau$ of $\mathfrak{D}(C)$ is encoded
into the two logical qubits of a single $\qcode{4,2,2}$ block. 
The composite map,
$\qcode{n,k,d}\mapsto\qcode{4n,2k,2d}$, yields a generally indecomposable self-dual code.

Reference~\cite{berthusen2025concatenated} studies several concatenated symplectic
doubles and verifies that the full logical Clifford group is reached two-fold transversally on the
$\qcode{16,4,4}$, $\qcode{20,2,6}$ (cf.~\cite{koh2026phantom,holmes2026quantum}), $\qcode{24,8,4}$, $\qcode{36,6,6}$ and
$\qcode{40,8,6}$ codes.
We apply the map to most of their seed codes as well as our full codes at \(n \leq 20\), with several concatenated doublings yielding longer full codes as shown in the table.
We anticipate that applying the map to full codes at larger \(n\) will generate more full codes.

We observe that fullness of a concatenated symplectic double is not always inherited from fullness of the seed code. 
Of the doubles we
certified, roughly half reach $\Sp{2k}$ at depth one and half do not.
We postulate that
such selective fullness is determined by how the seed code's automorphism group (see Sec.~\ref{sec:permutations}) lifts to the doubled logical space.

\emph{Cut-complement codes.} 
To construct this family, associate the physical qubits with the $\binom{m}{2}$ edges of the complete graph
on $m=2^{r}$ vertices for $r\ge3$.
For a vertex set $V$ of even
size, let $\mathrm{cut}(V)$ be the set of edges with exactly one end in $V$. Let
both classical codes be the span of these even cuts together with the all-ones
word. Each cut has weight $|V|(m-|V|)$, which is a multiple of four, so the code
is doubly even and self-dual. The parameters are
\begin{equation}
  \qcodeb{\,\tfrac{m(m-1)}{2},\ \tfrac{(m-1)(m-4)}{2},\ 4\,},
  \qquad m=2^{r},
  \label{eq:fam-johnson}
\end{equation}
where the distance is four because the lightest undetectable pattern is a
four-cycle. 
The rate tends to one while the distance stays at four.
We conjecture that this family is full for every $r\ge3$.
Table~\ref{tab:full-depth1} lists
$\qcode{28,14,4}$ at $r=3$, listed there under its quadric-tower name, and
$\qcode{120,90,4}$ at $r=4$.

\emph{Bipartite grid codes.} To construct this family, take the $ab$ cells of an $a\times b$ grid as the
qubits.
For a row index set $I$ and a column index set $J$ of equal parity, $|I|\equiv|J|
\bmod 2$, let the corresponding word be the symmetric difference of the chosen
rows and columns. Let both classical codes be the span of those words, which is
also the span of the single sums $\mathrm{row}_{i}+\mathrm{col}_{j}$. The parity
condition is essential, and the plain even-cut space gives a different code that
is not full. When $a$ and $b$ are even and $a+b\equiv2\bmod4$, the code is
doubly even and self-dual, with
\begin{equation}
  \qcodeb{\,ab,\ (a-2)(b-2),\ 4\,} .\qquad\text{(grid codes)}
  \label{eq:fam-grid}
\end{equation}
We conjecture that every such code is full. All fourteen members with
$ab\le120$ are certified and appear in Table~\ref{tab:full-depth1}: the
thirteen \texttt{cons:grid}$a$\texttt{x}$b$ entries, from $\qcode{40,16,4}$ at
$(4,10)$ to $\qcode{120,80,4}$ at $(10,12)$, and the
$\qcode{24,8,4}$~\texttt{ml:6457a562} entry, which is the grid code with
$(a,b)=(4,6)$. 

\emph{Quadric tower.} 
Fix
an elliptic quadratic form $Q_{0}$ on $\Ftwo^{2j}$, and take its zero set as the
qubits, of size $2^{2j-1}-2^{j-1}$. For both classical codes take the
Reed--Muller code of order $j-2$ in $2j$ variables, punctured to that zero set,
\begin{equation}
  \CX=\CZ=\mathrm{RM}(j-2,2j)\big|_{Q_{0}^{-1}(0)} .
  \label{eq:fam-quadric-C}
\end{equation}
The
resulting family is
\begin{equation}
  QT_{j}=\qcodeb{\,2^{2j-1}-2^{j-1},\ k_{j},\ 2^{j-1}\,},
  \label{eq:fam-quadric}
\end{equation}
for $j=3,4,5,6,\dots$ and $k_{j}=14,48,166,584,\dots$.

We verified that the first two members are full, shown in Table~\ref{tab:full-depth1}
as $QT_{3}=\qcode{28,14,4}$ and $QT_{4}=\qcode{120,48,8}$.
We do not believe the rest of the family is full, but note this code because its permutations form the minimal-permutation-degree representation of the $\Sp{2j}$ group.
Letting $B$ be the polar
form of $Q_{0}$, the map $a\mapsto Q_{0}+B(a,\cdot)$ then identifies the qubits
of $QT_{j}$ with one Arf class of quadratic forms on $\Ftwo^{2j}$, and that
class is the smallest set on which $\Sp{2j}$ can act faithfully~\cite{arf1941quadratische,cooperstein1978minimal}.
Its minimal
permutation degree is exactly $2^{j-1}(2^{j}-1)=2^{2j-1}-2^{j-1}$, the length
in Eq.~\eqref{eq:fam-quadric}. 
This saturates a bound discussed in Ref.~\cite{morris2026constraints}.

\emph{Rank-3 group codes.}  The \texttt{ar} entries of
Table~\ref{tab:full-depth1} come from a census of self-dual even CSS codes
that are invariant under a rank-3 permutation group.  The qubits carry a
transitive group action whose point stabilizer has three orbits, and the
classical code is an invariant self-orthogonal submodule of the permutation
module.  Some of these codes are doubly even, and the rest are singly even.

This class is not disjoint from the families above.  It contains the middle
Reed--Muller codes, the cut-complement codes, and the quadric tower, and the
census rediscovered members of all three.  The cut-complement code at $m$ is
the rank-3 code of $\Perm{m}$ acting on pairs.  The group $\Sp{2j}$ acts
2-transitively on the qubits of $QT_{j}$, and its subgroup $\mathrm{O}^{+}(2j)$
acts with rank three.  The bipartite grid codes lie
outside the class, because their automorphism group $\Perm{a}\times\Perm{b}$
has rank four on the cells.

A sweep of all such groups of degree at most $300$ contributes $69$ further
full codes.  Every full code we found with $n>120$ belongs to this class.
Examples include the thirty-five $\qcode{256,16}$ codes and $\qcode{276,230,4}$,
which the census found through $M_{24}$ acting on pairs.  That code is
invariant under all of $\Perm{24}$.  It is the cut-complement construction of
Eq.~\eqref{eq:fam-johnson} at $m=24$, outside the range $m=2^{r}$ of that family.
Gauge fixing gives a smaller example from a sporadic group.
The rank-3 code of $M_{22}$ on $77$ points is not full, but the
$\qcode{77,34,5}$ code obtained from it by gauge fixing is.

%% file: sec_ldpcdepthone.tex
\begin{table*}[t]
\caption{\label{tab:ldpc-depth1}%
A list of \Nldpccodes\ indecomposable codes, most of them instances of
low-density parity-check families, that
admit depth-one two-local code-preserving circuits without realizing
the full logical Clifford group.  For each we give the check weight and
qubit degree of a generating set of the stabilizer group,
and the order of the logical image 
of the group generated by the depth-one two-local circuits.  
That order is a lower
bound because our sampling of matchings is not exhaustive.  Codes are ordered by distance, then by $k$, then by the image
order; thick rules separate distance groups and thin rules separate $k$
groups within a distance. Each code is named by its
parameters followed by its source database and the leading octet of its
uuid or tag; Appendix~\ref{app:codes} expands the database names and discusses the algorithms used.}
\begin{ruledtabular}
\footnotesize
\setlength{\tabcolsep}{3pt}
\setlength{\aboverulesep}{0pt}\setlength{\belowrulesep}{0pt}
\renewcommand{\arraystretch}{1.25}
\begin{tabular}{lrrr@{\hspace{1.5em}}lrrr}
Code & Wt & Deg & \#Gates & Code & Wt & Deg & \#Gates \\
\hline
\lcode{\qcode{18,6,3}}{se:61280883} & 8 & 6 & 4423680 & \lcode{\qcode{36,4,6}}{coset2bga:c8dc53eb} & 6 & 6 & 60 \\
\lcode{\qcode{18,6,3}}{se:46ef7ce7} & 6 & 8 & 9.06e9 & \lcode{\qcode{54,4,6}}{kasai:e4d756a1} & 6 & 6 & 60 \\
\cmidrule[0.4pt]{1-4}
\lcode{\qcode{20,8,3}}{se:fa9a8d51} & 10 & 6 & 1.81e10 & \lcode{\qcode{108,4,6}}{toricdir:e0a5b7a6} & 5 & 5 & 120 \\
\lcode{\qcode{24,8,3}}{2bga:60fab289} & 8 & 8 & 3.40e25 & \lcode{\qcode{36,4,6}}{coset2bga:2f346b16} & 8 & 8 & 1536 \\
\cmidrule[0.4pt]{5-8}
\lcode{\qcode{24,8,3}}{cc:dbf27139} & 8 & 8 & 1.50e26 & \lcode{\qcode{40,6,6}}{cons:sdbb40} & 8 & 8 & 6144 \\
\cmidrule[0.4pt]{1-4}\cmidrule[0.4pt]{5-8}
\lcode{\qcode{36,12,3}}{2bga:b6c7c9b2} & 8 & 8 & 35389440 & \lcode{\qcode{42,8,6}}{coset2bga:a7b06b2a} & 8 & 8 & 3024 \\
\cmidrule[1.1pt]{1-4}
\lcode{\qcode{32,2,4}}{mm:f63d251a} & 4 & 4 & 48 & \lcode{\qcode{48,8,6}}{2bga:baa0c07b} & 8 & 8 & 3072 \\
\cmidrule[0.4pt]{1-4}
\lcode{\qcode{72,4,4}}{toricdir:28cebc04} & 5 & 5 & 720 & \lcode{\qcode{100,8,6}}{toricdir:9920a121} & 7 & 7 & 8160 \\
\lcode{\qcode{36,4,4}}{toricdir:f04e33b3} & 5 & 5 & 1440 & \lcode{\qcode{100,8,6}}{toricdir:dadba010} & 7 & 7 & 8160 \\
\lcode{\qcode{108,4,4}}{toricdir:1a18b6dd} & 5 & 5 & 1440 & \lcode{\qcode{48,8,6}}{coset2bga:3711a944} & 6 & 6 & 737280 \\
\lcode{\qcode{20,4,4}}{ml:f593c885} & 8 & 8 & 8847360 & \lcode{\qcode{48,8,6}}{2bga:83cff0f7} & 8 & 8 & 2.01e8 \\
\cmidrule[0.4pt]{1-4}
\lcode{\qcode{32,8,4}}{2bga:606c7fbd} & 8 & 8 & 25165824 & \lcode{\qcode{54,8,6}}{coset2bga:077dcc3e} & 6 & 6 & 1.97e9 \\
\cmidrule[1.1pt]{5-8}
\lcode{\qcode{32,8,4}}{2bga:2ee93e2b} & 8 & 8 & 7.25e10 & \lcode{\qcode{56,8,7}}{cc:5bb93516} & 8 & 8 & 48 \\
\lcode{\qcode{54,8,4}}{kasai:ec88c3e3} & 6 & 6 & 5.44e11 & \lcode{\qcode{56,8,7}}{2bga:eb2280f4} & 8 & 8 & 3072 \\
\lcode{\qcode{24,8,4}}{ml:50c0d47a} & 8 & 8 & 7.60e16 & \lcode{\qcode{56,8,7}}{2bga:8d5bc2dd} & 8 & 8 & 49152 \\
\cmidrule[0.4pt]{1-4}\cmidrule[0.4pt]{5-8}
\lcode{\qcode{48,12,4}}{2bga:c5f2a50c} & 8 & 8 & 6.44e42 & \lcode{\qcode{62,10,7}}{coset2bga:3fa86309} & 7 & 7 & 32736 \\
\cmidrule[0.4pt]{1-4}\cmidrule[0.4pt]{5-8}
\lcode{\qcode{64,16,4}}{2bga:502aeab6} & 8 & 8 & 1.29e10 & \lcode{\qcode{62,12,7}}{coset2bga:10961d90} & 8 & 8 & 196416 \\
\cmidrule[0.4pt]{1-4}\cmidrule[1.1pt]{5-8}
\lcode{\qcode{56,18,4}}{kasai:b824a37f} & 8 & 6 & 1524096 & \lcode{\qcode{40,4,8}}{2bga:25a18c5e} & 8 & 8 & 48 \\
\cmidrule[1.1pt]{1-4}
\lcode{\qcode{24,4,5}}{2bga:c64fd1aa} & 8 & 7 & 48 & \lcode{\qcode{48,4,8}}{coset2bga:3e15d372} & 6 & 6 & 60 \\
\lcode{\qcode{72,4,5}}{toricdir:66ff67b2} & 5 & 5 & 120 & \lcode{\qcode{60,4,8}}{kasai:98c88188} & 6 & 6 & 60 \\
\cmidrule[0.4pt]{5-8}
\lcode{\qcode{54,4,5}}{toricdir:62ccead5} & 5 & 5 & 720 & \lcode{\qcode{56,6,8}}{coset2bga:75935f40} & 6 & 6 & 504 \\
\cmidrule[0.4pt]{1-4}\cmidrule[0.4pt]{5-8}
\lcode{\qcode{40,8,5}}{2bga:b7c13b5b} & 8 & 8 & 3072 & \lcode{\qcode{64,8,8}}{2bga:b79b8ef6} & 8 & 8 & 196608 \\
\lcode{\qcode{40,8,5}}{coset2bga:396a167c} & 8 & 8 & 3072 & \lcode{\qcode{64,8,8}}{cons:sdbb64} & 8 & 9 & 3.22e9 \\
\cmidrule[1.1pt]{5-8}
\lcode{\qcode{40,8,5}}{2bga:58149515} & 8 & 8 & 49152 & \lcode{\qcode{54,6,9}}{2bga:4fc3bad7} & 8 & 8 & 360 \\
\cmidrule[1.1pt]{5-8}
\lcode{\qcode{40,8,5}}{cc:100b4262} & 8 & 8 & 786432 & \lcode{\qcode{56,4,10}}{2bga:38f3aaf6} & 8 & 8 & 48 \\
\cmidrule[1.1pt]{5-8}
\lcode{\qcode{50,8,5}}{toricdir:69165196} & 7 & 7 & 979200 & \lcode{\qcode{120,8,12}}{cons:sdbb120} & 8 & 8 & 34560 \\
\cmidrule[0.4pt]{1-4}\cmidrule[0.4pt]{5-8}
\lcode{\qcode{42,10,5}}{coset2bga:5d49d71b} & 7 & 7 & 30240 & \lcode[gross]{\qcode{144,12,12}}{cons:gross144} & 6 & 6 & 460800 \\
\cmidrule[0.4pt]{1-4}\cmidrule[1.1pt]{5-8}
\lcode{\qcode{60,12,5}}{2bga:6bce1312} & 8 & 8 & 17280 & \lcode{\qcode{152,6,16}}{cons:sdbb152} & 8 & 8 & 6144 \\
\cmidrule[1.1pt]{1-4}\cmidrule[0.4pt]{5-8}
\lcode{\qcode{22,2,6}}{ml:f660ad96} & 8 & 8 & 48 & \lcode{\qcode{160,8,16}}{cons:sdbb160} & 8 & 8 & 49152 \\
\end{tabular}
\end{ruledtabular}
\end{table*}

\section{Other codes with two-fold transversal gates}
\label{sec:ldpcdepthone}

This section provides lower bounds on the logical image of the two-fold transversal group for instances of several QLDPC code families, together with a few further database entries whose image falls short of the full logical Clifford group. 
We write
\begin{equation}
  \#\mathrm{Gates}=\bigl|\lact(\Ndep)\bigr|\qquad\text{(logical image size)}
  \label{eq:ldpc-gates}
\end{equation}
for the number of distinct logical Clifford actions that stacks of depth-one
two-local layers produce. The largest value this can take is the order of the
logical Clifford group,
\begin{equation}
  \bigl|\Sp{2k}\bigr|=2^{k^{2}}\prod_{i=1}^{k}\bigl(4^{i}-1\bigr),
  \label{eq:ldpc-spmax}
\end{equation}
which the full codes of Section~\ref{sec:depth1twolocal} attain.
Details of our numerical techniques are described in Appendix~\ref{app:codes}.

The lower bound on the logical reach of $\Ndep$, quantified by the above number of gates, varies wildly across
these instances. 
Table~\ref{tab:ldpc-depth1} lists a few codes with at least $48$ logical gates. This is not meant to be exhaustive, as we also found codes for which our sampling yielded no depth-one gates at all. 
The orders presented span
$41$ decimal orders of magnitude, from $48$ to about $6\times10^{42}$.

\subsection{The gross code}
\label{sec:ldpcdepthone-gross}

The bivariate-bicycle $\qcode{144,12,12}$ ``gross'' code~\cite{bravyi2024high} is a prominent instance of a QLDPC family; it is not full based on our estimates, yet it admits a large set of Clifford gates generated by depth-one two-local circuits.
The only gates we were able to find stem from a symmetry of the code. 
The code's $XZ$ duality $\sigma$ that exchanges the
two blocks of the construction is fixed-point free, so its diagonal-circuit pattern is a
perfect matching of $\mathrm{CZ}$ gates. All $72$ translates of $\sigma$ are
valid, for both the $Z$-diagonal and the $X$-diagonal circuits, and $71$ of the translations are permutation
automorphisms and hence CNOT gates.

Together, these generators produce a subgroup of order $460\,800$ inside the logical
image,
\begin{equation}
  \lact(\Ndep)\supseteq \Perm{5}\times\bigl((C_{2})^{6}\!:\!A_{5}\bigr).
  \label{eq:gross-group}
\end{equation}
Above, $A_{n}$ is the alternating group, while $\Perm{n}$ is the symmetric group.
The depth-one $Z$- and $X$-diagonal families each span only a
six-dimensional subspace of the $78$-dimensional space of symmetric
$12\times12$ matrices (corresponding to diagonal circuits of arbitrary depth). 
For
comparison, $|\Sp{24}|\approx1.4\times10^{90}$, about $3\times10^{84}$ times larger than the size of the group in Eq.~\eqref{eq:gross-group}.

The comparison with the automorphism group of
Section~\ref{sec:permutations} is instructive. That group has order
$|\Gaut|=288$ for this code, and its logical image has order
$144$~\cite{sayginel2025faulttolerant,sayginel2024autqec}. Both are subgroups of~\eqref{eq:gross-group}, because the two-fold transversal group includes any
transversal \(Z,X\)-diagonal circuits as well as, in the case of this code, the
permutations.

\subsection{Codes with addressable gates}
\label{sec:ldpcdepthone-addressable}

We highlight three entries of Table~\ref{tab:ldpc-depth1} in the context of addressable gates.

The $\qcode{24,8,4}$ code \texttt{ml:50c0d47a} addresses nine of its $28$ logical
pairs with $\overline{\mathrm{CZ}}$, and nine further pairs with the dual $X$-basis
gate. 
Ten of those $18$ layers are realized by a circuit with four $\mathrm{CZ}$ physical gates and
eight single-qubit gates. 
That is the same pattern as every gate of
$\mathrm{QRM}(4)$ in Section~\ref{sec:depth1twolocal}, reproduced here at check
weight $8$ and qubit degree $8$. 

The $\qcode{22,2,6}$ code \texttt{ml:f660ad96} has a small logical image with
$48$ gates, but addresses four
of its six targets: $\overline{\sqrt{Z}}_{1}$, $\overline{\sqrt{X}}_{0}$,
$\overline{\mathrm{CZ}}_{01}$ and $\overline{\mathrm{CZ}^{H}}_{01}$. 
All four gates are realized by depth-one two-local circuits on the \emph{same} seven-pair matching $\Mt$
and differ only in their single-qubit physical gates. 

In contrast, the $\qcode{24,8,3}$ code \texttt{2bga:60fab289} has a large depth-one logical image of order $3.4\times10^{25}$,
and its automorphism group (see Sec.~\ref{sec:permutations}) has order $48\,384$.
Nonetheless, we were not able to find any addressable gates in any Pauli basis.

%% file: sec_permutations.tex
\section{Including permutations}
\label{sec:permutations}

Single-qubit Clifford operations and permutations
preserve the weight of every Pauli operator. They thus do not change the distance of
any code.
These form the group $\Sp{2}\wr\Perm{n}$,
the wreath product of the single-qubit Clifford group with the permutation group.
Those transformations that preserve the code form the code's \emph{automorphism group}
\begin{equation}
  \Gaut=\bigl(\Sp{2}\wr\Perm{n}\bigr)\cap\Nfull .\qquad\text{(automorphism group)}
  \label{eq:aut-gaut}
\end{equation}
The automorphism group is traditionally \textit{the} symmetry group of a code,
and it is relevant to fault tolerance in settings where permutations are simple to implement.
This section places $\Gaut$ among the gate families of the
paper. Appendix~\ref{app:permutations} gives the proof details.

The automorphism group of a CSS code has two key subgroups. Its single-qubit Cliffords are the
transversal gates from a previous section. Its pure permutations are the
symmetries the two
classical codes share~\cite{grassl2013leveraging},
\begin{equation}
\begin{aligned}
  \Sp{2}^{\,n}\cap\Gaut &= \Ntr,\\
  \Perm{n}\cap\Gaut &= \mathrm{Aut}(\CX)\cap\mathrm{Aut}(\CZ).
\end{aligned}
  \label{eq:aut-parts}
\end{equation}
But $\Gaut$ is larger than these two
subgroups put together. It also contains gates in which a permutation compensates a
single-qubit Clifford. 
In other words, each factor alone does not preserve the code, but the
product is code-preserving.

Whether a permutation costs anything depends on the model of computation. Some
architectures allow one to move qubits efficiently~\cite{bluvstein2022transport,bluvstein2024logical}.
There a compensating
permutation is free as well, so adjoining permutations enlarges the set of
available transversal gates at no cost. 
Every permutation is a product of two involutions and hence of just two depth-one layers of disjoint SWAP gates~\cite{alon1994routing}. Those two layers, however, need not individually preserve the code, and depth-one two-local \emph{code-preserving} layers cannot construct some permutations at all (see Figure~\ref{fig:hasse}). The automorphism group and the two-fold transversal group are therefore different objects.

\subsection{Normal form and how to calculate it}
\label{sec:permutations-nf}

Recall the fourth family of Section~\ref{sec:classes}. A partial duality is a Hadamard on
a subset of the qubits followed by a permutation, and these elements form
a subgroup $\PDgp$ of $\Gaut$. The main result here is that Hadamards, together with
the two diagonal families, yield a normal form for the entire automorphism group.
In principle a permutation might be realizable only with diagonal circuits or other
single-qubit gates
as compensating elements. We show that a
Hadamard always suffices.

We show that every automorphism factors as a partial duality followed by a $Z$-diagonal
circuit and an $X$-diagonal circuit,
\begin{equation}
  \Gaut=\PDgp\cdot\Eupo\cdot\Edno\qquad\text{(normal form)}
  \label{eq:aut-nf}
\end{equation}
(Theorem~\ref{thm:autnf}). 
The three factors can be read off the local blocks
$A,B,C,D$ of the gate. 
Writing $(g;\pi)$ for the automorphism that dresses the
permutation $\pi$ with the single-qubit tuple $g$, we have
\begin{equation}
\begin{gathered}
  (g;\pi)=(\Hd{a};\pi)\,\Ush{q}\,\Lsh{p},\\
  a=B\pw C,\quad q=B\pw D,\quad p=A\pw C,
\end{gathered}
  \label{eq:aut-abc}
\end{equation}
and the factorization is unique once we require the two diagonal-circuit supports to be
disjoint. The Hadamard subset $a=B\pw C$ is a single pointwise product of blocks,
and the two remaining diagonal circuits are automatically code-preserving. As with
the transversal normal form of Section~\ref{sec:transversal}, three factors
suffice. But there all three were diagonal circuits, a subset Hadamard being a product of
three of them, whereas here the single
partial-duality factor carries the permutation content.

To create a set of generators for the automorphism group,
one need not carry the partial dualities as a whole family. It
suffices to adjoin to the two diagonal families one partial duality for each generator
of the permutations that actually occur, that is, the permutations preserving the
code up
to single-qubit Cliffords. 
The
lifting lemma behind the normal form (Lemma~\ref{lem:pureH}) is exactly the
guarantee that Hadamards alone can always lift a permutation occurring in $\Gaut$.

We cover two notable edge cases of the above normal form. When the Hadamard subset $a$ is empty, the permutation
preserves each classical code separately, an ordinary symmetry of $\CX$ and of
$\CZ$~\cite{grassl2013leveraging}. When the subset is everything, the
permutation carries $\CX$ onto $\CZ$ and back, a full duality.

The normal form splits the automorphism group into an efficiently obtainable transversal part and an expensive permutation-based part.
The two parameter codes of the transversal group from Sec.~\ref{sec:transversal} are nullspaces found by
Gaussian elimination, and the closed formula $a=B\pw C$ gives the Hadamard subset
lifting any permutation, with no search over single-qubit dressings. 
Everything expensive
sits in the permutations themselves. To determine them is to compute the code's
automorphisms modulo single-qubit Cliffords, a graph canonical-form problem~\cite{cross2025small}.
In this approach, the permutation search remains the computational bottleneck.

%% file: sec_extension.tex
\section{Possible extensions}
\label{sec:extensions}

We identify several natural extensions of the formalism: physical circuit depth optimization, interpretation of the results from the homological perspective, extension to non-CSS qubit stabilizer codes and to stabilizer codes over other alphabets, and a potential definition of a two-local analogue of the automorphism group $\Gaut$.

\subsection{Optimizing depth}
\label{sec:ext-depth}

Our general gate group results outline which circuits preserve a code, but they do not specify which of those
circuits has the lowest circuit volume.
One should further be able to optimize in the defined space to obtain desired circuits with favorable parameters.

The block
decomposition~\eqref{eq:N-structure} determines the free parameters of a circuit with fixed logical action.
The encoded qubits see only the $\Sp{2k}$ factor, so the padding $\Urad$
and the relabeling $\GLg{r}$ can vary freely. 
Each logical Clifford thus can be rendered in $|\Urad|\,|\GLg{r}|$ different ways as a physical circuit~\eqref{eq:N-fibre}.
We can therefore fix a desired logical block and optimize those remaining blocks to obtain minimal-depth representatives.
Specializing to diagonal circuits should simplify the optimization without losing the ability to do arbitrary logical Clifford gates (since the two diagonal families generate the whole code-preserving group, per our result).

We note several related but complementary approaches.
Reference~\cite{bauer2026finding} caps the ansatz gates per qubit and thus bounds the depth, but only over the prescribed ansatz.
Reference~\cite{benhemou2026automated} puts the depth cap into the solver
itself.
Fixing the physical representatives of the logical operators leaves the $2^{r(r+1)/2}$ destabilizer paddings of Ref.~\cite{rengaswamy2018synthesis}, a slice of the ``fiber'' \eqref{eq:N-fibre} over which Ref.~\cite{popov2026optimized} compiles logical Clifford circuits.
Reference~\cite{patra2025targeted} builds targeted logical Clifford circuits for hypergraph product codes while holding the stabilizers fixed pointwise, i.e., setting the relabeling factor $\GLg{r}$ to the identity, and notes that releasing it may shorten the circuits.
Reference~\cite{popov2026optimized} notes that most methods filter for
depth only after they find a particular circuit. We leave open how to combine these methods
with the parametrization above, or with a code-agnostic
optimizer~\cite{bravyi2021templates,bravyi2022optimal}.

\subsection{Homology}
\label{sec:ext-homology}

A CSS code is often presented as the middle term of a three-term chain complex
of binary vector spaces \cite{kitaev1997quantum,bombin2007homological,bravyi2014homological,breuckmann2018thesis,zou2025algebraic}.
One boundary map records the $X$-checks and the other records
the $Z$-checks. In that presentation, the supports of the $Z$-stabilizers are
the boundaries of the complex, and the supports of the $X$-stabilizers are the
coboundaries.
We can then interpret the gate families of Section~\ref{sec:classes} as maps
on this complex.
A valid $Z$-diagonal circuit is a symmetric mapping from coboundaries to boundaries.
A valid $X$-diagonal circuit is a symmetric mapping from boundaries to coboundaries.
A CNOT network is an
invertible map preserving the coboundaries and boundaries individually.
A partial duality exchanges the chain and
cochain descriptions.

Many constructions develop gates for codes using this language, for example through
cup products and Poincar\'e
duality~\cite{breuckmann2026cups,li2025poincare,zhu2025nonclifford,haruna2026homological}
and through chain maps between complexes~\cite{benhemou2026automated}.
Such constructions
select distinguished points inside the group of diagonal circuits, but it is not clear whether such specific formulations yield all possible diagonal circuits.
It will be useful to recast the present formalism in terms of geometric properties of the underlying chain complex and to
determine which gates admit sparse geometric representatives.

\subsection{Extension to non-CSS codes and other alphabets}
\label{sec:ext-noncss}

Several of our results are straightforwardly generalizable to non-CSS codes.
The Pauli dressing theorem
holds for an arbitrary stabilizer code (Theorem~\ref{thm:dressing}).
The structural description of the entire group $\Nfull$ of code-preserving Clifford circuits in
Section~\ref{sec:general} is equally general, since its only condition is that
stabilizers commute.
For any stabilizer code, linear conditions on the code's label space still determine
the groups of valid diagonal circuits, CNOT networks, and partial dualities.

However, diagonal families do not generate the full code-preserving group of more general codes.
For example, consider a single qubit stabilized by the Pauli $Y$. Its label space is
spanned by $(1|1)$, which does not split into $X$-type and $Z$-type components.
No nonzero diagonal circuit of either kind preserves this space, so both diagonal families are
trivial.
On the other hand, the code-preserving group is nontrivial. 
It contains the Hadamard,
which exchanges $X$ and $Z$ and fixes $Y$ up to a sign, and a Pauli correction
removes that sign. 
For such a code, the diagonal circuits generate a proper subgroup of
$\Nfull$.
A possible fix may be to add the duality generators, which are redundant for CSS codes but may become
essential for non-CSS codes.

Several of our results should generalize to qudit codes over Galois fields $\mathbb{F}_q$ for $q$ a prime power as well as other qudit and bosonic alphabets.
We expect the reduction modulo Pauli operators of Appendix~\ref{app:setup} to extend to any abelian group, not only $\Ftwo$~\cite{bauer2026quadratic}, since the stabilizer formalism and its normalizer gates already exist for such groups~\cite{hostens2005stabilizer,vandennest2013normalizer,bermejovega2014normalizer,bermejovega2016normalizer,moses2024clifford}.
The four families should survive that change of alphabet, because each is defined by the block it occupies in the symplectic matrix rather than by anything specific to $\Ftwo$. 
The two diagonal families are the off-diagonal blocks, the CNOT family the diagonal ones, and the partial dualities the block exchange.
What does not obviously survive is the generation result itself, which rests, in part, on every diagonal circuit being its own inverse.

\subsection{Two-fold automorphism group}
\label{sec:ext-perm}

The code automorphism group of Section~\ref{sec:permutations} is not merely the
transversal gates and the permutation symmetries put together. It also contains
products of such gates that are code-preserving only when combined.
We propose to extend this compensation mechanism to depth-one two-local layers.

Let $(g;\pi)$ be the permutation $\pi$ compensated by the depth-one layer $g$. We
now draw $g$ from the two-local layers of some matching $\Mt$.
We define the \textit{two-fold automorphism group},
\begin{equation}
  \bigl\langle\,(g;\pi):\;\pi\in\Perm{n},\ g\in\DM\ \text{for some }\Mt,\
  (g;\pi)\in\Nfull\,\bigr\rangle ,
  \label{eq:ext-perm-group}
\end{equation}
where two-local circuits now compensate the permutations.

Setting $\pi$ to be the trivial permutation implies that this group contains the two-fold transversal group $\Ndep$. Setting the matching to the empty one implies that the
group contains the automorphism group $\Gaut$.
But this group is not just a union of these two subgroups.
We show in Appendix~\ref{app:twolocalaut} that the $\qcode{9,3,2:186339}$ code~\cite{cross2025small,crossqiskitqec2025} admits a depth-one two-local circuit that preserves the code only when paired with a permutation.
The two-fold automorphism group of that code realizes the full-Clifford logical action, while the two-fold transversal group does not.
Such a circuit provides a new avenue for fault-tolerant logical gates in computation models where permutations are a cheap resource.

%% file: sec_conclusion.tex
\section{Conclusion}
\label{sec:conclusion}

We have characterized and, in some cases, defined several groups of Clifford circuits relevant to fault-tolerant computation with Calderbank--Shor--Steane (CSS) codes.
We sort most current literature on such gates into four families, each with its own ingredients and its own block of the binary symplectic matrix: $Z$- and $X$-diagonal circuits, CNOT circuits, and Hadamards+permutations (partial dualities).
Their code-preserving conditions can be stated simply in terms of the two classical codes defining the CSS code.

Importing long-established results from symplectic geometry and representation theory over binary spaces \cite{artin1957geometric,carter1972simple,taylor1992geometry,grove2002classical,malle2011linear}, we show that the two diagonal families generate every code-preserving physical Clifford for a CSS code with nonzero \(X\)- and \(Z\)-type stabilizer spaces. 
They consequently
realize every logical Clifford operation. 
The valid diagonal-circuit spaces can be obtained efficiently by Gaussian elimination~\cite{webster2023transversal,koh2026phantom,bauer2026finding}, and we quantify the physical representatives of the same logical gate that can be used for circuit optimization.

We then focus on subgroups consisting of circuits with restricted locality. We show that every depth-one one-local (read: transversal) gate
can be expressed using at most three diagonal-circuit layers, while two layers suffice for connected non-self-dual codes. 
Going beyond this group to depth-one two-local circuits on a fixed matching of qubits, we show that all code-preserving depth-one
two-local circuits are generated by diagonal circuits and CNOT layers on the same matching.
The group generated by such circuits over \textit{all} possible matchings --- the two-fold transversal group --- can therefore be studied by sampling depth-one diagonal circuits and CNOT gates on arbitrary matchings. 

The two-fold transversal group can be smaller than the full code-preserving
group, but its logical image can nevertheless be the complete symplectic group over the \(k\) logical qubits due to the availability of multiple matchings.
We identify \Nfullcodes\ codes with at most 120 qubits (and 45 more with up to 276 qubits, recorded in machine-readable form) that admit the full logical Clifford group via stacks of depth-one two-local code-preserving circuits, and study the size of this logical image for a sample of \Nldpccodes\ further codes, mostly instances of quantum low-density parity-check (QLDPC) families.

We also determine that the automorphism group --- the group of transversal gates and permutations --- is incomparable to the two-fold transversal group, i.e., neither group is a subgroup of the other.
We show that one partial duality and two transversal \(Z,X\)-diagonal circuits give a normal
form for elements of this group.

In summary, we provide an underlying structure
to the group of all available Clifford circuits, supply recipes for finding depth-one code-preserving circuits, and give simple circuit decompositions in the
transversal and automorphism-group settings. 
We leave the problem of optimizing the resulting circuits for gate count, depth, geometry, or fault propagation to future platform-dependent analyses.

%% file: app_codes.tex
\section{Code glossary and details about numerical methods}
\label{app:codes}

\begin{table*}[t]
\caption{\label{tab:databases}%
The code databases and constructions that supply the entries of
Tables~\ref{tab:full-depth1} and~\ref{tab:ldpc-depth1}. Database entries are
named by their parameters, a source shorthand, and the leading octet of their
uuid. A \texttt{cons:} tag names a construction and replaces the uuid.
The complete records behind Tables~\ref{tab:full-depth1}
and~\ref{tab:ldpc-depth1} are stored as machine-readable JSON files in the certificate
repository of Ref.~\cite{albert2026certificate}.}
\begin{ruledtabular}
\footnotesize
\renewcommand{\arraystretch}{1.15}
\begin{tabular}{lp{0.66\textwidth}}
Source & Description \\
\hline
\texttt{se}~\cite{bombin2006topological,bombin2007exact}
  & Self-dual even codes: $\CX=\CZ$, with every codeword of even weight. The $\qcode{18,4,4}$
    $6.6.6$ color code also appears in the Error Correction Zoo. \\
\texttt{ml}~\cite{hastings2017small,haah2017magic,delfosse2020short,prabhu2024distance,hastings2025cyclic,berthusen2025concatenated,koh2026phantom,reichardt2026fireandice}
  & Self-dual doubly-even codes: $\CX=\CZ$, with every codeword weight a
    multiple of four. Three entries also appear in the Error Correction Zoo:
    the $\qcode{16,6,4}$ tesseract color code, the $\qcode{18,2,5}$ BCC
    code, and the $\qcode{20,2,6}$ B\&C phantom code. \\
\texttt{eczoo}~\cite{albert2026handbook,kovalev2012improved,kovalev2013hyperbicycle,koh2026phantom,lai2025faulttolerant,paetznick2024demonstration}
  & The Error Correction Zoo, which covers the $\qcode{10,2,3}$
    rotated toric, $\qcode{10,2,3}$ binarized Galois-qudit,
    $\qcode{14,3,3}$ constant-excitation/phantom, and $\qcode{12,2,4}$
    carbon codes. \\
\texttt{qecdb}~\cite{qecdb}
  & An online database of stabilizer codes run by Simon Burton. \\
\texttt{2bga}~\cite{lin2023twoblock,goto2024manyhypercube,liang2025selfdual,berthusen2025concatenated}
  & Two-block group-algebra codes, which cover the $\qcode{16,4,4}$
    symplectic-double code. \\
\texttt{coset2bga}~\cite{aydin2026coset}
  & Coset two-block group-algebra codes.\\
\texttt{toricdir}~\cite{gu2026nearest}
  & Directional tile codes. \\
\texttt{kasai}~\cite{kasai2026breaking}
  & Block-circulant codes whose blocks are permutation matrices of affine
    $\mathbb{Z}_{P}$ maps $x\mapsto ax+b$. \\
\texttt{cc}~\cite{gu2026qgpu}
  & Clustered-cyclic codes. \\
\texttt{copycup}~\cite{tiew2026copycup}
  & Abelian balanced-product codes. \\
\texttt{mm}~\cite{mian2026multivariate,quantumclifford}
  & Multivariate multicycle codes, with some taken from the
    \texttt{QuantumClifford.jl} package. \\
\texttt{ar}
  & Self-dual even CSS codes, some doubly even, invariant under rank-3
    permutation groups.  They come from a census of all such groups of degree
    at most 300. \\
\hline
\texttt{cons:blk45k12}, \texttt{cons:blk45k18}
  & The indecomposable blocks of two decomposable \texttt{qecdb}~\cite{qecdb}
    codes: $\qcode{45,12,3}$ and $\qcode{45,18,3}$ are direct sums of three copies of
    the $\qcode{15,4,3}$ and $\qcode{15,6,3}$ codes listed here. \\
\texttt{cons:qrm6}
  & The middle Reed--Muller code $\mathrm{QRM}(6)$ of Eq.~\eqref{eq:fam-qrm}. \\
\texttt{cons:johnson16}
  & The cut-complement code at $m=16$, Eq.~\eqref{eq:fam-johnson}. \\
\texttt{cons:grid}$a$\texttt{x}$b$
  & Bipartite grid codes of Eq.~\eqref{eq:fam-grid}, at the thirteen
    admissible $(a,b)$ with $ab\le120$ beyond the $\qcode{24,8,4}$ member. \\
\texttt{cons:quadric}$j$
  & The quadric tower members $QT_{3}$ and $QT_{4}$ of
    Eq.~\eqref{eq:fam-quadric}. \\
\texttt{cons:concat112}
  & The concatenation lift $\qcode{112,6,12}$. It places $\mathrm{QRM}(4)$ inside
    the even $[7,3,4]$ subcode of the Hamming code, one inner block per outer
    qubit, yielding the map $\qcode{n,k,d}\mapsto\qcode{7n,k,3d}$. \\
\texttt{cons:mcr}$n$~\cite{mcr2026generalized}
  & Generalized-bicycle codes with block-circulant symmetry, at lengths
    $n=54$ and $66$.  The fullness certificates use the block-fold
    matchings of their two-block structure. \\
\texttt{cons:m22gauged}
  & The $\qcode{77,34,5}$ code obtained from the $M_{22}$ rank-3 code on
    $77$ points by adding its $X$-type all-ones logical operator to the
    stabilizer group, a form of gauge fixing.  The parent code is not
    full.  The gauged code is. \\
\texttt{dd422}~\cite{bravyi2010majorana,liu2024subsystem,burton2024genons,berthusen2025concatenated}
  & The $\qcode{4,2,2}$-concatenated symplectic doubles. \\
\texttt{cons:sdbb}$n$~\cite{liang2025selfdual}
  & Self-dual bivariate bicycle codes. \\
\texttt{cons:gross144}~\cite{bravyi2024high}
  & The gross code, discussed in Section~\ref{sec:ldpcdepthone-gross}. \\
\end{tabular}
\end{ruledtabular}
\end{table*}

\subsection{Code glossary.}
The code tables of Sections~\ref{sec:depth1twolocal} and~\ref{sec:ldpcdepthone}
name each database entry by its parameters, source shorthand, and the leading
octet of its uuid or tag.
Table~\ref{tab:databases}
gives the source or definition of each shorthand.
The \texttt{ml}, \texttt{se}, and \texttt{ar} shorthands denote databases currently in preparation.
Several other constructions are members of the infinite
families of Section~\ref{sec:depth1twolocal-families}, and the table points to
the equation that defines each one.
Some constructions are one-off
lifts, and the table describes them in place.

The automorphism data of each database entry lists generators that
are products of single-qubit Clifford gates and qubit permutations.  The recorded group order is the order of the group that the
listed generators create. We do not compute the full automorphism group for large codes, whose elements can combine a nontrivial single-qubit
layer with a permutation. The recorded order is then a lower bound.

\subsection{Numerical methods}

The gate searches of Sections~\ref{sec:depth1twolocal}
and~\ref{sec:ldpcdepthone} combine exact calculations with bounded searches.

Since the number of matchings grows superexponentially with $n$, we
enumerate maximum-matching orbits exactly when 
feasible, but otherwise are forced to sample.
We can also generate matchings from code symmetries such as involutions in the simultaneous permutation automorphism group and involutive $XZ$ dualities.

We compute the simultaneous permutation automorphisms of $\CX$ and $\CZ$
from a colored incidence graph. Its vertices represent coordinates and
codewords. A codeword color records its CSS side and its weight. For each side,
we add complete weight classes in increasing order until their span is the full
code. We enumerate these classes with the Brouwer--Zimmermann method of
Refs.~\cite{zimmermann1996integral,grassl2006searching}. Its information-set
and residual bounds certify that each class is complete. The coordinate action
of the graph automorphism group is then exactly the simultaneous permutation
automorphism group.
For self-dual codes, if the low-weight search does not span within its budget,
we fall back to enumerating the full codeword set. This exact route is
exponential in the code dimension~\cite{leon1982computing}.
We use \textsc{bliss} and \textsc{nauty} for our calculations~\cite{junttila2007bliss,mckay2014practical}.

We use the same graphs for permutation equivalence and for $XZ$
dualities. For an $XZ$ duality, the graph map interchanges the two CSS codes.
We find such maps with the VF2 algorithm~\cite{cordella2004subgraph}.

For each matching, we construct code-preserving gates from the three
families in Eq.~\eqref{eq:dep-gen}. The conditions on the two diagonal families
are linear. One nullspace calculation for each family gives its generators
exactly. A CNOT network has the symplectic block form
$\left(\begin{smallmatrix}K&0\\0&K^{\itp}\end{smallmatrix}\right)$ of
Eq.~\eqref{eq:setup-lev}. For a fixed matching, the linear preservation
conditions define an $\Ftwo$-matrix algebra $\mathcal{A}_{\Mt}$. The invertible
elements $K\in\mathcal{A}_{\Mt}$ give the valid CNOT gates on that matching. We
search this algebra for invertible elements. For each
structured involution, we also solve the affine lifting equations exactly and
retain a compensating partial Hadamard lift when one is found.
For targeted diagonal-layer search, Holmes uses a related separation of
linear code-preservation and logical-reach equations from the combinatorial
matching constraint \cite[Def.~6, Supplementary Material]{holmes2026quantum},
enforcing the latter with SAT.

We then calculate the logical image of each physical circuit and run a Schreier--Sims calculation~\cite{sims1970computational} to determine the logical image. We run it on the action on $\Ftwo^{2k}\setminus\{0\}$, falling back to the randomized stabilizer chains of the \textsc{genss} package~\cite{genss} once that degree becomes unwieldy. If the calculation finishes,
it gives the exact order of the subgroup generated by the gates found so far.
Equality with the order of $\Sp{2k}$ proves that the logical image is full. A
smaller order is only a lower bound on the full logical image.

For large logical spaces, we also use order-free tests. MeatAxe tests
irreducibility and absolute irreducibility
\cite{parker1984meataxe,holt1994testing}. A linear solve tests whether the group
preserves a quadratic form; such a form proves that the group is not full. An
explicit word that equals a transvection, together with a full orbit on
$\Ftwo^{2k}\setminus\{0\}$, proves fullness on its own, because the
transvections generate $\Sp{2k}$~\cite{taylor1992geometry}. The
classification of Ref.~\cite{mclaughlin1969subgroups} gives a positive result
only after one proves that the subgroup is generated by transvections. A
spanning connected set of transvections that admits no invariant quadratic form
supplies that hypothesis, once it is also larger than the $\binom{2k+2}{2}$
transvections that $\Perm{2k+2}$ carries in its natural module. Transitivity on
$\Ftwo^{2k}\setminus\{0\}$ together with absolute irreducibility leaves only
$\mathrm{SL}(2k,2)$, which is too large to lie in $\Sp{2k}$, and $\Sp{2k}$
itself~\cite{hering1974transitive}. The
recognition theorem of Ref.~\cite{guralnick1999linear} instead requires
primitive-prime divisors in element orders and the exclusion of its exceptional
cases. The one-sided algorithm of Ref.~\cite{niemeyer1998recognition},
implemented in the \textsc{recog} package~\cite{recogpackage} of
\textsc{gap}~\cite{gap4}, reports the same containment from random elements.
Sampling irreducible-factor degrees of minimal polynomials does not by
itself establish these hypotheses. We therefore treat that sampling as
diagnostic unless an independent certificate is available.

For the weight and degree calculations, we search for generating bases with small check weight and
small combined qubit degree. The check weight is the largest row weight over the
$X$- and $Z$-blocks together, and the qubit degree is
$\max_q\bigl(\deg_X(q)+\deg_Z(q)\bigr)$. The minima of these two quantities over
generating sets are properties of the rowspaces $\CX$ and $\CZ$. We search
each rowspace for a generating set that is better than the stored rows. We choose the $X$- and $Z$-bases independently,
which is what allows an odd combined degree.

The search has two phases. First we enumerate every codeword of weight at most
$w$ in each rowspace by the Brouwer--Zimmermann method described above
\cite{zimmermann1996integral,grassl2006searching}. With $h$ disjoint information
sets and all combinations of at most $p$ generator rows enumerated in each
systematic form, every codeword of weight below $h(p+1)$ is found, since such a
word places at most $p$ of its support in some information set. We choose $p$
so that $h(p+1)>w$.

If the weight-$\leq w$ codewords span
a space of dimension smaller than $\ranko\CX$, then no generating set of check
weight $w$ exists at \emph{any} qubit degree, because any such set would lie
inside that span.

Second, we select a basis. The enumerated words form a linear
matroid~\cite{whitney1935abstract}, so sorting by weight and greedily keeping
independent words yields a basis of minimum total
weight~\cite{rado1957note,edmonds1971matroids}; this minimum is exact. The
maximum degree is not a matroid objective and the greedy basis carries no
guarantee for it, so we then descend on degree by basis exchange. Adding a pool
word to a basis creates a unique fundamental circuit, and swapping the new word
for any element of that circuit preserves independence. We accept exchanges that
lower the largest column count, restarting from fresh minimum-weight bases on
stalls. A flow formulation with simulated annealing
\cite{kirkpatrick1983optimization,metropolis1953equation} serves as a fallback
when the matroid route fails to span.

The routine answers whether a generating set exists at a \emph{given} pair
$(w,\Delta)$ of check weight and qubit degree, so the pair itself must be
searched. We first bisect downward on $\Delta$, with $w$ fixed at the check
weight of the stored rows. We then bisect downward on $w$ at the resulting
degree. We keep the best witness. Every witness is then re-checked independently: weight and
degree are re-measured, the CSS condition $W_XW_Z^{\tp}=0$ is verified, and the
witness is confirmed to span the stored code by comparing the ranks of the
stored rows, the witness, and their union. A low-weight set that does not span
is not a generating set. An exhibited generating set achieves each reported pair. If the
enumeration was certified complete, the reported check weight is the exact
minimum. If it was not, the reported check weight is an upper bound. The
reported qubit degree is an upper bound at that weight.

%% file: app_setup.tex
\section{Notation, conventions, and the reduction modulo Pauli operators}
\label{app:setup}

This appendix and the five that follow it examine one object: the group of
Clifford gates that preserve a fixed Calderbank--Shor--Steane (CSS) code.
Appendices~\ref{app:general}--\ref{app:permutations}
examine it from four points of view, each more restrictive than
the last. Appendix~\ref{app:general} treats the unrestricted case. It shows that
the two families of diagonal circuits introduced in
Sec.~\ref{app:setup-shears} below already generate the whole code-preserving
group $\Nfull$, provided both classical codes are nonzero. Under that
hypothesis, every other valid gate class is a group word in those two families.
Appendix~\ref{app:depthone} fixes a matching $\Mt$ of the $n$ qubits.
That appendix then studies the two-fold transversal slice $\NM=\DM\cap\Nfull$. That slice is
the part of $\Nfull$ that a single layer of gates supported on the cells of $\Mt$
is able to reach. Appendix~\ref{app:transversal} specializes to the
all-singleton matching. At that matching the slice becomes the transversal group
$\Ntr$, and the two diagonal families collapse onto the parameter codes $\Pup$ and
$\Pdn$.
Appendix~\ref{app:permutations} enlarges the transversal group by qubit
permutations, producing $\Gaut=(\Sp{2}\wr\Perm{n})\cap\Stb(\lab)$ together with
its subgroup $\PDgp$ of partial dualities.
Appendix~\ref{app:twolocalaut} applies this machinery to a single small
code and exhibits a two-fold automorphism beyond
$\langle\Ndep,\Gaut\rangle$.

All of these discussions work inside the symplectic group $\Sp{2n}$, not inside the
unitary group. Theorem~\ref{thm:dressing} below licenses this choice: a Clifford
unitary can be multiplied by a suitable Pauli operator so as to preserve a
stabilizer code, signs included, exactly when its symplectic action preserves
the label space of that code. Code preservation is therefore a purely
linear-algebraic condition on a $2n\times 2n$ matrix over $\Ftwo$. No
compensating Pauli operator ever has to be exhibited. The group of
code-preserving gates modulo Pauli operators is literally the setwise
stabilizer $\Nfull=\Stb_{\Sp{2n}}(\lab)$. The present appendix collects the
notation, sets up the conventions, and proves that theorem. From that theorem
the appendix derives the two elementary gate families out of which everything
later is built.

Tables~\ref{tab:setup-ambient}--\ref{tab:setup-depthone} form the master symbol
list for all six appendices. The list is grouped into ambient structures, the
code and its label space, elementary symplectic maps, groups of gates, the
parabolic structure of $\Nfull$, and the depth-one machinery. Symbols that are
introduced only later are collected there as well, each carrying a pointer to
the appendix in which it is defined. A reader who meets an unfamiliar letter
should look in these tables first.

\begin{table*}[t]
\caption{\label{tab:setup-ambient}\textit{Ambient structures.} Master symbol
list, part one of six. Entries carrying a pointer to a later appendix are defined
there and are listed here only for reference. All matrix groups and matrix spaces
are over $\Ftwo$, which we suppress from the notation: $\Sp{2n}$ is the standard
$\mathrm{Sp}(2n,2)$, $\GLg{m}$ is $\mathrm{GL}(m,2)$, and $\Symm{m}$ the symmetric
$m\times m$ matrices over $\Ftwo$.}
\begin{ruledtabular} 
\begin{tabular}{lp{0.62\textwidth}}
Symbol & Meaning \\
\colrule
$\Ftwo$ & the binary field; all linear algebra below is over $\Ftwo$ \\
$\Vsp$ & phase space, equivalently label space, of $n$ qubits \\
$\Xh$, $\Zh$ & the pure-$X$ and the pure-$Z$ half of the polarization
$\Vsp=\Xh\oplus\Zh$ \\
$\sfm{u}{v}$ & the symplectic form, Eq.~\eqref{eq:setup-form} \\
$\Id$, $\Jm$ & identity matrix; Gram matrix
$\Jm=\left(\begin{smallmatrix}0&\Id\\ \Id&0\end{smallmatrix}\right)$ of the
form \\
$g^{\tp}$, $g^{\itp}$ & transpose and inverse transpose of a matrix $g$ \\
$\Sp{m}$ & symplectic group of $\Ftwo^{m}$; $\Sp{2}\cong\Perm{3}$ (order $6$) and
$\Sp{4}\cong\Perm{6}$ (order $720$) \\
$\GLg{m}$ & general linear group of $\Ftwo^{m}$ \\
$\Symm{m}$ & the $m\times m$ symmetric matrices \\
$\Perm{m}$ & symmetric group on $m$ letters; $\Perm{n}$ permutes the qubits \\
$\Stb$ & setwise stabilizer of a subspace inside an ambient group \\
$\ones$, $\ind{A}$ & the all-ones vector; the indicator vector of a subset
$A\subseteq\{1,\dots,n\}$, so that $\ind{\{1,\dots,n\}}=\ones$ \\
$\pw$ & pointwise product, $(a\pw b)_{i}=a_{i}b_{i}$ \\
$\suppo$, $\diago$, $\spano$, $\Homo$, $\ranko$ & support of a vector; diagonal
matrix built from a vector; linear span; group of homomorphisms; rank \\
\end{tabular}
\end{ruledtabular}
\end{table*}

\begin{table*}[t]
\caption{\label{tab:setup-code}\textit{The code and its label space.} Master
symbol list, part two of six.}
\begin{ruledtabular}
\begin{tabular}{lp{0.62\textwidth}}
Symbol & Meaning \\
\colrule
$\CX$, $\CZ$ & the two classical codes of the CSS code, with $\CX\perp\CZ$ \\
$r_{X}$, $r_{Z}$, $r$, $k$ & $\dim\CX$, $\dim\CZ$, $r=r_{X}+r_{Z}=\dim\lab$, and
the number $k=n-r$ of logical qubits, Eq.~\eqref{eq:setup-params} \\
$\lab$ & label space $(\CX|0)\oplus(0|\CZ)\subseteq\Vsp$,
Eq.~\eqref{eq:setup-lab} \\
$\stabgp$ & the stabilizer group, a group of operators with label space $\lab$,
Eq.~\eqref{eq:setup-stab} \\
$\Pau{a}{b}$ & the Pauli operator with label $(a|b)$,
Eq.~\eqref{eq:setup-pauli} \\
$\dst$ & chosen destabilizer complement (Appendix~\ref{app:general}) \\
$\logsp$ & chosen logical complement, representing $\lab^{\perp}/\lab$
(Appendix~\ref{app:general}) \\
$\Pup$, $\Pdn$ & the parameter codes $\{a:a\pw\CX\subseteq\CZ\}$ and
$\{a:a\pw\CZ\subseteq\CX\}$ (Appendix~\ref{app:transversal}) \\
$\IX$, $\IZ$, $\IL$ & the qubit-index sets of the adapted basis
(Appendix~\ref{app:general}) \\
$\Endl$ & the algebra of linear maps $f$ of $\Vsp$ with $\lab f\subseteq\lab$ \\
\end{tabular}
\end{ruledtabular}
\end{table*}

\begin{table*}[t]
\caption{\label{tab:setup-elementary}\textit{Elementary symplectic maps.} Master
symbol list, part three of six.}
\begin{ruledtabular}
\begin{tabular}{lp{0.62\textwidth}}
Symbol & Meaning \\
\colrule
$\Ush{S}$ & $Z$-diagonal circuit
$\left(\begin{smallmatrix}\Id&S\\ 0&\Id\end{smallmatrix}\right)$ with
$S\in\Symm{n}$, Eqs.~\eqref{eq:setup-circuit} and
\eqref{eq:setup-shearlabel} \\
$\Lsh{T}$ & $X$-diagonal circuit
$\left(\begin{smallmatrix}\Id&0\\ T&\Id\end{smallmatrix}\right)$ with
$T\in\Symm{n}$, Corollary~\ref{cor:dualfamily} \\
$\Hd{a}$ & subset Hadamard: $\sw$ on the qubits of $\suppo(a)$ and the identity
elsewhere (Appendix~\ref{app:permutations}) \\
$\sw$ & the $X\leftrightarrow Z$ swap
$\left(\begin{smallmatrix}0&\Id\\ \Id&0\end{smallmatrix}\right)$,
Eq.~\eqref{eq:setup-swapblocks} \\
$\PX$ & the projector onto the $X$-half, $(x|z)\mapsto(x|0)$
(Appendix~\ref{app:depthone}) \\
$\pperm{\pi}$ & the qubit-permutation matrix of $\pi\in\Perm{n}$
(Appendix~\ref{app:permutations}) \\
\end{tabular}
\end{ruledtabular}
\end{table*}

\begin{table*}[t]
\caption{\label{tab:setup-groups}\textit{Groups of gates.} Master symbol list,
part four of six. A superscript $Z$ or $X$ names the family throughout, and a
subscript $\Mt$ (or $\emptyset$ for the empty
matching, whose cells are the $n$ singletons) restricts a family to the layers
supported on that matching.}
\begin{ruledtabular}
\begin{tabular}{lp{0.62\textwidth}}
Symbol & Meaning \\
\colrule
$\Nfull$ & $\Stb_{\Sp{2n}}(\lab)$, the code-preserving group modulo Pauli
operators, Eq.~\eqref{eq:setup-Ndef} \\
$\Eup$, $\Edn$ & the valid $Z$-diagonal and $X$-diagonal circuits, Eqs.~\eqref{eq:setup-eup} and \eqref{eq:setup-edn}, respectively \\
$\Lev$ & the valid Levi gates, Eq.~\eqref{eq:setup-lev} \\
$\Mt$ & a matching of the qubits, i.e.\ a set of disjoint pairs; its
\emph{cells} are those pairs together with the unmatched qubits as singletons, so
that $\Mt=\emptyset$ is the all-singleton case (Appendix~\ref{app:depthone}) \\
$\DM$ & the subgroup of $\Sp{2n}$ block-diagonal for the cells of $\Mt$
(Appendix~\ref{app:depthone}) \\
$\NM$ & $\DM\cap\Nfull$, the two-fold transversal slice at $\Mt$
(Appendix~\ref{app:depthone}) \\
$\Ndep$ & the two-fold transversal group $\langle\bigcup_{\Mt}\NM\rangle$ generated by
all slices (Appendix~\ref{app:depthone}) \\
$\Ntr$ & the transversal group, $\NM$ at the all-singleton matching
(Appendix~\ref{app:transversal}) \\
$\EupM$, $\EdnM$ & $\Eup\cap\DM$ and $\Edn\cap\DM$
(Appendix~\ref{app:depthone}) \\
$\LevM$ & $\Lev\cap\DM$, the Levi family at $\Mt$
(Appendix~\ref{app:depthone}) \\
$\Eupo$, $\Edno$ & the two diagonal families at the all-singleton matching
(Appendix~\ref{app:transversal}) \\
$\Gaut$ & the local-Clifford and permutation automorphism group $(\Sp{2}\wr\Perm{n})\cap\Stb(\lab)$
(Appendix~\ref{app:permutations}) \\
$\PDgp$ & the partial dualities, the elements $(\Hd{a};\pi)$ of $\Gaut$
(Appendix~\ref{app:permutations}) \\
$\rhom$ & the map $\Gaut\to\Perm{n}$ forgetting the local gates, with kernel
$\Ntr$ (Appendix~\ref{app:permutations}) \\
\end{tabular}
\end{ruledtabular}
\end{table*}

\begin{table*}[t]
\caption{\label{tab:setup-parabolic}\textit{Parabolic structure.} Master symbol
list, part five of six. All of these symbols are defined in
Appendix~\ref{app:general}.}
\begin{ruledtabular}
\begin{tabular}{lp{0.62\textwidth}}
Symbol & Meaning \\
\colrule
$\Urad$ & the unipotent radical of $\Nfull$ \\
$\gau$ & the stabilizer relabeling induced by $g\in\Nfull$ on $\lab$, an
element of $\GLg{r}$ (Appendix~\ref{app:general}) \\
$\lgc$ & the logical action of $g\in\Nfull$ on $\lab^{\perp}/\lab$, an element of
$\Sp{2k}$ \\
$\radF$, $\radT$ & the two parameters of the unipotent radical, whose elements
are written $u(\radF,\radT)$ \\
$\lact$ & the surjection $\Nfull\twoheadrightarrow\Sp{2k}$ onto the logical
action \\
\end{tabular}
\end{ruledtabular}
\end{table*}

\begin{table*}[t]
\caption{\label{tab:setup-depthone}\textit{Depth-one machinery.} Master symbol
list, part six of six. All of these symbols are defined in
Appendix~\ref{app:depthone}.}
\begin{ruledtabular}
\begin{tabular}{lp{0.62\textwidth}}
Symbol & Meaning \\
\colrule
$\shp{G}$, $\shm{G}$ & the two automatic parameter spaces of a symplectic
$G$, spanned respectively by $AB^{\tp}$, $D^{\tp}B$, $B+B^{\tp}$ and by
$CD^{\tp}$, $C^{\tp}A$, $C+C^{\tp}$ \\
$\cell$ & a cell of the matching $\Mt$; $g_{\cell}$ is the block of $G$ there \\
$\eps$ & a coefficient vector in $\Ftwo^{3}$ selecting a combination of the three
formulas \\
$\Gterm$ & a blockwise-terminal layer \\
$\Lterm$ & the Levi factor of $\Gterm$ \\
$\Word$ & the group word $\Ush{B}\Lsh{C}\Ush{B}$ \\
\end{tabular}
\end{ruledtabular}
\end{table*}

\subsection{Conventions}
\label{app:setup-conventions}

Vectors are rows and matrices act on them from the right, so that a composition
of linear maps is read from left to right. The phase space of $n$ qubits is
$\Vsp$, carrying the fixed polarization $\Vsp=\Xh\oplus\Zh$ into pure-$X$ and
pure-$Z$ labels. A vector of $\Vsp$ is accordingly written $(x|z)$ with
$x,z\in\Ftwo^{n}$. The symplectic form is
\begin{equation}
\sfm{(x|z)}{(x'|z')}=x\cdot z'+z\cdot x'=(x|z)\,\Jm\,(x'|z')^{\tp},
\label{eq:setup-form}
\end{equation}
all arithmetic being in $\Ftwo$, where
$\Jm=\left(\begin{smallmatrix}0&\Id\\ \Id&0\end{smallmatrix}\right)$ is the Gram
matrix of the form. The form is nondegenerate, so that
$v\mapsto\sfm{\,\cdot\,}{v}$ identifies $\Vsp$ with its dual space; we use this
identification in the proof of Theorem~\ref{thm:dressing}.

A symplectic matrix is written in the block form
$G=\left(\begin{smallmatrix}A&B\\ C&D\end{smallmatrix}\right)$, so that it acts
by
\begin{equation}
(x|z)\,G=(xA+zC\,|\,xB+zD),
\label{eq:setup-blocks}
\end{equation}
in which the four $n\times n$ blocks $A$, $B$, $C$, $D$ are reserved for this
purpose throughout and are never used for anything else. Being symplectic means
preserving \eqref{eq:setup-form}, so that $G\in\Sp{2n}$ if and only if
\begin{equation}
G\Jm G^{\tp}=\Jm,
\qquad\text{equivalently}\qquad
G^{\tp}\Jm G=\Jm .
\label{eq:setup-symplectic}
\end{equation}
The two matrix identities in Eq.~\eqref{eq:setup-symplectic} are equivalent because
$\Jm^{2}=\Id$: if $G\Jm G^{\tp}=\Jm$, then $(\Jm G\Jm)G^{\tp}=\Id$, so
$\Jm G\Jm$ is a two-sided inverse of $G^{\tp}$ and hence
$G^{\tp}\Jm G\Jm=\Id$, which is the second identity; the converse is the same
computation with $G$ and $G^{\tp}$ interchanged. Although the two identities are
equivalent as conditions on $G$, they say different things when written out in
blocks, and both readings are used below. The first gives
\begin{equation}
\begin{gathered}
AB^{\tp}=BA^{\tp},
\qquad
CD^{\tp}=DC^{\tp},\\
AD^{\tp}+BC^{\tp}=\Id,
\end{gathered}
\label{eq:setup-blockeqs}
\end{equation}
and the second gives
\begin{equation}
\begin{gathered}
A^{\tp}C=C^{\tp}A,
\qquad
B^{\tp}D=D^{\tp}B,\\
A^{\tp}D+C^{\tp}B=\Id .
\end{gathered}
\label{eq:setup-blockeqs2}
\end{equation}
In either reading the two diagonal products are symmetric matrices and the two
mixed products are complementary; what differs is which four products are
certified symmetric. Appendix~\ref{app:depthone} draws on all four, and builds
its generators out of exactly these products. We write $\lab g$ for the image of a subspace
$\lab\subseteq\Vsp$ under $g$, again in keeping with the right action.

Pauli operators are labeled by $\Vsp$ with the convention that $X$ stands to
the left of $Z$ on each qubit,
\begin{equation}
\Pau{a}{b}:=\bigotimes_{i=1}^{n}X_{i}^{a_{i}}Z_{i}^{b_{i}},
\qquad a,b\in\Ftwo^{n},
\label{eq:setup-pauli}
\end{equation}
where $X_i$ and $Z_i$ denote the two single-qubit Pauli operators on qubit $i$;
these operators are not to be confused with the two halves $\Xh$ and $\Zh$ of the
polarization, which are subspaces of $\Vsp$. The ordering convention in
\eqref{eq:setup-pauli} fixes the multiplication rule
\begin{equation}
\Pau{a}{b}\,\Pau{a'}{b'}=(-1)^{b\cdot a'}\,\Pau{a+a'}{b+b'},
\label{eq:setup-paulimult}
\end{equation}
from which everything we need about signs follows. In particular, taking
$(a',b')=(a,b)$ in Eq.~\eqref{eq:setup-paulimult} gives
$\Pau{a}{b}^{2}=(-1)^{a\cdot b}\Id$, and taking the two products in the two
orders gives the conjugation rule
\begin{equation}
\begin{split}
\Pau{v}{u}\Pau{a}{b}\Pau{v}{u}^{\dagger}
&=(-1)^{u\cdot a+v\cdot b}\Pau{a}{b}\\
&=(-1)^{\sfm{(v|u)}{(a|b)}}\Pau{a}{b} .
\end{split}
\label{eq:setup-dress}
\end{equation}
Conjugation by a Pauli operator therefore changes only a sign and leaves the
label $(a|b)$ unchanged. That observation is what makes
Theorem~\ref{thm:dressing} true. Note also that the label map
$\Pau{a}{b}\mapsto(a|b)$ turns operator multiplication into vector addition, by
\eqref{eq:setup-paulimult}. Hence the image of a group of Pauli operators is an
$\Ftwo$-subspace of $\Vsp$ and not merely a set.

A CSS code is specified by a pair of subspaces $\CX,\CZ\subseteq\Ftwo^{n}$ that
are orthogonal to one another,
\begin{equation}
a\cdot b=0
\qquad\text{for all }a\in\CX,\ b\in\CZ,
\label{eq:setup-css}
\end{equation}
a relation we abbreviate as $\CX\perp\CZ$.
Orthogonality \eqref{eq:setup-css} is exactly what makes the group generated by
$\{\Pau{a}{0}:a\in\CX\}$ and $\{\Pau{0}{b}:b\in\CZ\}$ a stabilizer group: by
\eqref{eq:setup-paulimult} the product of $\Pau{a}{b}$ and $\Pau{a'}{b'}$ with
$a,a'\in\CX$ and $b,b'\in\CZ$ carries the sign $(-1)^{b\cdot a'}=+1$, so that the
set
\begin{equation}
\stabgp=\bigl\{\,\Pau{a}{b}\ :\ a\in\CX,\ b\in\CZ\,\bigr\}
\label{eq:setup-stab}
\end{equation}
is closed under multiplication and every one of its elements carries the sign
$+1$. Each element is moreover Hermitian and squares to $+\Id$, since
$\Pau{a}{b}^{2}=(-1)^{a\cdot b}\Id=\Id$ by Eq.~\eqref{eq:setup-css}. The label space
of Eq.~\eqref{eq:setup-stab} is
\begin{equation}
\lab=(\CX|0)\oplus(0|\CZ)\subseteq\Vsp ,
\label{eq:setup-lab}
\end{equation}
and \eqref{eq:setup-css} says precisely that $\lab$ is isotropic,
$\sfm{\lab}{\lab}=0$, because
$\sfm{(a|b)}{(a'|b')}=a\cdot b'+b\cdot a'$ vanishes for $a,a'\in\CX$ and
$b,b'\in\CZ$. We set
\begin{equation}
\begin{gathered}
r_{X}=\dim\CX,
\qquad
r_{Z}=\dim\CZ,\\
r=r_{X}+r_{Z}=\dim\lab,
\qquad
k=n-r,
\end{gathered}
\label{eq:setup-params}
\end{equation}
$k$ being the number of logical qubits.

The only structural feature of Eq.~\eqref{eq:setup-lab} that the arguments below use,
beyond isotropy, is that $\lab$ is \emph{split} with respect to the
polarization,
\begin{equation}
\begin{gathered}
\lab=(\lab\cap\Xh)\oplus(\lab\cap\Zh),\\
\lab\cap\Xh=(\CX|0),
\qquad
\lab\cap\Zh=(0|\CZ).
\end{gathered}
\label{eq:setup-split}
\end{equation}
Splitness is the entire content of the CSS hypothesis as far as these appendices
are concerned: a stabilizer code is CSS precisely when its label space admits the
decomposition \eqref{eq:setup-split}, and no other consequence of the CSS
structure is invoked anywhere below. It is worth stressing that
\eqref{eq:setup-split} is a genuine restriction and not a normalization. A
$Y$-type stabilizer, for instance, has a label space meeting neither $\Xh$ nor
$\Zh$. The diagonal families of Sec.~\ref{app:setup-shears} cannot see such a label
space at all. The generation statement of Appendix~\ref{app:general} carries the
further standing hypothesis $r_{X},r_{Z}\ge 1$, that is $\CX\neq 0$ and
$\CZ\neq 0$. That appendix isolates the single step of the proof of that
statement at which the hypothesis is consumed.

Finally we fix what it means for a gate to preserve the code. A unitary $U$ is a
\emph{valid gate} if $U\stabgp U^{\dagger}=\stabgp$ as a set, signs included. A
unitary $U$ is \emph{valid up to Pauli dressing} if $U P$ is a valid gate for
some Pauli operator $P$. The distinction matters at the level of unitaries.
At the level of labels the distinction does not matter, as we now show: by
\eqref{eq:setup-dress} a dressing alters no symplectic datum whatsoever.

\subsection{Universal Pauli dressing}
\label{app:setup-dressing}

The following theorem (cf.~\cite{dehaene2003clifford}) is the reason the rest of
these appendices can dispense
with unitaries entirely. It holds for an arbitrary stabilizer group, with no
CSS or splitness hypothesis. The label space $\lab$ appearing in the theorem is
simply the image of $\stabgp$ under the label map, which is an isotropic subspace
of $\Vsp$.

\begin{theorem}[Universal Pauli dressing]
\label{thm:dressing}
Let $\stabgp$ be any stabilizer group on $n$ qubits, with label space
$\lab\subseteq\Vsp$, and let $U$ be any Clifford unitary whose action on labels
is the symplectic map $g\in\Sp{2n}$. Then
\begin{equation}
\begin{split}
&\text{some Pauli dressing makes }U\text{ normalize }\stabgp\\
&\qquad\qquad\Longleftrightarrow\qquad
\lab g=\lab .
\end{split}
\label{eq:setup-dressing}
\end{equation}
\end{theorem}

\begin{proof}
Necessity is immediate. Suppose $P U$ normalizes $\stabgp$ for some Pauli
operator $P$. The label space of $U\stabgp U^{\dagger}$ is $\lab g$, since
conjugation by $U$ implements $g$ on labels; and conjugation by $P$ leaves every
label unchanged by Eq.~\eqref{eq:setup-dress}. Hence the label space of
$PU\stabgp U^{\dagger}P^{\dagger}=\stabgp$ is $\lab g$ on the one hand and
$\lab$ on the other, so $\lab g=\lab$. The same argument applies verbatim if the
dressing sits on the right of $U$.

For the converse, suppose $\lab g=\lab$. Both $\stabgp$ and
$U\stabgp U^{\dagger}$ are then stabilizer groups with the same label space
$\lab$. The group $U\stabgp U^{\dagger}$ is a stabilizer group because a Clifford
unitary conjugates Pauli operators to Pauli operators.

We first record that a stabilizer group contains exactly one operator per label.
Indeed, two of its elements with the same label differ by a scalar multiple of
the identity, and that scalar again lies in the group. The scalars occurring in a
stabilizer group form a subgroup of $\{\pm\Id,\pm i\,\Id\}$ that does not contain
$-\Id$, by the defining property of a stabilizer group. Since $(\pm i\,\Id)^{2}
=-\Id$, that subgroup cannot contain $\pm i\,\Id$ either, and is thus trivial.
Two elements with the same label thus coincide, and the label map is a bijection
from $\stabgp$ onto $\lab$, and likewise from $U\stabgp U^{\dagger}$ onto $\lab$.
Write $s_{\ell}$ and $s'_{\ell}$ for the unique elements of $\stabgp$ and of
$U\stabgp U^{\dagger}$ with label $\ell\in\lab$. Every element $s$ of a
stabilizer group moreover satisfies $s^{2}=\Id$. Indeed $s^{2}$ is a scalar
multiple of the identity lying in the group, and the only such scalar is $\Id$.
A unitary with $s^{2}=\Id$ obeys $s=s^{-1}=s^{\dagger}$ and is Hermitian.
Consequently $s_{\ell}$ and $s'_{\ell}$ are Hermitian operators with the same
label, and they differ by a real sign,
\begin{equation}
s'_{\ell}=\chi(\ell)\,s_{\ell},
\qquad
\chi(\ell)\in\{\pm 1\} .
\label{eq:setup-chi}
\end{equation}

The function $\chi$ is a character of the additive group $\lab$. To see this,
note that $s_{\ell}s_{\ell'}$ lies in $\stabgp$ and has label $\ell+\ell'$,
because the label map turns multiplication into addition. By uniqueness of the
element with a given label it follows that
$s_{\ell}s_{\ell'}=s_{\ell+\ell'}$, and in the same way
$s'_{\ell}s'_{\ell'}=s'_{\ell+\ell'}$. Substituting \eqref{eq:setup-chi} into the
latter and comparing with the former gives
\begin{equation}
\chi(\ell)\chi(\ell')=\chi(\ell+\ell')
\qquad\text{for all }\ell,\ell'\in\lab .
\label{eq:setup-character}
\end{equation}

Every character of a subspace extends to the ambient space. Writing
$\chi=(-1)^{\phi}$ with $\phi:\lab\to\Ftwo$ linear by
\eqref{eq:setup-character}, choose any complement of $\lab$ in $\Vsp$ and extend
$\phi$ by zero on it. The extension is a linear functional on all of $\Vsp$.
Since
the symplectic form \eqref{eq:setup-form} is nondegenerate, every linear
functional on $\Vsp$ is of the form $\sfm{p}{\,\cdot\,}$ for a unique
$p=(v|u)\in\Vsp$. Hence
\begin{equation}
\chi(\ell)=(-1)^{\sfm{p}{\ell}}
\qquad\text{for all }\ell\in\lab .
\label{eq:setup-pfound}
\end{equation}

Conjugation by the Pauli operator $\Pau{v}{u}$ with label $p$ now supplies
exactly the signs \eqref{eq:setup-pfound}: by Eq.~\eqref{eq:setup-dress} it
multiplies the element with label $\ell$ by $(-1)^{\sfm{p}{\ell}}=\chi(\ell)$.
Applying it to $U\stabgp U^{\dagger}$ and using
$\chi(\ell)^{2}=1$ together with Eq.~\eqref{eq:setup-chi},
\begin{equation}
\begin{split}
\Pau{v}{u}\,\bigl(U\stabgp U^{\dagger}\bigr)\,\Pau{v}{u}^{\dagger}
&=\bigl\{\chi(\ell)s'_{\ell}:\ell\in\lab\bigr\}\\
&=\bigl\{s_{\ell}:\ell\in\lab\bigr\}=\stabgp ,
\end{split}
\label{eq:setup-dressed}
\end{equation}
so that the dressed unitary $\Pau{v}{u}U$ normalizes $\stabgp$. Finally the
dressing may be moved to the other side. Since $U$ is Clifford,
$U^{\dagger}\Pau{v}{u}U$ is a Pauli operator up to an overall phase. The identity
$\Pau{v}{u}U=U\,(U^{\dagger}\Pau{v}{u}U)$ then exhibits the dressed unitary as
$U$ times a Pauli operator.
\end{proof}

\emph{Only labels matter.} Modulo Pauli operators, the code-preserving physical
Clifford group \emph{is} the setwise stabilizer
\begin{equation}
\begin{split}
\Nfull&=\Stb_{\Sp{2n}}(\lab)\\
&=\bigl\{\,g\in\Sp{2n}:\lab g=\lab\,\bigr\} .
\end{split}
\label{eq:setup-Ndef}
\end{equation}
Theorem~\ref{thm:dressing} says that a Clifford unitary admits a valid dressing
precisely when its symplectic action lies in $\Nfull$, and, once global phases
are discarded, the Pauli group is exactly the kernel of the symplectic action of
the Clifford group. The group theory of Appendices~\ref{app:general}
through~\ref{app:permutations} may therefore be conducted entirely on labels,
which is what we shall do. A statement about $\Nfull$ translates back into a
statement about physical gates with no loss and no additional hypotheses.

\emph{The Pauli frame is immaterial.} Equation \eqref{eq:setup-stab} fixes what
one may call the trivial frame, in which every stabilizer element carries the
sign $+1$. Conjugating $\stabgp$ by a Pauli operator changes only those signs,
by Eq.~\eqref{eq:setup-dress}, and leaves $\lab$, the group $\Nfull$, and the
criterion $\lab g=\lab$ unchanged. There are $2^{r}$ frames in all, one per
assignment of signs to a set of $r=\dim\lab$ generators, and Pauli conjugation
permutes them transitively. The dressing produced in the proof of
Theorem~\ref{thm:dressing} is simply the operator that undoes the assignment
that $U$ happens to induce. That operator has more structure than the proof
shows. Pauli operators commuting with all of $\stabgp$ leave every sign alone.
The dressing therefore matters only through its symplectic pairings with $\lab$:
$r$ bits, one per generator, recording which generators $U$ returns with a
flipped sign. Those bits are the syndrome. A Pauli operator realizing a
prescribed set of pairings with the generators is typically called a
destabilizer. The dressing may always be taken to be a product of
destabilizers, one factor for each flipped generator. Concretely, fix
destabilizers dual to the chosen generators. One may then take the compensating
Pauli operator to be the product of those destabilizers indexed by the flipped
generators. The symplectic pairings of that operator with the generators are
then the syndrome. The label of that operator is a vector in $\Ftwo^{2n}$,
whereas the syndrome is the $r$-bit functional it induces on $\lab$. The label
and the syndrome should not be identified. None of this is needed below, where
no compensating Pauli operator is ever computed.

\subsection{The two Siegel families}
\label{app:setup-shears}

We can now identify the two elementary families of physical gates that
Appendix~\ref{app:general} shows to generate all of $\Nfull$. Following the
correspondence of Sec.~\ref{app:setup-classes}, we call them the
\emph{$Z$-diagonal} and \emph{$X$-diagonal circuits}: both are products of
commuting gates, the first diagonal in the computational basis and the second
diagonal in the conjugate basis. In that correspondence these families are the
unipotent radicals of the two opposite Siegel parabolic subgroups fixed by the
polarization, and are the \emph{Siegel shears} of the classical
literature~\cite{artin1957geometric,taylor1992geometry}. The two families are not in
general circuits of depth one: the gates in a product commute, but they need not
act on disjoint qubits. The product becomes a single layer only when the support
graph described below is a matching. That restriction is imposed, and exploited,
in Appendix~\ref{app:depthone}.

Let $S\in\Symm{n}$ be a symmetric matrix over $\Ftwo$. Its off-diagonal entries
describe a graph on the qubit indices, with an edge $\{i,j\}$ for $i<j$ whenever
$S_{ij}=1$. The diagonal entries $S_{ii}$ mark a subset of qubits. The
associated circuit is the commuting product
\begin{equation}
\Ush{S}:=\prod_{i<j,\ S_{ij}=1}\mathrm{CZ}_{ij}
\ \cdot\!\!
\prod_{i,\ S_{ii}=1}\sqrt{Z}_{i} .
\label{eq:setup-circuit}
\end{equation}
Here $\sqrt{Z}=\diago(1,i)$ is the single-qubit phase gate. That gate is
usually written $S$. In these appendices, however, the letter $S$ denotes the
diagonal-circuit parameter of Eq.~\eqref{eq:setup-circuit}, so we write the gate as $\sqrt{Z}$
throughout. For the same reason we write $\sqrt{X}$ for the conjugate-basis
partner of $\sqrt{Z}$, the gate obtained from it by conjugation with a
Hadamard. We also allow ourselves to write $\Ush{S}$ both for the circuit
\eqref{eq:setup-circuit} and for the symplectic matrix it induces, an abuse
justified by Lemma~\ref{lem:shearaction}. When the parameter is a vector rather
than a matrix we use the width-one shorthand $\Ush{a}:=\Ush{\diago a}$ and
$\Lsh{a}:=\Lsh{\diago a}$ for $a\in\Ftwo^{n}$. The type of the argument tells
the two readings apart. In circuit terms $\Ush{a}$ is the transversal layer
$\prod_{i\in\suppo(a)}\sqrt{Z}_{i}$.

\begin{lemma}[Label action of a diagonal layer]
\label{lem:shearaction}
For every $S\in\Symm{n}$ the circuit \eqref{eq:setup-circuit} is Clifford, and
its action on labels is
\begin{equation}
\begin{gathered}
(a|b)\longmapsto(a\,|\,b+aS),\\
\text{that is}\qquad
\Ush{S}=\begin{pmatrix}\Id&S\\ 0&\Id\end{pmatrix},
\end{gathered}
\label{eq:setup-shearlabel}
\end{equation}
the $Z$-diagonal circuit with parameter $S$.
\end{lemma}

\begin{proof}
The circuit \eqref{eq:setup-circuit} is a product of Clifford gates and is
therefore Clifford. Every factor in it is diagonal in the computational basis
and hence commutes with each $Z_{i}$, so pure-$Z$ labels are fixed,
$(0|b)\mapsto(0|b)$.

For the $X$-sector it suffices, by Eq.~\eqref{eq:setup-paulimult}, to follow the
generators $X_{i}$. A controlled-$Z$ gate acts by
$\mathrm{CZ}_{ij}X_{i}\mathrm{CZ}_{ij}^{\dagger}=X_{i}Z_{j}$ and fixes $X_{l}$
for $l\notin\{i,j\}$. Each edge $\{i,j\}$ of the graph of $S$ with
$i\in\suppo(a)$ hence contributes one factor $Z_{j}$. The phase gate acts by
$\sqrt{Z}_{i}X_{i}\sqrt{Z}_{i}^{\dagger}=i\,X_{i}Z_{i}$, so each
$i\in\suppo(a)$ with $S_{ii}=1$ contributes one factor $Z_{i}$. Collecting the
$Z$-output on a fixed coordinate $j$ therefore gives the exponent
\begin{equation}
\sum_{i\neq j}a_{i}S_{ij}+a_{j}S_{jj}=(aS)_{j},
\label{eq:setup-collect}
\end{equation}
where the first sum accounts for the edges and the second term for the diagonal
entries. Hence $\Pau{a}{0}$ is mapped to a phase times $\Pau{a}{aS}$. The
phases are immaterial here, since Theorem~\ref{thm:dressing} needs only the label
map. Combining the two sectors with Eq.~\eqref{eq:setup-paulimult} gives
$(a|b)\mapsto(a|b+aS)$, which is the right action of the displayed matrix by the
block convention \eqref{eq:setup-blocks}.

It remains to observe that this matrix is symplectic. With $A=\Id$, $B=S$,
$C=0$ and $D=\Id$, the block conditions \eqref{eq:setup-blockeqs} reduce to the
single requirement that $AB^{\tp}=S^{\tp}$ be symmetric, the other two being
automatic. The $Z$-diagonal circuit is thus symplectic precisely because its parameter is
symmetric, which is the reason for restricting to $S\in\Symm{n}$ in
\eqref{eq:setup-circuit}.
\end{proof}

\begin{corollary}[Validity criterion]
\label{cor:validity}
For $S\in\Symm{n}$, the circuit $\Ush{S}$ is valid up to Pauli dressing if and
only if
\begin{equation}
\CX S\subseteq\CZ ,
\label{eq:setup-criterion}
\end{equation}
that is, if and only if $aS\in\CZ$ for every $a\in\CX$.
\end{corollary}

\begin{proof}
By Lemma~\ref{lem:shearaction} the label action of the circuit is $g=\Ush{S}$, and by Eq.~\eqref{eq:setup-lab},
\begin{equation}
\lab g=\bigl\{(a\,|\,b+aS)\ :\ a\in\CX,\ b\in\CZ\bigr\} .
\label{eq:setup-image}
\end{equation}
Since $g$ is invertible and $\lab$ is finite-dimensional, $\lab g$ and $\lab$
have the same dimension, so $\lab g=\lab$ holds if and only if
$\lab g\subseteq\lab$. Now $(a|b+aS)$ lies in $\lab$ if and only if its $X$-part
lies in $\CX$ and its $Z$-part lies in $\CZ$, by the splitness
\eqref{eq:setup-split}. The condition on the $X$-part holds by hypothesis. The
condition on the $Z$-part says $b+aS\in\CZ$, which for $b\in\CZ$ is equivalent to
$aS\in\CZ$. Thus
$\lab g=\lab$ is equivalent to Eq.~\eqref{eq:setup-criterion}, and
Theorem~\ref{thm:dressing} converts this into validity up to Pauli dressing.
\end{proof}

Corollary~\ref{cor:validity} thus identifies the valid $Z$-diagonal circuits with
the family
\begin{equation}
\Eup:=\bigl\{\,\Ush{S}:S\in\Symm{n},\ \CX S\subseteq\CZ\,\bigr\},
\label{eq:setup-eup}
\end{equation}
which by that same corollary is a subgroup of $\Nfull$. The corollary speaks
only of the layers \eqref{eq:setup-circuit}. Level two of the Clifford hierarchy
supplies the converse, that every Clifford gate diagonal in the computational
basis is such a layer (cf.~\cite{rengaswamy2019unifying}). Only with that
converse is \eqref{eq:setup-eup} the whole valid $Z$-diagonal group.
Because diagonal-circuit parameters simply add, $\Ush{S}\Ush{S'}=\Ush{S+S'}$ and
$\Ush{S}^{2}=\Id$, the family \eqref{eq:setup-eup} is an elementary abelian
$2$-group isomorphic to the additive group of admissible parameters. At the level
of circuits the corresponding statement holds only modulo Pauli operators.
Indeed the square of Eq.~\eqref{eq:setup-circuit} is the Pauli operator
$\prod_{i:S_{ii}=1}Z_{i}$, which is exactly the sort of discrepancy that
Theorem~\ref{thm:dressing} instructs us to ignore.

The dressing of Theorem~\ref{thm:dressing} is not vacuous for diagonal circuits. The proof
of Lemma~\ref{lem:shearaction} discards the signs. If we keep them, we get
\begin{equation}
\Ush{S}\,\Pau{a}{b}\,\Ush{S}^{\dagger}
  =i^{\,aSa^{\tp}}\,\Pau{a}{b+aS},
\label{eq:setup-shearsign}
\end{equation}
with the exponent evaluated over the integers. A valid diagonal circuit can thus return a
stabilizer element with a flipped sign. Take $n=2$, $\CX=\CZ=\spano\{\ones\}$,
and the valid pattern $S=\Id$. The circuit $\sqrt{Z}_{1}\sqrt{Z}_{2}$ sends the
stabilizer $X_{1}X_{2}=\Pau{11}{00}$ to $-\Pau{11}{11}$, but $\stabgp$ contains
$\Pau{11}{11}$. The dressing $Z_{1}$ restores all the signs.
The equivalent exact phase formula for diagonal two-local layers,
together with the same two-qubit sign obstruction, appears in
\cite[Lem.~6(i), Supplementary Material]{holmes2026quantum}.

The second family is obtained from the first by a duality of the whole setup.
Let $\sw$ denote the $X\leftrightarrow Z$ swap, that is the symplectic map
$(x|z)\mapsto(z|x)$. As a matrix it happens to coincide with the Gram matrix
$\Jm$ of Eq.~\eqref{eq:setup-form}. The swap and the Gram matrix play quite
different roles, however, one being a group element and the other the matrix of
a bilinear form. Conjugation by $\sw$ acts on the block decomposition
\eqref{eq:setup-blocks} by
\begin{equation}
\sw
\begin{pmatrix}A&B\\ C&D\end{pmatrix}
\sw
=
\begin{pmatrix}D&C\\ B&A\end{pmatrix},
\label{eq:setup-swapblocks}
\end{equation}
as one checks by multiplying out, using $\sw^{2}=\Id$. It carries label spaces
of CSS form to label spaces of CSS form with the two classical codes
interchanged,
\begin{equation}
\lab\sw=(\CZ|0)\oplus(0|\CX),
\label{eq:setup-swaplab}
\end{equation}
which is again isotropic and split because $\CZ\perp\CX$ is the same condition as
\eqref{eq:setup-css}. In particular, by Eq.~\eqref{eq:setup-swapblocks},
conjugation by $\sw$ exchanges $Z$-diagonal circuits with $X$-diagonal circuits,
\begin{equation}
\sw\,\Ush{T}\,\sw=\Lsh{T}
=\begin{pmatrix}\Id&0\\ T&\Id\end{pmatrix} .
\label{eq:setup-swapshear}
\end{equation}

It is essential to understand $\sw$ as a symmetry of the setup rather than as a
gate of a fixed code. Comparing \eqref{eq:setup-swaplab} with
\eqref{eq:setup-lab} shows that $\sw$ itself lies in $\Nfull$ if and only if
$\CX=\CZ$, which is a strong and usually false assumption. The swap $\sw$ always
maps the data of the code with classical codes $(\CX,\CZ)$ onto the data of the
code with classical codes $(\CZ,\CX)$. The swap always maps the corresponding
groups $\Nfull$ onto one another as well. Consequently every statement made in
these appendices comes in a dual pair. One obtains the dual by
applying $\sw$ throughout and interchanging $\CX$ with $\CZ$, upper with lower,
and $\Xh$ with $\Zh$. This is why the proofs below establish only one side of
each such pair.

\begin{corollary}[The dual family]
\label{cor:dualfamily}
For $T\in\Symm{n}$, the $X$-diagonal circuit $\Lsh{T}$ preserves $\lab$, equivalently the
corresponding circuit is valid up to Pauli dressing, if and only if
\begin{equation}
\CZ T\subseteq\CX .
\label{eq:setup-dualcriterion}
\end{equation}
\end{corollary}

\begin{proof}
By Eq.~\eqref{eq:setup-swapshear} and $\sw^{2}=\Id$, the $X$-diagonal circuit preserves
$\lab$ if and only if the $Z$-diagonal circuit $\Ush{T}=\sw\Lsh{T}\sw$ preserves
$\lab\sw$. By Eq.~\eqref{eq:setup-swaplab} the space $\lab\sw$ is the label space of
the CSS code whose classical codes are $\CZ$ and $\CX$ in that order.
Corollary~\ref{cor:validity} applied to that CSS code therefore says that
$\Ush{T}$ preserves $\lab\sw$ if and only if $\CZ T\subseteq\CX$.
Theorem~\ref{thm:dressing}
converts preservation of $\lab$ into validity up to Pauli dressing.
\end{proof}

The valid $X$-diagonal circuits accordingly form the family
\begin{equation}
\Edn:=\bigl\{\,\Lsh{T}:T\in\Symm{n},\ \CZ T\subseteq\CX\,\bigr\},
\label{eq:setup-edn}
\end{equation}
again a subgroup of $\Nfull$, and an elementary abelian $2$-group for the same
reason as $\Eup$. Three remarks
about this family are worth recording. First, it is available on every CSS code,
since \eqref{eq:setup-dualcriterion} is a condition on $T$ and not on the code,
and $T=0$ always satisfies it. Second, it is the image of an upper family under
the duality. Conjugation by $\sw$ carries the family \eqref{eq:setup-eup} of the
code with classical codes $(\CZ,\CX)$ onto the family \eqref{eq:setup-edn} of the
code with classical codes $(\CX,\CZ)$. When $\CX=\CZ$ these two codes coincide,
so that $\Edn=\sw\Eup\sw$ is then literally the $H^{\otimes n}$-conjugate of
$\Eup$. For $\CX\neq\CZ$ the two families belong to two different codes, and no
such identification is available.

Third, and most importantly for the physical reading, the circuit realizing
$\Lsh{T}$ requires no code-preserving transversal Hadamard. By
\eqref{eq:setup-swapshear} that circuit may be written as
$H^{\otimes n}\Ush{T}H^{\otimes n}$. The two Hadamard layers are not
themselves valid gates, however, and need not be: conjugating the product
\eqref{eq:setup-circuit} by them simply rewrites it in the conjugate basis. That
rewriting produces a product of commuting gates. On the diagonal support of $T$
the product uses $\sqrt{X}$. On the off-diagonal support of $T$ it uses the
Hadamard conjugates of $\mathrm{CZ}$, which are $XX$-type two-qubit gates. The
result is again a product of commuting gates, of the same depth and locality as
\eqref{eq:setup-circuit}, whose validity criterion is
\eqref{eq:setup-dualcriterion}. These appendices never assume that a
code-preserving transversal Hadamard exists.

\subsection{Where the gate classes come from}
\label{app:setup-classes}

It is worth recording at this point where the two diagonal families come from. The
reason is that the same source supplies every other class of gates used below.

A polarization is precisely a pair of opposite Lagrangian subspaces. Both $\Xh$
and $\Zh$ are maximal isotropic --- the form \eqref{eq:setup-form} pairs $X$
labels only with $Z$ labels --- and $\Vsp=\Xh\oplus\Zh$. Now the stabilizer of a
Lagrangian subspace is a maximal parabolic subgroup of $\Sp{2n}$, the
\emph{Siegel} parabolic of that subspace. A polarization therefore determines two
opposite Siegel parabolics, $\Stb(\Zh)$ and $\Stb(\Xh)$. Their unipotent
radicals are exactly the two diagonal families. Indeed a symplectic map fixes $\Zh$
pointwise and acts trivially on $\Vsp/\Zh$ precisely when it is $\Ush{S}$ for
some symmetric $S$. The dual characterization holds for $\Xh$ and $\Lsh{T}$. One
reads both off \eqref{eq:setup-shearlabel}. The common Levi subgroup of the two
parabolics is $\Stb(\Xh)\cap\Stb(\Zh)\cong\GLg{n}$, realized by CNOT circuits.
The associated Weyl group is the monomial group of order $2^{n}n!$ generated by
the subset Hadamards and the qubit permutations.

Cutting the common Levi down to the code, exactly as
Corollaries~\ref{cor:validity} and~\ref{cor:dualfamily} cut down the two
radicals, produces the third family used below,
\begin{equation}
  \Lev:=\Bigl\{\begin{pmatrix}K&0\\ 0&K^{\itp}\end{pmatrix}
        :K\in\GLg{n}\Bigr\}\cap\Nfull .
  \label{eq:setup-lev}
\end{equation}
Such an element acts on labels by $(a|b)\mapsto(aK\,|\,bK^{\itp})$, so lying in
$\Nfull$ amounts to $\CX K\subseteq\CX$ and $\CZ K^{\itp}\subseteq\CZ$, the same
inclusion form as the two diagonal-circuit criteria \eqref{eq:setup-criterion} and
\eqref{eq:setup-dualcriterion}. Here, however, $K$ is invertible, so each
inclusion is automatically an equality. Both conditions are needed: the first
forces $K^{\itp}$ to preserve $\CX^{\perp}$, which contains $\CZ$ but in general
properly. In circuit terms $\Lev$ is the group of CNOT
circuits preserving both classical codes.

\begin{table*}[t]
\caption{\label{tab:setup-classes}\textit{The four gate classes and the
polarization.} The polarization $\Vsp=\Xh\oplus\Zh$ is a pair of opposite
Lagrangian subspaces, and each family of gates used in these appendices is one
of the four pieces of $\Sp{2n}$ that such a pair determines, cut down by the
requirement that the code be preserved and, in Appendix~\ref{app:depthone}, by a
support restriction. The same four pieces arise for any symplectic space
equipped with a polarization, and in Appendix~\ref{app:general} they are used a
second time on the logical space $\logsp$. Here
$\mathrm{CZ}^{H}:=H^{\otimes2}\,\mathrm{CZ}\,H^{\otimes2}$ is the
$H^{\otimes2}$-conjugate of $\mathrm{CZ}$, a two-qubit gate of $XX$ type.}
\begin{ruledtabular}
\begin{tabular}{llll}
piece determined by the polarization & subgroup of $\Sp{2n}$ & gate family &
physical gates \\
\colrule
Siegel radical of $\Stb(\Zh)$ &
$\{\Ush{S}:S\in\Symm{n}\}$ &
$\Eup$, Eq.~\eqref{eq:setup-eup} &
$\sqrt{Z}$ and $\mathrm{CZ}$ \\
opposite radical, of $\Stb(\Xh)$ &
$\{\Lsh{T}:T\in\Symm{n}\}$ &
$\Edn$, Eq.~\eqref{eq:setup-edn} &
$\sqrt{X}$ and $\mathrm{CZ}^{H}$ \\
common Levi, $\Stb(\Xh)\cap\Stb(\Zh)$ &
$\bigl\{\left(\begin{smallmatrix}K&0\\ 0&K^{\itp}\end{smallmatrix}\right):
K\in\GLg{n}\bigr\}$ &
$\Lev$, Eq.~\eqref{eq:setup-lev} &
$\mathrm{CNOT}$ \\
Weyl group, of order $2^{n}n!$ &
$\{\Hd{a}\pperm{\pi}:a\in\Ftwo^{n},\ \pi\in\Perm{n}\}$ &
$\PDgp$ (Appendix~\ref{app:permutations}) &
Hadamard and swap \\
\end{tabular}
\end{ruledtabular}
\end{table*}

Table~\ref{tab:setup-classes} collects the four pieces.
Every family of gates appearing in these appendices is one of these four,
intersected with $\Nfull$. Where a depth restriction is in force, the
intersection is taken with $\DM$ as well. Theorem~\ref{thm:sheargen}
says that the first two pieces already suffice for the whole of $\Nfull$ when
$r_{X},r_{Z}\ge1$. Theorem~\ref{thm:depthone} says that the first three suffice
at any fixed matching. In both statements the Weyl piece is redundant. That
redundancy is a statement about a \emph{fixed} set of qubits. On a fixed set of
qubits, a code-preserving Hadamard or swap, when one exists, is already a group
word in the other three classes. As soon as qubit permutations are admitted the
Weyl piece becomes indispensable, and Appendix~\ref{app:permutations} is devoted
to it. Theorem~\ref{thm:autnf} factors every element of the code's automorphism
group $\Gaut$ as a partial duality from $\PDgp$ followed by two transversal
diagonal circuits. In that factorization the $\PDgp$ factor is the one that carries the
permutation.

%% file: app_general.tex
\section{The physical Clifford group is generated by the two Siegel families}
\label{app:general}

Appendix~\ref{app:setup} reduced the question of which physical Clifford gates
preserve a stabilizer code to a question about labels alone. Specifically, a
Clifford admits a Pauli dressing making it preserve the code if and only if its
symplectic action preserves the label space. The code-preserving Clifford
group, taken modulo Pauli operators, is therefore exactly the setwise
stabilizer $\Nfull=\Stb_{\Sp{2n}}(\lab)$ of Theorem~\ref{thm:dressing}. The
same appendix identified the two elementary families that a CSS code always
supplies. The first family consists of circuits of
controlled-$Z$'s and phase gates whose label action is $\Ush{S}$
with $S$ symmetric and $\CX S\subseteq\CZ$ (Corollary~\ref{cor:validity}). The
second family consists of their $X$-basis mirrors, whose label action is
$\Lsh{T}$ with $T$ symmetric and $\CZ T\subseteq\CX$
(Corollary~\ref{cor:dualfamily}). We write $\Eup$ and $\Edn$ for these two
groups. Everything below is a statement about $\Nfull$, $\Eup$ and $\Edn$; Pauli
dressings are never mentioned again.

Throughout this appendix $\lab\subseteq\Vsp$ is an \emph{isotropic and split}
subspace, meaning that
\begin{equation}
  \sfm{\lab}{\lab}=0
  \qquad\text{and}\qquad
  \lab=(\lab\cap\Xh)\oplus(\lab\cap\Zh),
  \label{eq:general-split}
\end{equation}
and we let $\CX,\CZ\subseteq\Ftwo^{n}$ be the supports of the two summands,
$\lab\cap\Xh=(\CX|0)$ and $\lab\cap\Zh=(0|\CZ)$. We abbreviate
$r_X=\dim\CX$, $r_Z=\dim\CZ$, $r=r_X+r_Z=\dim\lab$ and $k=n-r$. Isotropy of a
split $\lab$ says precisely that $\sfm{(a|0)}{(0|b)}=a\cdot b$ vanishes for all
$a\in\CX$ and $b\in\CZ$, i.e.\ that $\CX\perp\CZ$. This orthogonality
$\CX\perp\CZ$ is the condition that makes a pair of classical codes into a CSS
code. Conversely, the label space of a CSS code is split isotropic. The two
diagonal families are then defined by the same formulas as in
Appendix~\ref{app:setup}. The result of this appendix is the following.

\begin{theorem}[Generation by diagonal circuits]
\label{thm:sheargen}
Let $\lab\subseteq\Vsp$ be isotropic and split with $r_X\geq1$ and $r_Z\geq1$.
Then
\begin{equation}
  \Nfull=\langle\Eup,\Edn\rangle .
  \label{eq:general-main}
\end{equation}
\end{theorem}

The hypotheses are worth emphasizing because of what is \emph{not} among them.
The argument uses no self-duality, no divisibility, no distance bound, and
indeed no reference to a code at all. The argument needs only that $\lab$ is
isotropic and split. ``CSS'' enters the proof solely as the splitness of
$\lab$, which is what lets a diagonal circuit see the two halves of $\lab$ separately. The
two-sidedness $r_X,r_Z\geq1$ enters the argument at exactly one point, which is
identified where it occurs.

By construction $\Nfull$ is the largest group of label actions available to a
code-preserving Clifford. Theorem~\ref{thm:sheargen} therefore says that every
other valid gate class is already a group word in the two diagonal families. This
applies in particular to the valid CNOT-type gates, that is to the Levi family
$\Lev$ of Eq.~\eqref{eq:setup-lev}. The same applies to the valid
permutation-automorphism gates of $\Gaut$. Adjoining either class to $\Eup$ and
$\Edn$ does not enlarge the generated group. This is a statement about label
actions, hence about gates modulo Pauli. The statement does not assert that a
CNOT circuit literally equals a product of diagonal and $X$-basis circuits as a
unitary. The statement asserts only that the circuit and the product agree up
to a Pauli factor, which is all that matters for code preservation. We recall
too from the discussion following Corollary~\ref{cor:dualfamily} that a circuit
realizing an $X$-diagonal circuit may be written $H^{\otimes n}\Ush{T}H^{\otimes n}$.
That circuit is a product of commuting $\sqrt{X}$- and $XX$-type gates. No
code-preserving transversal Hadamard is needed for $\Edn$ to be available.

Every $g\in\Nfull$ preserves $\lab$ and therefore also $\lab^{\perp}$. Every
such $g$ thus induces an action on the quotient $\lab^{\perp}/\lab$, a
space of dimension $2k$ on which the symplectic form descends nondegenerately.
This defines a homomorphism
\begin{equation}
  \lact:\Nfull\longrightarrow\Sp{2k},
  \label{eq:general-lact}
\end{equation}
which reads off the logical action of a code-preserving gate and forgets
everything invisible to the encoded qubits;
Section~\ref{app:general-stab} identifies $\lact$ with the projection onto one
factor of the block-diagonal subgroup $\GLg{r}\times\Sp{2k}$, and shows in
particular that it is surjective.
Applying $\lact$ to Eq.~\eqref{eq:general-main} gives
$\lact(\langle\Eup,\Edn\rangle)=\Sp{2k}$: under the hypotheses of
Theorem~\ref{thm:sheargen} the two diagonal families realize \emph{every} logical
Clifford operation on the $k$ encoded qubits. Logical completeness is therefore
not an additional theorem but the image of the present one.

The proof occupies Sections~\ref{app:general-basis}--\ref{app:general-proof}. In
outline: we build a symplectic basis adapted to $\lab$
(Section~\ref{app:general-basis}). We then determine $\Nfull$ completely in that
basis, finding it to be a maximal parabolic subgroup of $\Sp{2n}$ with an
explicit semidirect decomposition (Section~\ref{app:general-stab}). We compute
which diagonal circuits survive, block by block (Section~\ref{app:general-blocks}).
Finally we check that the surviving blocks generate first the block-diagonal
factor and then the unipotent radical (Section~\ref{app:general-proof}).

\subsection{An adapted symplectic basis}
\label{app:general-basis}

It is worth being explicit about what kind of object $\lab$ is, since the entire
argument is linear algebra performed on it. The label map sends $X\mapsto(1|0)$,
$Z\mapsto(0|1)$ and $Y\mapsto(1|1)$ on each qubit, and forgets phases. Under
that map, multiplication of Pauli operators becomes addition of vectors modulo
two, as one sees from the composition rule
$\Pau{a}{b}\Pau{a'}{b'}=(-1)^{b\cdot a'}\Pau{a+a'}{b+b'}$ of
Appendix~\ref{app:setup}. Consequently the image of the \emph{entire} stabilizer
group $\stabgp$ is the $\Ftwo$-linear span of the labels of any generating set.
That is, the multiplicative closure of the group is the additive closure of its
labels. Thus $\lab$ is an $\Ftwo$-subspace of $\Vsp$ --- a classical binary
code --- and not an algebra. The only piece of Pauli data not carried by the
linear structure of $\lab$ is commutation, which is recorded by the symplectic
form. Moreover, the statement that the operators of $\stabgp$ commute pairwise
is exactly the statement $\sfm{\lab}{\lab}=0$ appearing
in Eq.~\eqref{eq:general-split}.

We now choose coordinates in which $\lab$ is a coordinate subspace. Pick a basis
$u_1,\dots,u_n$ of $\Ftwo^{n}$ and let $w_1,\dots,w_n$ be its Euclidean dual
basis, $u_i\cdot w_j=\delta_{ij}$. Setting
\begin{equation}
  X_i:=(u_i|0),\qquad Z_j:=(0|w_j),
  \label{eq:general-symplbasis}
\end{equation}
we obtain a symplectic basis of $\Vsp$, since
$\sfm{X_i}{Z_j}=u_i\cdot w_j=\delta_{ij}$ while
$\sfm{X_i}{X_{i'}}=\sfm{Z_j}{Z_{j'}}=0$.

The point of splitness is that the $u_i$ may be chosen so that $\lab$ is spanned
by basis vectors. Since $\lab$ is isotropic and split we have $\CX\perp\CZ$,
that is $\CX\subseteq\CZ^{\perp}$. Choose a basis of $\CX$. Extend that basis
to a basis of the larger space $\CZ^{\perp}$, and then extend further to a basis
$u_1,\dots,u_n$ of all of $\Ftwo^{n}$. Index the first $r_X$ vectors by a set
$\IX$, the $r_Z$ vectors added in the last step by a set $\IZ$, and the
$k=n-r_X-r_Z$ vectors added in the middle step by a set $\IL$; after reordering
we may take
\begin{equation}
\begin{split}
  \IX&=\{1,\dots,r_X\},\qquad \IZ=\{r_X+1,\dots,r\},\\
  \IL&=\{r+1,\dots,n\},
\end{split}
\label{eq:general-indexsets}
\end{equation}
so that $\{1,\dots,n\}=\IX\sqcup\IZ\sqcup\IL$. By construction
$\CX=\spano(u_a:a\in\IX)$ and $\spano(u_i:i\notin\IZ)=\CZ^{\perp}$. The dual
basis then produces $\CZ$ with no further work: a vector
$v=\sum_{b\in\IZ}c_bw_b$ satisfies $v\cdot u_i=c_i$ for $i\in\IZ$ and
$v\cdot u_i=0$ for $i\notin\IZ$, whence
$\spano(w_b:b\in\IZ)\subseteq(\spano(u_i:i\notin\IZ))^{\perp}$; the two sides
have the same dimension $r_Z$, so
\begin{equation}
  \spano(w_b:b\in\IZ)=(\CZ^{\perp})^{\perp}=\CZ .
  \label{eq:general-dualbasis}
\end{equation}
In the resulting basis $\CX=\spano(u_a:a\in\IX)$ and $\CZ=\spano(w_b:b\in\IZ)$,
and correspondingly $\lab\cap\Xh=\spano(X_a:a\in\IX)$ and
$\lab\cap\Zh=\spano(Z_b:b\in\IZ)$.

The sets $\IX,\IZ,\IL$ are sets of
\emph{qubit indices}, subsets of $\{1,\dots,n\}$; they are not sets of
operators. At \emph{each} index $i$ the adapted basis contains \emph{two
distinct} vectors, the pure-$X$ vector $X_i$ and the pure-$Z$ vector $Z_i$,
which are labels of different Pauli operators. What distinguishes the three
index sets is not how many vectors sit at an index but which of the two vectors
belongs to $\lab$. Table~\ref{tab:general-adapted} records this distinction.
Thus each of $\IX,\IZ,\IL$ indexes a \emph{conjugate pair} of basis vectors. A
stabilizer and its destabilizer share only an index: as vectors $X_a\neq Z_a$,
and indeed
$\sfm{X_a}{Z_a}=1$, so they anticommute as operators.

\begin{table}[htb]
\caption{\label{tab:general-adapted}%
The two adapted basis vectors sitting at a single qubit index $i$, and the
subspace each belongs to, for $i$ in each of the three index
sets~\eqref{eq:general-indexsets}. An index labels a conjugate pair; a
stabilizer and its destabilizer share only that index and are different Paulis.}
\begin{tabular}{lll}
\toprule
index & $X_i$ & $Z_i$\\
\midrule
$a\in\IX$ & $X$-stabilizer $(\in\lab)$ & destabilizer $(\in\dst)$\\
$b\in\IZ$ & destabilizer $(\in\dst)$ & $Z$-stabilizer $(\in\lab)$\\
$p\in\IL$ & logical $X$ $(\in\logsp)$ & logical $Z$ $(\in\logsp)$\\
\bottomrule
\end{tabular}
\end{table}

Throughout, $\lab^{\perp}$ denotes the \emph{symplectic} perpendicular of
$\lab$,
\begin{equation}
  \lab^{\perp}=\bigl\{v\in\Vsp:\sfm{v}{\ell}=0\ \text{ for all }\ell\in\lab\bigr\},
  \label{eq:general-perpdef}
\end{equation}
equivalently the set of labels of the Pauli operators commuting with every
element of $\stabgp$. It is not the Euclidean orthogonality of
\eqref{eq:setup-css}, which relates the two classical codes and happens to be
written with the same symbol. Isotropy of $\lab$ is the statement
$\lab\subseteq\lab^{\perp}$.

With this convention the four subspaces we shall use are
\begin{subequations}
\label{eq:general-spaces}
\begin{align}
  \lab &=\spano\{X_a:a\in\IX\}\nonumber\\
  &\qquad\oplus\spano\{Z_b:b\in\IZ\},
  \label{eq:general-spaces-L}\\
  \lab^{\perp}&=\spano\{X_i:i\in\IX\cup\IL\}\nonumber\\
  &\qquad\oplus\spano\{Z_j:j\in\IZ\cup\IL\},
  \label{eq:general-spaces-Lperp}\\
  \dst &=\spano\{Z_a:a\in\IX\}\nonumber\\
  &\qquad\oplus\spano\{X_b:b\in\IZ\},
  \label{eq:general-spaces-D}\\
  \logsp &=\spano\{X_p,Z_p:p\in\IL\},
  \label{eq:general-spaces-Q}
\end{align}
\end{subequations}
of dimensions $r$, $r+2k$, $r$ and $2k$ respectively. The description of
$\lab^{\perp}$ is immediate: writing $v=\sum_ix_iX_i+\sum_jz_jZ_j$ we have
$\sfm{v}{X_a}=z_a$ and $\sfm{v}{Z_b}=x_b$, so $v\in\lab^{\perp}$ iff $z_a=0$ for
all $a\in\IX$ and $x_b=0$ for all $b\in\IZ$. The subspace $\dst$ is the chosen
\emph{destabilizer} complement, isotropic and paired with $\lab$ by the
symplectic form. The subspace $\logsp$ is the chosen \emph{logical} complement,
a concrete representative of the quotient $\lab^{\perp}/\lab$ on which the form
restricts nondegenerately. Counting dimensions, $r+2k+r=2n$, and one checks
directly from Table~\ref{tab:general-adapted} that the three subspaces intersect
trivially, so that
\begin{equation}
  \Vsp=\lab\oplus\logsp\oplus\dst .
  \label{eq:general-decomp}
\end{equation}

\subsection{The stabilizer of a split isotropic subspace}
\label{app:general-stab}

Every $g\in\Nfull$ preserves $\lab$ setwise and therefore also
$\lab^{\perp}$. It consequently induces two quotient actions. The first is a
\emph{stabilizer relabeling} $\gau:=g|_{\lab}\in\GLg{r}$, which changes the
chosen generating set of $\stabgp$ and is invisible to the encoded information.
We call $\gau$ simply the relabeling below. The second is a \emph{logical}
action $\lgc:=g|_{\lab^{\perp}/\lab}$, represented on $\logsp$. In the ordered
basis $(X_p)_{p\in\IL},(Z_p)_{p\in\IL}$ of $\logsp$ the symplectic form has Gram
matrix $\Jm$ of size $2k$, so $\lgc\in\Sp{2k}$. We now determine all of
$\Nfull$, not merely these two quotients.

Order the decomposition~\eqref{eq:general-decomp} as $\lab,\logsp,\dst$ and use
the bases listed in Eq.~\eqref{eq:general-spaces}. Order the basis of $\dst$ so
that $Z_a$ is dual to $X_a$ and $X_b$ is dual to $Z_b$. In these coordinates
the symplectic form has Gram matrix
\begin{equation}
  \begin{pmatrix}
    0 & 0 & \Id\\
    0 & \Jm & 0\\
    \Id & 0 & 0
  \end{pmatrix},
  \label{eq:general-gram}
\end{equation}
the three vanishing diagonal blocks recording that $\lab$ and $\dst$ are
isotropic and that $\logsp\perp\lab$ and $\logsp\perp\dst$, and the two identity
blocks recording the duality of $\lab$ and $\dst$. (Producing such a
decomposition uses only the isotropy of $\lab$. Splitness is not needed again
until Section~\ref{app:general-blocks}.)

\begin{lemma}[Block form of $\Nfull$]
\label{lem:general-blockform}
With rows and columns ordered $\lab,\logsp,\dst$, a matrix $g\in\Sp{2n}$ lies in
$\Nfull$ if and only if it has the form
\begin{equation}
  g=\begin{pmatrix}
      \gau & 0 & 0\\
      \lgc\Jm\radF^{\tp}\gau & \lgc & 0\\
      \radT & \radF & \gau^{\itp}
    \end{pmatrix},
  \label{eq:general-blockform}
\end{equation}
where $\gau\in\GLg{r}$ and $\lgc\in\Sp{2k}$ are arbitrary, $\radF$ is an
arbitrary $r\times 2k$ matrix, and the $r\times r$ matrix $\radT$ satisfies
\begin{equation}
  \gau^{\itp}\radT^{\tp}+\radT\gau^{-1}=\radF\Jm\radF^{\tp}.
  \label{eq:general-radconstraint}
\end{equation}
Moreover~\eqref{eq:general-radconstraint} is solvable for every $\radF$, its
solution set being an affine space of dimension $r(r+1)/2$.
\end{lemma}

\begin{proof}
Vectors are rows and matrices act on the right. Writing
$v=(v_1,v_2,v_3)$ for the components of $v$ in $\lab\oplus\logsp\oplus\dst$, the
image $vg$ has $\lab$-component $v_1g_{11}+v_2g_{21}+v_3g_{31}$ and so on.
Preservation of $\lab$ forces $g_{12}=g_{13}=0$, since a vector $(v_1,0,0)$
must map into $\lab$. Preservation of $\lab^{\perp}=\lab\oplus\logsp$ then
forces $g_{23}=0$. Hence $g$ is block lower triangular, and we name the
surviving blocks $\gau:=g_{11}$, $\lgc:=g_{22}$, $\radT:=g_{31}$,
$\radF:=g_{32}$, leaving
$g_{21}$ and $g_{33}$ to be determined. That $\gau$ is invertible follows
because $g$ is and $\lab$ is $g$-invariant, and $\gau$ is by definition the
relabeling. Likewise $\lgc$ is the induced action on $\lab^{\perp}/\lab$.

Now impose that $g$ preserve the form~\eqref{eq:general-gram}. Writing
$\sfm{u}{v}=u_1\cdot v_3+u_2\Jm v_2^{\tp}+u_3\cdot v_1$ and denoting by
$R_1,R_2,R_3$ the three block rows of $g$, the condition is that
$\sfm{R_m}{R_{m'}}$ reproduce~\eqref{eq:general-gram}. Because the form is
symmetric, it suffices to impose the six conditions with $m\leq m'$, which we
take in turn. The conditions
$\sfm{R_1}{R_1}=0$ and $\sfm{R_1}{R_2}=0$ hold identically, since the second and
third components of $R_1$ vanish. The condition $\sfm{R_1}{R_3}=\Id$ reads
$\gau g_{33}^{\tp}=\Id$ and therefore fixes the contragredient block
$g_{33}=\gau^{\itp}$. The condition $\sfm{R_2}{R_2}=\Jm$ reads
$\lgc\Jm\lgc^{\tp}=\Jm$, i.e.\ $\lgc\in\Sp{2k}$. The condition
$\sfm{R_2}{R_3}=0$ reads $g_{21}g_{33}^{\tp}+\lgc\Jm\radF^{\tp}=0$. Since
$g_{33}^{\tp}=\gau^{-1}$ and we are in characteristic two, that condition
determines the lower-left block of the middle row,
$g_{21}=\lgc\Jm\radF^{\tp}\gau$, as displayed in Eq.~\eqref{eq:general-blockform}.
Finally $\sfm{R_3}{R_3}=0$ reads
$\radT\gau^{-1}+\radF\Jm\radF^{\tp}+\gau^{\itp}\radT^{\tp}=0$, which
is~\eqref{eq:general-radconstraint}. Conversely, a matrix of the displayed shape
satisfying these equations preserves both~\eqref{eq:general-gram} and $\lab$,
hence lies in $\Nfull$.

It remains to see that~\eqref{eq:general-radconstraint} can always be solved.
Substituting $Y:=\radT\gau^{-1}$ turns it into $Y+Y^{\tp}=\radF\Jm\radF^{\tp}$.
The right-hand side is symmetric, because $\Jm$ is, and has zero diagonal. The
$i$th diagonal entry of $\radF\Jm\radF^{\tp}$ is the symplectic form of the
$i$th row of $\radF$ with itself, which vanishes because that form is
alternating. The map $Y\mapsto Y+Y^{\tp}$ sends $r\times r$ matrices onto
exactly the symmetric matrices with zero diagonal. A preimage of such a $W$ is
the strictly lower triangular part of $W$. The kernel of the map is
$\Symm{r}$. Hence a solution $Y$ exists for every $\radF$, and the solutions
form a coset of $\Symm{r}$, of dimension
$r(r+1)/2$. Multiplying on the right by $\gau$ returns $\radT$.
\end{proof}

Setting $\gau=\Id$ and $\lgc=\Id$ in Eq.~\eqref{eq:general-blockform} isolates the
subgroup $\Urad$ of all matrices
\begin{equation}
  u(\radF,\radT):=
  \begin{pmatrix}
    \Id & 0 & 0\\
    \Jm\radF^{\tp} & \Id & 0\\
    \radT & \radF & \Id
  \end{pmatrix},
  \quad \radT+\radT^{\tp}=\radF\Jm\radF^{\tp},
  \label{eq:general-urad}
\end{equation}
with $\radF$ free, the displayed constraint
being~\eqref{eq:general-radconstraint} at $\gau=\Id$. By the dimension count in
Lemma~\ref{lem:general-blockform},
\begin{equation}
\begin{split}
  |\Urad|&=2^{2kr+r(r+1)/2},\\
  |\Nfull|&=|\Urad|\cdot|\GLg{r}|\cdot|\Sp{2k}| .
\end{split}
\label{eq:general-orders}
\end{equation}
Explicitly, over $\Ftwo$ one has $|\GLg{r}|=2^{r(r-1)/2}\prod_{i=1}^{r}(2^{i}-1)$
and $|\Sp{2k}|=2^{k^{2}}\prod_{i=1}^{k}(4^{i}-1)$, so the order is
\begin{equation}
  |\Nfull|=2^{n^{2}}\prod_{i=1}^{r}\bigl(2^{i}-1\bigr)
                    \prod_{i=1}^{k}\bigl(4^{i}-1\bigr),
  \label{eq:general-orderclosed}
\end{equation}
in the code parameters alone, $n=r+k$ being the number of qubits. As a check,
$r=0$ returns $|\Nfull|=|\Sp{2n}|$, consistent with $\lab=0$ being preserved by
all of $\Sp{2n}$. Likewise $k=0$ returns $2^{n^{2}}\prod_{i=1}^{n}(2^{i}-1)$,
the order of the Siegel parabolic that stabilizes a Lagrangian.

The block-diagonal matrices $\gau\oplus\lgc\oplus\gau^{\itp}$, which
are exactly the elements of Eq.~\eqref{eq:general-blockform} with $\radF=0$ and
$\radT=0$, form a subgroup isomorphic to $\GLg{r}\times\Sp{2k}$. Given any
$g\in\Nfull$ with data $(\gau,\lgc,\radF,\radT)$, multiply on the right by the
inverse of $\gau\oplus\lgc\oplus\gau^{\itp}$. The result is an element with
trivial relabeling and logical parts, that is an element of $\Urad$. A
block-diagonal element lying in $\Urad$ is the identity. Hence
\begin{equation}
  \Nfull=\Urad\rtimes\left(\GLg{r}\times\Sp{2k}\right).
  \label{eq:general-levi}
\end{equation}
The subgroup $\Urad$ is unipotent, normal, and a $2$-group, and it is the
unipotent radical of $\Nfull$, while the block-diagonal subgroup is a complement
to it; \eqref{eq:general-levi} exhibits $\Nfull$ as the maximal parabolic
subgroup of $\Sp{2n}$ stabilizing the isotropic subspace $\lab$, split over its
radical. The general theory of parabolic subgroups, their structure theory, and
the stabilizers of isotropic subspaces in the classical groups are treated in
the references below. The reader may consult Ref.~\cite{carter1972simple},
Ch.~8;
Ref.~\cite{malle2011linear}, \S\S12 and~17;
Ref.~\cite{taylor1992geometry}, Ch.~8; and
Ref.~\cite{grove2002classical}.

The consequence of Eq.~\eqref{eq:general-levi} that we shall use is the following.
That consequence is the reason the whole proof reduces to a finite list of block
computations.

\begin{corollary}[Block-diagonal-plus-radical criterion]
\label{cor:general-criterion}
A subgroup of $\Nfull$ that contains $\Urad$ and contains every block-diagonal
element $\gau\oplus\lgc\oplus\gau^{\itp}$ with $\gau\in\GLg{r}$ and
$\lgc\in\Sp{2k}$ is equal to $\Nfull$.
\end{corollary}

\begin{proof}
By the factorization established just above, every $g\in\Nfull$ is a product of
an element of $\Urad$ and a block-diagonal element. A subgroup containing both
sets contains all such products.
\end{proof}

\subsection{Stabilizer-logical-destabilizer block form of $Z$-diagonal circuits}
\label{app:general-blocks}

We must now express the two diagonal families in the adapted basis. We then read
off which parameters of those families survive the requirement that $\lab$ be
preserved.

Two bookkeeping remarks come first. The parameter $S$ of a $Z$-diagonal circuit
$\Ush{S}$ is an $n\times n$ symmetric matrix indexed by \emph{qubit indices},
that is by $\IX\sqcup\IZ\sqcup\IL$ with block sizes $r_X,r_Z,k$. The row $i$ of
$S$ refers to $X_i$ and the column $j$ to $Z_j$. The parameter $S$ is therefore
\emph{not} indexed by the three subspaces $\lab,\logsp,\dst$ of
Section~\ref{app:general-stab}, even though both decompositions carry three
blocks. The two decompositions must never be identified position by position.
Indeed, by Table~\ref{tab:general-adapted} each qubit index carries a conjugate
pair. So in $S$ the $\IZ$-rows refer to the $X$-type \emph{destabilizers} $X_b$
and the $\IX$-columns to the $Z$-type destabilizers $Z_a$. That is how
destabilizers enter a diagonal-circuit parameter at all. Second, passing from the standard
basis to the adapted basis
of Eq.~\eqref{eq:general-symplbasis} is effected by the symplectic change of
coordinates $(x|z)\mapsto(xK|zK^{\itp})$. The matrix $K$ here has rows
$u_1,\dots,u_n$. Under that change of coordinates a $Z$-diagonal circuit remains an
$Z$-diagonal circuit, with parameter $S\mapsto KSK^{\tp}$. This is an invertible
congruence of $\Symm{n}$, so describing the valid parameters in the adapted
basis describes the whole family.

\begin{lemma}[Valid $Z$-diagonal circuits]
\label{lem:general-uppershears}
In the adapted basis, $\Ush{S}\in\Nfull$ if and only if the $(\IX,\IX)$ and
$(\IX,\IL)$ blocks of $S$ vanish, i.e.\ if and only if
\begin{equation}
  S=\begin{pmatrix}
      0 & S_{XZ} & 0\\
      S_{XZ}^{\tp} & S_{ZZ} & S_{ZL}\\
      0 & S_{ZL}^{\tp} & S_{LL}
    \end{pmatrix}
  \label{eq:general-shearblocks}
\end{equation}
with rows and columns ordered $\IX,\IZ,\IL$, where $S_{XZ}$ is an arbitrary
$r_X\times r_Z$ matrix, $S_{ZL}$ an arbitrary $r_Z\times k$ matrix, and
$S_{ZZ}\in\Symm{r_Z}$, $S_{LL}\in\Symm{k}$ arbitrary symmetric matrices.
Consequently
\begin{equation}
\begin{split}
  \log_2|\Eup|&=r_Xr_Z+r_Zk\\
    &\quad+\tfrac{1}{2}r_Z(r_Z+1)+\tfrac{1}{2}k(k+1).
\end{split}
\label{eq:general-shearcount}
\end{equation}
\end{lemma}

\begin{proof}
The $Z$-diagonal circuit acts by $(x|z)\mapsto(x|z+xS)$. It therefore fixes every $Z_j$
and sends $X_i\mapsto X_i+\sum_jS_{ij}Z_j$. Since $\Ush{S}$ is invertible, it
preserves $\lab$ iff it maps the basis~\eqref{eq:general-spaces-L} of $\lab$
into $\lab$. The vectors $Z_b$, $b\in\IZ$, are fixed, so the only condition
comes from $X_a$ with $a\in\IX$: we need
$X_a+\sum_jS_{aj}Z_j\in\lab$, equivalently $\sum_jS_{aj}Z_j\in\lab$. This is a
pure-$Z$ vector, and here splitness enters decisively: $\lab\cap\Zh$ is spanned
by the $Z_b$ with $b\in\IZ$. So the condition is $S_{aj}=0$ for all $a\in\IX$
and all $j\notin\IZ$. By symmetry of $S$ the transposed blocks vanish
too, which is exactly~\eqref{eq:general-shearblocks}. The four remaining blocks
are unconstrained apart from the symmetry of $S_{ZZ}$ and $S_{LL}$, and distinct
parameters give distinct diagonal circuits, whence the
count~\eqref{eq:general-shearcount}.
\end{proof}

Diagonal circuit parameters add, in the sense that $\Ush{S}\Ush{S'}=\Ush{S+S'}$. Hence the
group $\Eup$ is elementary abelian, and every valid $Z$-diagonal circuit is a product of
four commuting diagonal circuits. We obtain the four diagonal circuits by turning on one of
$S_{XZ},S_{ZZ},S_{ZL},S_{LL}$ at a time, together with the transpose that
symmetry forces. It therefore suffices to locate each of these four one-block
diagonal circuits inside the block
form~\eqref{eq:general-blockform}. Reading off the images of the basis vectors
and consulting Table~\ref{tab:general-adapted} for the subspace each image
component belongs to, we obtain the following dictionary.

The block $S_{LL}$ sends $X_p\mapsto X_p+\sum_qS_{LL,pq}Z_q$ and fixes every
other basis vector, so it acts inside $\logsp$ alone. The resulting element is
block diagonal, with
$\lgc=\left(\begin{smallmatrix}\Id&S_{LL}\\0&\Id\end{smallmatrix}\right)$,
a $Z$-diagonal circuit of $\Sp{2k}$, and with $\gau=\Id$, $\radF=0$, $\radT=0$.
The block $S_{XZ}$ sends $X_a\mapsto X_a+\sum_bS_{XZ,ab}Z_b$, which stays inside
$\lab$, and $X_b\mapsto X_b+\sum_aS_{XZ,ab}Z_a$, which stays inside $\dst$. The
block $S_{XZ}$ is therefore also block diagonal, a pure relabeling element with
$\gau=\left(\begin{smallmatrix}\Id&S_{XZ}\\0&\Id\end{smallmatrix}\right)$ in the
basis $(X_a;Z_b)$ of $\lab$. That element carries the contragredient action on
$\dst$, and has $\lgc=\Id$. The block $S_{ZZ}$ sends
$X_b\mapsto X_b+\sum_{b'}S_{ZZ,bb'}Z_{b'}$ and nothing else, mapping $\dst$ into
$\lab$. The block $S_{ZZ}$ is thus the element $u(0,\radT)$ of $\Urad$.
Here $\radT$ is the $r\times r$ matrix with $S_{ZZ}$ in the $(\IZ,\IZ)$ position
and zeros elsewhere, and it is symmetric because $S_{ZZ}$ is. Finally the block
$S_{ZL}$ sends $X_b\mapsto X_b+\sum_pS_{ZL,bp}Z_p$, mapping $\dst$ into
$\logsp$, and $X_p\mapsto X_p+\sum_bS_{ZL,bp}Z_b$, mapping $\logsp$ into $\lab$.
The block $S_{ZL}$ is accordingly the element $u(\radF,0)$ of $\Urad$. Here
$\radF$ has the entries of $S_{ZL}$ in the rows indexed by the $X_b$ and in the
columns indexed by the $Z_p$.

This last case illustrates why the two block decompositions must be kept apart.
The block $S_{ZL}$ joins the qubit-index sets $\IZ$ and $\IL$. Suppose one
misread those index sets as the subspaces $\lab$ and $\logsp$. One would then
conclude that $S_{ZL}$ occupies the forbidden $(\lab,\logsp)$ position
of Eq.~\eqref{eq:general-blockform}, which vanishes identically for every element of
$\Nfull$. What the block $S_{ZL}$ actually does is add logical $Z_p$ components
to the destabilizers $X_b$ and stabilizer $Z_b$ components to the logical
vectors $X_p$. So that block occupies the allowed $(\dst,\logsp)$ and
$(\logsp,\lab)$ positions --- the parameter $\radF$ and the block
$\lgc\Jm\radF^{\tp}\gau$ that $\radF$ determines. (Consistently,
$\radT=0$ is permitted here: with $\radF$ supported in the $Z_p$-columns only,
$\radF\Jm\radF^{\tp}=0$, so~\eqref{eq:general-radconstraint} is satisfied.) The
one-block diagonal circuits therefore distribute as follows. The block $S_{LL}$ is a purely
logical diagonal circuit, and $S_{XZ}$ a purely relabeling element. The blocks $S_{ZZ}$ and
$S_{ZL}$ lie in $\Urad$. In particular the map sending a $Z$-diagonal circuit to its
relabeling component is already surjective onto the maps $\CX\to\CZ$, since
$S_{XZ}$ is unconstrained. Moreover, any one block may be turned on with the
other three set to zero.

For the $X$-diagonal circuits we do not repeat the computation. Conjugation by the swap
$\sw$ interchanges $Z$-diagonal and $X$-diagonal circuits, $\sw\Ush{S}\sw=\Lsh{S}$. That
conjugation also carries a split isotropic $\lab$ to the split isotropic
subspace $\lab\sw$. The two halves of $\lab\sw$ are those of $\lab$ with the
roles of $\CX$ and $\CZ$, hence of $\IX$ and $\IZ$, interchanged. Applying
Lemma~\ref{lem:general-uppershears} to $\lab\sw$ and conjugating back gives at
once: $\Lsh{T}\in\Nfull$ if and only if the
$(\IZ,\IZ)$ and $(\IZ,\IL)$ blocks of $T$ vanish,
\begin{equation}
  T=\begin{pmatrix}
      T_{XX} & T_{XZ} & T_{XL}\\
      T_{XZ}^{\tp} & 0 & 0\\
      T_{XL}^{\tp} & 0 & T_{LL}
    \end{pmatrix},
  \label{eq:general-lowerblocks}
\end{equation}
with $T_{XX}\in\Symm{r_X}$, $T_{LL}\in\Symm{k}$ and $T_{XZ}$, $T_{XL}$
arbitrary of sizes $r_X\times r_Z$ and $r_X\times k$; consequently
$\log_2|\Edn|=r_Xr_Z+\frac{1}{2}r_X(r_X+1)+r_Xk+\frac{1}{2}k(k+1)$. The mirrored
dictionary reads as follows. The block $T_{LL}$ gives the $X$-diagonal circuits
$\left(\begin{smallmatrix}\Id&0\\T_{LL}&\Id\end{smallmatrix}\right)$ of
$\Sp{2k}$ as block-diagonal elements. The block $T_{XZ}$ gives the
block-diagonal relabeling elements with
$\gau=\left(\begin{smallmatrix}\Id&0\\T_{XZ}^{\tp}&\Id\end{smallmatrix}\right)$,
that is the relabelings $\CZ\to\CX$. The blocks $T_{XX}$ and $T_{XL}$ lie in
$\Urad$, contributing respectively to $\radT$ and to $\radF$.

\subsection{Proof of the theorem}
\label{app:general-proof}

By Corollary~\ref{cor:general-criterion} it suffices to prove that
$\langle\Eup,\Edn\rangle$ contains the block-diagonal subgroup and the
unipotent radical $\Urad$. We treat the two in turn.

\begin{proof}[Proof of part (a)]
For the block-diagonal subgroup we must show the following. The group generated
by the two diagonal families contains the element
$\gau\oplus\lgc\oplus\gau^{\itp}$ for every $\gau\in\GLg{r}$ and every
$\lgc\in\Sp{2k}$. The block-diagonal subgroup is the
direct product $\GLg{r}\times\Sp{2k}$. By the dictionary of
Section~\ref{app:general-blocks} the generators used below are genuinely block
diagonal, with all other parameter blocks set to zero. Because that subgroup is
a direct product and those generators are block diagonal, it is enough to
realize each factor separately.

Consider first the logical factor. If $k=0$ there is nothing to prove, as
$\Sp{0}$ is trivial. If $k\geq1$, the $S_{LL}$ blocks of
Lemma~\ref{lem:general-uppershears} realize every $Z$-diagonal circuit
$\left(\begin{smallmatrix}\Id&S_{LL}\\0&\Id\end{smallmatrix}\right)$ of
$\Sp{2k}$ with $S_{LL}\in\Symm{k}$. Still assuming $k\geq1$, the $T_{LL}$
blocks of Eq.~\eqref{eq:general-lowerblocks} realize every $X$-diagonal circuit. These two
families are the unipotent radicals of the two opposite Siegel parabolic
subgroups of $\Sp{2k}$, in the sense of Section~\ref{app:setup-classes}. The
polarization of $\logsp$ determines those parabolic subgroups. We claim that
the two families generate $\Sp{2k}$.

Write $\Uupg$ and $\Udng$ for the two families. Recall that $\Sp{2k}$ is
generated by its symplectic transvections
$T_{u}:v\mapsto v+\sfm{v}{u}\,u$, one for each nonzero $u$~\cite[Ch.~III]{artin1957geometric}\cite[Ch.~8]{taylor1992geometry}. A
transvection with $u$ in the pure-$Z$ half of $\logsp$, say $u=(0|u_{0})$,
satisfies $\sfm{(x|z)}{u}=x\cdot u_{0}$ and therefore acts by
$(x|z)\mapsto(x\,|\,z+x\,u_{0}^{\tp}u_{0})$, so that
$T_{u}=\Ush{u_{0}^{\tp}u_{0}}\in \Uupg$. Dually $T_{u}\in \Udng$ when $u$ lies
in the pure-$X$ half. Next, $\langle \Uupg,\Udng\rangle$ is transitive on the
nonzero vectors of $\logsp$. The key point is that for a fixed nonzero $x$ the
vector $xS$ runs over all of $\Ftwo^{k}$ as $S$ runs over $\Symm{k}$: given a
target $y$ and an index $i$ with $x_{i}=1$, the matrix
\begin{equation}
  S=e_{i}^{\tp}y+y^{\tp}e_{i}+(x\cdot y)\,e_{i}^{\tp}e_{i}
  \label{eq:general-witness}
\end{equation}
is symmetric and satisfies $xS=y$, as one checks using $x\cdot e_{i}=1$. An
$Z$-diagonal circuit therefore carries $(x|z)$ with $x\neq0$ to any $(x|z')$, and an $X$-diagonal
circuit carries $(x|z')$ with $z'\neq0$ to any $(y|z')$. Alternating the two
diagonal circuits joins any nonzero vector to any other. Finally $gT_{u}g^{-1}=T_{ug^{-1}}$ for
$g\in\Sp{2k}$, since $\sfm{vg}{ug}=\sfm{v}{u}$. Conjugating the transvections
already in $\langle \Uupg,\Udng\rangle$ by that transitive action therefore
produces
$T_{u}$ for \emph{every} nonzero $u$, and these generate $\Sp{2k}$. Hence
$\Id\oplus\lgc\oplus\Id\in\langle\Eup,\Edn\rangle$ for every $\lgc\in\Sp{2k}$.

Consider next the relabeling factor. Order the basis of $\lab$ as $(X_a)_{a\in\IX}$
followed by $(Z_b)_{b\in\IZ}$, so that $\GLg{r}$ is written in $(r_X,r_Z)$-block
form. By Section~\ref{app:general-blocks} the $S_{XZ}$ blocks realize
\emph{every} relabeling of the form $X_a\mapsto X_a+\sum_bS_{XZ,ab}Z_b$. These
relabelings are the full unipotent radical
$\left(\begin{smallmatrix}\Id&*\\0&\Id\end{smallmatrix}\right)$ of the
$(r_X,r_Z)$-block parabolic subgroup of $\GLg{r}$. The $T_{XZ}$ blocks
realize the opposite radical
$\left(\begin{smallmatrix}\Id&0\\ *&\Id\end{smallmatrix}\right)$. Over $\Ftwo$
the group $\GLg{r}$ coincides with the special linear group and is generated by
the elementary transvections $\Id+e_{ij}$ with $i\neq j$. Here $e_{ij}$ denotes
the matrix unit with a single one in position $(i,j)$.
Those transvections whose two indices lie in different blocks are precisely the
one-entry $S_{XZ}$ and $T_{XZ}$ generators just exhibited. For a transvection
whose two indices $i\neq l$ both lie in the $\IX$-block, choose any index $j$ in
the $\IZ$-block and use the \emph{group} commutator identity
\begin{equation}
  [\Id+e_{ij},\,\Id+e_{jl}]=\Id+e_{il},
  \qquad i,j,l \text{ distinct},
  \label{eq:general-commutator}
\end{equation}
in which $[x,y]:=xyx^{-1}y^{-1}$. One checks \eqref{eq:general-commutator}
directly. Since $i\neq j$ we have $e_{ij}^{2}=0$, so $(\Id+e_{ij})^{2}=\Id$ and
each of the two factors is its own inverse; hence $[x,y]=(xy)^{2}$ for
$x=\Id+e_{ij}$ and $y=\Id+e_{jl}$. Now $e_{ij}e_{jl}=e_{il}$ while every other
product of two of the three matrix units $e_{ij},e_{jl},e_{il}$ vanishes,
because $i,j,l$ are distinct. Writing $xy=\Id+N$ with $N=e_{ij}+e_{jl}+e_{il}$
we therefore get $N^{2}=e_{il}$ and $(xy)^{2}=\Id+N^{2}=\Id+e_{il}$, as
claimed. (The ring commutator $xy-yx$ equals $e_{il}$ rather than $\Id+e_{il}$,
and would not be an element of $\GLg{r}$.)
Both factors on the left of Eq.~\eqref{eq:general-commutator} are off-block
transvections, hence lie in $\langle\Eup,\Edn\rangle$ as block-diagonal
elements. The commutator of those two factors is again block diagonal, and its
relabeling component is $\Id+e_{il}$ because the relabeling map
$\Nfull\to\GLg{r}$ is a homomorphism. The same argument with the roles of the
blocks exchanged produces the transvections with both indices in the
$\IZ$-block. That argument uses an index in the $\IX$-block as the intermediate
one. This is one place where the hypotheses $r_X\geq1$ and $r_Z\geq1$ are
used, and they are used exactly to supply these two bridge indices. Hence
$\gau\oplus\Id\oplus\gau^{\itp}\in\langle\Eup,\Edn\rangle$ for every
$\gau\in\GLg{r}$, and multiplying the two families of block-diagonal elements
gives the whole block-diagonal subgroup.
\end{proof}

\begin{proof}[Proof of part (b)]
The unipotent radical carries a two-step central filtration, which we now
describe. Multiplying two elements~\eqref{eq:general-urad} gives
\begin{equation}
\begin{split}
  &u(\radF,\radT)\,u(\radF',\radT')\\
  &\qquad=u\!\left(\radF+\radF',\;\radT+\radT'+\radF\Jm\radF'^{\tp}\right),
\end{split}
\label{eq:general-mult}
\end{equation}
as one checks by multiplying the block matrices: the middle row of the product
is $(\Jm\radF^{\tp}+\Jm\radF'^{\tp},\Id,0)$ and its bottom row is
$(\radT+\radF\Jm\radF'^{\tp}+\radT',\radF+\radF',\Id)$. Setting $\radF=0$
in Eq.~\eqref{eq:general-urad} leaves the subgroup
$\Urad_0:=\{u(0,\radT):\radT\in\Symm{r}\}$, which by Eq.~\eqref{eq:general-mult} is
central in $\Urad$. The map $u(\radF,\radT)\mapsto\radF$ is a homomorphism onto
the additive group $\Homo(\dst,\logsp)$ of $r\times2k$ matrices with kernel
$\Urad_0$. Thus there is an exact sequence
\begin{equation}
  1\longrightarrow\Urad_0\longrightarrow\Urad
   \longrightarrow\Homo(\dst,\logsp)\longrightarrow1 .
  \label{eq:general-exact}
\end{equation}
When $k=0$ the right-hand term is trivial and $\Urad=\Urad_0$.

The block-diagonal subgroup acts on the two layers by conjugation. A direct computation
with Eq.~\eqref{eq:general-blockform} and~\eqref{eq:general-urad} gives, for
$\ell=\gau\oplus\lgc\oplus\gau^{\itp}$,
\begin{equation}
  \ell\,u(\radF,\radT)\,\ell^{-1}
  =u\!\left(\gau^{\itp}\radF\lgc^{-1},\;\gau^{\itp}\radT\gau^{-1}\right),
  \label{eq:general-conj}
\end{equation}
where the consistency of the middle row uses $\Jm\lgc^{\itp}=\lgc\Jm$, itself a
restatement of $\lgc\Jm\lgc^{\tp}=\Jm$. Since $\Urad$ is normal in $\Nfull$ and
part~(a) placed the entire block-diagonal subgroup inside $\langle\Eup,\Edn\rangle$, the
right-hand side of Eq.~\eqref{eq:general-conj} lies in
$\langle\Eup,\Edn\rangle\cap\Urad$ whenever $u(\radF,\radT)$ does.

We first show $\Urad_0\subseteq\langle\Eup,\Edn\rangle$. By the dictionary of
Section~\ref{app:general-blocks}, turning on a single diagonal entry of the
block $S_{ZZ}$ produces an element $u(0,\radT)$. Here $\radT$ is a nonzero
symmetric matrix of rank one. Over $\Ftwo$ a symmetric rank-one matrix is
necessarily of the form $v^{\tp}v$ with $v\neq0$: writing it as $v^{\tp}w$ and
imposing symmetry forces $v$ and $w$ to be proportional, hence equal. The
action~\eqref{eq:general-conj} on this layer is the congruence
$v^{\tp}v\mapsto(v\gau^{-1})^{\tp}(v\gau^{-1})$, and $\GLg{r}$ is transitive on
the nonzero vectors of $\Ftwo^{r}$. The $\GLg{r}$-orbit of our element
therefore consists of \emph{all} nonzero symmetric rank-one matrices.
These span $\Symm{r}$: the diagonal unit $e_{ii}$ equals $e_i^{\tp}e_i$ and is
rank one, while an off-diagonal symmetric pair satisfies
\begin{equation}
  e_{il}+e_{li}=(e_i+e_l)^{\tp}(e_i+e_l)+e_i^{\tp}e_i+e_l^{\tp}e_l ,
  \label{eq:general-rankone}
\end{equation}
a sum of three rank-one symmetric matrices. As $\Urad_0$ is a group under
addition of the parameter $\radT$, by Eq.~\eqref{eq:general-mult}, we conclude
$\Urad_0\subseteq\langle\Eup,\Edn\rangle$.

Suppose now $k\geq1$. Turning on a single entry of the block $S_{ZL}$ produces
an element of $\langle\Eup,\Edn\rangle\cap\Urad$. The image of that element in
$\Homo(\dst,\logsp)$ is a nonzero rank-one matrix $v^{\tp}w$. Here $v$ is the
indicator of one $X_b$ row and $w$ the indicator of one $Z_p$ column. The
action~\eqref{eq:general-conj} on this layer sends
$v^{\tp}w\mapsto(v\gau^{-1})^{\tp}(w\lgc^{-1})$. The group $\GLg{r}$ is
transitive on the nonzero vectors of $\Ftwo^{r}$, and $\Sp{2k}$ is transitive on
the nonzero vectors of $\Ftwo^{2k}$. The orbit is thus all of the nonzero
rank-one $r\times2k$ matrices. Every matrix is the sum of its single-entry matrices, each of which has
rank one, so the rank-one matrices span $\Homo(\dst,\logsp)$. The image of
$\langle\Eup,\Edn\rangle\cap\Urad$ in Eq.~\eqref{eq:general-exact} is a subgroup
containing all rank-one matrices and is therefore everything.

Consequently $\langle\Eup,\Edn\rangle$ contains the kernel $\Urad_0$
of Eq.~\eqref{eq:general-exact} and meets every coset of that kernel, so it contains
$\Urad$. When $k=0$ the quotient in Eq.~\eqref{eq:general-exact} is trivial and the
previous paragraph is vacuous, but $\Urad=\Urad_0$ and the argument is already
complete. Together with part~(a) and Corollary~\ref{cor:general-criterion} this
gives $\langle\Eup,\Edn\rangle=\Nfull$, which is~\eqref{eq:general-main}.
\end{proof}

%% file: app_depthone.tex
\section{Depth-one layers at a fixed matching}
\label{app:depthone}

The cheapest gates a code can be asked to support are those that act in a single
layer. A single layer is a collection of one- and two-qubit Cliffords with
pairwise disjoint supports, applied simultaneously. Such a layer costs one unit
of circuit depth no matter how many qubits it touches, so it is the natural unit
in which to ask what a code can do fault-tolerantly. This
appendix determines the group of all such layers exactly. Fix which qubits are
paired with which. Three families then generate the layers that preserve the
code. One can write down each family by inspection: (i) the $Z$-diagonal circuits;
(ii) the $X$-diagonal circuits; (iii) the Levi (linear) part. There is no fourth
family and no residue.
In particular the partial dualities $\PDgp$ are redundant. One might expect to
need them, because a Hadamard is neither a diagonal circuit nor Levi. Every valid depth-one
partial duality is nevertheless a group word in the other three families.

Throughout, $\lab=(\CX|0)\oplus(0|\CZ)\subseteq\Vsp$ is a CSS-form subspace, so
that $\lab$ is split by the polarization $\Vsp=\Xh\oplus\Zh$. Orthogonality
$\CX\perp\CZ$ is what makes $\lab$ the label space of a stabilizer code, and it
is used in the logical corollary below. Orthogonality is not used anywhere in
the linear algebra, which is a statement about split subspaces and nothing more.
We keep the standing block convention throughout: a symplectic matrix is written
$G=\left(\begin{smallmatrix}A&B\\C&D\end{smallmatrix}\right)$ and acts on row
vectors from the right, $(x|z)G=(xA+zC\,|\,xB+zD)$.

\subsection{The setting and the statement}
\label{sec:d1-setting}

Let $\Mt$ be a matching of the $n$ qubits, that is, a set of disjoint
pairs, possibly empty. Its \emph{cells} are the pairs of $\Mt$ together with the
unmatched qubits, each taken as a singleton. The cells thus partition the qubits
into blocks of size one or two, and $\Mt=\emptyset$ is the case in which every
cell is a singleton. We write $\cell$ for a cell and $w_{\cell}\in\{1,2\}$
for its width. We let $m_1$ and $m_2$ denote the number of cells of width one
and two respectively, so that $n=m_1+2m_2$. Let $\DM$ be the subgroup of
$\Sp{2n}$ consisting of the $\Mt$-block-diagonal matrices. These are the
matrices carrying an arbitrary $\Sp{4}$ block on each matched pair and an
arbitrary $\Sp{2}$ block on each unmatched singleton. Concretely, in
the $\Xh|\Zh$ convention, $g\in\DM$ means that each of the four blocks
$A,B,C,D$ is block-diagonal for the cells of $\Mt$. Now reorder the coordinates
of $\Vsp$ so that the $X$ and $Z$ labels of each qubit sit together.
Equivalently, in those coordinates $g\in\DM$ means that $g$ is block-diagonal
with one $\Sp{4}$ block per pair and one
$\Sp{2}$ block per singleton. This is exactly what \emph{depth one 2-local}
means: such a $g$ is a layer of gates, each supported on at most two qubits,
with disjoint supports. The object of study is the two-fold transversal slice
\begin{equation}
  \NM \;:=\; \DM\cap\Nfull
  \label{eq:d1-NM}
\end{equation}
of the code-preserving group $\Nfull=\Stb_{\Sp{2n}}(\lab)$. Its three
subfamilies are the $\Mt$-block-diagonal parts of the families of
Appendix~\ref{app:general},
\begin{subequations}
\label{eq:d1-families}
\begin{align}
  \EupM &= \Eup\cap\DM, \label{eq:d1-familiesa}\\
  \EdnM &= \Edn\cap\DM, \label{eq:d1-familiesb}\\
  \LevM &= \Lev\cap\DM . \label{eq:d1-familiesc}
\end{align}
\end{subequations}
Explicitly, $\EupM$ consists of the $Z$-diagonal circuits $\Ush{S}$ whose parameter is
symmetric and $\Mt$-block-diagonal and satisfies $\CX S\subseteq\CZ$; dually
$\EdnM$ consists of the $X$-diagonal circuits $\Lsh{T}$ with $T$ symmetric,
$\Mt$-block-diagonal and $\CZ T\subseteq\CX$; and $\LevM$ consists of the Levi
gates \eqref{eq:setup-lev} whose matrix $K$ is in addition $\Mt$-block-diagonal.
All three have the same shape, an intersection of subgroups of $\Sp{2n}$, and
this settles their basic properties immediately. $\Eup$, $\Edn$ and $\Lev$ are
subgroups, as is $\DM$, so each family is a subgroup. Moreover, since all three
of $\Eup$, $\Edn$ and $\Lev$ lie in $\Nfull$, each family is contained in
$\DM\cap\Nfull=\NM$.

\begin{theorem}[Fixed-matching generation]
\label{thm:depthone}
For every CSS-form $\lab$ and every matching $\Mt$,
\begin{equation}
  \NM=\big\langle\,\EupM,\ \EdnM,\ \LevM\,\big\rangle .
  \label{eq:d1-thm}
\end{equation}
\end{theorem}

See Ref.~\cite{axiomatic2026cliffordcss} for a formalization of this theorem in Lean.
Every valid depth-one Clifford layer is therefore a group word in valid diagonal circuits and
Levi layers \emph{on the same matching}. By the definitions
\eqref{eq:d1-families} all three families already lie inside $\NM$, so no
intersection with $\NM$ need be written in Eq.~\eqref{eq:d1-thm}. We note also
what the theorem does not claim: it is a statement about generation. The
lengths of the group words play no role in the theorem, although the proof below
happens to be constructive.

The theorem fixes a matching, but the gates it produces generate a group that
depends on no single one. Write $\Ndep:=\langle\bigcup_{\Mt}\NM\rangle$ for the
subgroup of $\Nfull$ generated by all the two-fold transversal slices at once. Equivalently,
$\Ndep$ is generated by every code-preserving Clifford that is $2$-local in some
pairing of the qubits.
$\Ndep$ is the join of those slices in the subgroup lattice, sitting in the chain
$\Ntr\le\NM\le\Ndep\le\Nfull$. Modulo Pauli operators, $\Ndep$ is the group of
code-preserving actions reachable by a finite-depth circuit whose every layer is
$2$-local. The logical image of $\Ndep$ is what such circuits achieve on the
encoded qubits.
Whether $\Ndep$ is all of $\Nfull$ is a subtler question than
Theorem~\ref{thm:depthone}, and is not needed below. Section~\ref{sec:d1-dense}
shows that the depth-one diagonal circuits, even taken over all matchings, span only a proper
subspace of the dense family. The inclusion $\Ndep\le\Nfull$ therefore need not
be an equality. What the next corollary records is that the logical image of
$\Ndep$ is already generated by the logical images of the three depth-one
families.

Passing to logical actions needs no new argument.

\begin{corollary}[Depth-one logical generation]
\label{cor:depthone}
Let $\lab$ in addition be isotropic, so that
$\lact:\Nfull\twoheadrightarrow\Sp{2k}$ is defined. Then
\begin{equation}
  \lact(\Ndep)=\Big\langle\bigcup\nolimits_{\Mt}
     \lact\big(\EupM\cup\EdnM\cup\LevM\big)\Big\rangle ,
  \label{eq:d1-log}
\end{equation}
the union being over all matchings of the $n$ qubits.
\end{corollary}

\begin{proof}
Since $\lact$ is a homomorphism, $\lact(\Ndep)=\big\langle\bigcup_{\Mt}\lact(\NM)\big\rangle$,
so it is enough to treat one slice at a time. The inclusion $\supseteq$ is
trivial. For $\subseteq$, fix $\Mt$ and $G\in\NM$. By
Theorem~\ref{thm:depthone}, $G$ is a group word in elements of $\EupM$, $\EdnM$
and $\LevM$. Since $\lact$ is a homomorphism, $\lact(G)$ is the corresponding
group word in the images of those elements, all of which lie in the right-hand
side. Taking the union over all matchings gives \eqref{eq:d1-log}.
\end{proof}

\subsection{Relation to the dense families}
\label{sec:d1-dense}

Theorem~\ref{thm:depthone} looks like a restriction of the generation theorem of
Appendix~\ref{app:general}. It is worth saying precisely why
Theorem~\ref{thm:depthone} is not a corollary of that theorem. That theorem works
in the unrestricted group $\Nfull$ and uses the \emph{dense} diagonal families
$\Eup$ and $\Edn$, which correspond to circuits of arbitrary depth. Here we are
confined to a fixed $\Mt$ and may use only the $\Mt$-block-diagonal part of each
of those families. The relation between the two settings is
\begin{equation}
  \NM = \Nfull\cap\DM\,\qquad
  \EupM=\Eup\cap\DM,
  \label{eq:d1-slice}
\end{equation}
and similarly $\EdnM=\Edn\cap\DM$: thus $\NM$ is the two-fold transversal slice of $\Nfull$
and $\EupM$ is the $\Mt$-block-diagonal part of $\Eup$.

Because $\Ush{S}\Ush{S'}=\Ush{S+S'}$, the family $\Eup$ is canonically an
$\Ftwo$-vector space, namely the space of admissible parameters
$\{S\in\Symm{n}:\CX S\subseteq\CZ\}$, and $\EupM$ is the subspace of those
parameters that are $\Mt$-block-diagonal. One may then ask whether the depth-one
diagonal circuits, taken over \emph{all} matchings at once, already see the whole of $\Eup$:
\begin{equation}
  \spano\bigcup_{\Mt}\EupM \;\subseteq\; \Eup .
  \label{eq:d1-span}
\end{equation}
The inclusion \eqref{eq:d1-span} is strict in general. On the doubly-even
self-dual $\qcode{6,2,2}$ code of Sec.~\ref{sec:d1-example} below, with
$\Mt=\{\{0,1\},\{2,3\},\{4,5\}\}$, the three parameter spaces in
\eqref{eq:d1-span} have computed dimensions
\begin{equation}
  \begin{gathered}
  \dim\EupM = 4,\qquad \dim\Eup = 14,\\
  \dim\spano\bigcup\nolimits_{\Mt}\EupM = 13 :
  \end{gathered}
  \label{eq:d1-dims}
\end{equation}
a single matching sees very little, and even the span over all matchings falls
one dimension short of the dense family. The corresponding computed gaps
$\dim\Eup-\dim\spano\bigcup_{\Mt}\EupM$ are $0$ for
$\qcode{4,2,2}$ and $4$ for the Steane code $\qcode{7,1,3}$. Thus
Theorem~\ref{thm:depthone} neither follows from nor implies
Theorem~\ref{thm:sheargen}: the two use different stocks of generators, and the
depth-one stock is genuinely poorer.

\subsection{Architecture of the proof}
\label{sec:d1-architecture}

The strategy is \emph{reduction by multiplication}. Take an arbitrary $G\in\NM$
and multiply $G$ on the left and on the right by valid diagonal circuits on the
same matching. Stop when what is left is manifestly a product of such layers. Two
obstacles stand in the way, and there is one lemma for each.

The first is \emph{supply}. One cannot multiply by a diagonal circuit without knowing that
some particular diagonal circuit is valid, and validity depends on the code. There is no
a-priori stock of moves. Lemma~\ref{lem:auto} manufactures the moves out of $G$
itself. Six explicit matrices built from $G$'s own blocks are automatically
symmetric, automatically supported on the same matching, and automatically valid
for the same code. These properties hold under no hypothesis on the code
whatsoever. This lemma drives the argument. Every later step uses the output of
Lemma~\ref{lem:auto}.

The second obstacle is \emph{termination}, and termination for many cells at
once. Lemma~\ref{lem:local} answers the question that algebra alone does not
settle: whether those moves get anywhere. Restricted to a single cell the state
space is finite: $|\Sp{2}|=6$ and $|\Sp{4}|=720$, since two-locality caps cells
at width two. The question is therefore decidable, and the answer is that every
block reaches a \emph{terminal} block in at most three moves. This is the one
place where a machine check enters, and it enters because the question is
finite, not because the algebra is hard.
Lemma~\ref{lem:restrict} then records why a local move is always induced by a
global one. Lemma~\ref{lem:global-reduction} globalizes: since terminality
is absorbing, cells can be processed one at a time with the number of finished
cells strictly increasing. Finally Lemma~\ref{lem:terminal} classifies the
terminal blocks and shows that each is already Levi or already a product
$\Ush{B}\Lsh{C}\Ush{B}$ of diagonal circuits.

The theorem concerns a
\emph{group} of gates, but the proofs take place in the larger associative
algebra
\begin{equation}
  \Endl = \{f:\Vsp\to\Vsp \ \text{linear},\ f(\lab)\subseteq\lab\},
  \label{eq:d1-endl}
\end{equation}
of which $\NM$ is only the invertible, symplectic, $\Mt$-supported part. The
algebra $\Endl$ is closed under composition and addition, and --- because
$G(\lab)=\lab$ exactly --- conjugation by any $G\in\NM$ is an automorphism of
$\Endl$: if $f(\lab)\subseteq\lab$ then
$\lab\,GfG^{-1}=\big((\lab G)f\big)G^{-1}\subseteq
(\lab)G^{-1}=\lab$. Intermediate objects may therefore be singular,
non-symplectic, and not gates at all; only the conclusions are statements about
gates. Lemma~\ref{lem:auto} exploits exactly this freedom.

\subsection{Validity and automatically valid diagonal circuits}
\label{sec:d1-supply}

For $G=\left(\begin{smallmatrix}A&B\\C&D\end{smallmatrix}\right)$, preservation
of the CSS-form space $\lab$ is equivalent to the four inclusions
\begin{equation}
  \begin{gathered}
  \CX A\subseteq\CX,\qquad \CX B\subseteq\CZ,\\
  \CZ C\subseteq\CX,\qquad \CZ D\subseteq\CZ ,
  \end{gathered}
  \label{eq:d1-valid}
\end{equation}
obtained by feeding $(c|0)$ with $c\in\CX$ and $(0|d)$ with $d\in\CZ$ into $G$
and demanding that the images lie in $\lab$. Inclusions suffice: $G$ is
invertible, so $\lab G$ has the same dimension as $\lab$, and
$\lab G\subseteq\lab$ forces $\lab G=\lab$. We call a layer satisfying
\eqref{eq:d1-valid} \emph{valid}. Symplecticity, $G\Jm G^{\tp}=\Jm$, expands to
\begin{equation}
  \begin{pmatrix} BA^{\tp}+AB^{\tp} & BC^{\tp}+AD^{\tp}\\
                  DA^{\tp}+CB^{\tp} & DC^{\tp}+CD^{\tp}\end{pmatrix}
  =\begin{pmatrix}0&\Id\\ \Id&0\end{pmatrix},
  \label{eq:d1-gjg}
\end{equation}
whose diagonal blocks say exactly that $AB^{\tp}$ and $CD^{\tp}$ are symmetric,
and whose off-diagonal blocks say $AD^{\tp}+BC^{\tp}=\Id$. Dually
$G^{\tp}\Jm G=\Jm$ makes $A^{\tp}C$ and $B^{\tp}D$ symmetric and gives
$A^{\tp}D+C^{\tp}B=\Id$. Collecting,
\begin{equation}
  \begin{gathered}
  AB^{\tp},\ CD^{\tp},\ A^{\tp}C,\ B^{\tp}D \ \ \text{symmetric},\\
  AD^{\tp}+BC^{\tp}=\Id,\qquad A^{\tp}D+C^{\tp}B=\Id ,
  \end{gathered}
  \label{eq:d1-sympl}
\end{equation}
and, since $G^{-1}=\Jm G^{\tp}\Jm$ over $\Ftwo$,
\begin{equation}
  G^{-1}=\begin{pmatrix}D^{\tp}&B^{\tp}\\ C^{\tp}&A^{\tp}\end{pmatrix}.
  \label{eq:d1-inv}
\end{equation}

We shall manufacture diagonal circuits out of $G$. A diagonal circuit
$\Ush{S}=\left(\begin{smallmatrix}\Id&S\\0&\Id\end{smallmatrix}\right)$ is
symplectic \emph{if and only if} $S=S^{\tp}$, so symmetry of a manufactured
parameter is met
automatically. Note that $AB^{\tp}$ is symmetric for \emph{every} symplectic $G$,
with no assumption that $A$ or $B$ be individually symmetric. This symmetry is
the content of the diagonal blocks of Eq.~\eqref{eq:d1-gjg}. Likewise
$D^{\tp}B=(B^{\tp}D)^{\tp}$
and $C^{\tp}A=(A^{\tp}C)^{\tp}$ are symmetric by Eq.~\eqref{eq:d1-sympl}, while
$B+B^{\tp}$ and $C+C^{\tp}$ are symmetric by construction. Symmetry is preserved
by $\Ftwo$-linear combination, so any span of these matrices lies in $\Symm{n}$.

\begin{lemma}[Automatic parameters]
\label{lem:auto}
Let $G\in\NM$. Then every matrix in the two spans
\begin{equation}
  \begin{gathered}
  \shp{G}=\spano\{AB^{\tp},\,D^{\tp}B,\,B+B^{\tp}\},\\
  \shm{G}=\spano\{CD^{\tp},\,C^{\tp}A,\,C+C^{\tp}\}
  \end{gathered}
  \label{eq:d1-spans}
\end{equation}
is symmetric and $\Mt$-block-diagonal, and
\begin{equation}
  \begin{gathered}
  \CX S\subseteq\CZ \quad\text{for}\ S\in\shp{G},\\
  \CZ T\subseteq\CX \quad\text{for}\ T\in\shm{G}.
  \end{gathered}
  \label{eq:d1-shvalid}
\end{equation}
Consequently $\Ush{S}\in\EupM$ and $\Lsh{T}\in\EdnM$ for all such $S$ and $T$:
each is a valid diagonal-circuit supported on $\Mt$.
\end{lemma}

Each span is spanned by three matrices, hence has dimension at most three and
contains at most eight elements, at most seven of them nonzero. The subscript
records upper versus lower; the content of the lemma is that the assignment
$G\mapsto\shp{G},\shm{G}$ manufactures valid diagonal circuits out of an
arbitrary valid layer.

\begin{proof}
Let $\PX=\left(\begin{smallmatrix}\Id&0\\0&0\end{smallmatrix}\right)$ be the
projector onto the $X$-half, $(x|z)\mapsto(x|0)$. Since it is the only non-gate
in the argument, three points about it deserve to be made explicitly.

First, $\PX$ is idempotent of rank exactly $n$, its image being the $X$-half.
It is therefore singular and not symplectic: it is not a gate, and is never
claimed to be one. The rank is $n$ by definition, for every $n$; no genericity or
specialization is involved. In addition, $\PX$ depends only on the fixed
polarization $\Vsp=\Xh\oplus\Zh$, not on the code, on $\Mt$, or on $G$.

Second, only one property of $\PX$ is used, namely $\PX(\lab)\subseteq\lab$.
This is ``maps into'', not ``maps onto''; for a singular map the two differ, and
only the weak form is true or needed. It holds exactly because $\lab$ is split:
$(c|d)\mapsto(c|0)\in(\CX|0)\subseteq\lab$. For a non-split label space it
fails: the $Y$-type space $\lab=\spano\{(1|1)\}$ has
$\PX(1|1)=(1|0)\notin\lab$. This is therefore the step at which the CSS
hypothesis is consumed, exactly as in Appendix~\ref{app:general}.

Third, by Eq.~\eqref{eq:d1-endl} we have $\PX\in\Endl$, and conjugation by $G$ is an
automorphism of $\Endl$, so the conjugate constructed below lies in $\Endl$ too.
Its singularity is not a defect but the point of the construction. Conjugating an
\emph{invertible} element of the group would merely return another gate and yield
no new information. Conjugating a \emph{projector}, by contrast, is what exposes
the block products $AB^{\tp}$ and $CD^{\tp}$ as fresh data. This is the one place
where the argument steps outside the group and into the algebra.

Now compute, using Eq.~\eqref{eq:d1-inv},
\begin{equation}
  \begin{split}
  G\,\PX\,G^{-1}
  &=\begin{pmatrix}A&0\\ C&0\end{pmatrix}
    \begin{pmatrix}D^{\tp}&B^{\tp}\\ C^{\tp}&A^{\tp}\end{pmatrix}\\
  &=\begin{pmatrix}AD^{\tp}&AB^{\tp}\\ CD^{\tp}&CB^{\tp}\end{pmatrix}.
  \end{split}
  \label{eq:d1-conj}
\end{equation}
By the previous paragraph this matrix maps $\lab$ into $\lab$. Applying it to
$(c|0)$ with $c\in\CX$ gives $(cAD^{\tp}\,|\,cAB^{\tp})\in\lab$, and applying it
to $(0|d)$ with $d\in\CZ$ gives $(dCD^{\tp}\,|\,dCB^{\tp})\in\lab$. Reading off
the components that must lie in $\CZ$ and in $\CX$ respectively,
\begin{equation}
  \CX\,AB^{\tp}\subseteq\CZ,\qquad \CZ\,CD^{\tp}\subseteq\CX .
  \label{eq:d1-probe}
\end{equation}

Since $\NM$ is a group, $G^{-1}\in\NM$ as well, and \eqref{eq:d1-inv} exhibits
its four blocks as $D^{\tp},B^{\tp},C^{\tp},A^{\tp}$ in the roles of $A,B,C,D$.
Applying \eqref{eq:d1-probe} to $G^{-1}$ therefore gives
$\CX D^{\tp}B\subseteq\CZ$ and $\CZ C^{\tp}A\subseteq\CX$, which is the validity
of the second parameter on each side.

Finally, \eqref{eq:d1-valid} applied to $G$ and to $G^{-1}$ gives the four
inclusions
\begin{equation}
  \begin{gathered}
  \CX B\subseteq\CZ,\qquad \CX B^{\tp}\subseteq\CZ,\\
  \CZ C\subseteq\CX,\qquad \CZ C^{\tp}\subseteq\CX ,
  \end{gathered}
  \label{eq:d1-plus}
\end{equation}
so that $\CX(B+B^{\tp})\subseteq\CZ$ and $\CZ(C+C^{\tp})\subseteq\CX$: the third
parameter on each side is valid too. Validity is closed under binary sums, so
\eqref{eq:d1-shvalid} holds throughout the spans. Symmetry of the six matrices
was established before the statement. Each of the six is a sum of products of
the $\Mt$-block-diagonal matrices $A,B,C,D$, hence is $\Mt$-block-diagonal
itself. A symmetric $\Mt$-block-diagonal parameter satisfying
$\CX S\subseteq\CZ$ is precisely an element of $\EupM$ by
\eqref{eq:d1-familiesa}, and dually for $\EdnM$.
\end{proof}

\paragraph*{The two dualities.} Two symmetries of the construction account for
the six parameters, and both were used above. Inversion, $G\mapsto G^{-1}$ of
\eqref{eq:d1-inv}, turns the pair $AB^{\tp},CD^{\tp}$ into $D^{\tp}B,C^{\tp}A$.
The $X\leftrightarrow Z$ swap, conjugation by $\sw$, sends
$\left(\begin{smallmatrix}A&B\\C&D\end{smallmatrix}\right)\mapsto
\left(\begin{smallmatrix}D&C\\B&A\end{smallmatrix}\right)$ and hence exchanges
the two spans,
\begin{equation}
  \shp{G}=\shm{\sw G\sw},
  \label{eq:d1-sigma}
\end{equation}
so that Lemma~\ref{lem:auto} is $\sw$-equivariant and only its upper half ever
needs proving. The remaining two parameters, $B+B^{\tp}$ and $C+C^{\tp}$, come
from neither symmetry: they are read off the validity conditions
\eqref{eq:d1-plus} for $G$ and for $G^{-1}$ directly. Those two parameters are
what remove the nonsymmetric residue the other four leave behind.

\subsection{Exact local reduction}
\label{sec:d1-local}

\paragraph*{Global versus local blocks.} Globally, $G$ is $2n\times 2n$ and
$G=\left(\begin{smallmatrix}A&B\\C&D\end{smallmatrix}\right)$ with each block an
$n\times n$ matrix; each is $\Mt$-block-diagonal, $A=\bigoplus_{\cell}A_{\cell}$
with $A_{\cell}$ of size $w_{\cell}\times w_{\cell}$. Locally, the block of $G$
at the cell $\cell$ is
\begin{equation}
  g_{\cell}=\begin{pmatrix}A_{\cell}&B_{\cell}\\ C_{\cell}&D_{\cell}\end{pmatrix}
  \in\Sp{2w_{\cell}},\quad w_{\cell}\in\{1,2\}.
  \label{eq:d1-gb}
\end{equation}
A move will select one of the six formulas of Eq.~\eqref{eq:d1-spans} and evaluate it
on the \emph{global} $A,B,C,D$. Because these are $\Mt$-block-diagonal, the entry
$(AB^{\tp})_{ij}=\sum_k A_{ik}B_{jk}$ forces $k\in \cell\cap \cell'$ for
$i\in \cell$, $j\in \cell'$, which is empty unless $\cell=\cell'$. Hence the
resulting parameter satisfies
\begin{equation}
  \begin{gathered}
  S=\bigoplus\nolimits_{\cell}\text{formula}(g_{\cell}),\\
  S[\cell,\cell']=0\quad\text{for}\ \cell\neq \cell' .
  \end{gathered}
  \label{eq:d1-blockdiag}
\end{equation}
Each cell's diagonal block of $S$ is that one formula evaluated on that cell's
own current block, and no entry of $S$ ever links two different cells. Inside a
pair cell $\cell'$, the off-diagonal \emph{entry} of $S[\cell',\cell']$ --- the
$\mathrm{CZ}$ between $\cell'$'s two qubits --- is fixed entirely by
$g_{\cell'}$. This is exactly why the other cells are affected: the formula chosen to
advance one cell is in general nonzero when evaluated on the others.

\paragraph*{Basic moves.} For a local block
$g=\left(\begin{smallmatrix}A&B\\C&D\end{smallmatrix}\right)\in\Sp{2w}$,
$w\in\{1,2\}$ --- here $A,B,C,D$ are the local $w\times w$ blocks, the cell
subscript having been dropped --- we call any left or right multiplication by
\begin{equation}
  \begin{gathered}
  \Ush{S},\quad 0\neq S\in\shp{g},\\
  \Lsh{T},\quad 0\neq T\in\shm{g}
  \end{gathered}
  \label{eq:d1-moves}
\end{equation}
a \emph{move}. A move is thus specified by a family ($Z$- or
$X$-diagonal), a nonzero parameter in the relevant span, and a side (left or right).
It is convenient to index the parameter by a coefficient vector: writing
\begin{equation}
  \begin{aligned}
  \Spz(G,\eps)&:=\eps_1 AB^{\tp}+\eps_2 D^{\tp}B+\eps_3(B+B^{\tp}),\\
  \Spx(G,\eps)&:=\eps_1 CD^{\tp}+\eps_2 C^{\tp}A+\eps_3(C+C^{\tp}),
  \end{aligned}
  \label{eq:d1-eps}
\end{equation}
with $\eps\in\Ftwo^{3}$,
we have $\shp{G}=\{\Spz(G,\eps):\eps\in\Ftwo^{3}\}$ and
$\shm{G}=\{\Spx(G,\eps):\eps\in\Ftwo^{3}\}$, and a move is the choice of a
family, a nonzero $\eps$-value of the parameter, and a side. The same notation
applied to a local block $g_{\cell}$ gives $\Spzx(g_{\cell},\eps)$.

\paragraph*{The reduction as a walk.} It helps to picture what follows as a walk
on a graph. The vertices of that graph are the blocks, that is the elements of
$\Sp{2w}$, and its edges are the moves. The reduction is then a walk that has to
reach the terminal blocks. One feature of this graph is unusual and is the source
of every difficulty below. The moves available at a block are manufactured from
that block itself, so the edges leaving a vertex change as the walk proceeds.
This is not a graph with a fixed set of steps, which is why reachability has to
be decided rather than read off.

The selected nonzero parameter stays a valid move: the layer stays inside
$\NM$, so Lemma~\ref{lem:auto} applies again at each step. A move is by
definition indexed by a \emph{nonzero} parameter, so the walk halts exactly when
all six parameters vanish.

\paragraph*{Terminal blocks.} Call $g$ \emph{terminal} if all six generators of
\eqref{eq:d1-spans} vanish on it:
\begin{equation}
  \begin{gathered}
  AB^{\tp}=D^{\tp}B=B+B^{\tp}=0,\\
  CD^{\tp}=C^{\tp}A=C+C^{\tp}=0 .
  \end{gathered}
  \label{eq:d1-term}
\end{equation}
A terminal block is fixed by every move, since every parameter available
at it is zero. Terminality is precisely the property of being invisible to every
possible move. By exhaustive enumeration we show that every non-terminal block
admits a move that decreases its distance to a terminal block. Every block of
width one and every block of width two therefore reaches a terminal block.

\begin{lemma}[Local reduction]
\label{lem:local}
Every element of $\Sp{2}$ reaches a terminal block in at most one move;
there are two terminal blocks. Every element of $\Sp{4}$ reaches a terminal block
in at most three moves; there are sixteen terminal blocks.
\end{lemma}

Define the \emph{distance} of
$g$ to be the least number of moves carrying it into the terminal set. What
must be established is that each non-terminal element admits a
distance-decreasing move, and this is checked for every one of the $720+6$
elements.

It is worth pausing on what a terminal block is. The answer says
something about the construction of Lemma~\ref{lem:auto} rather than about any
particular code. Everything that construction can see is the $X/Z$ splitting: the
four parameters $AB^{\tp}$, $CD^{\tp}$, $D^{\tp}B$ and $C^{\tp}A$ are, by
\eqref{eq:d1-conj}, nothing but the off-diagonal blocks of the projector $\PX$
carried along by $g$ and by $g^{-1}$. A projector has vanishing off-diagonal
blocks exactly when it commutes with $\PX$. Those four parameters therefore
vanish exactly when the splitting $\Vsp=\Xh\oplus\Zh$ and its image under $g$ are
compatible. Compatible here means that $\Vsp$ is the direct sum of the four
intersections $\Xh\cap\Xh g$, $\Xh\cap\Zh g$, $\Zh\cap\Xh g$ and
$\Zh\cap\Zh g$. A gate that moves the splitting to one refining it in this way
respects the very structure out of which the moves were manufactured, and so is
invisible to all four of those parameters. At width two there are $48$ such
blocks: six that preserve the two halves, six that interchange them, and
thirty-six that cut each half in two.

The remaining two parameters look at $g$ itself rather than at the transported
projector. The blocks $B$ and $C$ already satisfy the validity conditions
\eqref{eq:d1-valid}, but it is a Clifford gate only when its corresponding matrix $S$ or $T$ is symmetric,
so the blocks $B$ and $C$ are not usable as they stand. What the construction
can use is the symmetrizations of $B$ and $C$, which are valid for the same
reason and symmetric by construction. Those vanish exactly when $B$ and $C$ were
already legal diagonal-circuit parameters. This is what separates the $16$ genuine
terminals of
Table~\ref{tab:d1-terminals} from the $48$: a block is terminal when it both
carries the splitting to a compatible one \emph{and} leaks between the two halves
only through blocks that are themselves already moves. Dropping
$B+B^{\tp}$ and $C+C^{\tp}$ would leave $32$ blocks that are not terminal but
have no move available, and the reduction would genuinely stall on them.

\begin{table*}[t]
\caption{\label{tab:d1-terminals}\textit{The sixteen terminal blocks at width
two.} These are the elements of $\Sp{4}$ on which all six parameters of
Lemma~\ref{lem:auto} vanish, so that no move can move them. Each is
written as a gate acting on two qubits, meaning its label action; Pauli factors
are ignored throughout, and $\mathrm{CZ}^{H}:=H^{\otimes2}\,\mathrm{CZ}\,H^{\otimes2}$
is the $H^{\otimes2}$-conjugate of $\mathrm{CZ}$. Only the images that differ from their
argument are listed. The blocks fall into three groups according to
$\ranko B$, which by Lemma~\ref{lem:terminal} equals $\ranko C$ and counts how
much of the splitting $\Vsp=\Xh\oplus\Zh$ the gate interchanges: none of it for
the six Levi blocks, a single dimension for the next six, and all of it for the
last four. Together the sixteen generate a subgroup of $\Sp{4}$ of order $72$,
isomorphic to $\Sp{2}\wr\Perm{2}$, which is the stabilizer of a decomposition of
$\Ftwo^{4}$ into two orthogonal symplectic planes and is conjugate in $\Sp{4}$ to
the group generated by the single-qubit Cliffords together with the swap; for the
classification of the subgroups of the two-qubit Clifford group see
Ref.~\cite{kubischta2025classification}.}
\begin{ruledtabular}
\footnotesize
\begin{tabular}{clccl}
$\ranko B$ & gate & $B$ & $C$ & action on labels\\
\colrule
$0$ & $\Id$ & $\left(\begin{smallmatrix}0&0\\0&0\end{smallmatrix}\right)$ & $\left(\begin{smallmatrix}0&0\\0&0\end{smallmatrix}\right)$ & identity\\
$0$ & $\mathrm{CNOT}_{12}$ & $\left(\begin{smallmatrix}0&0\\0&0\end{smallmatrix}\right)$ & $\left(\begin{smallmatrix}0&0\\0&0\end{smallmatrix}\right)$ & $X_1\!\mapsto\!X_1X_2,\ Z_2\!\mapsto\!Z_1Z_2$\\
$0$ & $\mathrm{CNOT}_{21}$ & $\left(\begin{smallmatrix}0&0\\0&0\end{smallmatrix}\right)$ & $\left(\begin{smallmatrix}0&0\\0&0\end{smallmatrix}\right)$ & $X_2\!\mapsto\!X_1X_2,\ Z_1\!\mapsto\!Z_1Z_2$\\
$0$ & $\mathrm{SWAP}$ & $\left(\begin{smallmatrix}0&0\\0&0\end{smallmatrix}\right)$ & $\left(\begin{smallmatrix}0&0\\0&0\end{smallmatrix}\right)$ & $X_1\!\mapsto\!X_2,\ X_2\!\mapsto\!X_1,\ Z_1\!\mapsto\!Z_2,\ Z_2\!\mapsto\!Z_1$\\
$0$ & $\mathrm{CNOT}_{12}\,\mathrm{CNOT}_{21}$ & $\left(\begin{smallmatrix}0&0\\0&0\end{smallmatrix}\right)$ & $\left(\begin{smallmatrix}0&0\\0&0\end{smallmatrix}\right)$ & $X_1\!\mapsto\!X_2,\ X_2\!\mapsto\!X_1X_2,\ Z_1\!\mapsto\!Z_1Z_2,\ Z_2\!\mapsto\!Z_1$\\
$0$ & $\mathrm{CNOT}_{12}\,\mathrm{SWAP}$ & $\left(\begin{smallmatrix}0&0\\0&0\end{smallmatrix}\right)$ & $\left(\begin{smallmatrix}0&0\\0&0\end{smallmatrix}\right)$ & $X_1\!\mapsto\!X_1X_2,\ X_2\!\mapsto\!X_1,\ Z_1\!\mapsto\!Z_2,\ Z_2\!\mapsto\!Z_1Z_2$\\
\colrule
$1$ & $H_1$ & $\left(\begin{smallmatrix}1&0\\0&0\end{smallmatrix}\right)$ & $\left(\begin{smallmatrix}1&0\\0&0\end{smallmatrix}\right)$ & $X_1\!\mapsto\!Z_1,\ Z_1\!\mapsto\!X_1$\\
$1$ & $H_2$ & $\left(\begin{smallmatrix}0&0\\0&1\end{smallmatrix}\right)$ & $\left(\begin{smallmatrix}0&0\\0&1\end{smallmatrix}\right)$ & $X_2\!\mapsto\!Z_2,\ Z_2\!\mapsto\!X_2$\\
$1$ & $\mathrm{CZ}^{H}H_1\,\mathrm{CZ}^{H}$ & $\left(\begin{smallmatrix}1&0\\0&0\end{smallmatrix}\right)$ & $\left(\begin{smallmatrix}1&1\\1&1\end{smallmatrix}\right)$ & $X_1\!\mapsto\!X_2Z_1,\ Z_1\!\mapsto\!X_1X_2,\ Z_2\!\mapsto\!X_1X_2Z_1Z_2$\\
$1$ & $\mathrm{CZ}^{H}H_2\,\mathrm{CZ}^{H}$ & $\left(\begin{smallmatrix}0&0\\0&1\end{smallmatrix}\right)$ & $\left(\begin{smallmatrix}1&1\\1&1\end{smallmatrix}\right)$ & $X_2\!\mapsto\!X_1Z_2,\ Z_1\!\mapsto\!X_1X_2Z_1Z_2,\ Z_2\!\mapsto\!X_1X_2$\\
$1$ & $\mathrm{CZ}\,H_1\,\mathrm{CZ}$ & $\left(\begin{smallmatrix}1&1\\1&1\end{smallmatrix}\right)$ & $\left(\begin{smallmatrix}1&0\\0&0\end{smallmatrix}\right)$ & $X_1\!\mapsto\!Z_1Z_2,\ X_2\!\mapsto\!X_1X_2Z_1Z_2,\ Z_1\!\mapsto\!X_1Z_2$\\
$1$ & $\mathrm{CZ}\,H_2\,\mathrm{CZ}$ & $\left(\begin{smallmatrix}1&1\\1&1\end{smallmatrix}\right)$ & $\left(\begin{smallmatrix}0&0\\0&1\end{smallmatrix}\right)$ & $X_1\!\mapsto\!X_1X_2Z_1Z_2,\ X_2\!\mapsto\!Z_1Z_2,\ Z_2\!\mapsto\!X_2Z_1$\\
\colrule
$2$ & $H_1H_2$ & $\left(\begin{smallmatrix}1&0\\0&1\end{smallmatrix}\right)$ & $\left(\begin{smallmatrix}1&0\\0&1\end{smallmatrix}\right)$ & $X_1\!\mapsto\!Z_1,\ X_2\!\mapsto\!Z_2,\ Z_1\!\mapsto\!X_1,\ Z_2\!\mapsto\!X_2$\\
$2$ & $\mathrm{CZ}\,\mathrm{CZ}^{H}\mathrm{CZ}$ & $\left(\begin{smallmatrix}0&1\\1&0\end{smallmatrix}\right)$ & $\left(\begin{smallmatrix}0&1\\1&0\end{smallmatrix}\right)$ & $X_1\!\mapsto\!Z_2,\ X_2\!\mapsto\!Z_1,\ Z_1\!\mapsto\!X_2,\ Z_2\!\mapsto\!X_1$\\
$2$ & $\mathrm{CZ}\,\mathrm{CZ}^{H}\mathrm{CZ}\,\mathrm{CNOT}_{12}$ & $\left(\begin{smallmatrix}1&1\\1&0\end{smallmatrix}\right)$ & $\left(\begin{smallmatrix}0&1\\1&1\end{smallmatrix}\right)$ & $X_1\!\mapsto\!Z_1Z_2,\ X_2\!\mapsto\!Z_1,\ Z_1\!\mapsto\!X_2,\ Z_2\!\mapsto\!X_1X_2$\\
$2$ & $\mathrm{CZ}\,\mathrm{CZ}^{H}\mathrm{CZ}\,\mathrm{CNOT}_{21}$ & $\left(\begin{smallmatrix}0&1\\1&1\end{smallmatrix}\right)$ & $\left(\begin{smallmatrix}1&1\\1&0\end{smallmatrix}\right)$ & $X_1\!\mapsto\!Z_2,\ X_2\!\mapsto\!Z_1Z_2,\ Z_1\!\mapsto\!X_1X_2,\ Z_2\!\mapsto\!X_1$\\
\end{tabular}
\end{ruledtabular}
\end{table*}

For the same reason it is not enough merely to \emph{list} the terminal blocks,
which is Lemma~\ref{lem:terminal}'s task. A list of terminals says nothing about
whether some non-terminal element is stuck, able to move but only ever to other
non-terminals.

Two facts splice Lemmas~\ref{lem:auto} and~\ref{lem:local} together. First, the
two notions of move coincide: Lemma~\ref{lem:local} permits any nonzero parameter
in the \emph{local} span. A global move, by contrast, selects a binary
combination of the three generators. Since the formulas commute with restriction,
that global choice restricts at the cell to the same combination of the local
generators, and as $\eps$ ranges over $\Ftwo^{3}$ these restrictions exhaust the
local span. This is Lemma~\ref{lem:restrict} below. Second, the layer never
leaves $\NM$, since only valid $\Mt$-supported gates are ever multiplied in and
$\NM$ is a group. Lemma~\ref{lem:auto} may therefore be re-applied at every
step.

Before the proof we isolate the two cases that need no enumeration at all.

\begin{lemma}[Two closed-form cases]
\label{lem:strataC}
Let $g=\left(\begin{smallmatrix}A&B\\C&D\end{smallmatrix}\right)\in\Sp{2w}$.
If $C=0$, then $S:=D^{\tp}B=A^{-1}B$ is symmetric, lies in $\shp{g}$, and the
single move $g\mapsto g\,\Ush{S}$ lands on a Levi block, which is terminal.
Dually, if $B=0$, then $T:=C^{\tp}A=A^{\tp}C$ is symmetric, lies in $\shm{g}$,
and the single move $g\mapsto g\,\Lsh{T}$ lands on a Levi block.
\end{lemma}

\begin{proof}
Suppose $C=0$. The relation $AD^{\tp}+BC^{\tp}=\Id$ of Eq.~\eqref{eq:d1-sympl}
becomes $AD^{\tp}=\Id$, so $A$ and $D$ are invertible with $D=A^{\itp}$, whence
$S=D^{\tp}B=A^{-1}B$ and $B=AS$. Now $AB^{\tp}=AS^{\tp}A^{\tp}$ is symmetric by
\eqref{eq:d1-sympl}, and $(AS^{\tp}A^{\tp})^{\tp}=ASA^{\tp}$, so
$AS^{\tp}A^{\tp}=ASA^{\tp}$ and, $A$ being invertible, $S^{\tp}=S$. Thus
$S\in\shp{g}$, and the move is a move whenever $S\neq0$; if $S=0$ then
$B=0$ and $g$ is already Levi. Computing,
\begin{equation}
  \begin{split}
  g\,\Ush{S}&=\begin{pmatrix}A&B\\0&D\end{pmatrix}
              \begin{pmatrix}\Id&S\\0&\Id\end{pmatrix}\\
            &=\begin{pmatrix}A&AS+B\\0&D\end{pmatrix}
             =\begin{pmatrix}A&0\\0&D\end{pmatrix},
  \end{split}
  \label{eq:d1-2a}
\end{equation}
since $AS+B=A(A^{-1}B)+B=B+B=0$. A block with $B=C=0$ satisfies all six
conditions \eqref{eq:d1-term}, hence is terminal.

The second statement is the image of the first under the $X\leftrightarrow Z$
swap \eqref{eq:d1-sigma}, and we record the direct argument as well. If $B=0$
then $A^{\tp}D+C^{\tp}B=\Id$ gives $A^{\tp}D=\Id$, so $D=A^{\itp}$, and
$T=C^{\tp}A=(A^{\tp}C)^{\tp}=A^{\tp}C$ is symmetric by Eq.~\eqref{eq:d1-sympl}. Then
$g\,\Lsh{T}$ has lower-left block $C+DT=C+A^{\itp}A^{\tp}C=C+C=0$ and unchanged
$B=0$, so it is Levi and terminal.
\end{proof}

These closed-form rules cover every block with $\ranko B=0$ or $\ranko C=0$.
Those blocks are $90$ of the $720$ elements at $w=2$, all of them at distance at
most one. The remaining $630$, together with all six elements of $\Sp{2}$, are
certified elementwise.

\begin{table*}[t]
\caption{\label{tab:d1-strata}
The local groups, with their blocks grouped by $(\ranko B,\ranko C)$. ``Case''
records whether the group of blocks is handled in closed form by
Lemma~\ref{lem:strataC} (symbolic) or by the finite certificate of
Lemma~\ref{lem:local} (numerical). In the distance column the entry $d{:}m$
means that $m$ of those blocks lie at distance $d$ from the terminal
set. The last column records the type of the terminal
blocks reached, ``Levi'' meaning $B=C=0$ and $\Word$ meaning
$\Ush{B}\Lsh{C}\Ush{B}$ in the sense of Lemma~\ref{lem:terminal}; a block of
type $\Word$ is thus a product of diagonal circuits, and a Levi block contains none.}
\footnotesize
\begin{tabular}{ccrllll}
\toprule
$\ranko B$ & $\ranko C$ & \# & case & rule / bound & distances & landing\\
\midrule
\multicolumn{7}{c}{$w=1$: $\Sp{2}$, order $6$, two terminal blocks, max distance $1$}\\
\midrule
$0$&$0$&$1$&symbolic& already terminal & $0{:}1$ & Levi $1$\\
$0$&$1$&$1$&symbolic& $g\mapsto g\Lsh{C^{\tp}A}$ & $1{:}1$ & Levi $1$\\
$1$&$0$&$1$&symbolic& $g\mapsto g\Ush{D^{\tp}B}$ & $1{:}1$ & Levi $1$\\
$1$&$1$&$3$&numerical& $\leq1$ move & $0{:}1,\ 1{:}2$ & $\Word$ $3$\\
\midrule
\multicolumn{7}{c}{$w=2$: $\Sp{4}$, order $720$, sixteen terminal blocks, max distance $3$}\\
\midrule
$0$&$0$&$6$&symbolic& already terminal & $0{:}6$ & Levi $6$\\
$0$&$1$&$18$&symbolic& $g\mapsto g\Lsh{C^{\tp}A}$ & $1{:}18$ & Levi $18$\\
$0$&$2$&$24$&symbolic& $g\mapsto g\Lsh{C^{\tp}A}$ & $1{:}24$ & Levi $24$\\
$1$&$0$&$18$&symbolic& $g\mapsto g\Ush{D^{\tp}B}$ & $1{:}18$ & Levi $18$\\
$1$&$1$&$126$&numerical& $\leq3$ moves & $0{:}6,\ 1{:}48,\ 2{:}54,\ 3{:}18$ & Levi $18$, $\Word$ $108$\\
$1$&$2$&$144$&numerical& $\leq3$ moves & $1{:}36,\ 2{:}102,\ 3{:}6$ & Levi $36$, $\Word$ $108$\\
$2$&$0$&$24$&symbolic& $g\mapsto g\Ush{D^{\tp}B}$ & $1{:}24$ & Levi $24$\\
$2$&$1$&$144$&numerical& $\leq3$ moves & $1{:}36,\ 2{:}102,\ 3{:}6$ & Levi $54$, $\Word$ $90$\\
$2$&$2$&$216$&numerical& $\leq3$ moves & $0{:}4,\ 1{:}56,\ 2{:}148,\ 3{:}8$ & Levi $48$, $\Word$ $168$\\
\bottomrule
\end{tabular}
\end{table*}

\begin{proof}[Proof of Lemma~\ref{lem:local}]
This is an exact finite verification. Enumerate all binary $2w\times 2w$ matrices
and retain precisely those satisfying $g\Jm g^{\tp}=\Jm$: there are six for $w=1$
and $720$ for $w=2$. For every retained matrix, enumerate the at most seven
nonzero members of each of the two spans in Eq.~\eqref{eq:d1-moves}. Take each such
member on both the left and the right. This gives at most $28$ moves out of each
block, and hence a directed graph on the group. The edges of that graph run from
a block to its images. The moves are genuinely directed, since the two spans are
recomputed from the current block. The search therefore runs in three steps:
\begin{enumerate}
\item Construct the whole edge set.
\item Reverse every edge.
\item Run one breadth-first search on the reversed graph, started from the set
  of states satisfying \eqref{eq:d1-term}.
\end{enumerate}
Step three assigns to every element its distance to that set in a single pass.
One forward search per element is not needed. The search reaches the whole
group, with exact distance distributions
\begin{equation}
  \begin{aligned}
  w=1:&\ \ 2\ \text{terminal},\ 4\ \text{at distance }1;\\
  w=2:&\ \ 16\ \text{terminal},\ 260\ \text{at distance }1,\\[-2pt]
      &\ \ 406\ \text{at distance }2,\ 38\ \text{at distance }3 .
  \end{aligned}
  \label{eq:d1-dist}
\end{equation}
The sums are $2+4=6$ and $16+260+406+38=720$, the full groups.
The maxima, one and three, are the move bounds in the statement. A certificate,
available at Ref.~\cite{albert2026certificate}, records a canonical shortest move
sequence for every
one of the $726$ blocks. The certificate also provides a script that replays each
of those sequences. The script re-derives the six parameters at every step and
checks that the recorded one lies in the relevant span.
\end{proof}

It is convenient to group the blocks of $\Sp{2w}$ by the pair
$(\ranko B,\ranko C)$, which is the relative-position invariant of the two
Lagrangians: $\dim\big(\Xh g\cap\Xh\big)=w-\ranko B$ and
$\dim\big(\Zh g\cap\Zh\big)=w-\ranko C$. Table~\ref{tab:d1-strata} collects both widths.
Two structural facts are visible in the table.
Terminality implies $\ranko B=\ranko C$. The sixteen terminal blocks
of $\Sp{4}$ consist of six Levi blocks at $(0,0)$, six rank-one blocks at
$(1,1)$, and four blocks with $C=B^{-1}$ at $(2,2)$. Hence the pair of ranks already
determines the type of the target. Furthermore the closed-form blocks always land
on Levi blocks, while terminals of type $\Word$ are reached only from blocks with
$\ranko B,\ranko C\geq1$.

It remains to connect the local moves just certified with the global gates one is
actually allowed to apply.

\begin{lemma}[Restriction]
\label{lem:restrict}
Let $G\in\NM$ have cell blocks $g_{\cell}$. Then for every cell $\cell$ and every
$\eps\in\Ftwo^{3}$,
\begin{equation}
  \begin{gathered}
  \Spz(G,\eps)_{\cell}=\Spz(g_{\cell},\eps),\\
  \Spx(G,\eps)_{\cell}=\Spx(g_{\cell},\eps).
  \end{gathered}
  \label{eq:d1-restrict}
\end{equation}
Hence restriction carries $\shp{G}$ onto $\shp{g_{\cell}}$ and $\shm{G}$ onto
$\shm{g_{\cell}}$: every move available at $g_{\cell}$ is the restriction
of a move on $G$, obtained by re-evaluating the same $\eps$ on the whole
layer, and by Lemma~\ref{lem:auto} that global move is a valid diagonal circuit
on $\Mt$.
\end{lemma}

\begin{proof}
The blocks $A,B,C,D$ are $\Mt$-block-diagonal, so products and sums of them are
too. Restriction to a cell is multiplicative and additive on those blocks:
$(XY^{\tp})_{\cell}=X_{\cell}Y_{\cell}^{\tp}$ and
$(X+X^{\tp})_{\cell}=X_{\cell}+X_{\cell}^{\tp}$. Each of the three global
generators therefore restricts to its local counterpart,
\begin{equation}
  \begin{gathered}
  (AB^{\tp})_{\cell}=A_{\cell}B_{\cell}^{\tp},\qquad
  (D^{\tp}B)_{\cell}=D_{\cell}^{\tp}B_{\cell},\\
  (B+B^{\tp})_{\cell}=B_{\cell}+B_{\cell}^{\tp},
  \end{gathered}
  \label{eq:d1-restrict3}
\end{equation}
and dually for the three generators of $\shm{G}$. Restriction is
$\Ftwo$-linear, so it commutes with forming the $\eps$-combination, which is
\eqref{eq:d1-restrict}; surjectivity follows because $\eps$ ranges over all of
$\Ftwo^{3}$ on both sides. Finally a move requires a nonzero parameter,
and $\Spz(G,\eps)_{\cell}\neq0$ forces $\Spz(G,\eps)\neq0$, so the induced
global move is itself a legitimate move.
\end{proof}

\subsection{Global reduction}
\label{sec:d1-global-reduction}

We now implement the above local (single-block) moves as parts of global moves 
to map any layer $G\in\NM$ to a terminal layer.

\begin{lemma}[Global reduction]
\label{lem:global-reduction}
Every $G\in\NM$ can be carried to a layer $\Gterm\in\NM$ all of whose cell blocks
are terminal, by left and right multiplication with elements of $\EupM$ and
$\EdnM$. At most $m_1+3m_2$ such multiplications are required.
\end{lemma}

\begin{proof}
Choose a cell $\cell$ whose block is not terminal, and a shortest local path from
$g_{\cell}$ to a terminal block, supplied by Lemma~\ref{lem:local}. Each edge of
that path specifies a side, a family, and a coefficient vector $\eps$. That
vector selects a binary combination of the three corresponding formulas in
\eqref{eq:d1-spans}.
Evaluate that same combination on the whole current layer. By
Lemma~\ref{lem:auto} the resulting global diagonal circuit is valid and $\Mt$-supported,
that is, it lies in $\EupM$ or $\EdnM$. By Lemma~\ref{lem:restrict} the
restriction of that diagonal circuit to $\cell$ is exactly the desired local move.
Multiplying by that diagonal circuit on the prescribed side therefore advances $\cell$ one
step along its path. The other cells may
move arbitrarily. After at most one such step if $\cell$ is a singleton, and at
most three if $\cell$ is a pair, the block $g_{\cell}$ is terminal. The counts
one and three are the bounds of Lemma~\ref{lem:local}.

Once $g_{\cell}$ is terminal, every formula in Eq.~\eqref{eq:d1-spans} vanishes on
that block by Eq.~\eqref{eq:d1-term}. For a terminal $g_{\cell}$, therefore,
$S_{\cell}=0$ for every binary combination of those formulas, and every later
move restricts to the identity on $\cell$. A finished cell can never be
disturbed again. Processing the cells one at a time, the number of terminal
cells strictly increases at the end of each phase and never decreases. Under
this cell-by-cell processing, every cell is terminal after at most one phase per
cell.
Since the layer is multiplied only by elements of $\EupM$ and $\EdnM$, both
contained in the group $\NM$, the result $\Gterm$ lies in $\NM$. Counting a
singleton cell at one move and a pair cell at three gives the bound
\begin{equation}
  \text{total moves}\ \leq\ m_1+3m_2 .
  \label{eq:d1-bound}
\end{equation}
\end{proof}

Three features of this argument are worth isolating, since none of them is a
matter of choice. Because the $Z$-diagonal circuits supplied by
Lemma~\ref{lem:auto} are themselves $\Mt$-block-diagonal, multiplication acts
blockwise:
$\big(\Ush{S}G\big)\big|_{\cell}=\Ush{S_{\cell}}\,g_{\cell}$, with $S_{\cell}$
the same formula evaluated on $g_{\cell}$. What one cannot do is choose the
$S_{\cell}$ independently. A move fixes one binary combination of the three
formulas, evaluated globally, and every cell then receives whatever that
combination does there. Steering one cell along its shortest path therefore drags
the others, which is why a naive simultaneous induction fails. This is also why
Lemma~\ref{lem:local} had to be established for the whole of $\Sp{2w}$ rather
than for some reachable subset. However badly a cell has been dragged, it is
still at most three moves from terminal. A formula chosen to advance one cell
can push another from distance two to distance three, and explicit examples do
this. The induction is on the number of terminal cells, which is monotone
because terminality is absorbing, and on nothing finer. Over the cells
that are neither finished nor being processed no control is needed at all.

\subsection{Terminal layers factor}
\label{sec:d1-terminal}

Reaching a canonical form is worthless unless the canonical form itself factors.
This canonical form does factor. The following is a statement about a single
local block $g_{\cell}\in\Sp{2w_{\cell}}$, so its $A,B,C,D$ are the local
$w_{\cell}\times w_{\cell}$ blocks with the cell subscript dropped. The assembly
afterwards works with the global $\Gterm$.

\begin{lemma}[Terminal normal form]
\label{lem:terminal}
Let $g=\left(\begin{smallmatrix}A&B\\C&D\end{smallmatrix}\right)$ be terminal,
with local blocks. At width one, either $g=\Id$ or $g=\Ush{1}\Lsh{1}\Ush{1}$. At
width two, either $B=C=0$, in which case $g$ is Levi, or
\begin{equation}
  g=\Ush{B}\Lsh{C}\Ush{B}
  \label{eq:d1-uvu}
\end{equation}
with $B$ and $C$ symmetric.
\end{lemma}

\begin{proof}
At width one, all quantities are scalars, $\left(\begin{smallmatrix}a&b\\c&d
\end{smallmatrix}\right)$ with $ad+bc=1$, and the conditions \eqref{eq:d1-term}
reduce to $ab=db=cd=ca=0$, the two symmetry conditions being vacuous. If $b=0$
then $ad=1$, so $a=d=1$, and $ca=0$ forces $c=0$, giving $g=\Id$. If $b=1$ then
$ab=0$ and $db=0$ give $a=d=0$, so $bc=1$ forces $c=1$, and the remaining
conditions hold; thus $g=\left(\begin{smallmatrix}0&1\\1&0\end{smallmatrix}
\right)=\Ush{1}\Lsh{1}\Ush{1}$. These are the two terminal blocks of $\Sp{2}$.

Now let $w=2$. Terminality \eqref{eq:d1-term} gives $B=B^{\tp}$ and $C=C^{\tp}$,
whence $AB^{\tp}=0$ reads $AB=0$, $D^{\tp}B=0$ transposes to $BD=0$, and
$C^{\tp}A=0$ reads $CA=0$; combining with $CD^{\tp}=0$ and with the symplectic
relation $AD^{\tp}+BC^{\tp}=\Id$ of Eq.~\eqref{eq:d1-sympl}, we obtain
\begin{equation}
  \begin{gathered}
  B=B^{\tp},\quad C=C^{\tp},\quad AD^{\tp}+BC=\Id,\\
  AB=0,\quad BD=0,\quad CD^{\tp}=0,\quad CA=0 .
  \end{gathered}
  \label{eq:d1-13}
\end{equation}
A nonzero $2\times2$ matrix has rank one or two, and rank two means invertible,
so the following three cases are exhaustive.

\emph{Case (i): $B=0$ or $C=0$.} If $B=0$ then $AD^{\tp}=\Id$, so $A$ and $D$ are
invertible, and $CA=0$ forces $C=0$. If $C=0$ then again $AD^{\tp}=\Id$, so $D$
is invertible, and $BD=0$ forces $B=0$. Either way $B=C=0$ and $g$ is Levi.

\emph{Case (ii): $B,C\neq0$ with at least one of them invertible.} If $B$ is
invertible, $AB=0$ gives $A=0$ and $BD=0$ gives $D=0$, so $AD^{\tp}+BC=\Id$
becomes $BC=\Id$, that is $C=B^{-1}$. If instead $C$ is invertible, $CA=0$ gives
$A=0$ and $CD^{\tp}=0$ gives $D=0$, and again $BC=\Id$. In both subcases
\begin{equation}
  \begin{split}
  \Ush{B}\Lsh{C}\Ush{B}
  &=\begin{pmatrix}\Id+BC&BCB\\ C&\Id+CB\end{pmatrix}\\
  &=\begin{pmatrix}0&B\\ C&0\end{pmatrix}=g ,
  \end{split}
  \label{eq:d1-antidiag}
\end{equation}
where the first equality is the direct computation
$\Ush{B}\Lsh{C}=\left(\begin{smallmatrix}\Id+BC&B\\ C&\Id\end{smallmatrix}
\right)$ followed by right multiplication by $\Ush{B}$, using
$(\Id+BC)B+B=BCB$, and the second uses $BC=CB=\Id$ and $BCB=B$. This is
\eqref{eq:d1-uvu}.

\emph{Case (iii): $B,C\neq0$ with both singular,} hence both of rank exactly
one. Every nonzero singular symmetric $2\times2$ binary matrix has the form
$vv^{\tp}$ for a nonzero column $v$. Writing the matrix as
$\left(\begin{smallmatrix}\alpha&\beta\\ \beta&\gamma\end{smallmatrix}\right)$,
singularity reads $\alpha\gamma=\beta$ over $\Ftwo$. The nonzero solutions of
$\alpha\gamma=\beta$ are the three matrices realized by $v=e_1$, $v=e_2$ and
$v=e_1+e_2$. (The remaining
four nonzero symmetric matrices are invertible and belong to case (ii); there are
seven nonzero symmetric $2\times2$ matrices in all.) So write
\begin{equation}
  B=uu^{\tp},\qquad C=vv^{\tp}
  \label{eq:d1-rankone}
\end{equation}
for nonzero columns $u,v$. Then $AB=0$ reads $(Au)u^{\tp}=0$, that is $Au=0$;
similarly $CA=0$ gives $v^{\tp}A=0$, $BD=0$ gives $u^{\tp}D=0$, and $CD^{\tp}=0$
gives $Dv=0$. If $u^{\tp}v=0$ then $BC=u(u^{\tp}v)v^{\tp}=0$, so
$AD^{\tp}=\Id$ and $A$ is invertible, contradicting $Au=0$ with $u\neq0$. Hence
$u^{\tp}v=1$, and consequently
\begin{equation}
  \begin{gathered}
  BCB=u(u^{\tp}v)(v^{\tp}u)u^{\tp}=uu^{\tp}=B,\\
  CBC=C,\qquad BC=uv^{\tp},\qquad CB=vu^{\tp}.
  \end{gathered}
  \label{eq:d1-bcb}
\end{equation}
Consider now the space of $2\times2$ matrices $X$ with $Xu=0$ and
$v^{\tp}X=0$. The condition $Xu=0$ says every row of $X$ lies in the
one-dimensional space $u^{\perp}$, so $X=yz^{\tp}$ with $z$ spanning $u^{\perp}$
and $y\in\Ftwo^{2}$ arbitrary. The further condition
$v^{\tp}X=(v^{\tp}y)z^{\tp}=0$ restricts $y$ to the one-dimensional space
$v^{\perp}$. The space is therefore
one-dimensional, and since $(\Id+uv^{\tp})u=u+u(v^{\tp}u)=0$ and
$v^{\tp}(\Id+uv^{\tp})=v^{\tp}+(v^{\tp}u)v^{\tp}=0$ while $\Id+uv^{\tp}\neq0$, it
is spanned by $\Id+uv^{\tp}=\Id+BC$. Now $A$ lies in that space, and
$AD^{\tp}=\Id+BC\neq0$ forces $A\neq0$, so
\begin{equation}
  A=\Id+BC .
  \label{eq:d1-Aform}
\end{equation}
The same argument applied to the conditions $Dv=0$ and $u^{\tp}D=0$ gives
$D=\Id+vu^{\tp}=\Id+CB$. Substituting into the first equality of
\eqref{eq:d1-antidiag} and using $BCB=B$ from Eq.~\eqref{eq:d1-bcb},
\begin{equation}
  \begin{split}
  \Ush{B}\Lsh{C}\Ush{B}
  &=\begin{pmatrix}\Id+BC&BCB\\ C&\Id+CB\end{pmatrix}\\
  &=\begin{pmatrix}A&B\\ C&D\end{pmatrix}=g ,
  \end{split}
  \label{eq:d1-case3}
\end{equation}
which proves \eqref{eq:d1-uvu} in the last case.
\end{proof}

\paragraph*{Assembly.} Apply Lemma~\ref{lem:terminal} simultaneously to all cells
of the blockwise-terminal layer
$\Gterm=\left(\begin{smallmatrix}A&B\\C&D\end{smallmatrix}\right)$ produced by
Lemma~\ref{lem:global-reduction}. The blocks of $\Gterm$ are now the global
$n\times n$ matrices $B=\bigoplus_{\cell}B_{\cell}$ and
$C=\bigoplus_{\cell}C_{\cell}$. Every local
$B_{\cell}$ and $C_{\cell}$ is symmetric, so the global $B$ and $C$ are
symmetric. Validity of $\Gterm$, that is \eqref{eq:d1-valid} applied to it,
contains in particular
\begin{equation}
  \CX B\subseteq\CZ,\qquad \CZ C\subseteq\CX ,
  \label{eq:d1-BC}
\end{equation}
so $\Ush{B}\in\EupM$ and $\Lsh{C}\in\EdnM$. It is worth being explicit about why.
For $\Ush{B}=\left(\begin{smallmatrix}\Id&B\\0&\Id\end{smallmatrix}\right)$ the
four conditions of Eq.~\eqref{eq:d1-valid} reduce to $\CX B\subseteq\CZ$ alone. The
other three are $\CX\Id\subseteq\CX$, $\CZ\cdot0\subseteq\CX$ and
$\CZ\Id\subseteq\CZ$. Moreover, $\CX B\subseteq\CZ$ is \emph{literally} one of
the four conditions already satisfied by $\Gterm$. The diagonal-circuit part of a valid
layer thus inherits validity with no further argument, and dually for $\Lsh{C}$.

The inclusion $\CX B\subseteq\CZ$ holds for \emph{every} valid layer,
terminal or not. One might wonder why Lemma~\ref{lem:auto} was needed at all,
rather than simply taking the parameter $B$ from the outset. The answer
is that validity is not sufficient: $\Ush{B}$ is a gate only if it is
symplectic, that is only if $B=B^{\tp}$, and for a general valid layer the block
$B$ is not symmetric. Here it is symmetric, precisely because $\Gterm$ is
terminal and terminality includes $B+B^{\tp}=C+C^{\tp}=0$.

Put
\begin{equation}
  \begin{split}
  \Word&:=\Ush{B}\Lsh{C}\Ush{B}\\
       &\;=\begin{pmatrix}\Id+BC&BCB\\ C&\Id+CB\end{pmatrix}\!,
  \end{split}
  \label{eq:d1-W}
\end{equation}
the second equality being the computation in Eq.~\eqref{eq:d1-antidiag}, valid for any
symmetric $B$ and $C$. Since all three factors are $\Mt$-block-diagonal, so is
$\Word$, and its restriction to a cell is the corresponding local product. On a
terminal cell of type $\Word$ we get
$\Word|_{\cell}=\Ush{B_{\cell}}\Lsh{C_{\cell}}\Ush{B_{\cell}}=g_{\cell}$ by
Lemma~\ref{lem:terminal}; on a Levi terminal cell, $B_{\cell}=C_{\cell}=0$ and
$\Word|_{\cell}=\Id$. Note that $\Word$ is an \emph{involution}: applying the
inversion formula \eqref{eq:d1-inv} to Eq.~\eqref{eq:d1-W} and using $B=B^{\tp}$,
$C=C^{\tp}$ returns $\Word$ itself, since $(\Id+CB)^{\tp}=\Id+BC$ and
$(BCB)^{\tp}=BCB$. Hence $\Word^{-1}=\Word$ and no inverse is needed anywhere in
the assembly. Define
\begin{equation}
  \Lterm:=\Gterm\,\Word^{-1}=\Gterm\,\Word .
  \label{eq:d1-Lstar}
\end{equation}
Restricting \eqref{eq:d1-Lstar} to a cell and using the two cases just computed
shows that $\Lterm$ is blockwise Levi. On a Levi cell $\Word|_{\cell}=\Id$, so
$\Lterm|_{\cell}=g_{\cell}$ is the original Levi block; on a cell of type $\Word$
we have $\Word|_{\cell}=g_{\cell}$, so
$\Lterm|_{\cell}=g_{\cell}g_{\cell}^{-1}=\Id$. Each kind of cell is thus
transparent to the factor it does not need. The two kinds in general both occur
in the same layer, and one might expect to need a separate treatment per cell.
One does not, because the single global group word $\Word$ already does the right
thing everywhere at once. Consequently $\Lterm$
has the global Levi form $\left(\begin{smallmatrix}K&0\\0&K^{\itp}\end{smallmatrix}
\right)$ with $K=\bigoplus_{\cell}K_{\cell}$, where $K_{\cell}$ is the original
Levi block on a Levi cell and the identity elsewhere.

That $\Lterm$ is a legitimate gate is seen differently from the way one sees it
for $\Ush{B}$ and $\Lsh{C}$. The difference is worth spelling out. Being of
Levi form, $\Lterm$ must satisfy
\begin{equation}
  \CX K\subseteq\CX,\qquad \CZ K^{\itp}\subseteq\CZ ,
  \label{eq:d1-levi}
\end{equation}
which together with the Levi form is exactly the membership $\Lterm\in\LevM$
required by Eq.~\eqref{eq:d1-thm}. But \eqref{eq:d1-levi} is \emph{derived}, not read
off: one cannot quote $\CX A\subseteq\CX$ from $\Gterm$, because $K\neq A$. The
matrix $K$ agrees with $A$ only on the Levi cells and is the identity on the
cells of type $\Word$. What gives \eqref{eq:d1-levi} is closure. The elements
$\Gterm$, $\Ush{B}$ and $\Lsh{C}$ all lie in $\NM$, which is a group, so $\Word$
and $\Word^{-1}$ do too, hence so does $\Lterm=\Gterm\Word^{-1}$. In addition, a
valid layer of Levi form satisfies \eqref{eq:d1-levi} by the definition
\eqref{eq:d1-valid} of validity. In the
same vein, Lemma~\ref{lem:terminal} factors the \emph{abstract local matrix}
$g_{\cell}$. But the local gate $\Ush{B_{\cell}}$ appearing there is in general
not a valid gate. Here $\Ush{B_{\cell}}$ means $B_{\cell}$ placed in cell
$\cell$ with zeros elsewhere. What is valid is the global $\Ush{B}$ assembled
from all the $B_{\cell}$ at once. Lemma~\ref{lem:terminal} is used only to prove
the identity
$\Word|_{\cell}=g_{\cell}$; it never supplies a gate.

Rearranging \eqref{eq:d1-Lstar} we obtain the terminal factorization
\begin{equation}
  \Gterm=\Lterm\,\Ush{B}\,\Lsh{C}\,\Ush{B},
  \label{eq:d1-factor}
\end{equation}
a product of four valid depth-one layers on $\Mt$: one from $\LevM$,
two from $\EupM$ and one from $\EdnM$.

\begin{proof}[Proof of Theorem~\ref{thm:depthone}]
The inclusion $\supseteq$ in Eq.~\eqref{eq:d1-thm} is immediate, all three families
being subgroups of $\NM$. For $\subseteq$, let $G\in\NM$. By
Lemma~\ref{lem:global-reduction} there are elements
$X_1,\dots,X_p\in\EupM\cup\EdnM$ and $Y_1,\dots,Y_q\in\EupM\cup\EdnM$ with
$p+q\leq m_1+3m_2$ such that
$\Gterm=X_p\cdots X_1\,G\,Y_1\cdots Y_q$ is blockwise terminal, and
$\Gterm\in\NM$. By Eq.~\eqref{eq:d1-factor},
$\Gterm=\Lterm\Ush{B}\Lsh{C}\Ush{B}$ with $\Lterm\in\LevM$, $\Ush{B}\in\EupM$
and $\Lsh{C}\in\EdnM$. Therefore
\begin{equation}
  \begin{split}
  G=X_1^{-1}\cdots X_p^{-1}\;&\Lterm\,\Ush{B}\,\Lsh{C}\,\Ush{B}\\
    &\times Y_q^{-1}\cdots Y_1^{-1},
  \end{split}
  \label{eq:d1-assembly}
\end{equation}
which exhibits $G$ as a group word in the three families. (Over $\Ftwo$ every diagonal circuit is
its own inverse, $\Ush{S}^{-1}=\Ush{S}$, so the inverses in
\eqref{eq:d1-assembly} may be dropped, but nothing depends on that.) Since $G$
was arbitrary, \eqref{eq:d1-thm} holds.
\end{proof}

The argument is constructive and gives a group word of length at most
$(m_1+3m_2)+4$, the sandwich of Lemma~\ref{lem:global-reduction} together with the four
factors of Eq.~\eqref{eq:d1-factor}. Nothing above or below depends on that count,
and no synthesis routine is claimed. The theorem itself is a generation statement
in which group-word length plays no role.

\subsection{A worked example: the \texorpdfstring{$\qcode{6,2,2}$}{[[6,2,2]]} code}
\label{sec:d1-example}

A single small code shows every mechanism in the proof at once: dragging;
freezing; block-diagonality; the impossibility of acting cell by cell; and the
looseness of the move bound. Take the doubly-even self-dual
$\mathrm{CSS}(C,C)$ with
\begin{equation}
  \CX=\CZ=C=\spano\{111100,\ 110011\},
  \label{eq:d1-622code}
\end{equation}
a connected $\qcode{6,2,2}$: both generators have weight four and share qubits $0$
and $1$, so the code is indecomposable. Fix the perfect matching
$\Mt=\{\{0,1\},\{2,3\},\{4,5\}\}$, three pair cells $\cell_1,\cell_2,\cell_3$, so
$m_1=0$ and $m_2=3$. At this matching
\begin{equation}
  |\NM|=3072
  \quad\text{inside}\quad
  |\DM|=720^{3},
  \label{eq:d1-622count}
\end{equation}
that is, $3072$ of the $373\,248\,000$ depth-one layers are valid. Exactly $15$
valid $Z$-diagonal circuits, $15$ valid $X$-diagonal circuits and $7$ nontrivial Levi elements are
available. Validity is therefore a strong constraint at each matching.

Take the element $G\in\NM$ whose three cell blocks, written in the ordered basis
$(x_i,x_j\,|\,z_i,z_j)$ of each cell, are
\begin{equation}
  g_{\cell_1}\!=\!g_{\cell_2}\!=\!\!
  \begin{pmatrix}0&0&0&1\\0&0&1&0\\0&1&0&1\\1&0&1&0\end{pmatrix}\!\!,
  \ \
  g_{\cell_3}\!=\!\!
  \begin{pmatrix}1&1&0&1\\1&1&1&0\\1&0&0&1\\0&1&1&0\end{pmatrix}\!\!,
  \label{eq:d1-622start}
\end{equation}
with local distances $1$, $1$ and $2$ respectively. The reduction of
Lemma~\ref{lem:global-reduction}, run with cells processed in the order
$\cell_1,\cell_2,\cell_3$, takes three moves, all of them $Z$-diagonal circuits. We call
the parameters of those moves $S^{(1)},S^{(2)},S^{(3)}$ and list them in
Table~\ref{tab:d1-622run}.

\begin{table*}[t]
\caption{\label{tab:d1-622run}
The three moves reducing \eqref{eq:d1-622start} to a blockwise-terminal layer.
Each row gives the working cell, the family and formula selected, the side on
which the resulting gate is multiplied, the global parameter as a direct sum of
its three cell blocks, the physical circuit it realizes, and the resulting vector
of local distances.}
\footnotesize
\begin{tabular}{clllll}
\toprule
move & working cell & formula, side & global parameter & circuit & distances\\
\midrule
start & --- & --- & --- & --- & $(1,1,2)$\\
$S^{(1)}$ & $\cell_1$ & $D^{\tp}B$, right
  & $\Id_2\oplus \Id_2\oplus \Id_2=\Id_6$
  & $\prod_{i=0}^{5}\sqrt{Z}_i$ & $(0,0,2)$\\
$S^{(2)}$ & $\cell_3$ & $AB^{\tp}$, left
  & $0\oplus0\oplus\left(\begin{smallmatrix}1&1\\1&1\end{smallmatrix}\right)$
  & $\sqrt{Z}_4\,\sqrt{Z}_5\,\mathrm{CZ}_{45}$ & $(0,0,1)$\\
$S^{(3)}$ & $\cell_3$ & $D^{\tp}B$, right
  & $0\oplus0\oplus\left(\begin{smallmatrix}1&1\\1&1\end{smallmatrix}\right)$
  & $\sqrt{Z}_4\,\sqrt{Z}_5\,\mathrm{CZ}_{45}$ & $(0,0,0)$\\
\bottomrule
\end{tabular}
\end{table*}

Several features of the general argument are visible in this one run.

\emph{Dragging.} The cell $\cell_2$ is never given a phase of its own. The first
move, chosen to advance $\cell_1$, has global parameter $\Id_6$, whose
restriction to $\cell_2$ is the nonzero matrix $\Id_2$. That move therefore acts
on $\cell_2$ as well, and clears it as a side effect. The distance at $\cell_2$
drops from $1$ to $0$ without any move having been selected for it. The same move
also alters the block at $\cell_3$, though without changing its distance.

\emph{Freezing.} On moves $S^{(2)}$ and $S^{(3)}$ the global parameter restricts to
\emph{zero} on the finished cells $\cell_1$ and $\cell_2$. This is not arranged.
All six formulas of Eq.~\eqref{eq:d1-spans} vanish identically on a terminal block,
so every binary combination of them does too. On such a terminal block the
corresponding gate restricts to the identity.

\emph{The move bound is loose.} Three moves are used against the guarantee
$m_1+3m_2=9$ of Eq.~\eqref{eq:d1-bound}, because one global formula can advance
several cells at once. The bound is a worst case, not a prediction.

The resulting terminal layer has cell blocks
\begin{equation}
  g^{\star}_{\cell_1}\!=\!g^{\star}_{\cell_2}\!=\!\!
  \begin{pmatrix}0&0&0&1\\0&0&1&0\\0&1&0&0\\1&0&0&0\end{pmatrix}\!\!,
  \ \
  g^{\star}_{\cell_3}\!=\!\!
  \begin{pmatrix}0&0&1&0\\0&0&0&1\\1&0&0&0\\0&1&0&0\end{pmatrix}\!\!,
  \label{eq:d1-622term}
\end{equation}
all three of type $\Word$ in the sense of Lemma~\ref{lem:terminal}: each has
$A_{\cell}=D_{\cell}=0$ with $B_{\cell}$ invertible and
$C_{\cell}=B_{\cell}^{-1}$, this being case (ii) of that lemma. The global
parameters are therefore
\begin{equation}
  B=C=
  \begin{pmatrix}0&1\\1&0\end{pmatrix}\oplus
  \begin{pmatrix}0&1\\1&0\end{pmatrix}\oplus
  \begin{pmatrix}1&0\\0&1\end{pmatrix},
  \label{eq:d1-622BC}
\end{equation}
both symmetric as terminality requires, so that $\Ush{B}$ is the circuit
$\mathrm{CZ}_{01}\,\mathrm{CZ}_{23}\,\sqrt{Z}_4\,\sqrt{Z}_5$ and $\Lsh{C}$ is its
mirror image in the $X$ basis. The group word $\Word=\Ush{B}\Lsh{C}\Ush{B}$ is an
involution. Since no cell is Levi one has $\Lterm=\Gterm\Word=\Id$ here: the
whole terminal layer is carried by the $\Word$ group word. Undoing the three moves,
which are involutions as well, gives the identity
\begin{equation}
  G=\Ush{S^{(2)}}\Ush{B}\Lsh{C}\Ush{B}\Ush{S^{(3)}}\Ush{S^{(1)}},
  \label{eq:d1-622word}
\end{equation}
a group word of six valid diagonal circuits on the same matching $\Mt$, in
accordance with Eq.~\eqref{eq:d1-assembly}. Here $S^{(2)}=S^{(3)}$, so the group word uses
only four distinct parameters.

Two phenomena that the argument allows do \emph{not} occur for this input. It
is worth saying so, since a reader might otherwise take the example for the
general case. No cell is ever dragged to a strictly \emph{worse} distance here.
In general a formula chosen to advance one cell can increase another cell's
distance. The induction is therefore on the number of terminal cells rather
than on any aggregate of distances. No mixed terminal layer occurs: for this
code and this matching every terminal cell reached is of type $\Word$. The Levi
factor is trivial for this matching. Mixed terminal
layers, with some cells Levi and some of type $\Word$, are handled by the
assembly of Sec.~\ref{sec:d1-terminal} without
any per-cell case distinction.

%% file: app_transversal.tex
\section{Three-block decomposition for transversal gates}
\label{app:transversal}

The depth-one theorem of Appendix~\ref{app:depthone} holds for an arbitrary partial
matching $\Mt$ of the $n$ qubits. The content of that theorem grows with the size of the
cells that $\Mt$ is allowed to use. It is worth descending to the finest partition, in
which no two qubits are matched at all and every cell is a singleton. The layers
available there are the strictly transversal ones --- a single-qubit Clifford on each
qubit. In this case, the generation statement can be sharpened to a decomposition of
\emph{fixed length}, valid uniformly in $n$ and in the code. That decomposition, and the
exact criterion for when it can be shortened, are the subject of this appendix.

Throughout, $\lab=(\CX|0)\oplus(0|\CZ)\subseteq\Vsp$ is a CSS-form label space,
vectors are rows, and matrices act on the right. Only this splitting of $\lab$
is used below. In particular, orthogonality $\CX\perp\CZ$ is what makes $\lab$ a
stabilizer label space, but it is never invoked in any of the arguments below.
No argument below uses any distance hypothesis, and none assumes that $\CX$ and
$\CZ$ are nonzero.

Set $\Mt=\emptyset$: no two qubits are matched, so every cell is a singleton. The corresponding
block-diagonal subgroup $\DM$ of $\Sp{2n}$ is then a product of one copy of $\Sp{2}$ per
qubit, so that the transversal group,
\begin{equation}
  \Ntr=\DM\cap\Nfull=\Sp{2}^{n}\cap\Nfull ,
  \label{eq:trans-ntr}
\end{equation}
consists of exactly those tuples of single-qubit Cliffords that preserve $\lab$. The
three families of Appendix~\ref{app:depthone} specialize as follows. An element of
$\LevM$ is one of the Levi gates \eqref{eq:setup-lev} whose $K$ is in addition
$\Mt$-block-diagonal. At $\Mt=\emptyset$ each diagonal block of $K$ is a single
invertible element of $\Ftwo$, so $K=\Id$. At $\Mt=\emptyset$ the Levi slice $\LevM$ is
the trivial group, for the elementary reason that $\GLg{1}$ is trivial. The
two slices $\EupM$ and $\EdnM$ become the transversal families
$\Eupo$ and $\Edno$. The parameters of $\Eupo$ and $\Edno$ are diagonal matrices, and
are therefore automatically symmetric. Writing a diagonal parameter as the
vector on its diagonal, in the width-1 shorthand $\Ush{a}:=\Ush{\diago a}$ and
$\Lsh{a}:=\Lsh{\diago a}$, these families are
\begin{equation}
  \begin{aligned}
    \Eupo&=\{\Ush{a}:a\in\Pup\},\\
    \Edno&=\{\Lsh{b}:b\in\Pdn\},
  \end{aligned}
  \label{eq:trans-families}
\end{equation}
where the two \emph{parameter codes} are
\begin{equation}
  \begin{aligned}
    \Pup&=\{a\in\Ftwo^{n}:a\pw\CX\subseteq\CZ\},\\
    \Pdn&=\{b\in\Ftwo^{n}:b\pw\CZ\subseteq\CX\}.
  \end{aligned}
  \label{eq:trans-parcodes}
\end{equation}
In circuit terms the two families are the transversal layers of
Appendix~\ref{app:setup}: $\Ush{a}$ is a $\sqrt{Z}$ on each qubit of $\suppo(a)$ and
the identity elsewhere, and dually $\Lsh{b}$ is a $\sqrt{X}$ on each qubit of
$\suppo(b)$. The parameter codes thus record which subsets of qubits may carry a phase
gate, respectively its conjugate-basis partner, without disturbing the code.
For fixed $w$ the map $a\mapsto a\pw w$ is $\Ftwo$-linear in $a$. Each defining condition
in Eq.~\eqref{eq:trans-parcodes} is the requirement that a linear image of $a$ land in a
subspace. Hence $\Pup$ and $\Pdn$ are binary linear codes, obtained from generator
matrices of $\CX$ and $\CZ$ by Gaussian elimination. Since $\Ush{a}\Ush{a'}=\Ush{a+a'}$
and $\Lsh{b}\Lsh{b'}=\Lsh{b+b'}$, the parametrizations in Eq.~\eqref{eq:trans-families} are
group isomorphisms from $(\Pup,+)$ and $(\Pdn,+)$. It follows that $\Eupo$ and
$\Edno$ are elementary abelian $2$-groups whose orders are read off two ranks.

It is convenient to have the block form of a transversal element on record. Write
\begin{equation}
  G=\begin{pmatrix}A&B\\C&D\end{pmatrix},
  \label{eq:trans-blocks}
\end{equation}
where all four blocks are diagonal because $G$ is transversal, and identify each block
with the vector in $\Ftwo^{n}$ on its diagonal, so that matrix products of blocks become
pointwise products of vectors. Since rows act on the right,
$(x|z)G=(xA+zC\,|\,xB+zD)$. A vector $(w|0)$ with $w\in\CX$ is then sent to
$(A\pw w\,|\,B\pw w)$, and a vector $(0|w)$ with $w\in\CZ$ is sent to
$(C\pw w\,|\,D\pw w)$. Applying the requirement $\lab G\subseteq\lab$ to spanning
vectors of the two summands of $\lab$, membership $G\in\Ntr$ is therefore exactly
\begin{equation}
  \begin{aligned}
    A\pw\CX&\subseteq\CX, &\qquad B\pw\CX&\subseteq\CZ,\\
    C\pw\CZ&\subseteq\CX, &\qquad D\pw\CZ&\subseteq\CZ ,
  \end{aligned}
  \label{eq:trans-valid}
\end{equation}
together with symplecticity, which in this identification is the single pointwise
relation $A\pw D+B\pw C=\ones$.

With the three families identified, the fixed-matching theorem specializes at once.

\begin{corollary}[Transversal generation]
\label{cor:transgen}
For every CSS-form label space $\lab$,
\begin{equation}
  \Ntr=\langle\,\Eupo,\Edno\,\rangle .
  \label{eq:trans-gen}
\end{equation}
\end{corollary}

\begin{proof}
Apply Theorem~\ref{thm:depthone} at $\Mt=\emptyset$: it gives
$\NM=\langle\EupM,\EdnM,\LevM\rangle$. By the specialization just described,
$\NM=\Ntr$, $\EupM=\Eupo$, $\EdnM=\Edno$, and $\LevM=\{\Id\}$, which is exactly
\eqref{eq:trans-gen}.
\end{proof}

Only the width-one case of Lemma~\ref{lem:local} is used here. At width one every block
is a scalar, so $B+B^{\tp}$ and $C+C^{\tp}$ vanish identically and terminality reduces
to $AB=DB=CD=CA=0$. At width one the two terminal blocks are $\Id$ and $\sw$. Each of
the four remaining elements of $\Sp{2}$ has a nonzero parameter and is cleared by one
move. The finite verification behind Theorem~\ref{thm:depthone} therefore enters this
appendix only through a six-element check.

Equation \eqref{eq:trans-gen} says that every transversal gate is \emph{some}
group word in valid diagonal circuits. Equation \eqref{eq:trans-gen} says nothing at all
about how long that group word must be: a generation statement carries no bound
on group-word length. A priori the length could grow with $n$ or with the
dimensions of $\CX$ and $\CZ$. The point of the following theorem is precisely
that the length grows with neither. Three letters always suffice, in a fixed
alternating pattern.

\begin{theorem}[Transversal decomposition]
\label{thm:transnf}
For every CSS-form label space $\lab$, every $G\in\Ntr$ can be written
\begin{equation}
  G=\Ush{a}\,\Lsh{b}\,\Ush{a'},\quad a,a'\in\Pup,\; b\in\Pdn ;
  \label{eq:trans-nf}
\end{equation}
equivalently, as a product of subsets of $\Sp{2n}$,
\begin{equation}
  \Ntr=\Eupo\cdot\Edno\cdot\Eupo .
  \label{eq:trans-nfset}
\end{equation}
\end{theorem}

Normal forms of this genre are standard. Every element of the unconstrained
symplectic group factors into Siegel, Levi and Weyl letters. Templates of the same
shape compile logical Clifford circuits
(cf.~\cite{dehaene2003clifford,rengaswamy2018synthesis,chen2024tailoring,popov2026optimized}).
Theorem~\ref{thm:transnf} adds that all three letters are transversal, valid for
the code, and diagonal circuits.

The inclusion $\supseteq$ in Eq.~\eqref{eq:trans-nfset} is immediate, since each of the three
factors already lies in $\Ntr$ by Eq.~\eqref{eq:trans-families}; the content is the inclusion
$\subseteq$. The proof occupies Secs.~\ref{sec:trans-product}
and~\ref{sec:trans-closure}. Corollary~\ref{cor:transgen} already supplies group words of
unbounded length. Given those group words, it is enough to show that the three-block set
is closed under multiplication, and one exchange identity does that.
Section~\ref{sec:trans-sharp} then shows that the length three is sharp, and identifies
exactly the codes for which it drops to at most two.

\subsection{Pointwise products}
\label{sec:trans-product}

The proof rests on one closure property of the parameter codes. The two conditions in
\eqref{eq:trans-parcodes} move vectors in opposite directions between $\CX$ and $\CZ$.
Applying the two conditions alternately therefore returns a vector to the code it started
in. A parameter that is valid in both senses thus pushes each code into the intersection
$\CX\cap\CZ$.

\begin{lemma}[Product lemma]
\label{lem:product}
If $a\in\Pup$ and $b\in\Pdn$, then
\begin{equation}
  a\pw b\pw\CX\subseteq\CX\cap\CZ ,\quad
  a\pw b\pw\CZ\subseteq\CX\cap\CZ ;
  \label{eq:trans-prodstrong}
\end{equation}
in particular $a\pw b\in\Pup\cap\Pdn$.
\end{lemma}

\begin{proof}
Let $w\in\CX$. Since $a\in\Pup$ we have $a\pw w\in\CZ$. Since $b\in\Pdn$ we may apply the
second condition of Eq.~\eqref{eq:trans-parcodes} to this vector, giving
$b\pw a\pw w\in\CX$. Applying $a\in\Pup$ once more to that vector gives
$a\pw b\pw a\pw w\in\CZ$. Pointwise multiplication is commutative and idempotent,
$a\pw a=a$, so the last two vectors are both equal to $a\pw b\pw w$. The two memberships
just derived therefore read $a\pw b\pw w\in\CX$ and $a\pw b\pw w\in\CZ$, which is the
first half of Eq.~\eqref{eq:trans-prodstrong}. The same chain run in the other direction
proves the second half: for $w\in\CZ$ we get $b\pw w\in\CX$, then $a\pw b\pw w\in\CZ$,
then $b\pw a\pw b\pw w\in\CX$. That last vector is again $a\pw b\pw w$ by commutativity
and idempotence.

Finally, \eqref{eq:trans-prodstrong} contains $a\pw b\pw\CX\subseteq\CZ$ and
$a\pw b\pw\CZ\subseteq\CX$, which are the defining conditions
$a\pw b\in\Pup$ and $a\pw b\in\Pdn$.
\end{proof}

\subsection{Closure of the three-block set}
\label{sec:trans-closure}

The obstruction to shortening group words in $\Eupo\cup\Edno$ is that valid
$Z$-diagonal and $X$-diagonal circuits need not commute. There is therefore in general no
two-factor exchange rule of the form $\Lsh{b}\Ush{a}=\Ush{\cdot}\Lsh{\cdot}$.
There is, however, an exchange rule with three factors on each side, and its
right-hand side is again in the required upper--lower--upper pattern. For all
$a,b,c\in\Ftwo^{n}$,
\begin{equation}
  \begin{aligned}
    \Lsh{b}\,\Ush{a}\,\Lsh{c}
      &=\Ush{a+a\pw b+a\pw b\pw c}\\
      &\quad\cdot\Lsh{b+c+a\pw b\pw c}\\
      &\quad\cdot\Ush{a\pw b}.
  \end{aligned}
  \label{eq:trans-exchange}
\end{equation}

Proving this involves checking eight
cases.

\begin{proof}[Proof of Eq.~\eqref{eq:trans-exchange}]
Fix a coordinate $i$ and drop it from the notation, so that $a,b,c\in\Ftwo$ and every
factor is one of the six elements of $\Sp{2}$.

If $a=0$ the left-hand side is $\Lsh{b}\Lsh{c}=\Lsh{b+c}$. On the right-hand side the
first and third exponents, $a+a\pw b+a\pw b\pw c$ and $a\pw b$, both vanish, and the
middle exponent is $b+c+a\pw b\pw c=b+c$. At $a=0$ the right-hand side is
$\Lsh{b+c}$ as well.

If $a=1$ there are four subcases, according to the pair $(b,c)$. The left-hand sides are,
respectively,
\begin{equation}
  \begin{aligned}
    &\Ush{1},\qquad \Ush{1}\Lsh{1},\\
    &\Lsh{1}\Ush{1},\qquad \Lsh{1}\Ush{1}\Lsh{1},
  \end{aligned}
  \label{eq:trans-fourwords}
\end{equation}
for $(b,c)=(0,0),(0,1),(1,0),(1,1)$, since a factor with exponent $0$ is the identity. On
the right-hand side the three exponents
$\bigl(a+a\pw b+a\pw b\pw c,\;b+c+a\pw b\pw c,\;a\pw b\bigr)$ evaluate to $(1,0,0)$,
$(1,1,0)$, $(0,1,1)$ and $(1,1,1)$ in the four subcases, giving
$\Ush{1}$, $\Ush{1}\Lsh{1}$, $\Lsh{1}\Ush{1}$ and $\Ush{1}\Lsh{1}\Ush{1}$. The first
three agree with Eq.~\eqref{eq:trans-fourwords} exactly. In the fourth we need
\begin{equation}
  \Lsh{1}\Ush{1}\Lsh{1}=\Ush{1}\Lsh{1}\Ush{1},
  \label{eq:trans-braid}
\end{equation}
which holds because over $\Ftwo$ both sides equal the swap $\sw$:
\begin{equation}
  \begin{aligned}
    \begin{pmatrix}1&1\\0&1\end{pmatrix}
    \begin{pmatrix}1&0\\1&1\end{pmatrix}
    \begin{pmatrix}1&1\\0&1\end{pmatrix}
    &=\begin{pmatrix}0&1\\1&0\end{pmatrix},\\[2pt]
    \begin{pmatrix}1&0\\1&1\end{pmatrix}
    \begin{pmatrix}1&1\\0&1\end{pmatrix}
    \begin{pmatrix}1&0\\1&1\end{pmatrix}
    &=\begin{pmatrix}0&1\\1&0\end{pmatrix}.
  \end{aligned}
  \label{eq:trans-swapword}
\end{equation}
This exhausts the eight cases.
\end{proof}

Identity \eqref{eq:trans-exchange} is only useful if it preserves validity, and this is
where Lemma~\ref{lem:product} enters. Suppose $a\in\Pup$ and $b,c\in\Pdn$, so that the
left-hand side of Eq.~\eqref{eq:trans-exchange} is a product of three valid gates. Applying
Lemma~\ref{lem:product} to $a\in\Pup$ and $b\in\Pdn$ gives $a\pw b\in\Pup\cap\Pdn$, and a
second application, now to $a\pw b\in\Pup$ and $c\in\Pdn$, gives
\begin{equation}
  a\pw b\in\Pup\cap\Pdn,\qquad a\pw b\pw c\in\Pup\cap\Pdn .
  \label{eq:trans-twoprods}
\end{equation}
Since $\Pup$ and $\Pdn$ are linear, the first exponent in Eq.~\eqref{eq:trans-exchange} is a
sum of the three vectors $a$, $a\pw b$ and $a\pw b\pw c$, all in $\Pup$, and so lies in
$\Pup$; the second is a sum of $b$, $c$ and $a\pw b\pw c$, all in $\Pdn$, and so lies in
$\Pdn$; the third is $a\pw b\in\Pup$. Every factor on the right of
\eqref{eq:trans-exchange} is thus a valid gate of the required type.

\begin{proof}[Proof of Theorem~\ref{thm:transnf}]
The three-block set $\Eupo\cdot\Edno\cdot\Eupo$ is contained in $\Ntr$, since each of its
three factors is. We show that it is in fact a subgroup of $\Ntr$.

It contains the identity, which is the group word with $a=b=a'=0$. It is closed under
inversion: over $\Ftwo$ each diagonal circuit is an involution, $\Ush{a}^{-1}=\Ush{a}$ and
$\Lsh{b}^{-1}=\Lsh{b}$, so reversing a three-block group word inverts it,
\begin{equation}
  \bigl(\Ush{a}\,\Lsh{b}\,\Ush{a'}\bigr)^{-1}=\Ush{a'}\,\Lsh{b}\,\Ush{a},
  \label{eq:trans-inverse}
\end{equation}
and the right-hand side is again of three-block form.

For closure under multiplication, take two elements of $\Eupo\cdot\Edno\cdot\Eupo$ and
multiply them. Their two adjacent inner upper factors combine by
$\Ush{a'}\Ush{a''}=\Ush{a'+a''}$, so the product has the five-block form
\begin{equation}
  \Ush{a}\,\Lsh{b}\,\Ush{c}\,\Lsh{d}\,\Ush{e},
  \label{eq:trans-fiveblock}
\end{equation}
with $a,c,e\in\Pup$ and $b,d\in\Pdn$. Apply \eqref{eq:trans-exchange} to the middle three
factors $\Lsh{b}\Ush{c}\Lsh{d}$, whose parameters satisfy the hypotheses of the validity
discussion above. This replaces them by a valid group word of the pattern
upper--lower--upper, after which \eqref{eq:trans-fiveblock} has the pattern
upper--upper--lower--upper--upper. Combining the two adjacent upper factors at each end
returns a group word of the three-block pattern, with all parameters valid.

So $\Eupo\cdot\Edno\cdot\Eupo$ is a subgroup of $\Ntr$. It contains $\Eupo$ (take
$b=a'=0$) and $\Edno$ (take $a=a'=0$), hence contains $\langle\Eupo,\Edno\rangle$, which
is all of $\Ntr$ by Corollary~\ref{cor:transgen}. Combined with the inclusion
$\Eupo\cdot\Edno\cdot\Eupo\subseteq\Ntr$ noted at the outset, this gives
\eqref{eq:trans-nfset}, and \eqref{eq:trans-nf} is the elementwise restatement.
\end{proof}

The argument is effective as well as existential. Corollary~\ref{cor:transgen} expresses a
given $G$ as some group word in valid diagonal circuits. The two rewriting steps just used are
merging adjacent like factors and applying \eqref{eq:trans-exchange} to a
lower--upper--lower segment. These two steps shorten any group word in valid diagonal circuits to
three blocks, and each step produces parameters that are again valid. The triple
$(a,b,a')$ is thus produced by rewriting rather than in closed
form, but a closed form is available too: one checks coordinatewise, using
$A\pw D+B\pw C=\ones$, that
\begin{equation}
  b=C,\qquad a=A\pw B+B\pw C,\qquad a'=B\pw C+B\pw C\pw D
  \label{eq:trans-explicit}
\end{equation}
also satisfies $G=\Ush{a}\Lsh{b}\Ush{a'}$, the three parameters being valid by
\eqref{eq:trans-valid} together with Lemma~\ref{lem:product}.

\subsection{Sharpness of the block length}
\label{sec:trans-sharp}

Three blocks cannot in general be reduced to two, and the obstruction is visible in the
parameter codes alone.

\begin{proposition}[Two blocks suffice exactly when the parameter codes meet trivially]
\label{prop:twoblock}
For every CSS-form label space $\lab$,
\begin{equation}
  \Ntr=\Eupo\cdot\Edno
  \quad\Longleftrightarrow\quad
  \Pup\cap\Pdn=\{0\}.
  \label{eq:trans-twoblock}
\end{equation}
\end{proposition}

\begin{proof}
Assume first that $\Pup\cap\Pdn=\{0\}$. By Lemma~\ref{lem:product}, $a\pw b\in\Pup\cap\Pdn$
for every $a\in\Pup$ and $b\in\Pdn$, so $a\pw b=0$ for every such pair. Thus every valid
$Z$-diagonal circuit and every valid $X$-diagonal circuit have disjoint supports. Two transversal gates
supported on disjoint sets of qubits commute, so $\Ush{a}$ and $\Lsh{b}$ commute for all
$a\in\Pup$, $b\in\Pdn$. Consequently every group word in $\Eupo\cup\Edno$ may be
rearranged so that all upper letters precede all lower letters. Each of the two resulting
runs collapses to a single letter, because $\Eupo$ and $\Edno$ are groups. Hence
$\langle\Eupo,\Edno\rangle=\Eupo\cdot\Edno$, and Corollary~\ref{cor:transgen} identifies
the left-hand side with $\Ntr$.

For the converse we argue contrapositively: assume there exists a nonzero
$t\in\Pup\cap\Pdn$ and exhibit an element of $\Ntr$ outside $\Eupo\cdot\Edno$. Apply
Lemma~\ref{lem:product} with $a=b=t$; since $t\pw t=t$, its conclusion
\eqref{eq:trans-prodstrong} reads
\begin{equation}
  t\pw\CX\subseteq\CX\cap\CZ,\qquad t\pw\CZ\subseteq\CX\cap\CZ .
  \label{eq:trans-hadvalid}
\end{equation}
Consider the subset Hadamard $\Hd{t}$, acting as $\sw$ on $\suppo(t)$ and as the identity
elsewhere. Its blocks in the notation \eqref{eq:trans-blocks} are $A=D=\ones+t$ and
$B=C=t$, so the four conditions \eqref{eq:trans-valid} for $\Hd{t}$ read
$(\ones+t)\pw\CX\subseteq\CX$, $t\pw\CX\subseteq\CZ$, $t\pw\CZ\subseteq\CX$ and
$(\ones+t)\pw\CZ\subseteq\CZ$. The two middle ones are contained in
\eqref{eq:trans-hadvalid}. Since $w+t\pw w\in\CX$ together with $w\in\CX$ gives
$t\pw w\in\CX$, the first is equivalent to $t\pw\CX\subseteq\CX$, which is again part of
\eqref{eq:trans-hadvalid}, and symmetrically for the fourth. Hence $\Hd{t}\in\Ntr$.

On the other hand $\Hd{t}$ is not a two-block product. A product of one upper and one
$X$-diagonal circuit is
\begin{equation}
  \Ush{S}\,\Lsh{T}=\begin{pmatrix}\ones+S\pw T & S\\ T & \ones\end{pmatrix},
  \label{eq:trans-twofactor}
\end{equation}
whose lower-right block is $\ones$ for every choice of the parameters, whereas
the lower-right block of $\Hd{t}$ is $\ones+t$, which differs from $\ones$ on
$\suppo(t)\neq\emptyset$. Equivalently, at any coordinate of $\suppo(t)$ the
local gate is $\sw$, and by Eq.~\eqref{eq:trans-swapword} this is a group word of
length three and not a product of at most two of the local letters
$\Ush{1},\Lsh{1}$. Either way $\Hd{t}\notin\Eupo\cdot\Edno$, so
$\Ntr\neq\Eupo\cdot\Edno$.
\end{proof}

The computation in the proof characterizes validity of a subset Hadamard in both
directions. Validity of $\Hd{a}$ forces $a\pw\CX\subseteq\CZ$ and $a\pw\CZ\subseteq\CX$,
that is $a\in\Pup\cap\Pdn$. Conversely, \eqref{eq:trans-hadvalid} follows from
$a\in\Pup\cap\Pdn$. We may therefore record the criterion in the form it is
usually wanted:
\begin{equation}
  \Hd{a}\in\Ntr\qquad\Longleftrightarrow\qquad a\in\Pup\cap\Pdn ,
  \label{eq:trans-hadcrit}
\end{equation}
so that Proposition~\ref{prop:twoblock} says the block length drops from three to two
exactly when the code admits no valid subset Hadamard other than the identity.

It is worth recording alongside \eqref{eq:trans-hadcrit} what such a Hadamard costs. It
is not a fourth kind of generator: by Eq.~\eqref{eq:trans-swapword}, applied at each
coordinate of $\suppo(a)$ and with all three factors equal to the identity off that
support,
\begin{equation}
  \Hd{a}=\Ush{a}\,\Lsh{a}\,\Ush{a},
  \label{eq:trans-had}
\end{equation}
so a valid subset Hadamard is already a three-block group word in the two diagonal families, and
by Eq.~\eqref{eq:trans-hadcrit} its validity is precisely the pair of conditions that make
$\Ush{a}$ and $\Lsh{a}$ valid. This is the mechanism behind
Proposition~\ref{prop:twoblock}: the $X\leftrightarrow Z$ mixing that a subset Hadamard
supplies is not a separate ingredient. It is what the third block is for.

Structurally, a valid subset Hadamard says more than this: a valid partial
subset Hadamard cannot occur unless the code splits along its support
(cf.~\cite{guyot2026addressability}).

\begin{proposition}[A partial subset Hadamard splits the code]
\label{prop:trans-split}
Let $t\in\Pup\cap\Pdn$, and write $R=\suppo(t)$ with $R^{c}$ for the complementary
set of qubits. Then both classical codes split along the partition $R\sqcup R^{c}$,
\begin{equation}
\begin{aligned}
  \CX&=(t\pw\CX)\oplus\bigl((\ones+t)\pw\CX\bigr),\\
  \CZ&=(t\pw\CZ)\oplus\bigl((\ones+t)\pw\CZ\bigr),
\end{aligned}
  \label{eq:trans-codesplit}
\end{equation}
the first summand of each supported on $R$ and the second on $R^{c}$. The label
space is therefore a direct sum $\lab=\lab_{R}\oplus\lab_{R^{c}}$ of parts
supported on $R$ and on $R^{c}$, so the stabilizer code is the tensor product of
its restrictions to the two sets of qubits. The partition is nontrivial exactly
when $0\neq t\neq\ones$.
\end{proposition}

\begin{proof}
Applying Lemma~\ref{lem:product} with $a=b=t$ and using $t\pw t=t$ gives
$t\pw\CX\subseteq\CX\cap\CZ$ and $t\pw\CZ\subseteq\CX\cap\CZ$. Write $P$ and $Q$
for multiplication by $t$ and by $\ones+t$. The maps $P$ and $Q$ are complementary
coordinate projectors, with images supported on $R$ and on $R^{c}$. Each of $P$ and $Q$
is idempotent, their product is zero, and $P+Q=\Id$. The inclusion just derived gives
$P\CX\subseteq\CX$, and hence $Q\CX=(\Id+P)\CX\subseteq\CX$ as well, since for
$x\in\CX$ both $x$ and $Px$ lie in $\CX$. So $P$ and $Q$ each carry $\CX$ into itself.
Because they also sum to the identity, they restrict to complementary projectors
\emph{on} $\CX$: every $x\in\CX$ is $Px+Qx$ with both terms in $\CX$. Moreover
$P\CX\cap Q\CX=0$,
because $P$ fixes $P\CX$ and annihilates $Q\CX$. Hence $\CX=P\CX\oplus Q\CX$, the
first line of Eq.~\eqref{eq:trans-codesplit}, and the same argument gives the second.
The two summands of $\lab$ are then supported on the disjoint qubit sets $R$ and
$R^{c}$, on which the symplectic form \eqref{eq:setup-form} pairs nothing. The sum is
therefore orthogonal, and the code is the claimed tensor product.
\end{proof}

The proposition turns the criterion \eqref{eq:trans-hadcrit} into a statement
about the code itself. Call a CSS code \emph{connected}, equivalently
\emph{indecomposable}, if no nontrivial partition $R\sqcup R^{c}$ of the qubits
splits it as $\lab=\lab_{R}\oplus\lab_{R^{c}}$ with the summands supported on $R$
and $R^{c}$. This definition is the generating-set-independent form of Tanner-graph
connectedness.

\begin{corollary}
\label{cor:trans-connected}
On a connected CSS code $\Pup\cap\Pdn\subseteq\{0,\ones\}$, and $\ones$ belongs to
it if and only if $\CX=\CZ$. By Eq.~\eqref{eq:trans-hadcrit} the only valid
permutation-free subset Hadamards are then the identity and, in the self-dual case
$\CX=\CZ$, the global Hadamard $H^{\otimes n}$.
\end{corollary}

\begin{proof}
A vector $t\in\Pup\cap\Pdn$ with $0\neq t\neq\ones$ would split the code by
Proposition~\ref{prop:trans-split}, against connectedness, so
$\Pup\cap\Pdn\subseteq\{0,\ones\}$. Here $\ones\in\Pup$ reads
$\CX=\ones\pw\CX\subseteq\CZ$ and $\ones\in\Pdn$ reads $\CZ\subseteq\CX$, so
$\ones\in\Pup\cap\Pdn$ if and only if $\CX=\CZ$; the last sentence is
\eqref{eq:trans-hadcrit} read on $a=0$ and $a=\ones$.
\end{proof}

The self-dual case genuinely occurs: the connected $\qcode{4,2,2}$ code with
$\CX=\CZ=\spano\{\ones\}$ carries the permutation-free duality $H^{\otimes4}$. By
Proposition~\ref{prop:twoblock} the transversal group of that code needs three blocks
and not two.

The results above concern the permutation-free group $\Ntr$.
Appendix~\ref{app:permutations} extends \eqref{eq:trans-nfset} to the group $\Gaut$
obtained by adjoining qubit permutations. It also compares the resulting normal form
with the set product of $\PDgp$ with the two diagonal families.

%% file: app_permutations.tex
\section{Adding qubit permutations}
\label{app:permutations}

Every gate up to this point treats the $n$ physical qubits as labeled and
immovable. Such a gate may act on the qubits one at a time, or two at a time
along a fixed matching, but it may never move a qubit to a different position.
For a fixed hardware layout this is the right convention. It is not, however,
the convention under which codes are catalogued. A code-automorphism routine
reports the larger group. It consists of the operations that act by a
single-qubit Clifford on each qubit and then relabel the qubits. The standard
databases store the order of that group as the
automorphism-group size of a stabilizer code. This appendix determines that
group completely for a CSS code. It shows three things:
\begin{enumerate}
\item The group is an exact three-fold product of two diagonal families and one
  family of Hadamard-plus-permutation elements.
\item The factorization of a given element is unique, and a closed formula
  reads it off the local blocks of that element.
\item One ingredient of the group is the only irreducible cost in computing
  it. We identify that ingredient.
\end{enumerate}

We write
\begin{equation}
  \Gaut \;=\; \bigl(\Sp{2}\wr\Perm{n}\bigr)\cap\Stb(\lab),
  \label{eq:perm-gaut}
\end{equation}
which is the same as $(\Sp{2}\wr\Perm{n})\cap\Nfull$, and we adopt throughout the
convention that the pair $(g;\pi)$ with $g\in\Sp{2}^{n}$ and
$\pi\in\Perm{n}$ denotes the matrix product $\pperm{\pi}\,g$. The matrix
$\pperm{\pi}$ acts in the same way on the two halves of the polarization. We
therefore also write $u\pperm{\pi}$ for the induced permutation of the
coordinates of a vector $u\in\Ftwo^{n}$. Recall that our vectors are rows, and
that the action is defined from the right.
Entrywise, $(\pperm{\pi})_{i,\pi(i)}=1$ and all other entries vanish, so
that the label $u\pperm{\pi}$ carries the coordinate $u_{i}$ to position
$\pi(i)$.
The group $\Gaut$ is the physical automorphism group of the code modulo Pauli
operators. This group consists of the single-qubit Cliffords and qubit
permutations that preserve the code, and it is a subgroup of the full
code-preserving group $\Nfull$. The kernel of $\Gaut$ on the permutation side,
\begin{equation}
  \Ntr \;=\; \Gaut\cap\Sp{2}^{n},
  \label{eq:perm-kernel}
\end{equation}
is the transversal group already met in Appendix~\ref{app:transversal}; in the notation of
Appendix~\ref{app:depthone} it is the two-fold transversal slice $\NM$ at the empty matching, all of
whose cells are singletons. Thus $\Gaut$ is exactly what one obtains by
adjoining qubit permutations to the transversal theory. The question of this
appendix is how much the adjunction really adds.

A local tuple $g\in\Sp{2}^{n}$ is a block matrix
\begin{equation}
  g=\begin{pmatrix} A & B\\ C& D\end{pmatrix},
  \qquad A\pw D + B\pw C=\ones ,
  \label{eq:perm-localsp}
\end{equation}
in which all four blocks are diagonal; we identify each block with its diagonal in
$\Ftwo^{n}$, so that all products of blocks appearing below are pointwise. The
right-hand condition in Eq.~\eqref{eq:perm-localsp} is pointwise symplecticity. Because
rows act on the right, $g\in\Stb(\lab)$ is equivalent to the four inclusions
\begin{equation}
\begin{aligned}
  A\pw\CX&\subseteq\CX, &\qquad B\pw\CX&\subseteq\CZ,\\
  C\pw\CZ&\subseteq\CX, &\qquad D\pw\CZ&\subseteq\CZ .
\end{aligned}
  \label{eq:perm-valid}
\end{equation}
Each of these conditions is linear in the diagonal it constrains, so it is enough to
impose it on a generating set of the CSS code. For each such generator $u$ the
map $v\mapsto v\pw u$ is linear. Membership of the result in the target code is
the vanishing of a fixed linear functional on that result. Consequently the
parameter codes
\begin{equation}
\begin{aligned}
  \Pup&=\{a\in\Ftwo^{n} : a\pw\CX\subseteq\CZ\},\\
  \Pdn&=\{a\in\Ftwo^{n} : a\pw\CZ\subseteq\CX\}
\end{aligned}
  \label{eq:perm-params}
\end{equation}
of Appendix~\ref{app:transversal}, whose associated gate families are
$\Eupo=\{\Ush{a}:a\in\Pup\}$ and $\Edno=\{\Lsh{a}:a\in\Pdn\}$, are binary linear
spaces obtainable by Gaussian elimination, and so is the space of parameters of
valid subset Hadamards.

The new ingredient is the family of \emph{partial dualities}
\begin{equation}
  \PDgp=\bigl\{(\Hd{a};\pi)\in\Gaut \;:\; a\in\Ftwo^{n},\ \pi\in\Perm{n}\bigr\},
  \label{eq:perm-pd}
\end{equation}
the elements of $\Gaut$ whose local part consists of Hadamards only. This is a
subgroup. Indeed $\Hd{a}\Hd{a'}=\Hd{a+a'}$, and conjugating a subset Hadamard by a
qubit permutation only moves its support,
$\pperm{\tau}^{-1}\Hd{a}\pperm{\tau}=\Hd{a\pperm{\tau}}$, so that
\begin{equation}
  \bigl(\pperm{\sigma}\Hd{a}\bigr)\bigl(\pperm{\tau}\Hd{a'}\bigr)
  \;=\;\pperm{\sigma}\pperm{\tau}\;\Hd{a\pperm{\tau}+a'},
  \label{eq:perm-weylcomp}
\end{equation}
whose right-hand side is again a permutation matrix times a subset Hadamard, and
similarly $(\pperm{\sigma}\Hd{a})^{-1}=\pperm{\sigma}^{-1}\Hd{a\pperm{\sigma}^{-1}}$.
The family $\PDgp$ has two familiar endpoints. When $a=0$ the element is a pure
permutation preserving $\CX$ and $\CZ$ separately, the Levi-type case. When
$a=\ones$ the element is a full $\Zh\leftrightarrow\Xh$ duality, a permutation
that carries $\CX$ onto $\CZ$ and $\CZ$ onto $\CX$. Neither endpoint family
contains the other in general. The theorem below shows that the intermediate
values of $a$ are exactly what the two diagonal families need to close up into the
whole group.

Finally, write $\rhom(\Gaut)\le\Perm{n}$ for the group of qubit permutations
that occur in $\Gaut$. This group is the image of $\Gaut$ under the homomorphism
$\Sp{2}\wr\Perm{n}\to\Perm{n}$ that forgets the single-qubit gate on each qubit
and keeps only the permutation. The kernel of its restriction to $\Gaut$ is
$\Ntr$, and Section~\ref{app:perm-exact} develops the resulting exact sequence.

\begin{theorem}[Normal form for the automorphism group]
\label{thm:autnf}
Let $\lab=(\CX|0)\oplus(0|\CZ)$ be a label space presented in CSS form. Then
\begin{equation}
  \Gaut=\PDgp\cdot\Eupo\cdot\Edno
  \label{eq:perm-setproduct}
\end{equation}
as a product of sets. Every element of $\Gaut$ factors as
\begin{equation}
\begin{gathered}
  (g;\pi)=(\Hd{a};\pi)\,\Ush{q}\,\Lsh{p},\\
  a=B\pw C,\quad q=B\pw D,\quad p=A\pw C,
\end{gathered}
  \label{eq:perm-nf}
\end{equation}
where $A,B,C,D$ are the local blocks of $g$ as in Eq.~\eqref{eq:perm-localsp}, and this
factorization is unique once the two diagonal-circuit supports are required to be disjoint,
$q\pw p=0$. For generation it suffices to adjoin to $\langle\Eupo,\Edno\rangle$ one
partial duality lifting each element of a generating set of $\rhom(\Gaut)$.
\end{theorem}

The only structural hypothesis is the displayed CSS presentation of $\lab$. No
distance assumption is used, neither $\CX$ nor $\CZ$ is required to be nonzero,
and $k=0$ is allowed. Orthogonality $\CX\perp\CZ$ is of course needed for $\lab$
to be the label space of a stabilizer code, but no step of the argument invokes
it. What cannot be dropped is that the presentation is really CSS. Consider a
code that is merely local-Clifford equivalent to a CSS code, and that is stored
in a frame carrying $Y$-type generators. That code must first be rotated into
CSS form. Otherwise the blocks in Eq.~\eqref{eq:perm-nf} are not the blocks of
\eqref{eq:perm-localsp}.

The proof occupies the next three subsections. The permutation part is separated
off first, by an exact sequence that is pure wreath-product
formalism~\cite{dixon1996permutation}. The whole mathematical content then sits
in a single lemma about local tuples, which is the width-one case of the
automatic-validity lemma of Appendix~\ref{app:depthone}.

\subsection{The permutation exact sequence}
\label{app:perm-exact}

The base group $\Sp{2}^{n}$ is normal in $\Sp{2}\wr\Perm{n}$. The projection
$\rhom$ that forgets the local gates and keeps the permutation is a homomorphism
onto its image. Intersecting with $\Stb(\lab)$ therefore yields
\begin{equation}
\begin{gathered}
  1\longrightarrow \Ntr\longrightarrow \Gaut
   \stackrel{\rhom}{\longrightarrow}\rhom(\Gaut)\longrightarrow 1,\\
  |\Gaut|=|\Ntr|\cdot|\rhom(\Gaut)| .
\end{gathered}
  \label{eq:perm-exact}
\end{equation}
The image is characterized without reference to the local gates at all:
\begin{multline}
  \rhom(\Gaut)=\bigl\{\pi\in\Perm{n} :
     \lab\pperm{\pi}\ \text{is local-Clifford}\\
     \text{equivalent to}\ \lab\bigr\},
  \label{eq:perm-rhoimage}
\end{multline}
because a lift of $\pi$ is exactly a local tuple $g$ with
$\lab\pperm{\pi}g=\lab$. Thus $\rhom(\Gaut)$ is the group of permutations that
preserve the code up to single-qubit Cliffords. It is called the
\emph{Clifford-twisted automorphism group} of the
code~\cite{hao2021investigations}. The group $\Ntr$ consists of the
single-qubit Cliffords that preserve the code outright. The two factors on the
right of Eq.~\eqref{eq:perm-exact} are of quite different computational character, a
point we return to at the end of Sec.~\ref{app:perm-nf}.

It is worth recording what the lifting problem looks like for one fixed permutation,
before we prove that it is always solvable by Hadamards alone. Let $\pi$ be given and
ask for which subsets, that is for which indicator vectors $a\in\Ftwo^{n}$, the element
$(\Hd{a};\pi)=\pperm{\pi}\Hd{a}$ lies in $\Gaut$. Applying $\pperm{\pi}\Hd{a}$ to
$(u|0)$ with $u\in\CX$ gives $(u\pperm{\pi}\pw(\ones+a)\,|\,u\pperm{\pi}\pw a)$, and
applying it to $(0|w)$ with $w\in\CZ$ gives
$(w\pperm{\pi}\pw a\,|\,w\pperm{\pi}\pw(\ones+a))$. Hence $(\Hd{a};\pi)\in\Gaut$ if and
only if, for all $u\in\CX$ and all $w\in\CZ$,
\begin{subequations}
\label{eq:perm-affine}
\begin{align}
  u\pperm{\pi}\pw(\ones+a)&\in\CX, &
  u\pperm{\pi}\pw a&\in\CZ,
  \label{eq:perm-affine-x}\\
  w\pperm{\pi}\pw(\ones+a)&\in\CZ, &
  w\pperm{\pi}\pw a&\in\CX .
  \label{eq:perm-affine-z}
\end{align}
\end{subequations}
For fixed $\pi$ this is an affine-linear system in $a$: the conditions involving $a$
alone are linear, and those involving $\ones+a$ are linear up to the constant terms
$u\pperm{\pi}$ and $w\pperm{\pi}$. It is solved by Gaussian elimination once $\pi$ is
known.  

At the two endpoint values noted above the system reads as expected.
The case $a=0$ demands $\CX\pperm{\pi}\subseteq\CX$ and $\CZ\pperm{\pi}\subseteq\CZ$, a permutation
symmetry of the two classical codes separately --- that is, $\pi\in\mathrm{Aut}(\CX)$ and
$\pi\in\mathrm{Aut}(\CZ)$.
The other extreme, $a=\ones$, demands
$\CX\pperm{\pi}\subseteq\CZ$ and $\CZ\pperm{\pi}\subseteq\CX$, a full duality.  

Nothing so far guarantees that \eqref{eq:perm-affine} is consistent for a given
$\pi\in\rhom(\Gaut)$. That consistency is the content of Lemma~\ref{lem:pureH}
below. Consistency is not obvious, since a priori a permutation could be
realizable only with the help of diagonal circuits or three-cycles in its local dressing.

\subsection{The master lemma}
\label{app:perm-master}

The whole theorem rests on the following statement about a single local tuple.
It says that a local tuple carrying one CSS-form space onto another is, up to
diagonal circuits that are automatically valid, nothing but a subset Hadamard. It also says
that the subset in question is computed by one pointwise product of blocks.

\begin{lemma}[Master lemma at width one]
\label{lem:permmaster}
Let $\lab_{1}$ and $\lab_{0}$ be subspaces of $\Vsp$ that are both in CSS form, and let
$g\in\Sp{2}^{n}$ be a local tuple with $\lab_{1}g=\lab_{0}$. Write $A,B,C,D$ for the
blocks of $g$ as in Eq.~\eqref{eq:perm-localsp} and set
\begin{equation}
  a=B\pw C,\qquad q=B\pw D,\qquad p=A\pw C .
  \label{eq:perm-abc}
\end{equation}
Then $q\pw p=0$, the identity
\begin{equation}
  g=\Hd{a}\,\Ush{q}\,\Lsh{p}=\Hd{a}\,\Lsh{p}\,\Ush{q}
  \label{eq:perm-master-fact}
\end{equation}
holds in $\Sp{2}^{n}$, and $\Ush{q}\Lsh{p}\in\Stb(\lab_{0})$; writing
$\lab_{0}=(\CX|0)\oplus(0|\CZ)$, the two diagonal-circuit parameters obey
\begin{equation}
  q\pw\CX\subseteq\CZ,\qquad p\pw\CZ\subseteq\CX,
  \label{eq:perm-shearvalid}
\end{equation}
that is $q\in\Pup$ and $p\in\Pdn$. Consequently $\lab_{1}\Hd{a}=\lab_{1}g=\lab_{0}$.
\end{lemma}

\begin{proof}
We first establish \eqref{eq:perm-shearvalid}, which is where the CSS hypothesis on
$\lab_{1}$ enters. A subspace of $\Vsp$ is in CSS form precisely when it is preserved by
the projector $\PX=\left(\begin{smallmatrix}\ones&0\\0&0\end{smallmatrix}\right)$ onto
the $\Xh$-half. Suppose a subspace is carried into itself by $\PX$. Then that
subspace is the direct sum of its images under $\PX$ and under the complementary
projector $\left(\begin{smallmatrix}0&0\\0&\ones\end{smallmatrix}\right)$ onto
the $\Zh$-half. Those images are the pure-$\Xh$ and pure-$\Zh$ parts of that
subspace. Conversely a CSS-form subspace is $\PX$-invariant, since $\PX$ sends
$(c|d)$ to $(c|0)$, which lies in it. Consider the conjugated idempotent
\begin{equation}
  e:=g^{-1}\PX\,g .
  \label{eq:perm-idem}
\end{equation}
It preserves $\lab_{0}$: using $\lab_{0}g^{-1}=\lab_{1}$ and then the CSS form of
$\lab_{1}$,
\begin{equation}
  \lab_{0}\,e=\lab_{1}\PX g\subseteq\lab_{1}g=\lab_{0}.
  \label{eq:perm-idem-pres}
\end{equation}
Over $\Ftwo$ the pointwise symplecticity condition in Eq.~\eqref{eq:perm-localsp} gives the
inverse in closed form,
\begin{equation}
  g^{-1}=\begin{pmatrix} D& B\\ C& A\end{pmatrix},
  \label{eq:perm-ginv}
\end{equation}
as one checks by multiplying out and using $A\pw D+B\pw C=\ones$ together with
commutativity of the diagonal blocks. Substituting \eqref{eq:perm-ginv} into
\eqref{eq:perm-idem} and using $\PX=\left(\begin{smallmatrix}\ones&0\\0&0\end{smallmatrix}\right)$,
\begin{equation}
\begin{aligned}
  e&=\begin{pmatrix} D&0\\ C&0\end{pmatrix}
      \begin{pmatrix} A&B\\ C&D\end{pmatrix}\\[2pt]
   &=\begin{pmatrix} A\pw D & B\pw D\\ A\pw C & B\pw C\end{pmatrix}
    =\begin{pmatrix} \ones+a & q\\ p& a\end{pmatrix}.
\end{aligned}
  \label{eq:perm-e}
\end{equation}
Now read off \eqref{eq:perm-idem-pres} on the two halves of $\lab_{0}$. For $u\in\CX$
one has $(u|0)e=(u\pw(\ones+a)\,|\,u\pw q)\in\lab_{0}$, whence $u\pw q\in\CZ$; for
$w\in\CZ$ one has $(0|w)e=(w\pw p\,|\,w\pw a)\in\lab_{0}$, whence $w\pw p\in\CX$. This
is \eqref{eq:perm-shearvalid}.

Next we prove the purely matrix-theoretic identity \eqref{eq:perm-master-fact}.
Every matrix appearing in it is diagonal, so the assertion is coordinatewise. At
a single coordinate $g$ is one of the six elements of $\Sp{2}$. Tabulating them
settles both claims at once:
\begin{center}
\begin{tabular}{ccccl}
\hline\hline
$g$ & $a$ & $q$ & $p$ & $\Hd{a}\Ush{q}\Lsh{p}$\\
\hline
$\left(\begin{smallmatrix}1&0\\0&1\end{smallmatrix}\right)$ & $0$ & $0$ & $0$ & $\Id$\\
$\left(\begin{smallmatrix}1&1\\0&1\end{smallmatrix}\right)$ & $0$ & $1$ & $0$ & $\Ush{1}$\\
$\left(\begin{smallmatrix}1&0\\1&1\end{smallmatrix}\right)$ & $0$ & $0$ & $1$ & $\Lsh{1}$\\
$\left(\begin{smallmatrix}0&1\\1&0\end{smallmatrix}\right)$ & $1$ & $0$ & $0$ & $\Hd{1}$\\
$\left(\begin{smallmatrix}0&1\\1&1\end{smallmatrix}\right)$ & $1$ & $1$ & $0$ & $\Hd{1}\Ush{1}$\\
$\left(\begin{smallmatrix}1&1\\1&0\end{smallmatrix}\right)$ & $1$ & $0$ & $1$ & $\Hd{1}\Lsh{1}$\\
\hline\hline
\end{tabular}
\end{center}
In each row $a=B\pw C$, $q=B\pw D$ and $p=A\pw C$ are read off the blocks of
$g$, and the last column reproduces $g$. Moreover, in each row $q$ and $p$ are
never both $1$, so $q\pw p=0$. The six rows exhaust $\Sp{2}$, so both statements
hold at every coordinate and therefore globally.
Disjointness of the supports of $q$ and $p$ makes the two diagonal circuits commute,
$\Ush{q}\Lsh{p}=\Lsh{p}\Ush{q}$, which is the second equality in
\eqref{eq:perm-master-fact}.

Finally, a $Z$-diagonal circuit whose parameter is the diagonal matrix $\diago q$ sends $(u|w)$ to
$(u|w+u\pw q)$, and an $X$-diagonal circuit with parameter $p$ sends $(u|w)$ to
$(u+w\pw p|w)$. So \eqref{eq:perm-shearvalid} says exactly that $\Ush{q}$ and
$\Lsh{p}$ preserve $\lab_{0}$. Hence $\Ush{q}\Lsh{p}\in\Stb(\lab_{0})$, and
applying \eqref{eq:perm-master-fact} to
$\lab_{1}$,
\begin{equation}
  \lab_{0}=\lab_{1}g=\lab_{1}\Hd{a}\,\Ush{q}\Lsh{p},
  \label{eq:perm-lastline}
\end{equation}
so that $\lab_{1}\Hd{a}=\lab_{0}\bigl(\Ush{q}\Lsh{p}\bigr)^{-1}=\lab_{0}$.
\end{proof}

It is worth separating the two halves of this proof, because they have different scopes.
The factorization \eqref{eq:perm-master-fact} is an identity between matrices: it holds
for every local tuple whatsoever, with no hypothesis on any subspace. The CSS forms of
$\lab_{1}$ and $\lab_{0}$ give only the validity statement \eqref{eq:perm-shearvalid}.
That statement promotes the two diagonal circuits from arbitrary elements of
$\Sp{2}^{n}$ to gates of the code.

The validity half of Lemma~\ref{lem:permmaster} is the width-one specialization
of the automatic-validity lemma, Lemma~\ref{lem:auto} of
Appendix~\ref{app:depthone}. That specialization is the same projector
computation, performed with $\PX$ conjugated by the gate. The diagonal circuits that
computation manufactures are automatically valid for the same reason. The
factorization \eqref{eq:perm-master-fact} is not a specialization of that lemma
but an additional statement, proved coordinatewise below.
At general width the lemma produces the two parameter spaces $\shp{G}$ and $\shm{G}$ of
\eqref{eq:d1-spans}, namely $\shp{G}=\spano\{AB^{\tp},D^{\tp}B,B+B^{\tp}\}$ and
$\shm{G}=\spano\{CD^{\tp},C^{\tp}A,C+C^{\tp}\}$, six matrices in all. At width one the
blocks are diagonal, hence symmetric, so $B+B^{\tp}$ and $C+C^{\tp}$ vanish identically
and the six collapse to the four pointwise products
\begin{equation}
\begin{aligned}
  A\pw B,\ \ D\pw B &\qquad\text{(upper)},\\
  C\pw D,\ \ C\pw A &\qquad\text{(lower)}.
\end{aligned}
  \label{eq:perm-fourparams}
\end{equation}
The parameters $q=B\pw D$ and $p=A\pw C$ used above are two of these four. The
remaining two, $A\pw B$ and $C\pw D$, are what the same computation manufactures
when the roles of the two spaces are interchanged. The interchange puts
$g\PX g^{-1}$ in place of $g^{-1}\PX g$. The resulting diagonal circuits are valid for the
source space $\lab_{1}$ rather than for the target $\lab_{0}$. They are not
needed here.

Two features of the lemma deserve emphasis, because they are what makes the resulting
normal form canonical. First, the Hadamard support produced by the lemma is
$\suppo(B\pw C)$ and \emph{not} the complement of $\suppo(A)$. In the table above,
$a=B\pw C$ is nonzero on exactly the last three rows. $A$ vanishes on only the first
two of those three rows. The canonical Weyl support therefore contains the naive one,
and is strictly larger as soon as some coordinate carries the last row $\Hd{1}\Lsh{1}$.
Second, the last two rows exhibit the two three-cycles of $\Sp{2}$ as one Hadamard and
one diagonal circuit apiece. This is why no three-cycle ever survives into the normal form.
This is also why the local dressing of a permutation can always be taken to
consist of Hadamards and diagonal circuits alone.

\subsection{The normal form and its uniqueness}
\label{app:perm-nf}

We first take the two spaces of Lemma~\ref{lem:permmaster} to be equal, which resolves
the permutation-free part of $\Gaut$.

\begin{proposition}[Exact product and counting for the transversal group]
\label{prop:permkernel}
Let $\lab$ be in CSS form. Every $g\in\Ntr$ factors as
\begin{equation}
\begin{gathered}
  g=\Hd{a}\,\Ush{q}\,\Lsh{p},\\
  a=B\pw C,\quad q=B\pw D,\quad p=A\pw C,
\end{gathered}
  \label{eq:perm-kernelnf}
\end{equation}
with $\Hd{a}\in\Ntr$, $\Ush{q}\in\Eupo$, $\Lsh{p}\in\Edno$ and $q\pw p=0$. Hence
\begin{equation}
  \Ntr=\{\Hd{a}: a\in\Pup\cap\Pdn\}\cdot\Eupo\cdot\Edno
  \label{eq:perm-kernelproduct}
\end{equation}
as a product of sets. The factorization \eqref{eq:perm-kernelnf} is the unique one with
disjoint diagonal-circuit supports, and consequently
\begin{multline}
  |\Ntr|=\#\bigl\{(a,q,p)\in(\Pup\cap\Pdn)\times\Pup\times\Pdn\\
  \text{such that } q\pw p=0\bigr\}.
  \label{eq:perm-count}
\end{multline}
\end{proposition}

\begin{proof}
Apply Lemma~\ref{lem:permmaster} with $\lab_{1}=\lab_{0}=\lab$, which is legitimate
because $g\in\Ntr$ means precisely $\lab g=\lab$. It gives
\eqref{eq:perm-kernelnf} together with $q\in\Pup$, $p\in\Pdn$ and $q\pw p=0$; and since
$\Ush{q}\Lsh{p}\in\Ntr$, also
$\Hd{a}=g\bigl(\Ush{q}\Lsh{p}\bigr)^{-1}\in\Ntr$. A valid subset Hadamard has blocks
$A=D=\ones+a$ and $B=C=a$, so the second and third conditions of Eq.~\eqref{eq:perm-valid}
read $a\pw\CX\subseteq\CZ$ and $a\pw\CZ\subseteq\CX$, that is $a\in\Pup\cap\Pdn$.
Conversely every $a\in\Pup\cap\Pdn$ gives a valid $\Hd{a}$; this is the identification of
the valid subset-Hadamard family established in Appendix~\ref{app:transversal}, where it
is obtained from the pointwise product lemma. The two inclusions together
identify the first factor in Eq.~\eqref{eq:perm-kernelproduct} with
$\{\Hd{a}:a\in\Pup\cap\Pdn\}$. Elementwise validity of the three factors
upgrades \eqref{eq:perm-kernelnf} from a statement about individual elements to
the set equality \eqref{eq:perm-kernelproduct}.

For uniqueness we multiply a general group word out. Let $a,s,v\in\Ftwo^{n}$ be arbitrary,
with no disjointness assumed, and let $g=\Hd{a}\Ush{s}\Lsh{v}$. Working coordinatewise
one finds for the blocks of $g$
\begin{equation}
\begin{gathered}
  B\pw C=a+s\pw v,\qquad B\pw D=s,\\
  A\pw C=v\pw(\ones+s).
\end{gathered}
  \label{eq:perm-multout}
\end{equation}
Thus two distinct group words can represent the same element only where the two diagonal-circuit supports
overlap; the smallest instance is the coordinatewise coincidence
$\Ush{1}\Lsh{1}=\Hd{1}\Ush{1}$, both sides being the three-cycle $(0,1,1,1)$, which is
\eqref{eq:perm-multout} with $(a,s,v)=(0,1,1)$ and $(a,s,v)=(1,1,0)$ giving the same
triple of blocks. Imposing $s\pw v=0$ removes the slack: \eqref{eq:perm-multout} then
returns $(a,s,v)=(B\pw C,\,B\pw D,\,A\pw C)$, so the group word is recovered from the element
it represents. Hence the assignment $g\mapsto(a,q,p)$ of
\eqref{eq:perm-kernelnf} is a bijection from $\Ntr$ onto the set of triples
counted in Eq.~\eqref{eq:perm-count}. This bijection proves both the uniqueness
claim and the counting formula.
\end{proof}

It is instructive to evaluate \eqref{eq:perm-count} on the Steane code, for which
$\CX=\CZ$ is the $[7,3,4]$ simplex code all of whose nonzero words have weight four. For
$a\in\Pup=\Pdn$ and a nonzero $u\in\CX$ the vector $a\pw u$ lies in $\CX$ and is
supported inside $\suppo(u)$, hence is either $0$ or $u$. So for such $a$ and
$u$ the support $\suppo(a)$ either contains $\suppo(u)$ or is disjoint from it.
Any two distinct weight-four words meet in exactly two coordinates, so the two
alternatives cannot both occur for different $u$. Since the weight-four words
cover all seven coordinates we are left with $a\in\{0,\ones\}$ --- as
Corollary~\ref{cor:trans-connected} gives at once, the simplex code being
connected.
All three spaces $\Pup$, $\Pdn$ and $\Pup\cap\Pdn$ therefore equal $\{0,\ones\}$, and
\eqref{eq:perm-count} counts $2\cdot3=6$ triples, the factor three excluding only
$q=p=\ones$. This is consistent with $\Ntr=\langle\Eupo,\Edno\rangle$ being the diagonal
copy of $\Sp{2}$, of order six.

The permutation part now follows from the same lemma, applied with the two spaces
distinct.

\begin{lemma}[Pure-Hadamard lift]
\label{lem:pureH}
Let $(g;\pi)\in\Gaut$ and let $a=B\pw C$ be read off the blocks of $g$. Then
$(\Hd{a};\pi)\in\PDgp$. In particular every $\pi\in\rhom(\Gaut)$ admits a lift whose
local part consists of Hadamards only, and the system \eqref{eq:perm-affine} is
consistent for every such $\pi$.
\end{lemma}

\begin{proof}
By definition $(g;\pi)=\pperm{\pi}g$ stabilizes $\lab$, that is
$\lab\pperm{\pi}g=\lab$. Put $\lab_{1}:=\lab\pperm{\pi}$ and $\lab_{0}:=\lab$. A qubit
permutation acts on the two halves of the polarization separately, so
$\lab_{1}=(\CX\pperm{\pi}|0)\oplus(0|\CZ\pperm{\pi})$ is again in CSS form, and by
construction $\lab_{1}g=\lab_{0}$. Lemma~\ref{lem:permmaster} therefore applies and
yields $\lab\pperm{\pi}\Hd{a}=\lab_{1}\Hd{a}=\lab_{0}=\lab$, which says exactly that
$\pperm{\pi}\Hd{a}=(\Hd{a};\pi)$ stabilizes $\lab$.
\end{proof}

Note that the diagonal-circuit parameters produced alongside the lift are validated against
$\lab_{0}=\lab$ and not against the permuted space. Indeed, the dressing
$\Ush{q}\Lsh{p}$ is a valid gate of the code itself with trivial permutation
part, so it lies in the kernel $\Ntr$ of $\rhom$. Its parameters $q$ and $p$ are
read off the local witness $g$ and so do depend on $\pi$ through that choice.
What does not depend on $\pi$ is the space against which those parameters are
validated. This is what allows the permutation to be compensated by the Hadamard
factor alone.

We now turn to proving the main theorem of this section.

\begin{proof}[Proof of Theorem~\ref{thm:autnf}]
Let $(g;\pi)\in\Gaut$ and let $a,q,p$ be as in Eq.~\eqref{eq:perm-abc}. Lemma~\ref{lem:pureH}
gives $(\Hd{a};\pi)\in\PDgp$, and Lemma~\ref{lem:permmaster}, applied as in the proof of
that lemma with $\lab_{1}=\lab\pperm{\pi}$ and $\lab_{0}=\lab$, gives $q\in\Pup$,
$p\in\Pdn$, $q\pw p=0$ and
$g=\Hd{a}\Ush{q}\Lsh{p}$. Multiplying on the left by $\pperm{\pi}$,
\begin{equation}
\begin{aligned}
  (g;\pi)&=\pperm{\pi}g=\pperm{\pi}\Hd{a}\,\Ush{q}\Lsh{p}\\
         &=(\Hd{a};\pi)\,\Ush{q}\,\Lsh{p},
\end{aligned}
  \label{eq:perm-thmproof}
\end{equation}
which is \eqref{eq:perm-nf}. As every factor on the right lies in $\Gaut$, the reverse
inclusion in Eq.~\eqref{eq:perm-setproduct} is trivial, so \eqref{eq:perm-setproduct} holds.
For uniqueness, suppose $(g;\pi)=(\Hd{a'};\pi')\Ush{q'}\Lsh{p'}$ with
$(\Hd{a'};\pi')\in\PDgp$ and $q'\pw p'=0$. The two diagonal circuits have trivial
permutation part, so comparing permutation parts in the wreath product forces
$\pi'=\pi$, hence $g=\Hd{a'}\Ush{q'}\Lsh{p'}$. Then \eqref{eq:perm-multout} with
$q'\pw p'=0$ gives $a'=B\pw C=a$, $q'=B\pw D=q$ and $p'=A\pw C=p$.

It remains to prove the generation statement. Let $\pi_{1},\dots,\pi_{m}$ generate
$\rhom(\Gaut)$ and choose, by Lemma~\ref{lem:pureH}, one partial duality
$w_{i}=(\Hd{a_{i}};\pi_{i})$ lifting each $\pi_{i}$. Set
$\Gaux=\langle\Eupo,\Edno,w_{1},\dots,w_{m}\rangle\subseteq\Gaut$. Its image under
$\rhom$ contains every $\pi_{i}$ and therefore equals $\rhom(\Gaut)$. Its kernel
contains $\langle\Eupo,\Edno\rangle$, which is all of $\Ntr$ by
Theorem~\ref{thm:transnf}. Hence
\begin{equation}
\begin{aligned}
  |\Gaux|&=|\Gaux\cap\Ntr|\cdot|\rhom(\Gaux)|\\
          &=|\Ntr|\cdot|\rhom(\Gaut)|=|\Gaut|,
\end{aligned}
  \label{eq:perm-gencount}
\end{equation}
using Eq.~\eqref{eq:perm-exact} twice, so $\Gaux=\Gaut$. This proves the last assertion of the theorem.
\end{proof}

The theorem also makes visible which ingredient of the automorphism group is expensive
to compute. The kernel $\Ntr$ is cheap: the two parameter codes $\Pup$ and
$\Pdn$ are nullspaces of the explicit linear systems \eqref{eq:perm-params}. The
lift datum $a(\pi)=B\pw C$ of any witness is a closed formula requiring no
search over local dressings. Everything expensive sits in $\rhom(\Gaut)$, whose
determination by Eq.~\eqref{eq:perm-rhoimage} is the computation of the code
automorphisms modulo local
Cliffords, a graph canonical-form problem. No factorization of $\Gaut$ can remove this,
since $\rhom(\Gaut)$ is a quotient of $\Gaut$ and must be produced by any description of
the group. What the theorem achieves is an almost complete separation. The one
exception is that the kernel $\Ntr$ may itself carry a subset Hadamard $\Hd{a}$
with $a\in\Pup\cap\Pdn$, a duality that no permutation accompanies. On a
connected code Corollary~\ref{cor:trans-connected} confines that duality to the
global $H^{\otimes n}$, and only when $\CX=\CZ$. Apart from that one self-dual
exception the entire duality content is one pure-Hadamard lift per generator of
$\rhom(\Gaut)$, and everything else is linear algebra. A pure-permutation lift,
the case $a(\pi)=0$, exists only for permutations
preserving $\CX$ and $\CZ$ separately. A genuinely dual permutation, for
example a fold along a $\Zh\leftrightarrow\Xh$ duality, requires the Hadamard.
The Hadamard gives exactly the content that the diagonal families cannot give.

%% file: app_twolocalaut.tex
\section{A two-fold automorphism witness}
\label{app:twolocalaut}

This appendix exhibits an explicit depth-one two-local layer that preserves the $\qcode{9,3,2:186339}$~\cite{cross2025small,crossqiskitqec2025} code only when compensated by a permutation, and whose logical action lies outside the group generated by $\Ndep$ and $\Gaut$.

The code's stabilizer is generated by
$X_{1}X_{6}X_{7}$, $X_{0}X_{1}X_{2}X_{3}$, $X_{0}X_{4}X_{5}X_{6}$,
$X_{0}X_{2}X_{5}X_{8}$, $Z_{1}Z_{2}Z_{4}Z_{5}Z_{7}$, and
$Z_{0}Z_{3}Z_{6}Z_{7}Z_{8}$.  Take
$\Mt=\{\{0,3\},\{2,7\},\{4,8\},\{5,6\}\}$, with singleton $1$, and
$\pi=(0\,1)(2\,3)(4\,7\,5)(6\,8)$.  In the coordinate order
$(X_i,X_j\mid Z_i,Z_j)$ for each listed pair, let $g\in\DM$ have the trivial
blocks $g_{\{0,3\}}=I_{4}$ and $g_{\{1\}}=I_{2}$ together with
\begin{equation}
g_{\{2,7\}}=
\begin{pmatrix}1&0&0&0\\1&1&0&0\\0&0&1&1\\0&0&0&1\end{pmatrix},\qquad
g_{\{4,8\}}=
\begin{pmatrix}0&1&0&0\\1&0&0&0\\0&0&0&1\\0&0&1&0\end{pmatrix},
\label{eq:ext-perm-witness}
\end{equation}
and
\begin{equation}
g_{\{5,6\}}=
\begin{pmatrix}1&1&0&0\\1&0&0&0\\0&0&0&1\\0&0&1&1\end{pmatrix}.
\label{eq:ext-perm-witness-b}
\end{equation}
Neither $\pperm{\pi}$ nor $g$ preserves the code, whereas their product
$h=(g;\pi)=\pperm{\pi}g$ does.  This can be certified without enumerating the
physical group~(a machine-checkable certificate is provided in the
repository of Ref.~\cite{albert2026certificate}): generate $\Ndep$ from all $9\cdot7!!=945$ maximum matchings and
use the logical basis
$(\bar X_1,\bar X_2,\bar X_3)=(X_4X_5,X_4X_6X_7,X_6X_8)$ and
$(\bar Z_1,\bar Z_2,\bar Z_3)=(Z_0Z_3Z_5,Z_1Z_3Z_7,Z_2Z_3Z_8)$.
An exact Schreier--Sims~\cite{sims1970computational} calculation gives
$|\lact(\langle\Ndep,\Gaut\rangle)|=23\,040$, while
\begin{equation}
\lact(h)=
\begin{pmatrix}
0&1&0&0&0&0\\
1&1&1&0&0&0\\
1&0&0&0&0&0\\
0&0&0&0&1&1\\
0&0&0&0&0&1\\
0&0&0&1&0&1
\end{pmatrix}.
\label{eq:ext-perm-witness-logical}
\end{equation}
This matrix satisfies $\lact(h)\notin\lact(\langle\Ndep,\Gaut\rangle)$, and
adjoining it fills the logical group,
$|\langle\lact(\Ndep),\lact(\Gaut),\lact(h)\rangle|=1\,451\,520=|\Sp{6}|$.
Thus the group in Eq.~\eqref{eq:ext-perm-group} is strictly larger than
$\langle\Ndep,\Gaut\rangle$ for this code.